\documentclass[12pt]{article}

\usepackage{epsfig,color,amsmath,amssymb,euscript}
 \usepackage{mathabx}
\usepackage{subfigure}
\usepackage{fancyhdr}
\usepackage{natbib}
\usepackage{rotating}
\usepackage{amsmath}
\usepackage{amsfonts} 
\usepackage{setspace}
\usepackage{fancyhdr}

\usepackage{amsthm}
\usepackage{enumerate} 
\usepackage{threeparttable}
\usepackage{diagbox}
\usepackage{multirow}
\usepackage{makecell}
\usepackage{xcolor}
\usepackage{colortbl}
\usepackage[colorlinks,linkcolor = blue,anchorcolor=blue, citecolor= blue]{hyperref}
\usepackage{url} 
 
\usepackage{graphicx}
\usepackage{enumerate}
\usepackage{mathrsfs}
\usepackage{verbatim}
\usepackage{blkarray}
\usepackage{pdflscape}

\usepackage{authblk}
\usepackage{setspace}
\usepackage{threeparttable}
\usepackage{caption}
\usepackage{algorithm,algcompatible}

\usepackage{tabularx}
\usepackage{array}
\usepackage{makecell}

\usepackage{accents}
\newcommand{\ubar}[1]{\underaccent{\bar}{#1}}

\usepackage{bm}
\usepackage{url}
\usepackage{booktabs}
\usepackage{algorithm}
\usepackage{algpseudocode}
\algrenewcommand\algorithmicrequire{\textbf{Input:}}
\algrenewcommand\algorithmicensure{\textbf{Output:}}

\usepackage{rotating}[figuresright]
\usepackage{xcolor}
\definecolor{darkgreen}{rgb}{0,0.5,0}

\usepackage[resetlabels]{multibib}
\allowdisplaybreaks
\newcites{S}{References} 

\algnewcommand\INPUT{\item[\textbf{Input:}]}%
\algnewcommand\OUTPUT{\item[\textbf{Output:}]}%

\theoremstyle{plain}

\newtheorem{theorem}{Theorem}
\newtheorem{lemma}{Lemma} 

\theoremstyle{plain}
\newtheorem{assumption}{Assumption}
\newtheorem{remark}{Remark}

\newcommand{\ab}{\mathbf{a}}
\newcommand{\bbb}{\mathbf{b}}

\newcommand{\db}{\mathbf{d}}
\newcommand{\eb}{\mathbf{e}}
\newcommand{\fb}{\mathbf{f}}

\newcommand{\hb}{\mathbf{h}}

\newcommand{\sbb}{\mathbf{s}}

\newcommand{\ub}{\mathbf{u}}

\newcommand{\xb}{\mathbf{x}}
\newcommand{\yb}{\mathbf{y}}

\newcommand{\Ab}{\mathbf{A}}
\newcommand{\Bb}{\mathbf{B}}
\newcommand{\Cb}{\mathbf{C}}
\newcommand{\Db}{\mathbf{D}}
\newcommand{\Eb}{\mathbf{E}}
\newcommand{\Fb}{\mathbf{F}}
\newcommand{\Gb}{\mathbf{G}}
\newcommand{\Hb}{\mathbf{H}}
\newcommand{\Ib}{\mathbf{I}}
\newcommand{\Jb}{\mathbf{J}}
\newcommand{\Kb}{\mathbf{K}}

\newcommand{\Mb}{\mathbf{M}}

\newcommand{\Ob}{\mathbf{O}}
\newcommand{\Pb}{\mathbf{P}}
\newcommand{\Qb}{\mathbf{Q}}
\newcommand{\Rb}{\mathbf{R}}
\newcommand{\Sbb}{\mathbf{S}}

\newcommand{\Ub}{\mathbf{U}}
\newcommand{\Vb}{\mathbf{V}}
\newcommand{\Wb}{\mathbf{W}}
\newcommand{\Xb}{\mathbf{X}}
\newcommand{\Yb}{\mathbf{Y}}

\newcommand{\bbeta}{\bm{\beta}}
\newcommand{\bgamma}{\bm{\gamma}}

\newcommand{\bzeta}{\bm{\zeta}}

\newcommand{\bvartheta}{\bm{\vartheta}}

\newcommand{\bxi}{\bm{\xi}}

\newcommand{\bvarphi}{\bm{\varphi}}

\newcommand{\bomega}{\bm{\omega}}

\newcommand{\bGamma}{\bm{\Gamma}}

\newcommand{\bTheta}{\bm{\Theta}}
\newcommand{\bLambda}{\bm{\Lambda}}
\newcommand{\bXi}{\bm{\Xi}}

\newcommand{\bSigma}{\bm{\Sigma}}
\newcommand{\bUpsilon}{\bm{\Upsilon}}

\newcommand{\bOmega}{\bm{\Omega}}

\newcommand{\T}{\mathrm{\scriptscriptstyle \top }}
\newcommand{\MP}{\mathrm{\scriptscriptstyle + }}
\newcommand{\mminus}{\scalebox{1}{\text{-}}}

\makeatletter
\newcommand{\fdsy@scale}{1.0}
\newcommand\fdsy@mweight@normal{Book}
\newcommand\fdsy@mweight@small{Book}
\newcommand\fdsy@bweight@normal{Medium}
\newcommand\fdsy@bweight@small{Medium}

\DeclareFontFamily{U}{FdSymbolC}{}

\DeclareFontShape{U}{FdSymbolC}{m}{n}{
	<-7.1> s * [\fdsy@scale] FdSymbolC-\fdsy@mweight@small
	<7.1-> s * [\fdsy@scale] FdSymbolC-\fdsy@mweight@normal
}{}
\DeclareFontShape{U}{FdSymbolC}{b}{n}{
	<-7.1> s * [\fdsy@scale] FdSymbolC-\fdsy@bweight@small
	<7.1-> s * [\fdsy@scale] FdSymbolC-\fdsy@bweight@normal
}{}

\DeclareSymbolFont{arrows}{U}{FdSymbolC}{m}{n}
\SetSymbolFont{arrows}{bold}{U}{FdSymbolC}{b}{n}

\DeclareMathSymbol{\upvDash}{\mathrel}{arrows}{233}

\DeclareMathSymbol{\upmodels}{\mathrel}{arrows}{237}
\makeatother

\def\spacingset#1{\renewcommand{\baselinestretch}%
	{#1}\small\normalsize} \spacingset{1}

\newcommand{\blind}{1}

\def\singlespace{\def\baselinestretch{1}\@normalsize}

\begin{document}

\if1\blind
{
\spacingset{1.25}
  \title{\bf \Large CP-Factorization for High Dimensional Tensor Time Series and Double Projection Iterations}
\author[1,2]{Jinyuan Chang}
\author[1]{Guanglin Huang}
\author[3]{Qiwei Yao}
\author[4]{Long Yu}


\affil[1]{\it \small Joint Laboratory of Data Science and Business
Intelligence, Institute of Statistical
Interdisciplinary Research, Southwestern University of Finance and Economics, Chengdu, China}
\affil[2]{\it \small State Key Laboratory of Mathematical Sciences, Academy of Mathematics and Systems Science, Chinese Academy of Sciences, Beijing, China}
\affil[3]{\it \small Department of Statistics, The London School of Economics and Political Science, London, U.K.}
\affil[4]{\it \small School of Statistics and Data Science,  Institute of Data Science and Statistics, Shanghai University of Finance and Economics, Shanghai, China}

		\setcounter{Maxaffil}{0}
		
		\renewcommand\Affilfont{\itshape\small}
		\date{\vspace{-5ex}}
		\maketitle
	} \fi
	\if0\blind
	{
		\bigskip
		\bigskip
		\bigskip
		\begin{center}
			{
			\Large \bf CP-Factorization for High Dimensional Tensor Time Series and Double Projection Iterations
			}
		\end{center}
		\medskip
	} \fi
 
\spacingset{1.5}
\begin{abstract}
We adopt the canonical polyadic (CP) decomposition to model high-dimensional tensor time series. Our primary goal is to identify and estimate the factor loadings in the CP decomposition. We propose a one-pass estimation procedure through standard eigen-analysis for a matrix constructed based on the serial dependence structure of the data. The asymptotic properties of the proposed estimator are established under a general setting as long as the factor loading vectors are linearly independent,  allowing the factors to be correlated and the factor loading vectors to be not nearly orthogonal. The procedure adapts to the sparsity of the factor loading vectors, accommodates weak factors, and demonstrates strong performance across a wide range of scenarios. To further reduce estimation errors, we also introduce an iterative algorithm based on a novel double projection approach. We theoretically justify the improved convergence rate of the iterative estimator, and derive the associated limiting distribution. A consistent estimator of the asymptotic variance is also provided, which plays a key role in the related inference problems. All results are validated through extensive simulations and two real data applications.
\end{abstract}


\spacingset{1.69}
\setlength{\abovedisplayskip}{0.2\baselineskip}
\setlength{\belowdisplayskip}{0.2\baselineskip}
\setlength{\abovedisplayshortskip}{0.2\baselineskip}
\setlength{\belowdisplayshortskip}{0.2\baselineskip}

\section{Introduction}\label{sec: introduction}
Due to recent advances in information technology and data science, 
the demand for tensor analysis arises in a variety of fields, including but not limited to neuroimaging \citep{zhou2013tensor}, recommendation systems \citep{bi2018}, and dynamic transportation networks \citep{chen2021autoregressive,chen2022factor}. This motivates the rapid development of efficient inference tools and the associated theory for tensor analysis.
%
The size of tensor data is often large or extremely large. Low-rank approximation remains one of the most powerful tools for balancing between computational and statistical efficiencies. 
In tensor analysis, there are two types of frequently used  rank decompositions: the Tucker decomposition and the canonical polyadic (CP) decomposition, and both can be viewed as a natural extension of the singular value decomposition (SVD) for matrices \citep{kolda2009tensor}. 
The Tucker decomposition is often achieved by the SVD on the unfolded matrices \citep{de2000multilinear}. 
By contrast, computing the CP decomposition is
NP-hard, and the alternating least squares iteration remains as the workhorse method 
\citep{wang2017tensor}.

	In real applications, tensor data are often recorded in chronological order, and the dynamics of the data are often driven by a small number of factors.  
    Therefore, it is natural to extend the vector-valued factor models
    \citep{bai2003inferential, lam2012factor, chang2015high} for tensor time series, leading to 
    two types of tensor factor models based on, respectively, the Tucker decomposition and the CP decomposition. Most existing studies focus on tensor Tucker-factor models, which can be traced back to \cite{wang2019factor}, where a two-way factor structure was introduced for matrix time series  (i.e. a tensor with two modes). 
    See also \cite{elynn2020}, \cite{yu2022projected},  and \cite{chen2023statistical}. 
Extensions to higher-order tensor Tucker-factor models have been studied in
 \cite{chen2024rank}, \cite{han2024tensor},
\cite{chen2024semi}, \cite{barigozzi2023statistical}, and
\cite{HE2026105557}. Robust estimation methods for tensor Tucker-factor models are further considered in \cite{barigozzi2023robust} and \cite{barigozzi2025tail}.
    Note that both the factors and the factor loadings in the Tucker decomposition are not uniquely defined: the decomposition is invariant under general invertible linear transformations. In empirical practice, rotations are often applied to the estimated loadings and factors  to enhance interpretability.
    

In contrast, the factor loadings in tensor CP-factor models are uniquely defined up to the reflection and permutation indeterminacy,  where the reflection indeterminacy is also referred to as the sign indeterminacy. This facilitates  a straightforward and practically meaningful interpretation of the fitted models.
See, for example, the real data illustration in Section \ref{sec:application} below.
For tensor CP-factor models, \cite{han2024cp} propose an algorithm for the so-called High-Order Projection Estimator (HOPE),  which consists of a principal component analysis (PCA) based initialization followed by recursive iterations. Remarkably, the estimation accuracy improves progressively with each iteration even though the initial estimate is not consistent. \cite{chen2026estimation} extend the HOPE using a contemporary covariance matrix and randomized projection, and also derive the limiting distribution of the associated estimator.
From the theoretical  perspective, the HOPE is constructed under the following requirements:
    (a) the factor loading vectors are nearly orthogonal, and 
(b) the  factors are almost uncorrelated. 
Note that both factor 
loadings and factors are uniquely defined in CP decomposition (up to the reflection and permutation indeterminacy), and there is no guarantee that those two requirements fulfill.
Free from requirement (b)  and replacing ``nearly orthogonal" by ``linearly independent" in requirement (a),
\cite{chang2023modelling} propose a one-pass estimation procedure (i.e. without iterations) for matrix CP-factor models. Furthermore the ``linearly independent" requirement is freed in \cite{chang2024unified} which propose another one-pass estimation procedure applicable to matrix CP-factor models with rank-deficient factor loadings.
Note that the two estimation procedures are radically different: the method of \cite{chang2023modelling} is based on a generalized eigen-analysis while
the procedure of \cite{chang2024unified}  is more involved and a key step is
to identify the factor loadings by a joint diagonalization of several symmetric matrices defined by the basis vectors of a linear system. However, neither \cite{chang2023modelling} nor \cite{chang2024unified} provide results on statistical inference, and it remains unclear how to generalize their procedures to higher-order tensor settings.
	
	In this paper, we develop a unified framework for estimating  CP-factor models for tensor time series, including the models for matrix time series as special cases. We do not impose requirements (a) and (b) stated above in our analysis. Under the mild  assumption that the factor loading vectors are linearly independent, we propose two new estimation methods. Different from the method of \cite{chang2023modelling} which relies on the generalized eigen-analysis of certain matrices, our first method is established through  the standard eigen-analysis of a well-designed matrix constructed based on the serial dependence structure of the
data, which can  substantially attenuate the plug-in errors, and adapt to the higher (than two) mode tensor structures.
    The second new estimation, termed as double projection method, is an iterative algorithm. 
    This algorithm substantially outperforms state-of-the-art iterative approaches \citep{han2024cp, chen2026estimation} in terms of both statistical and computational efficiency, particularly in the presence of correlated factors. The superiority of the iterative estimator is rigorously established through theoretical analysis and extensive numerical studies.   Furthermore, we derive a tractable limiting representation for this new iterative estimator, along with its explicit asymptotic distribution. An estimator for the asymptotic variance is also provided, which makes the related statistical inference feasible. The \textsf{R}-function \texttt{CP\_TTS} for implementing our newly proposed methods is available publicly in the \texttt{HDTSA} package \citep{chang2024hdtsa}.
	
	The rest of the paper is organized as follows. Section \ref{sec: model} introduces the tensor CP-factor models and the basic settings. Section \ref{sec: methodology} presents our methodology, including the one-pass and iterative estimators together with the inference procedure. 
    Sections \ref{sec:numerical} and \ref{sec:application} validate the performance of our proposed methods through simulation studies and  a real data analysis, respectively. Section \ref{sec: theoretical}  develops the theoretical guarantees of the proposed estimators. Section \ref{sec: discuss} provides some discussion. 
    All technical details, and additional simulation and empirical results, are provided in the supplementary material. The replication code for both the simulations and the real data analysis is
available at the GitHub repository: \url{https://github.com/JinyuanChang-Lab/CPTensorTimeSeries}.
	
	\emph{Notation}. For any integer $p\ge 1$, let $[p]=\{1,\ldots,p\}$, and $\Ib_p$ be  the $p\times p$  identity matrix. Denote by $I(\cdot)$ the indicator function. For a vector $\ab=(a_1,\ldots,a_p)^\T$,  define $|\ab|_2=(\sum_{i=1}^p a_i^2)^{1/2}$ and $|\ab|_0 = \sum_{i=1}^p I(a_i \neq 0)$. For a complex vector $\ab$, $\operatorname{Re}(\ab)$ denotes the vector of its real parts taken entry-wise. For a matrix $\Ab=(a_{i,j})_{p_1 \times p_2}$, denote by $\sigma_i(\Ab)$, $\sigma_{\max}(\Ab)$, and $\sigma_{\min}(\Ab)$, respectively, its $i$-th largest, maximum, and minimum singular values. We write $\|\Ab\|_2=\sqrt{\sigma_1(\Ab^\T\Ab)}$, $\|\Ab\|_{\rm F}=\sqrt{\sum_{i=1}^{p_1}\sum_{j=1}^{p_2}a_{i,j}^2}$,  and $ | \Ab |_{\max}=\max_{i\in[p_1],j\in[p_2]}|a_{i,j}|$. For an $m$-mode tensor $\mathcal{A}=(a_{i_1,\ldots,i_m})_{p_1\times\cdots \times p_m}$, we write $[\mathcal{A}]_{i_1,\ldots,i_m}=a_{i_1,\ldots,i_m}$. The vectorization  $\textup{vec}(\mathcal{A})$ is an ${\textstyle{\prod}}_{j=1}^m p_j$-dimensional vector, with the $\{ 1 + \sum_{k = 1}^m(i_k - 1)\prod_{\ell=1}^{k-1}p_\ell\}$-th element being $a_{i_1,\ldots,i_m}$. The matricization $\textup{Mat}_q(\mathcal{A})$ is a $p_q \times \prod_{j \neq q}p_j$ matrix, with the $\{i_{q}, 1 + \sum_{k \neq q}^m(i_k - 1)\prod_{\ell \neq q}^{k-1}p_\ell\}$-th element being $a_{i_1,\ldots,i_m}$. For two sequences of positive numbers $\{a_n\}_{n\ge 1}$ and $\{b_n\}_{n\ge 1}$, we write $a_n\lesssim b_n$ or $b_n\gtrsim a_n$ if $\limsup_{n\rightarrow\infty}a_n/b_n< \infty$,  $a_n\asymp b_n$ if and only if $b_n\gtrsim a_n$ and $a_n\gtrsim b_n$ hold simultaneously, and $a_n\ll b_n$ or $b_n\gg a_n$ if $\limsup_{n\rightarrow\infty}a_n/b_n=0$.  For any $a, b\in\mathbb{R}$, let $a \vee b=\max(a,b)$ and $a \wedge b=\min(a,b)$. Denote by $\lfloor x \rfloor$ the largest integer less than or equal to $x$.   Let $\circ$ and $\otimes$ denote the vector outer product and the Kronecker product, respectively.
	
	\section{Model}\label{sec: model}
    Let $\mathcal{Y}_t\in \mathbb{R}^{d_1\times \cdots \times d_m}$ be an observed $m$-mode tensor with 
    $m\geq2$. We consider the tensor CP-factor model \citep{han2024cp} as follows:
	\begin{equation}\label{model cp}
		\mathcal{Y}_t= \sum_{i=1}^rw_if_{t,i}\,\ab_{i,1} \circ \ab_{i,2} \circ \cdots \circ \ab_{i,m}+\mathcal{E}_t\,,\quad t\geq1\,,
	\end{equation}
	where $1\leq r \leq \min_{j\in[m]}d_j$ is a fixed but unknown constant,  	 $\mathcal{E}_t \in \mathbb{R}^{d_1\times \cdots \times d_m}$ is the idiosyncratic error tensor,  $\fb_t = (f_{t,1},\ldots,f_{t,r})^\T$ is the $r$-dimensional  factor vector, $w_i$ is the strength of the $i$-th factor,  and $\ab_{i,j}$ is a $d_j$-dimensional factor loading vector corresponding to the $i$-th factor and $j$-th mode. 
    Without loss of generality, we assume $|\ab_{i,j}|_2=1$ for any $i\in[r]$ and $j\in[m]$. When $w_i=1$ for all $i\in[r]$, model \eqref{model cp} is an extension of the model considered in \cite{chang2023modelling} from matrix to tensor regimes.  In practice, $w_i$ can be either a constant or grow with the dimensions. To understand this, let us consider a toy example that $[\mathcal{Y}_{t}]_{h_1,\ldots,h_m} = \beta f_{t}+[\mathcal{E}_t]_{h_1,\ldots,h_m}$ for any   $h_j\in[d_j]$ with $j\in[m]$. This example can be formulated as model \eqref{model cp} with $r = 1$, $w_1=\beta(\prod_{j=1}^{m}d_j)^{1/2}$ and $\ab_{1,j}=(d_j^{-1/2},\ldots,d_j^{-1/2})^\T$ for each $j \in [m]$, where  $w_1$ will diverge if at least one $d_j$ grows to infinity as $n\rightarrow\infty$. 
    
    The key problem of interest for the tensor CP-factor model \eqref{model cp} is to identify the factor loading vectors $\{\ab_{i,j}\}_{i\in[r],j\in[m]}$ and also provide the suitable estimates for them. To do this, we first impose the following regularity assumption on the tensor CP-factor model \eqref{model cp}.
    
    	\begin{assumption}\label{error}
		It holds that $\mathbb{E}(\mathcal{E}_t)={\bf 0}$ for any $t\in[n]$, $\mathbb{E}(\mathcal{E}_t\otimes\mathcal{E}_s)={\bf 0}$ for any $t\ne s$, and $\mathbb{E}(f_{t,i}\mathcal{E}_s)={\bf 0}$ for any $i\in[r]$ and $t,s \in [n]$.
	\end{assumption}
    \begin{remark}
 \textup{(a)}  Assumption \textup{\ref{error}} is significantly weaker than the assumptions imposed in \textup{\cite{han2024cp}}, which is satisfied automatically under the assumptions of \textup{\cite{han2024cp}}.   More specifically, \textup{\cite{han2024cp}} require the error process $\{\mathcal{E}_t\}_{t \ge 1}$ to be independent Gaussian tensors conditional on the factor process $\{\fb_{t}\}_{t \ge 1}$. Furthermore, \textup{\cite{han2024cp}} require the factor process $\{\fb_{t}\}_{t \ge 1}$ to be  stationary with zero mean and also to satisfy $\mathbb{E}(f_{t,i}^2) = 1$ and $\mathbb{E}(f_{t,i}f_{t-k,j})=0$ for all $i\neq j$ and $k\ge 1$, while the stationarity and zero mean are not necessarily required in our framework. \textup{(b)} For the relationship between the factor process $\{\fb_t\}_{t \ge 1}$ and the error process $\{\mathcal{E}_t\}_{t \ge 1}$,   the theoretical analysis of the one-pass estimator introduced in Section \textup{\ref{sec: initial}} only requires $\mathbb{E}(f_{t,i}\mathcal{E}_s)={\bf 0}$ for any $i\in[r]$ and $t \le s$, while the iterative estimator  introduced in Section \textup{\ref{sec: Double projection iterations}} necessitates the stronger condition as stated in Assumption \textup{\ref{error}}. \textup{(c)} Write  $\sigma_{t,i}^2 = \mathbb{E}(f^2_{t,i})$. Different from \textup{\cite{han2024cp}}, we do not require $\sigma_{t,i}^2$ 
 equal to $1$, which allows $\sigma_{t,i}^2$ to vary with $t$ for each given $i$. If $\sigma_{t,i}^2 \equiv \sigma_{i}^2$ for all $t \in [n]$, we can assume $\mathbb{E}(f^2_{t,i}) = 1$ without loss of generality. From this perspective, our model setting is more general than that in \textup{\cite{han2024cp}}. For more general scenarios, $w_i$ and $f_{t,i}$ cannot be identified separately.  Nevertheless, the loading vectors $\ab_{i,1},\ldots,\ab_{i,m}$ remain identifiable up to the reflection and permutation indeterminacy.  \textup{(d)} The idiosyncratic errors are assumed to be serially uncorrelated, which enables a direct separation of the signal part and the noise part through the auto-covariances of the observed data. This is a common assumption in the literature on  factor modeling; see, for example, \textup{\cite{lam2012factor}} and \textup{\cite{han2024cp}}. By contrast, the factors are allowed to be serially correlated; see Assumption \textup{\ref{mixing}} in Section \textup{\ref{sec: asmp}}. 
    \end{remark}

  For each $j\in[m]$, write $d_{\mminus j} = \prod_{j^\prime\ne j}^md_{j^\prime}$. Following the tensor matricization in \cite{kolda2009tensor}, we can reshape $\mathcal{Y}_t$ into a $d_j\times d_{\mminus j}$ matrix as follows:
	\begin{equation}\label{unfold X}
\Yb_{t,j}=\text{Mat}_j(\mathcal{Y}_t)= \underbrace{ {\textstyle\sum\nolimits}_{i=1}^r w_if_{t,i}\ab_{i,j}\bbb_{i,j}^\T}_{\Cb_{t,j}} +\underbrace{\text{Mat}_j(\mathcal{E}_t)}_{\Eb_{t,j}}\,,~~~~j\in[m]\,,
	\end{equation}
	where $	\bbb_{i,j}=\ab_{i,m}\otimes\cdots\otimes \ab_{i,j+1}\otimes \ab_{i,j-1}\otimes\cdots\otimes \ab_{i,1} \in \mathbb{R}^{d_{\mminus j}}$. Write $\Ab_j=(\ab_{1,j},\ldots,\ab_{r,j})$, $\Bb_j=(\bbb_{1,j},\ldots,\bbb_{r,j})$, and  $\Xb_t = \textup{diag}(w_1f_{t,1},\ldots,w_r f_{t,r})$. Then model \eqref{unfold X} can be rewritten as the matrix CP-factor model considered in \cite{chang2023modelling}:
   \begin{equation}\label{eq:matrixform}
       \Yb_{t,j} = \Ab_j \Xb_{t} \Bb_j^{\T} + \Eb_{t,j}\,, ~~~~j\in[m]\,.
   \end{equation}
We assume $\textup{rank}(\Ab_j) = r$ for $j \in [m]$, which is also required in
\cite{han2024cp}. 
Under this assumption, Lemma \ref{pro:rank-B} in the supplementary material  shows that $\Bb_1,\ldots,\Bb_m$ also have full column rank.

We could estimate $\ab_{i,j}$  by the method proposed in \cite{chang2023modelling}. However, for the tensors with more than two modes, the matrix $\Bb_j$ in \eqref{eq:matrixform}  involves a Kronecker product structure. Direct application of the method in \cite{chang2023modelling} would discard this structural information, leading to substantial statistical inefficiency. Meanwhile, the two-stage estimation approach of \cite{chang2023modelling} introduces cross-step plug-in errors, which significantly complicates statistical inference for $\ab_{i,j}$ in high-dimensional settings.  In this paper, we propose a new one-pass estimation method that only requires the eigen-decomposition of a $d_j \times d_j$ matrix to obtain estimates of $\ab_{1,j},\ldots,\ab_{r,j}$. This novel one-pass design eliminates cross-stage plug-in errors, and further motivates an iterative estimation procedure introduced in Section \ref{sec: Double projection iterations}. Moreover, simulation studies in Section \ref{sec:numerical} show that even for matrix-variate cases, our newly proposed methods perform significantly better than the procedure in \cite{chang2023modelling}.
    
\section{Methodology}\label{sec: methodology}
 
\subsection{One-pass estimation of the factor loading vectors}\label{sec: initial}
 
    Let $\xi_{t}$ be a linear combination of the components of $\mathcal{Y}_t$. For any $k\ge 1$ and $t \ge k+1$, we define $\bXi_{t,k,j}=\mathbb{E}[\{\Yb_{t,j}-\mathbb{E}(\bar\Yb_j)\}\{\xi_{t-k}-\mathbb{E}(\bar\xi)\}]$ with $\bar \Yb_j=n^{-1}\sum_{t=1}^n\Yb_{t,j}$ and $\bar\xi=n^{-1}\sum_{t=1}^n\xi_{t}$. Given observations $\{\mathcal{Y}_t\}_{t=1}^n$, for any $k\geq1$, we write
 \begin{equation}\label{Xi and Sigma Y}
 \bSigma_{\Yb_j,\xi}(k)=\frac{1}{n-k}\sum_{t=k+1}^n\bXi_{t,k,j}\,,  
	\end{equation}
and let  $\Gb_{k,\xi}=\textup{diag}(g_{k,1,\xi},\ldots,g_{k,r,\xi})$ be an $r\times r$ diagonal matrix with
	\begin{equation}\label{gki}
		g_{k,i,\xi}=\frac{1}{n-k}\sum_{t=k+1}^n w_i\mathbb{E}[\{f_{t,i}-\mathbb{E}(\bar f_i)\}\{\xi_{t-k}-\mathbb{E}(\bar\xi)\}]\,,
	\end{equation}
	where $\bar f_i=n^{-1}\sum_{t=1}^n f_{t,i}$. Then $\bSigma_{\Yb_j,\xi}(k) = \Ab_j \Gb_{k,\xi} \Bb_j^{\T}$. For each given $j \in [m]$, by singular value decomposition of $\Bb_j$,  there exist a $d_{\mminus j} \times r$ column-orthogonal matrix $\Qb_j$ and an $r\times r$ invertible matrix $\Vb_j$ such that $\Bb_j=\Qb_j\Vb_j$. If  $\textup{rank}(\Gb_{1,\xi}) = r = \textup{rank}(\Gb_{2,\xi})$, we  define 
 \begin{align}\label{Kbj}
		\Kb_{1,2,j}&=\bSigma_{\Yb_j,\xi}(1)\Qb_j \{\Qb_j^\T\bSigma_{\Yb_j,\xi}(2)^\T\bSigma_{\Yb_j,\xi}(2)\Qb_j \}^{-1}\Qb_j^\T\bSigma_{\Yb_j,\xi}(2)^\T\,. 
	\end{align}
Since $\bSigma_{\Yb_j,\xi}(k) = \Ab_j \Gb_{k,\xi} \Bb_j^{\T}$,  we have $\Kb_{1,2,j} = \Ab_j\Gb_{1,\xi}\Gb_{2,\xi}^{-1}(\Ab_j^\T\Ab_j)^{-1}\Ab_j^\T$, which can be used to identify $\Ab_j$. Write $\bar\lambda_i=g_{2,i,\xi}^{-1}g_{1,i,\xi}$ with $g_{k,i,\xi}$ defined in \eqref{gki}. For each given $j\in[m]$, since $(\Ab_j^\T\Ab_j)^{-1}\Ab_j^\T\Ab_j=\Ib_r$, then
	$
	\Kb_{1,2,j}\ab_{i,j}=\bar\lambda_i\ab_{i,j}$ for any $i\in[r]$, 
	which implies that, as long as $\bar\lambda_1,\ldots,\bar\lambda_r$ are distinct, $\ab_{1,j},\ldots,\ab_{r,j}$ can be identified uniquely up to the reflection and permutation indeterminacy by solving the eigen-equation
\begin{equation}\label{generalized eigenequation}
		\Kb_{1,2,j}\ab=\lambda\ab\,.
	\end{equation}

	In practice, $\bSigma_{\Yb_j,\xi}(k)$ and $\Qb_j$ in \eqref{Kbj} are unknown. Given observations $\{\mathcal{Y}_t\}_{t=1}^n$, in the spirit of \cite{bickel2008covariance}, we can estimate $\bSigma_{\Yb_j,\xi}(k)$ by
\begin{equation}\label{hat Sigma kj}
	\tilde\bSigma_{k,j}=T_{\delta_1}\{\tilde\bSigma_{\Yb_j,\xi}(k)\}~~\textrm{with}~~\tilde\bSigma_{\Yb_j,\xi}(k)=\frac{1}{n-k}\sum_{t=k+1}^n(\Yb_{t,j}-\bar\Yb_j)(\xi_{t-k}-\bar\xi)\,,
	\end{equation}
where $T_{\delta_1}(\cdot)$ is a thresholding  operator, i.e., $[T_{\delta_1}(\Wb)]_{i,j} = W_{i,j}\, I(|W_{i,j}| \ge \delta_1)$ for any matrix $\Wb = (W_{i,j})$ with the threshold level $\delta_1 \ge 0$. To estimate $\Qb_j$, define 
    \begin{equation}\label{tilde Mj}
		\mathbf{M}_{j}=\sum_{k=1}^K \bSigma_{\Yb_j,\xi}(k)^\T\bSigma_{\Yb_j,\xi}(k) 
	\end{equation}
 for some predetermined integer $K\ge 1$. Under Assumption \ref{error} and  $\textup{rank}(\Ab_{j}) = r = \textup{rank}(\Gb_{1,\xi})$, we have $\mathbf{M}_{j} = \Bb_j (\sum_{k=1}^K\Gb_{k,\xi}\Ab_j^\T\Ab_j\Gb_{k,\xi})\Bb_j^\T$ with $\textup{rank}(\mathbf{M}_{j}) = r$,  which implies the columns of $\Qb_j$ are in the linear space spanned by the $r$ orthonormal eigenvectors of $\mathbf{M}_{j}$ corresponding to its $r$ largest eigenvalues. Let $\tilde{r}$ be a consistent estimate of $r$, which will be specified in Section \ref{sec:tuning}. Then we select $\tilde\Qb_j$ as a $d_{\mminus j} \times \tilde{r}$ matrix of which the columns are the $\tilde{r}$ orthonormal eigenvectors of $\tilde{\mathbf{M}}_{j} =\sum_{k=1}^K \tilde\bSigma_{k,j}^\T\tilde\bSigma_{k,j}$ corresponding to its $\tilde{r}$ largest eigenvalues.
By plugging $\tilde\Qb_j$ and $\tilde\bSigma_{k,j}$ into \eqref{Kbj}, we can estimate $\Kb_{1,2,j}$ by
	\begin{equation}\label{tilde K21j}
		\tilde\Kb_{1,2,j}=\tilde\bSigma_{1,j}\tilde\Qb_j(\tilde\Qb_j^\T\tilde\bSigma_{2,j}^\T\tilde\bSigma_{2,j}\tilde\Qb_j)^{-1}\tilde\Qb_j^\T\tilde\bSigma_{2,j}^\T\,.
	\end{equation} 
	Let $\tilde\lambda_{i,j}$ and $\tilde\ab_{i,j}$ be the $i$-th largest eigenvalue (in terms of absolute value) and the associated eigenvector of $\tilde\Kb_{1,2,j}$, respectively.  Note that $\Kb_{1,2,j}$ and $\tilde\Kb_{1,2,j}$ are in general nonsymmetric matrices. Although \eqref{generalized eigenequation} indicates that the eigenvectors of $\Kb_{1,2,j}$ are real vectors, those of the estimator $\tilde\Kb_{1,2,j}$ are not guaranteed to always be real vectors in practice.  If $\tilde\ab_{i,j}$ is a complex vector, we replace it by $\operatorname{Re}(\tilde\ab_{i,j})/|\operatorname{Re}(\tilde\ab_{i,j})|_2$,  which has a negligible effect on the consistency of the estimator.  The loading matrix $\Ab_j=(\ab_{1,j},\ldots,\ab_{r,j})$ is then estimated by $(\tilde{\ab}_{1,j},\ldots,\tilde{\ab}_{\tilde{r},j})$. When $\tilde r=r$, for each given $j\in[m]$, Theorem \ref{thm: aij} in Section \ref{sec: iterative theorem} shows that the proposed one-pass estimator $\{\tilde\ab_{i,j}\}_{i\in[\tilde r]}$ is consistent to $\{\ab_{i,j}\}_{i\in[r]}$ up to the reflection and permutation indeterminacy.

\subsection{Double projection estimation for the factor loading vectors}\label{sec: Double projection iterations} 
The one-pass estimation procedure entails thresholding the large $d_j\times d_{\mminus j}$ matrix $\tilde\bSigma_{\Yb_j,\xi}(k)$ in  \eqref{hat Sigma kj}, where the thresholding errors of all the elements in the matrix will accumulate.  Moreover, its performance depends on the choice of the linear combination $\xi_t$.  To address these issues, we introduce a novel double projection iteration method that achieves high accuracy without relying on the uncorrelated factor assumption or the near-orthogonality condition on factor loadings required in \cite{han2024cp}. More specifically, when a consistent initial estimator is available, by projecting the tensor-valued data into lower dimensions,  it will suffice to perform thresholding on a  $d_j$-dimensional vector rather than a large $d_j\times d_{\mminus j}$ matrix. Meanwhile, the initial estimator can be leveraged to construct a specific linear combination of 
 $\mathcal{Y}_t$, denoted by $\tilde{\xi}_t$, to further reduce the estimation error.

Write $(\bbb_{1,j}^{\MP},\ldots,\bbb_{r,j}^{\MP})^{\T}=(\Bb_j^\T\Bb_j)^{-1}\Bb_j^{\T}$. When $\bbb_{i,j}^{\MP}$ is given, it follows from \eqref{eq:matrixform} that, for each $j\in[m]$, the $d_j\times d_{\mminus j}$ matrix $\Yb_{t,j}$ can be projected into the $d_j$-dimensional vector
	\begin{equation}\label{ytij}
		\yb_{t,i,j}=\Yb_{t,j}\bbb_{i,j}^{\MP}=w_if_{t,i}\ab_{i,j}+\eb_{t,i,j}\,,
	\end{equation}
	where $\eb_{t,i,j}=\Eb_{t,j}\bbb_{i,j}^{\MP}$  with $\Eb_{t,j}$ specified in \eqref{unfold X}. Then, it reduces to a standard vector-variate factor model with only one factor and a much lower dimension.  Given $\xi_t$, a linear combination of $\mathcal{Y}_t$, for each  $i \in [r]$ and $j \in [m]$, similarly to \eqref{Xi and Sigma Y}, we let 
\begin{equation*}\label{Sigma small y}	
\bSigma_{\yb_{i,j},\xi}(1)=\frac{1}{n-1}\sum_{t=2}^n\mathbb{E}[\{\yb_{t,i,j}-\mathbb{E}(\bar\yb_{i,j})\}\{\xi_{t-1}-\mathbb{E}(\bar \xi)\}]\,,
\end{equation*}
where  $\bar\yb_{i,j}=n^{-1}\sum_{t=1}^n\yb_{t,i,j}$.
Under Assumption \ref{error} and $\min_{i \in [r]}|g_{1,i,\xi}| > 0$ for $g_{1,i,\xi}$ defined in \eqref{gki}, it holds that $ \bSigma_{\yb_{i,j},\xi}(1)=g_{1,i,\xi}\ab_{i,j}$ and $ \bSigma_{\yb_{i,j},\xi}(1)/|\bSigma_{\yb_{i,j},\xi}(1)|_2 \in \{ \ab_{i,j}, -\ab_{i,j}\}$. In practice, $\bbb_{i,j}^{\MP}$ is unknown. Based on the estimates $\{\tilde\ab_{i,j}\}_{i\in[\tilde{r}],j\in[m]}$, we can plug them into the definition of $\Bb_j$ to obtain $\tilde{\Bb}_j$, the estimate of $\Bb_j$. Set
$(\tilde\bbb_{1,j}^{\MP},\ldots, \tilde\bbb_{\tilde r,j}^{\MP})^{\T}
= (\tilde\Bb_j^\T\tilde\Bb_j)^{-1}\tilde\Bb_j^{\T}$ and define
	\[
    \begin{split}
\tilde\yb_{t,i,j}=\Yb_{t,j}\tilde\bbb_{i,j}^{\MP}= \underbrace{w_if_{t,i}\ab_{i,j}(\bbb_{i,j}^\T\tilde\bbb_{i,j}^{\MP})}_{\textup{``target" factor}} +\underbrace{\textstyle{\sum\nolimits_{\ell\ne i}}w_{\ell}f_{t,\ell}\ab_{\ell, j}(\bbb_{\ell, j}^\T\tilde\bbb_{i,j}^{\MP})}_{\textup{``noisy" factors}}+\underbrace{\Eb_{t,j}\tilde\bbb_{i,j}^{\MP}}_{\textup{error term}}\,, ~ i\in[\tilde{r}],\, j\in[m]\,.
    \end{split}
	\]
Based on a similar projection, \cite{han2024cp} update their estimators using the eigenvector of  $(n-k)^{-1}\sum_{t=k+1}^n(\tilde\yb_{t-k,i,j}\tilde\yb_{t,i,j}^\T+\tilde\yb_{t,i,j}\tilde\yb_{t-k,i,j}^\T)$ associated with the largest eigenvalue for some $k\ge 1$. However, the ``target" and ``noisy"  factors will interact with each other in their procedure. Therefore, they require the assumption of uncorrelated factors, i.e., $\mathbb{E}(f_{t,i}f_{t-k,j})=0$ for all $i\neq j$ and $k\ge 1$, to ensure the iteration works. When the factors are correlated, the iterative method in \cite{han2024cp} becomes inefficient, as shown in Figure \ref{fig:iter-step} in Section \ref{sec:numerical}. This motivates us to explore a new approach.

To reduce the effect of the ``noisy" factors for estimating $\ab_{i,j}$, we need to involve a new linear combination of $\mathcal{Y}_t$, denoted by $\tilde{\xi}_{t,i}$, such that $\tilde{\xi}_{t-1,i}$ is correlated to the ``target" factor $f_{t,i}$ but almost uncorrelated to the ``noisy'' factors $f_{t,\ell}$ for $\ell \ne i$. To this end, we first estimate the factor series. Notice that $(\ab_{i,m}^{\MP} \otimes \cdots \otimes \ab_{i,1}^{\MP})^{\T}\textup{vec}(\mathcal{Y}_t)=w_i f_{t,i}+(\ab_{i,m}^{\MP} \otimes \cdots \otimes \ab_{i,1}^{\MP})^{\T}\textup{vec}(\mathcal{E}_t)$, where $(\ab_{1,j}^{\MP},\ldots,\ab_{r,j}^{\MP})^{\T}=(\Ab_j^\T \Ab_j)^{-1}{\Ab}_j^{\T}$. We therefore estimate $w_i f_{t,i}$ by $\check f_{t,i}=(\tilde\ab_{i,m}^{\MP} \otimes \cdots \otimes \tilde\ab_{i,1}^{\MP})^{\T}\textup{vec}(\mathcal{Y}_t)$, where $(\tilde\ab_{1,j}^{\MP},\ldots,\tilde\ab_{\tilde r,j}^{\MP})^{\T}=(\tilde\Ab_j^\T\tilde\Ab_j)^{-1}\tilde{\Ab}_j^{\T}$ with $\tilde\Ab_j=(\tilde\ab_{1,j},\ldots,\tilde\ab_{\tilde r,j})$. Standardize the series and write
\begin{equation}\label{eq: standardized series}
    \tilde f_{t,i}=  (\check f_{t,i}- \bar{\check{f}}_{i})/\tilde{\sigma}_{\check f,i}\,,
\end{equation}
where $ \bar{\check{f}}_{i} = n^{-1}\sum_{t=1}^n \check f_{t,i}$ and $\tilde{\sigma}^2_{\check f,i} = (n-1)^{-1}\sum_{t=1}^n(\check f_{t,i}-\bar{\check{f}}_{i})^2$. If $\tilde r=1$, we let $\tilde \xi_{t,i}=\tilde f_{t,i}$ for $t\in[n]$ and $i\in[\tilde r]$. If $\tilde r\ge 2$,
 let $\tilde\fb_{i}=(\tilde f_{1,i},\ldots,\tilde f_{n-1,i})^\T$  and $\tilde\Fb_{\mminus i}$ be a $(n-1)\times (\tilde r-1)$ matrix of which the columns are composed of $(\tilde f_{2,\ell},\ldots,\tilde f_{n,\ell})^\T$ for  $\ell \ne i$. We project $\tilde\fb_{i}$ into the complementary space of $\tilde\Fb_{\mminus i}$ and obtain
\begin{equation}\label{eq: tilde xi}
    (\tilde\xi_{1,i},\ldots,\tilde\xi_{n-1,i})^{\T} = \{\Ib_{n-1}-\tilde\Fb_{\mminus i}(\tilde\Fb_{\mminus i}^\T\tilde\Fb_{\mminus i})^{-1}\tilde\Fb_{\mminus i}^{\T}\}\tilde\fb_{i}\,,
\end{equation}
where $\{\tilde\xi_{t,i}\}_{t=1}^{n-1}$ satisfies $ \sum_{t=2}^n \tilde\xi_{t-1,i} \tilde f_{t,\ell} = 0$ for $\ell\ne i$. Define
	\begin{equation*}\label{eq:tildebSigmatildeybij}
	\tilde\bSigma_{\tilde\yb_{i,j},\tilde\xi_{i}}(1) =\frac{1}{n-1}\sum_{t=2}^n(\tilde\yb_{t,i,j}-\bar{\tilde\yb}_{i,j})\tilde\xi_{t-1,i}
	\end{equation*}
    with $\bar{\tilde\yb}_{i,j} = n^{-1}\sum_{t=1}^n\tilde\yb_{t,i,j}$.
 Considering that the loading vector $\ab_{i,j}$ may be sparse,  we can update the estimator $\tilde\ab_{i,j}$ with $T_{\delta_{2,j}}\{\tilde\bSigma_{\tilde\yb_{i,j},\tilde\xi_{i}}(1)\}/|T_{\delta_{2,j}}\{\tilde\bSigma_{\tilde\yb_{i,j},\tilde\xi_{i}}(1)\}|_2$, where $\delta_{2,j}\geq0$ is the threshold level.
Such a double projection refinement can be naturally designed into an iterative procedure, as shown in  Algorithm \ref{alg1}. When $\tilde r=r$, for each given $j\in[m]$, Theorem \ref{thm: iterative} in Section \ref{sec: iterative theorem}  shows that the iterative estimator $\{\hat\ab_{i,j}\}_{i\in[\tilde r]}$ obtained in Algorithm \ref{alg1} is consistent to $\{\ab_{i,j}\}_{i\in[r]}$ up to the reflection and permutation indeterminacy.

   	\begin{algorithm}[H] 
       \footnotesize
		\caption{Double projection iterations for estimating $\{\Ab_j\}_{j=1}^m$}\label{alg1}
		\begin{algorithmic}[1]
			\Require tensor observations $\{\mathcal{Y}_t\}_{t=1}^n$,  number of factors $\tilde r$, initial estimates $\tilde\Ab^{(0)}_j = (\tilde\ab^{(0)}_{1,j},\ldots,\tilde\ab^{(0)}_{\tilde r,j}  )$ for $j \in [m]$, maximal number of iterations $L$, upper error bound $\epsilon_0$, threshold levels $\{\delta_{2,j}\}_{j=1}^m$.
			\Ensure Iterative estimates $\{\hat\Ab_j\}_{j=1}^m$
			\State  (\textbf{Initialization})  $\{ (\tilde\ab_{1,j}^{(0)})^\MP,\ldots,(\tilde\ab_{\tilde r,j}^{(0)})^\MP\}^{\T} \gets \{(\tilde\Ab^{(0)}_j)^\T\tilde\Ab^{(0)}_j\}^{-1}(\tilde\Ab^{(0)}_j)^{\T}$ for $j\in[m]$, $v=1$, $\epsilon^{(0)}=2\epsilon_0$;
			
			\While {$v \le L$ and $\epsilon^{(\textit{v}-1)}>\epsilon_0$}
			     \For {$j=1$ to $m$}
			     \For {$t=1$ to $n$, $i=1$ to $\tilde r$}
			     
			        \STATE $
			     \check f^{(\textit{v},j)}_{t,i}  \gets \{(\tilde\ab_{i,m}^{(\textit{v}-1)})^{\MP} \otimes \cdots \otimes (\tilde\ab_{i,j}^{(\textit{v}-1)})^{\MP} \otimes (\tilde\ab_{i,j-1}^{(\textit{v})})^{\MP} \otimes \cdots \otimes   (\tilde\ab_{i,1}^{(\textit{v})})^{\MP}\}^{\T}\textup{vec}(\mathcal{Y}_t)\,$;
			     
			        \STATE replace $\check f_{t,i}$ in \eqref{eq: standardized series} with $\check f^{(\textit{v},j)}_{t,i}$ to obtain the  standardized factors $\tilde f^{(\textit{v},j)}_{t,i}$;
			     
			     \EndFor
			     
			     \STATE $\tilde\bbb^{(\textit{v})}_{i,j} \gets \tilde\ab^{(\textit{v}-1)}_{i,m}  \otimes \cdots\otimes \tilde\ab^{(\textit{v}-1)}_{i,j+1}  \otimes \tilde\ab^{(\textit{v})}_{i,j-1}  \otimes \cdots \otimes \tilde\ab^{(\textit{v})}_{i,1}$,  $i \in [\tilde r]$;
			     
			     \STATE $\tilde \Bb^{(\textit{v})}_j \gets  (\tilde{\bbb}^{(\textit{v})}_{1,j} ,\ldots,\tilde{\bbb}^{(\textit{v})}_{\tilde r,j})$, $\{(\tilde\bbb^{(\textit{v})}_{1,j})^{\MP} ,\ldots,(\tilde\bbb^{(\textit{v})}_{\tilde r,j})^{\MP}\}^{\T}\gets\{(\tilde\Bb^{(\textit{v})}_j)^\T \tilde\Bb^{(\textit{v})}_j \}^{-1}(\tilde\Bb^{(\textit{v})}_j)^{\T} $;

			     \For {$i=1$ to $\tilde r$}
			     
			     \STATE if $\tilde r=1$, $\tilde \xi_{t,i}^{(\textit{v},j)}\gets \tilde f^{(\textit{v},j)}_{t,i},t\in[n],i\in[\tilde r]$;  if $\tilde r\ge 2$, replace $\tilde f_{t,i}$ in \eqref{eq: tilde xi} with $\tilde f^{(\textit{v},j)}_{t,i}$ to obtain $\{\tilde\xi^{(\textit{v},j)}_{t,i}\}_{t=1}^{n-1}$;
			     
			     \STATE $\tilde\yb^{(\textit{v})}_{t,i,j} \gets\Yb_{t,j}(\tilde\bbb_{i,j}^{(\textit{v})})^{\MP}$ for $t\in[n]$;
              
              \STATE $\tilde\bSigma^{(\textit{v},j)}_{\tilde\yb_{i,j},\tilde\xi_{i}}(1) \gets(n-1)^{-1}\sum_{t=2}^n\{\tilde\yb^{(\textit{v})}_{t,i,j}-n^{-1}\sum_{s=1}^n\tilde\yb^{(\textit{v})}_{s,i,j}\}\tilde\xi^{(\textit{v},j)}_{t-1,i}$;
		
		\STATE $			     	\tilde\ab^{(\textit{v})}_{i,j}\gets T_{\delta_{2,j}}\{\tilde\bSigma^{(\textit{v},j)}_{\tilde\yb_{i,j},\tilde\xi_{i}}(1)\} / |T_{\delta_{2,j}}\{\tilde\bSigma^{(\textit{v},j)}_{\tilde\yb_{i,j},\tilde\xi_{i}}(1) \}|_2$;

			   	 \EndFor  
			     \State  $ \tilde\Ab^{(\textit{v})}_j \gets  (\tilde\ab^{(\textit{v})}_{1,j} ,\ldots,\tilde\ab^{(\textit{v})}_{\tilde r,j} )$,  $\{(\tilde\ab^{(\textit{v})}_{1,j})^{\MP} ,\ldots,(\tilde\ab^{(\textit{v})}_{\tilde r,j})^{\MP}\}^{\T} \gets \{(\tilde\Ab^{(\textit{v})}_j)^\T\tilde\Ab^{(\textit{v})}_j\}^{-1}(\tilde\Ab^{(\textit{v})}_j)^{\T}$;
			     
			     \EndFor
			     
			     \STATE $\epsilon^{(\textit{v})}\gets \max_{j\in[m]}\max_{\ell\in[\tilde r]}\min_{i\in[\tilde r]}\{1-|(\tilde\ab^{(\textit{v})}_{i,j})^\T\tilde\ab^{(\textit{v} - 1)}_{\ell, j}|^2\}$, $v\gets v+1$;
            \EndWhile
            
            \STATE $\hat\ab_{i,j}\gets \tilde\ab^{(\textit{v})}_{i,j} $, 	$\hat\Ab_j\gets (\hat\ab_{1,j},\ldots,\hat\ab_{\tilde r,j})$ for $i\in[\tilde r]$ and $j\in[m]$.

		\end{algorithmic}      
	\end{algorithm}

\subsection{Inference procedure for the factor loading vectors}\label{sec:inference procedure}
To introduce the main idea of our inference procedure based on $\{\hat\ab_{i,j}\}_{i\in[\tilde{r}]}$, we assume $\hat{\ab}_{i,j}$ is consistent to $\ab_{i,j}$ for each $i\in[\tilde{r}]$ to simplify the notation. We consider a function $\Rb_{i,j}^*(\cdot): \mathbb{R}^{d_j}\rightarrow\mathbb{R}^{d_j}$ defined as 
$$
\Rb_{i,j}^*(\ab)= \bSigma_{\yb_{i,j},\xi}(1) - \{\ab^{\T}\bSigma_{\yb_{i,j},\xi}(1)\}\ab\,,\quad\ab \in \mathbb{R}^{d_j}\,.
$$
Since $\bSigma_{\yb_{i,j},\xi}(1) = g_{1,i,\xi}\ab_{i,j}$ and $ | \ab_{i,j} |_2 = 1$, we have $\Rb^*_{i,j}(\ab_{i,j}) = {\bf 0}$. 
For any deterministic vector $\hb\in\mathbb{R}^{d_j}$, under some regularity conditions, it follows from the Taylor expansion that 
\begin{align}\label{eq: debias taylor iterative}
        \hb^{\T}\bigg\{\frac{\partial \Rb^*_{i,j}(\ab_{i,j})}{\partial \ab^\T}\bigg\}^{-1}\Rb_{i,j}^*(\hat\ab_{i,j})&= \hb^{\T}\bigg\{\frac{\partial \Rb^*_{i,j}(\ab_{i,j})}{\partial \ab^\T}\bigg\}^{-1}\Rb_{i,j}^*(\ab_{i,j})+\hb^{\T}(\hat\ab_{i,j}-\ab_{i,j})\,\\
                &\quad +O(|\hb|_2|\hat\ab_{i,j}-\ab_{i,j}|_2^2)\,. \nonumber
\end{align}
Since $\Rb^*_{i,j}(\ab_{i,j}) = {\bf 0}$, 
 we have
  \begin{equation}\label{eq: iterative debis expansion 0}
     		\hb^{\T}(\hat\ab_{i,j} -\ab_{i,j}) = \hb^{\T}\bigg\{\frac{\partial \Rb^*_{i,j}(\ab_{i,j})}{\partial \ab^\T}\bigg\}^{-1} \Rb^*_{i,j}(\hat\ab_{i,j})+O(|\hb|_2|\hat\ab_{i,j}-\ab_{i,j}|_2^2)\,.
 \end{equation}
Therefore, the asymptotic representation of $\hb^\T\hat{\ab}_{i,j}$ is primarily driven by the leading term on the right-hand side of \eqref{eq: iterative debis expansion 0}. However, directly deriving the asymptotic distribution of this term creates significant difficulties. The thresholding technique involved in defining $T_{\delta_{2,j}}\{\tilde\bSigma^{(\textit{v},j)}_{\tilde\yb_{i,j},\tilde\xi_{i}}(1)\}$ in Algorithm \ref{alg1} introduces additional bias whose impact is difficult to characterize, making the derivation of an asymptotic distribution intractable. To guarantee a tractable asymptotic distribution,  we construct a quantity $\hat{\bvartheta}_{i,j}$ to account for the bias induced by thresholding, and consider the asymptotic distribution of $\hb^\T (\hat\ab_{i,j} -\ab_{i,j}- \hat{\bvartheta}_{i,j})$, where
\begin{equation}\label{eq:iterative debias}
    \hb^{\T}(\hat\ab_{i,j} -\ab_{i,j}-\hat{\bvartheta}_{i,j})=\hb^{\T}\bigg\{\frac{\partial \Rb^*_{i,j}(\ab_{i,j})}{\partial \ab^\T}\bigg\}^{-1} \Rb^*_{i,j}(\hat\ab_{i,j})-\hb^{\T}\hat{\bvartheta}_{i,j}+O(|\hb|_2|\hat\ab_{i,j}-\ab_{i,j}|_2^2)\,.
\end{equation}
Notice that
\begin{align}\label{eq: iterative debis expansion}
        		\hb^{\T}\bigg\{\frac{\partial \Rb^*_{i,j}(\ab_{i,j})}{\partial \ab^\T}\bigg\}^{-1} \Rb^*_{i,j}(\hat\ab_{i,j})&= \hb^{\T}\bigg[\frac{\{\hat\ab_{i,j}^{\T}\bSigma_{\yb_{i,j},\xi}(1)\}\hat\ab_{i,j}-\bSigma_{\yb_{i,j},\xi}(1)}{\ab_{i,j}^{\T}\bSigma_{\yb_{i,j},\xi}(1)}\bigg]\, \\
            &\quad +O(|\hb|_2|\hat\ab_{i,j}-\ab_{i,j}|_2^2)\,. \nonumber
\end{align}
To obtain a tractable asymptotic distribution in \eqref{eq:iterative debias}, 
we construct the bias-correction term $ \hat{\bvartheta}_{i,j}$ based on the leading term on the right-hand side of 
\eqref{eq: iterative debis expansion}. Specifically, we replace the unknown quantities 
$\ab_{i,j}$ and $\bSigma_{\yb_{i,j},\xi}(1)$ with their plug-in estimators 
$\hat\ab_{i,j}$ and 
$\tilde\bSigma^{(\textit{v}_{\max},j)}_{\tilde\yb_{i,j},\tilde\xi_i}(1)$, respectively, 
where $\textit{v}_{\max}$ denotes the stopping iteration of Algorithm~\ref{alg1}.
This leads to the following estimator 
\[
 \hat{\bvartheta}_{i,j}=
 \frac{
 \{\hat\ab_{i,j}^{\T} \tilde\bSigma^{(\textit{v}_{\max},j)}_{\tilde\yb_{i,j},\tilde\xi_i}(1)\}
 \hat\ab_{i,j}
 - \tilde\bSigma^{(\textit{v}_{\max},j)}_{\tilde\yb_{i,j},\tilde\xi_i}(1)
 }{
 \hat\ab_{i,j}^{\T} \tilde\bSigma^{(\textit{v}_{\max},j)}_{\tilde\yb_{i,j},\tilde\xi_i}(1)
 }\,.
\]  
Theorem \ref{thm: debias iterative} in Section \ref{sec: iterative theorem} shows that 
 $\sqrt{n}\{w_{i}\bar{\tau}^{-1}_{i,j}(\hb)\} \hb^\T (\hat\ab_{i,j} -\ab_{i,j}- \hat{\bvartheta}_{i,j})$ is asymptotically standard normal for $\bar{\tau}_{i,j}(\hb)$ specified in \eqref{iterative not degenerate}.

Finally, we provide two estimators of the asymptotic variance $w_i^{-2}\bar{\tau}_{i,j}^2(\hb)$ so that statistical inference based on the iterative estimator can be implemented in practice. The estimation of $w_i^{-2} \bar{\tau}_{i,j}^{2}(\hb)$ is essentially a long-run variance estimation problem. Section~\ref{sec:variance iter} in the supplementary material provides an estimator $\hat{w}^{-2}_{i,j}\tilde{\tau}^{2}_{i,j}(\hb)$ for $w_i^{-2} \bar{\tau}_{i,j}^{2}(\hb)$  based on the kernel-type long-run variance estimator $\tilde{\tau}^{2}_{i,j}(\hb)$. The consistency of such kernel-type long-run variance estimator is well-known. See, for example, \cite{andrews1991heteroskedasticity} and \cite{chang2018confidence}. Therefore,
\begin{equation}\label{eq:normality-plugin-longrun}
   \sqrt{n}\, \{\hat w_{i,j}\tilde\tau^{-1}_{i,j}(\hb) \} \hb^\T(\hat\ab_{i,j}-\ab_{i,j}-\hat{\bvartheta}_{i,j})
   \overset{{\rm d}}{\rightarrow}
   \mathcal{N}(0,1)\,.
\end{equation}
Furthermore, if the error process $\{\mathcal{E}_t\}_{t\ge 1}$ is independent of the factor process $\{\mathbf{f}_t\}_{t\ge 1}$, the asymptotic variance $w_i^{-2}\bar{\tau}_{i,j}^2(\hb)$ admits a simple form, which motivates a plug-in estimation method. Section~\ref{sec:variance iter} in the supplementary material further provides such plug-in estimator $\hat{w}^{-2}_{i,j}\hat{\tau}^{2}_{i,j}(\hb)$ for $w_i^{-2}\bar{\tau}_{i,j}^2(\hb)$.
Theorem \ref{thm: estimation of iteration variance} in the supplementary material establishes the consistency of this plug-in estimator, and hence
\begin{equation}\label{eq:normality-plugin}
   \sqrt{n}\, \{\hat w_{i,j}\hat\tau^{-1}_{i,j}(\hb) \} \hb^\T(\hat\ab_{i,j}-\ab_{i,j}-\hat{\bvartheta}_{i,j})
   \overset{{\rm d}}{\rightarrow}
   \mathcal{N}(0,1)\,.
\end{equation}
The simulation results in Table~\ref{table:var-iter-all} in the supplementary material further demonstrate the effectiveness of the proposed estimators for the asymptotic variance.

\subsection{Selection of tuning parameters}\label{sec:tuning}

There are some tuning parameters that need to be determined in our proposed methods. The key quantities include  the number of factors $r$ specified in 
\eqref{model cp}, the linear combination  $\xi_t$ used to construct $\bSigma_{\Yb_j,\xi}(k)$ 
in \eqref{Xi and Sigma Y}, the lag parameter $K$ specified in \eqref{tilde Mj}, and two threshold levels: 
$\delta_1$, used in the one-pass estimation as defined in 
\eqref{hat Sigma kj}, and $\delta_{2,j}$, employed in the 
iterative procedure described in Algorithm \ref{alg1}.
Write $d_{\min}=\min_{j\in[m]}d_j$ and $\Yb = \{\text{vec}(\mathcal{Y}_1),\ldots,\text{vec}(\mathcal{Y}_n)\}^{\T} \in \mathbb{R}^{n \times \prod_{j=1}^m d_j}$. 

First, we determine the lag parameter $K$. As discussed in Remark 3 of \cite{chang2023modelling}, choosing a larger $K$ makes it more likely that the condition $\textup{rank}(\mathbf{M}_j)=r$ holds, since more lagged information is incorporated. On the other hand, as shown in Section 5.1 of \cite{chang2023modelling}, an excessively large $K$ may reduce the estimation accuracy of both the number of factors and the factor loading vectors. Balancing these two considerations, \cite{chang2023modelling} recommend choosing $K \le 10$ and show through simulations that the estimation performance is robust to the choice of $K$ within a moderate range. Our additional simulations, reported in Figures \ref{fig:Krobust-acc} and \ref{fig:Krobust-error} in the supplementary material, further support this recommendation. Specifically, the estimation accuracy improves as $K$ increases initially and then stabilizes, with almost no visible change once $K > 10$. Therefore, in practice, we recommend setting $K = 10$.

Second, we introduce how to determine $r$. When $\xi_t$ is specified, for given $\delta_1$ and $j\in[m]$, \cite{chang2023modelling} employ the eigenvalue-ratio (ER) method to estimate $r$ in the matrix CP-factor model ($m = 2$):
\begin{equation}\label{hat rj}
\tilde r^{(\textup{er})}_j(\delta_1) =\arg\min_{1 \le i \le \lfloor 0.5d_{\min} \rfloor}\frac{ \sigma_{i+1}(\tilde{\mathbf{M}}_{j})+c_n}{ \sigma_i(\tilde{\mathbf{M}}_{j})+c_n}\,,\quad j\in[m]\,,
\end{equation}
where $c_n \rightarrow 0^{\MP}$ as $n\rightarrow\infty$, and $\tilde{\mathbf{M}}_{j}$ is the plug-in estimator of $\Mb_j$ specified above \eqref{tilde K21j}. Such defined ER method has also been used in \cite{chang2015high,chang2018principal,chang2025modeling} for solving other problems. In practice,  we can set $c_n=n^{-1}\hat\sigma_0^2$ with $\hat\sigma_0^2=(n\prod_{j=1}^m d_j)^{-1}\|\Yb\|_{\rm F}^2$. 
Notice that $\textup{rank}(\Mb_j) = r$, and Theorem 1 of \cite{chang2023modelling} implies that $\mathbb{P}\{\tilde r^{(\textup{er})}_j(\delta_1) = r\} \to 1$ as $n \to \infty$ under certain regularity conditions for each $j \in [m]$. For the tensor CP-factor model \eqref{model cp} with more than two modes ($m > 2$), to aggregate the information from the estimators $\tilde r^{(\textup{er})}_j(\delta_1)$ across all modes,  we may consider selecting $\tilde{r}$ as $\max_{j \in [m]}\tilde r^{(\textup{er})}_j(\delta_1)$. 
 However, when the factor loading vectors or the factor processes are highly correlated, the largest eigenvalue of $\tilde{\mathbf{M}}_{j}$ may be inflated relative to the remaining eigenvalues. As pointed out by \cite{brown1989number}, this may lead to the so-called ``one-factor bias'', under which the conventional ER method tends to favor a one-factor model even when the true number of factors is larger than one. This phenomenon is particularly pronounced when the sample size $n$ is small and is also consistent with the simulation results reported in Table \ref{table: rank} in Section \ref{sec:numerical}.
To avoid this issue, we suggest estimating $r$ by $\max_{j \in [m]}\tilde r_j^{(\log)}(\delta_1)$, where 
\begin{equation}\label{hat rj-log}
\tilde r^{(\text{log})}_j(\delta_1)=\arg\min_{1 \le i \le \lfloor 0.5d_{\min} \rfloor}\frac{ \log\{1+\sigma_{i+1}(\tilde{\mathbf{M}}_{j})\}+c_n}{ \log\{1+\sigma_{i}(\tilde{\mathbf{M}}_{j})\}+c_n}\,,\quad j\in[m]\,,
\end{equation}
with the same setting as in \eqref{hat rj}. Table \ref{table: rank} in Section \ref{sec:numerical} shows that the logarithmic eigenvalue-ratio (log-ER) method \eqref{hat rj-log} exhibits better finite-sample performance than the ER method \eqref{hat rj}. Specifically, when there is a high degree of correlation among factor loading vectors, the ER method tends to  underestimate the number of factors, whereas the log-ER method performs stably across all scenarios. Theorem~\ref{thm: factor number consistency} in the supplementary material establishes the consistency of the ER and log-ER estimators. Section~\ref{sec:misspecify r} in the supplementary material further examines the robustness of the proposed estimation procedures in Sections \ref{sec: initial} and \ref{sec: Double projection iterations} to misspecification of $r$.

Next, we consider how to select $\xi_t$. For the special case of the tensor CP-factor model \eqref{model cp} with $m = 2$, \cite{chang2023modelling} suggest selecting $\xi_t$ as the average of the principal components of $\{\text{vec}(\mathcal{Y}_t)\}_{t=1}^n$. Here, we propose a randomized projection approach to select $\xi_t$, which can be viewed as the extension of the method suggested by  \cite{chang2023modelling}. 
 For a prescribed integer $p > 1$, perform PCA on $\Yb$ and then obtain the first $p$ principal components, denoted by $\{\tilde{\eta}_{t,1},\ldots,\tilde{\eta}_{t,p}\}_{t=1}^n$. We then randomly generate a set of $p \times p$ orthonormal matrices $\bOmega^{(1)},\ldots,\bOmega^{(M)}$ and define $\xi^{(l)}_t = p^{-1} \sum_{i = 1}^{p}\eta^{(l)}_{t,i}$ with $(\eta^{(l)}_{t,1},\ldots,\eta^{(l)}_{t,p})^{\T} = \bOmega^{(l)}(\tilde{\eta}_{t,1},\ldots,\tilde{\eta}_{t,p})^{\T}$. Our goal is to choose the optimal candidate from $\xi^{(1)}_t,\ldots,\xi^{(M)}_{t}$ as the final $\xi_t$. For each $l \in [M]$ and a given $\breve r\geq r$, we obtain $\tilde{\ab}_{i,j}(l)$ in the same manner as $\tilde{\ab}_{i,j}$ defined in Section \ref{sec: initial} for $i \in [\breve{r}]$ and $j \in [m]$ but with replacing $(\xi_t,\tilde{r})$ by $(\xi^{(l)}_t,\breve{r})$. 
 Our guiding principle is to choose the index 
 $l$ for which the associated estimates $\{\tilde{\ab}_{i,j}(l)\}_{i\in[\breve{r}],j\in[m]}$ are most similar to the other estimates $\{\tilde{\ab}_{i,j}(\tilde{l})\}_{i\in[\breve{r}],j\in[m],\tilde{l}\neq l}$. 
 For any $l\in[M]$, we consider the measure 
$$D(l) = \sum_{\tilde{l} \neq l}\sum_{i = 1}^{\breve{r}} I\bigg\{\max_{j\in[m]}\min_{\ell \in[\breve{r}]}[1-|\{\tilde{\ab}_{i,j}(l)\}^\T \tilde{\ab}_{\ell,j}(\tilde{l}) |^2] < \varepsilon  \bigg\}
\,,$$
where $\varepsilon>0$ is a prescribed distance threshold. The measure $D(l)$ quantifies the similarity between $\{\tilde{\ab}_{i,j}(l)\}_{i\in[\breve{r}],\,j\in[m]}$ and $\{\tilde{\ab}_{i,j}(\tilde{l})\}_{i\in[\breve{r}],\,j\in[m]}$ with $\tilde{l}\neq l$, where larger values of $D(l)$ indicate higher similarity.  We then select $\xi_t$ as $\xi^{(l^*)}_t$ with $l^* = \arg\max_{l \in [M]} D(l)$. In practice, we set $p = 10$, $\breve r = 2 \tilde{r}^{*}$, $M = 50$ and  $\varepsilon = 0.1$, where $\tilde{r}^{*}$ is an initial estimate of $r$ obtained via the log-ER method \eqref{hat rj-log} with $\xi_t$ selected using the approach proposed in \cite{chang2023modelling}.

Finally, we determine the threshold levels $\delta_{1}$ and $\{\delta_{2,j}\}_{j = 1}^m$. 
Let $\tilde v^{(\textup{log})}_j(\delta_1)$ be the minimal ratio in \eqref{hat rj-log} corresponding to the $j$-th mode for a given $\delta_1$.
We can select $\delta_1$ as
\[
\delta_1 = \arg\min_{0<\delta < 0.1\hat\sigma_0 (n^{-1}\sum_{j = 1}^m\log d_j)^{1/2}} \frac{1}{m}\sum_{j=1}^m \tilde v^{(\textup{log})}_j(\delta)\,.
\]
Additionally, we suggest setting $\delta_{2,j} = \tilde C_* \hat\sigma_0 (n^{-1}\log d_j)^{1/2}$ with some prescribed constant $\tilde C_* \ge 0$. Extensive simulation studies demonstrate that the performance of Algorithm \ref{alg1} introduced in Section \ref{sec: Double projection iterations} with such selected $\delta_{2,j}$ is robust with respect to $\tilde C_* \in [0,1]$. We therefore recommend setting $\tilde C_* = 1$ in practice.



\section{Numerical studies}\label{sec:numerical}

We generate the observations $\{\mathcal{Y}_t\}_{t=1}^n$ via the tensor CP-factor model \eqref{model cp}. 
For each $j \in [m]$, we generate $\Ab^*_j = (\ab^*_{1,j},\ldots,\ab^*_{r,j}) \in \mathbb{R}^{d_j \times r}$ with elements drawn independently from the uniform distribution $U(-1,1)$, subject to the restriction $\text{rank}(\Ab^*_j) = r$, and let $\breve{\ab}_{1,j} = \ab^*_{1,j}$ and $\breve{\ab}_{i,j} = \ab^*_{i,j} + \phi \ab^*_{i-1,j}$ for $2 \le i \le r$. For each $i \in [r]$ and $j\in[m]$, we obtain $\bar{\ab}_{i,j}$ based on $\breve{\ab}_{i,j}$ by randomly setting its $\lfloor s d_j \rfloor$ components to be zero, and let $\ab_{i,j} = \bar{\ab}_{i,j}/|\bar{\ab}_{i,j}|_2$. Here, the parameters $s$ and $\phi$, respectively, control the sparsity of $\ab_{i,j}$ and the correlations among $(\ab_{1,j},\ldots,\ab_{r,j})$. We generate $\{f^*_{t,i}\}_{t = 1}^n$ for $i\in [r]$ as $r$ independent AR(1) sequences, i.e.\ $f^*_{t,i} = \beta_i f^*_{t-1,i} + v_{t,i}$, where the innovations $v_{t,i}$ are independently drawn from the standard normal distribution $\mathcal{N}(0,1)$, and let $(f_{t,1},\ldots,f_{t,r})^{\T} = \Jb^{1/2}(f^*_{t,1},\ldots,f^*_{t,r})^{\T}$, where $\Jb$ is an $r \times r$ matrix with $[\Jb]_{i,j} = I(i  = j) + \rho I(i \neq j)$. Here, the parameter $\rho$ governs the correlation among the factor processes. The elements of the error term sequence $\{\mathcal{E}_t\}_{t=1}^n$ are independently drawn from $\mathcal{N}(0,1)$. We set $m=2$ (matrix time series), $r=3$ (three factors), $w_i = 15$, $d_j = 20$, $\beta_i = 0.85 - 0.05i$, $n \in \{400,800\}$, $s \in \{0,0.3,0.6\}$, $\phi \in \{0.25,0.75\}$ and $\rho \in \{0,0.75\}$. We follow the methods described in Section \ref{sec:tuning} to select the tuning parameters involved in our proposed methods. 

Table \ref{table: rank} compares the performance of two estimation methods (the ER estimator and the log-ER estimator) introduced in Section \ref{sec:tuning} and the unfolded eigenvalue-ratio (Unfolded-ER) estimator considered in \cite{chen2026estimation} for estimating $r$. We can find that log-ER outperforms ER in estimating $r$ across all scenarios, and that, except for the case $(\rho,\phi,s) = (0.75,0.75,0)$, the performance of log-ER is comparable to that of Unfolded-ER. When $(\rho,\phi,s) = (0.75,0.75,0)$, both ER and Unfolded-ER tend to underestimate $r$, whereas log-ER still maintains high accuracy in estimating $r$.

\begin{table}[htbp]
\caption{
Relative frequency estimates of $\mathbb{P}(\tilde{r} < r)$, $\mathbb{P}(\tilde{r} = r)$ and $\mathbb{P}(\tilde{r} > r)$  with $\tilde{r}$ determined by the ER estimator \eqref{hat rj}, the log-ER estimator \eqref{hat rj-log} and the Unfolded-ER estimator based on 2000 repetitions. All numbers reported below are multiplied by 100.
}
\renewcommand\tabcolsep{2pt}
\label{table: rank}
\footnotesize
\centering
\begin{tabular}{c|c|c|c|ccc|ccc|ccc}
\hline\hline
\multirow{2}{*}{$\rho$} & \multirow{2}{*}{$\phi$} & \multirow{2}{*}{$s$} & \multirow{2}{*}{$n$} & \multicolumn{3}{c|}{\textbf{log-ER}}                                                    & \multicolumn{3}{c|}{\textbf{ER}}                                                        & \multicolumn{3}{c}{\textbf{Unfolded-ER}}                                                  \\ \cline{5-13} 
                        &                         &                      &                      & $\mathbb{P}(\tilde{r} < r)$ & $\mathbb{P}(\tilde{r} = r)$ & $\mathbb{P}(\tilde{r} > r)$ & $\mathbb{P}(\tilde{r} < r)$ & $\mathbb{P}(\tilde{r} = r)$ & $\mathbb{P}(\tilde{r} > r)$ & $\mathbb{P}(\tilde{r} < r)$ & $\mathbb{P}(\tilde{r} = r)$ & $\mathbb{P}(\tilde{r} > r)$ \\ \hline
\multirow{12}{*}{0}     & \multirow{6}{*}{0.25}   & \multirow{2}{*}{0}   & 400                  & 0.25                        & 99.75                       & 0.00                        & 6.30                        & 93.70                       & 0.00                        & 0.00                        & 100.00                      & 0.00                        \\
                        &                         &                      & 800                  & 0.65                        & 99.35                       & 0.00                        & 7.55                        & 92.45                       & 0.00                        & 0.00                        & 100.00                      & 0.00                        \\ \cline{3-13} 
                        &                         & \multirow{2}{*}{0.3} & 400                  & 0.40                        & 99.60                       & 0.00                        & 4.80                        & 95.20                       & 0.00                        & 0.00                        & 100.00                      & 0.00                        \\
                        &                         &                      & 800                  & 0.40                        & 99.60                       & 0.00                        & 6.85                        & 93.15                       & 0.00                        & 0.00                        & 100.00                      & 0.00                        \\ \cline{3-13} 
                        &                         & \multirow{2}{*}{0.6} & 400                  & 0.10                        & 99.90                       & 0.00                        & 2.80                        & 97.20                       & 0.00                        & 0.00                        & 100.00                      & 0.00                        \\
                        &                         &                      & 800                  & 0.55                        & 99.45                       & 0.00                        & 6.05                        & 93.95                       & 0.00                        & 0.00                        & 100.00                      & 0.00                        \\ \cline{2-13} 
                        & \multirow{6}{*}{0.75}   & \multirow{2}{*}{0}   & 400                  & 4.55                        & 95.45                       & 0.00                        & 44.60                       & 55.40                       & 0.00                        & 0.00                        & 100.00                      & 0.00                        \\
                        &                         &                      & 800                  & 1.60                        & 98.40                       & 0.00                        & 39.05                       & 60.95                       & 0.00                        & 0.00                        & 100.00                      & 0.00                        \\ \cline{3-13} 
                        &                         & \multirow{2}{*}{0.3} & 400                  & 0.60                        & 99.40                       & 0.00                        & 14.05                       & 85.95                       & 0.00                        & 0.00                        & 100.00                      & 0.00                        \\
                        &                         &                      & 800                  & 0.35                        & 99.65                       & 0.00                        & 11.85                       & 88.15                       & 0.00                        & 0.00                        & 100.00                      & 0.00                        \\ \cline{3-13} 
                        &                         & \multirow{2}{*}{0.6} & 400                  & 0.10                        & 99.90                       & 0.00                        & 4.25                        & 95.75                       & 0.00                        & 0.00                        & 100.00                      & 0.00                        \\
                        &                         &                      & 800                  & 0.20                        & 99.80                       & 0.00                        & 5.20                        & 94.80                       & 0.00                        & 0.00                        & 100.00                      & 0.00                        \\ \hline
\multirow{12}{*}{0.75}  & \multirow{6}{*}{0.25}   & \multirow{2}{*}{0}   & 400                  & 0.10                        & 99.90                       & 0.00                        & 5.30                        & 94.70                       & 0.00                        & 0.10                        & 99.90                       & 0.00                        \\
                        &                         &                      & 800                  & 0.00                        & 100.00                      & 0.00                        & 1.25                        & 98.75                       & 0.00                        & 0.00                        & 100.00                      & 0.00                        \\ \cline{3-13} 
                        &                         & \multirow{2}{*}{0.3} & 400                  & 0.00                        & 100.00                      & 0.00                        & 2.05                        & 97.95                       & 0.00                        & 0.00                        & 100.00                      & 0.00                        \\
                        &                         &                      & 800                  & 0.00                        & 100.00                      & 0.00                        & 0.80                        & 99.20                       & 0.00                        & 0.00                        & 100.00                      & 0.00                        \\ \cline{3-13} 
                        &                         & \multirow{2}{*}{0.6} & 400                  & 0.15                        & 99.85                       & 0.00                        & 1.75                        & 98.25                       & 0.00                        & 0.05                        & 99.95                       & 0.00                        \\
                        &                         &                      & 800                  & 0.00                        & 100.00                      & 0.00                        & 0.95                        & 99.05                       & 0.00                        & 0.00                        & 100.00                      & 0.00                        \\ \cline{2-13} 
                        & \multirow{6}{*}{0.75}   & \multirow{2}{*}{0}   & 400                  & 13.80                       & 86.20                       & 0.00                        & 71.20                       & 28.80                       & 0.00                        & 66.25                       & 33.75                       & 0.00                        \\
                        &                         &                      & 800                  & 0.55                        & 99.45                       & 0.00                        & 17.50                       & 82.50                       & 0.00                        & 37.90                       & 62.10                       & 0.00                        \\ \cline{3-13} 
                        &                         & \multirow{2}{*}{0.3} & 400                  & 0.65                        & 99.35                       & 0.00                        & 15.95                       & 84.05                       & 0.00                        & 7.40                        & 92.60                       & 0.00                        \\
                        &                         &                      & 800                  & 0.00                        & 100.00                      & 0.00                        & 2.55                        & 97.45                       & 0.00                        & 0.50                        & 99.50                       & 0.00                        \\ \cline{3-13} 
                        &                         & \multirow{2}{*}{0.6} & 400                  & 0.15                        & 99.85                       & 0.00                        & 3.45                        & 96.55                       & 0.00                        & 0.50                        & 99.50                       & 0.00                        \\
                        &                         &                      & 800                  & 0.00                        & 100.00                      & 0.00                        & 1.10                        & 98.90                       & 0.00                        & 0.05                        & 99.95                       & 0.00                        \\ \hline\hline
\end{tabular}
\end{table}

We also compare the performance of our proposed one-pass initial estimate (Pro.init) introduced in Section \ref{sec: initial} and iterative estimate (Pro.iter) introduced in Section \ref{sec: Double projection iterations} with the composite PCA method (cPCA) and High-Order Projection Estimator (HOPE) proposed by \cite{han2024cp}, the methods of Randomized Projection PCA (RP-PCA) and  Contemporary Covariance-based Iterative Simultaneous Orthogonalization (CC-ISO) proposed by \cite{chen2026estimation}, and the refined estimate for the matrix CP-factor model (RCP) proposed by \cite{chang2023modelling}. Notice that cPCA, RP-PCA and RCP are one-pass estimates, and HOPE and CC-ISO are iterative estimates. We set the tuning parameter $h = 1$ in cPCA and HOPE, as in the simulation studies of \cite{han2024cp}, and the tuning parameter $K = 10$ in RCP as suggested by \cite{chang2023modelling}. As shown in Section \ref{sec: addtional simulation results in main paper} in the supplementary material, our proposed methods are robust to the selection of $K$. For each method, the estimation error between the obtained estimates $\{\check{\ab}_{i,j}\}_{i\in[\tilde{r}],j\in[m]}$ and 
the true factor loading vectors $\{\ab_{i,j}\}_{i\in[r],j\in[m]}$ is measured by 
\begin{equation}\label{eq:estimation error}
 \psi^2(\{\check{\ab}_{i,j}\}_{i\in[\tilde{r}],j\in[m]}, \{\ab_{i,j}\}_{i\in[r],j\in[m]})  =  \max_{j\in[m]}\max_{\ell\in[r]}\min_{i\in[\tilde{r}]}(1-|\check{\ab}_{i,j}^\T \ab_{\ell, j} |^2)\,,
\end{equation}
where $\tilde{r}$ is the associated estimate of $r$. For methods without a dedicated procedure for estimating $r$, we substitute the value obtained from the log-ER estimator when implementing their methods. As shown in Table \ref{table:rf-all}, when $\rho = 0$, Pro.iter performs comparably to CC-ISO and significantly outperforms the other methods. When $\rho = 0.75$, both CC-ISO and HOPE exhibit poor performance, whereas Pro.iter remains effective across all scenarios. Moreover, Pro.init outperforms all other one-pass estimators in all scenarios. Given $\tilde{r} = r$, we further evaluate the iterative efficiency of Pro.iter against CC-ISO and HOPE. For Pro.iter, we consider three choices of the initialization: Pro.init, cPCA, and RP-PCA. As shown in Figure \ref{fig:iter-step}, the estimation errors of Pro.iter converge to nearly zero in very few iterations across different scenarios, irrespective of the initial estimates used. However, HOPE and CC-ISO  require more steps for iterative convergence. When $\rho=0.75$, the estimation errors of HOPE and CC-ISO cannot converge to zero even after a large number of iterations. This suggests that these two methods break down under such scenarios, whereas our iterative algorithm remains effective. 
Recall that $\rho$ measures the degree of correlation among factors, with larger values corresponding to stronger factor correlations. The simulation results demonstrate that HOPE and CC-ISO perform poorly in scenarios with highly correlated factors. Notice that HOPE proposed by \cite{han2024cp} explicitly requires the uncorrelated factor assumption, i.e. $\mathbb{E}(f_{t,i}f_{t-k,j})=0$ for all $i\neq j$ and $k\ge 1$, while our proposed methods do not rely on this assumption.

\begin{table} 
\scriptsize
\caption{
The averages and standard deviations (in parentheses) of the estimation errors \eqref{eq:estimation error} for different methods based on 2000 repetitions. Bold numbers indicate the smallest average estimation error among all competing methods. All numbers reported below are multiplied by 100.}

\centering
\renewcommand{\arraystretch}{1.15}
\setlength{\tabcolsep}{3pt}
\label{table:rf-all}
 
\begin{tabular}{c|c|c|c|ccc|cccc}
\hline\hline
\multirow{2}{*}{\textbf{$\rho$}} & \multirow{2}{*}{\textbf{$\phi$}} & \multirow{2}{*}{\textbf{$s$}} & \multirow{2}{*}{\textbf{$n$}} & \multicolumn{3}{c|}{\textbf{Iterative estimates}} & \multicolumn{4}{c}{\textbf{One-pass estimates}} \\ \cline{5-11}
 & & & & \textbf{Pro.iter} & \textbf{HOPE} & \textbf{CC-ISO} & \textbf{Pro.init} & \textbf{cPCA} & \textbf{RP-PCA} & \textbf{RCP} \\ \hline
\multirow{12}{*}{0} & \multirow{6}{*}{0.25} & \multirow{2}{*}{0} & 400 & \textbf{0.26} (4.47) & 0.67 (7.42) & 0.75 (7.67) & 4.44 (8.52) & 17.01 (17.16) & 19.28 (17.99) & 31.77 (38.92) \\
 & & & 800 & 0.63 (7.59) & 0.78 (8.34) & \textbf{0.39} (5.52) & 2.80 (8.82) & 14.40 (16.01) & 16.79 (16.27) & 27.06 (37.09) \\ \cline{3-11}
 & & \multirow{2}{*}{0.3} & 400 & \textbf{0.40} (5.87) & 0.93 (9.05) & 0.47 (6.43) & 4.23 (8.86) & 14.94 (17.15) & 16.54 (17.88) & 29.74 (38.40) \\
 & & & 800 & \textbf{0.36} (5.53) & 0.47 (6.23) & 0.48 (6.30) & 2.24 (6.81) & 11.72 (14.86) & 13.47 (15.81) & 26.22 (37.14) \\ \cline{3-11}
 & & \multirow{2}{*}{0.6} & 400 & \textbf{0.12} (2.89) & 0.82 (8.46) & 0.69 (7.64) & 3.46 (6.20) & 13.61 (17.82) & 14.85 (17.92) & 28.83 (37.90) \\
 & & & 800 & 0.50 (6.66) & 0.63 (7.46) & \textbf{0.48} (6.67) & 2.14 (7.65) & 9.99 (15.05) & 10.98 (15.13) & 25.11 (36.79) \\ \cline{2-11}
 & \multirow{6}{*}{0.75} & \multirow{2}{*}{0} & 400 & 1.55 (7.31) & 1.99 (7.91) & \textbf{0.53} (3.88) & 12.15 (15.88) & 32.07 (10.40) & 33.59 (10.33) & 52.67 (31.25) \\
 & & & 800 & \textbf{0.53} (4.39) & 1.35 (6.33) & 0.83 (5.01) & 5.09 (10.28) & 33.49 (10.27) & 34.49 (10.27) & 54.65 (31.55) \\ \cline{3-11}
 & & \multirow{2}{*}{0.3} & 400 & \textbf{0.45} (5.50) & 0.74 (6.73) & 0.54 (5.48) & 6.59 (11.90) & 25.73 (15.45) & 28.57 (14.94) & 40.10 (38.10) \\
 & & & 800 & \textbf{0.27} (4.26) & 0.57 (5.52) & 0.64 (5.76) & 2.92 (7.37) & 26.25 (14.52) & 30.18 (14.13) & 42.32 (39.31) \\ \cline{3-11}
 & & \multirow{2}{*}{0.6} & 400 & \textbf{0.12} (2.62) & 0.72 (7.72) & 0.37 (5.10) & 4.17 (7.60) & 17.99 (17.77) & 19.52 (17.66) & 31.90 (38.83) \\
 & & & 800 & \textbf{0.20} (4.13) & 0.40 (5.54) & 0.30 (4.73) & 2.10 (7.01) & 15.18 (16.13) & 17.47 (16.75) & 28.32 (37.64) \\ \hline
\multirow{12}{*}{0.75} & \multirow{6}{*}{0.25} & \multirow{2}{*}{0} & 400 & \textbf{0.37} (2.51) & 24.74 (37.81) & 27.41 (38.74) & 8.75 (13.39) & 48.35 (14.97) & 49.65 (14.48) & 22.31 (25.08) \\
 & & & 800 & \textbf{0.12} (0.05) & 23.64 (37.38) & 24.29 (37.29) & 4.50 (8.96) & 48.54 (14.26) & 48.89 (13.65) & 21.08 (25.36) \\ \cline{3-11}
 & & \multirow{2}{*}{0.3} & 400 & \textbf{0.22} (0.11) & 27.39 (39.77) & 29.57 (40.26) & 7.80 (12.18) & 49.57 (16.08) & 50.56 (15.15) & 21.85 (25.60) \\
 & & & 800 & \textbf{0.09} (0.04) & 29.66 (40.68) & 29.79 (40.63) & 3.58 (7.33) & 50.31 (15.64) & 51.27 (15.06) & 19.13 (23.97) \\ \cline{3-11}
 & & \multirow{2}{*}{0.6} & 400 & \textbf{0.32} (3.66) & 30.52 (41.50) & 31.97 (41.82) & 6.74 (10.93) & 51.31 (18.01) & 52.42 (16.77) & 20.94 (24.68) \\
 & & & 800 & \textbf{0.08} (0.03) & 32.94 (42.23) & 32.56 (41.97) & 3.15 (6.75) & 51.05 (17.24) & 52.10 (16.47) & 20.31 (25.42) \\ \cline{2-11}
 & \multirow{6}{*}{0.75} & \multirow{2}{*}{0} & 400 & \textbf{4.07} (9.69) & 6.57 (15.32) & 30.07 (21.34) & 21.52 (18.99) & 38.28 (8.65) & 42.07 (9.09) & 27.61 (20.44) \\
 & & & 800 & \textbf{0.34} (1.39) & 4.02 (14.52) & 18.10 (21.57) & 10.20 (13.27) & 38.59 (9.06) & 39.92 (8.32) & 23.73 (20.67) \\ \cline{3-11}
 & & \multirow{2}{*}{0.3} & 400 & \textbf{0.49} (2.52) & 13.68 (28.22) & 17.83 (30.08) & 12.23 (16.41) & 44.22 (10.84) & 45.84 (10.59) & 23.49 (23.76) \\
 & & & 800 & \textbf{0.12} (0.05) & 12.97 (27.65) & 14.70 (28.94) & 5.03 (8.92) & 44.37 (10.81) & 45.29 (10.61) & 20.89 (23.60) \\ \cline{3-11}
 & & \multirow{2}{*}{0.6} & 400 & \textbf{0.28} (2.41) & 22.73 (37.14) & 24.11 (37.76) & 8.39 (13.31) & 49.42 (14.94) & 50.52 (14.37) & 22.16 (25.61) \\
 & & & 800 & \textbf{0.09} (0.04) & 25.04 (38.18) & 25.42 (38.16) & 3.27 (6.39) & 50.41 (14.25) & 50.68 (13.61) & 19.13 (23.68) \\ \hline\hline
\end{tabular}
 
\end{table} 

\begin{figure}[htbp]
\centerline{\includegraphics[width= 12cm]{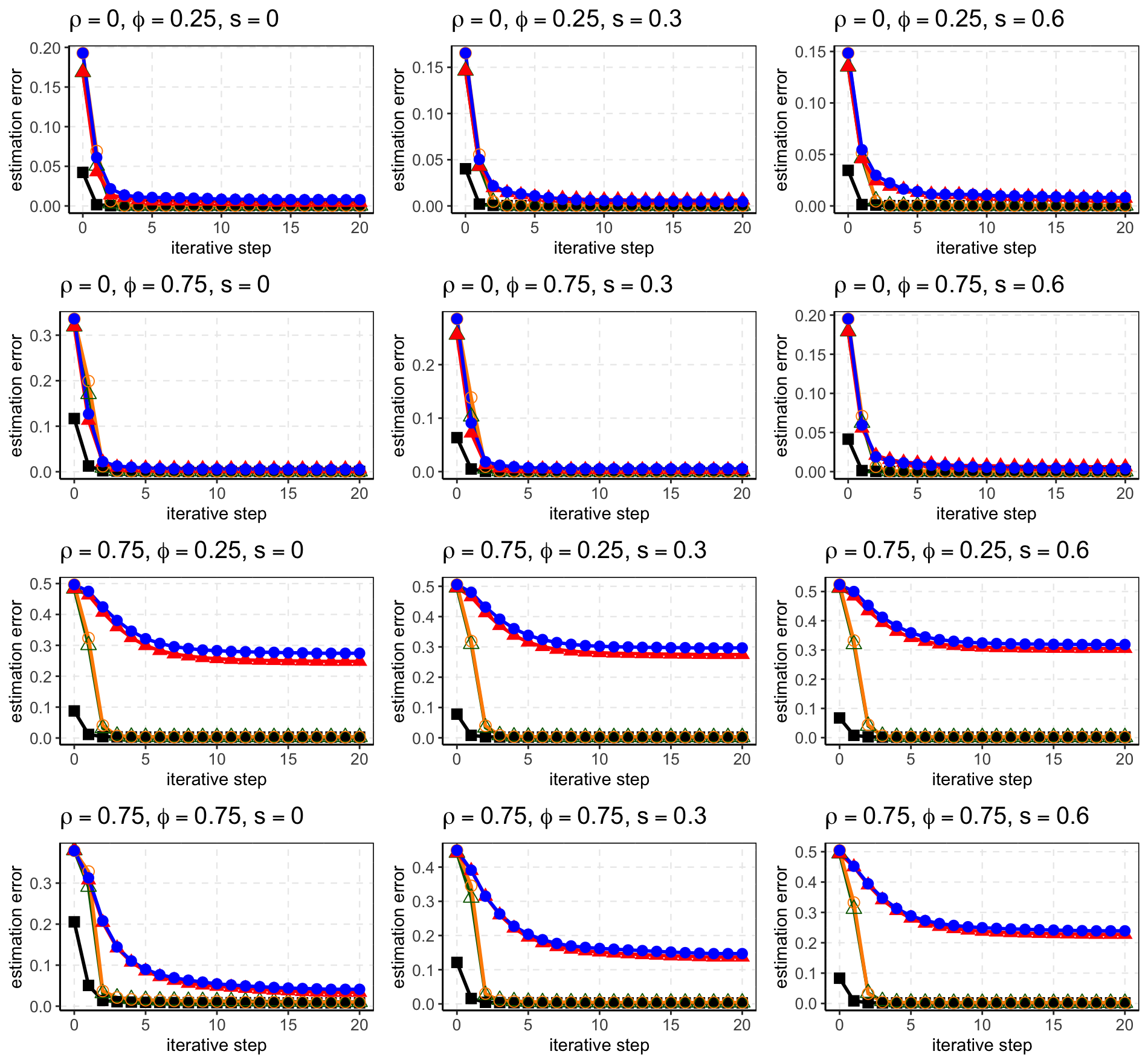}}
\caption{The lineplots for the averages of estimation errors \eqref{eq:estimation error} in the first 20 iterative steps based on 2000 repetitions. The sample size $n = 400$. The legend is defined as follows: (i) Pro.iter initialized with Pro.init (\color{black}{$-$\scalebox{0.75}{$\blacksquare$}$-$}), (ii) Pro.iter initialized with cPCA ($\color{green}{-\vartriangle-}$), (iii) Pro.iter initialized with RP-PCA ({\color{brown}$- \circ-$}), (iv) HOPE ($\color{red}{-\blacktriangle-}$), and (v) CC-ISO ($\color{blue}{-\bullet-}$).}
\label{fig:iter-step}
\end{figure}

We also evaluate the asymptotic normality \eqref{eq:normality-plugin-longrun} and \eqref{eq:normality-plugin} for the iterative estimator $\hat{\ab}_{i,j}$ obtained in Algorithm \ref{alg1} for two choices of $\hb$: (i) $\hb_{1} = (1,0,\ldots,0)^\T$ and (ii) $\hb_{2} = (d_j^{-1/2},\ldots, d_j^{-1/2})^\T$. It should be noted that there exists the reflection and permutation indeterminacy between the estimates and the true factor loadings. Here we set $(i,j) = (1,1)$ and  impose $z_1 = \arg\min_{i \in [\tilde r]}\{1 - |\ab_{1,1}^{\T}\hat{\ab}_{i,1}|^2\}$, thereby eliminating the reflection and permutation indeterminacy between $\hat{\ab}_{z_1,1}$ and $\text{sign}(\ab_{1,1}^{\T}\hat{\ab}_{z_1,1}) \cdot \ab_{1,1}$. We exclude the replications with $\tilde{r} \neq r$ to avoid outliers. 
 Figures \ref{fig:normality-h1-hist-iter-est} and \ref{fig:normality-h2-hist-iter-est} 
present  the histograms of $\{\hat w_{z_1,1}\hat \tau_{z_1,1}^{-1}(\hb_k)\}\sqrt{n}\,\hb_{k}^{\T}\{\hat{\ab}_{z_1,1}-\text{sign}(\ab_{1,1}^{\T}\hat{\ab}_{z_1,1})\cdot \ab_{1,1}-\hat{\bvartheta}_{z_1,1}\}$ for $k \in \{1,2\}$ based on 2000 repetitions, which verify the asymptotic normality of our iterative estimator based on the asymptotic variance estimation $\hat w_{z_1,1}^{-2} \hat\tau_{z_1,1}^{2}(\hb)$.  Figures 
\ref{fig:normality-h1-hist-iter-est-longrun} and 
\ref{fig:normality-h2-hist-iter-est-longrun} in the supplementary material also verify the asymptotic normality of our iterative estimator based on the asymptotic variance estimation 
$\hat w_{z_1,1}^{-2}\tilde\tau_{z_1,1}^{2}(\hb)$.

\begin{figure}[htbp]
\centerline{\includegraphics[width= 12cm]{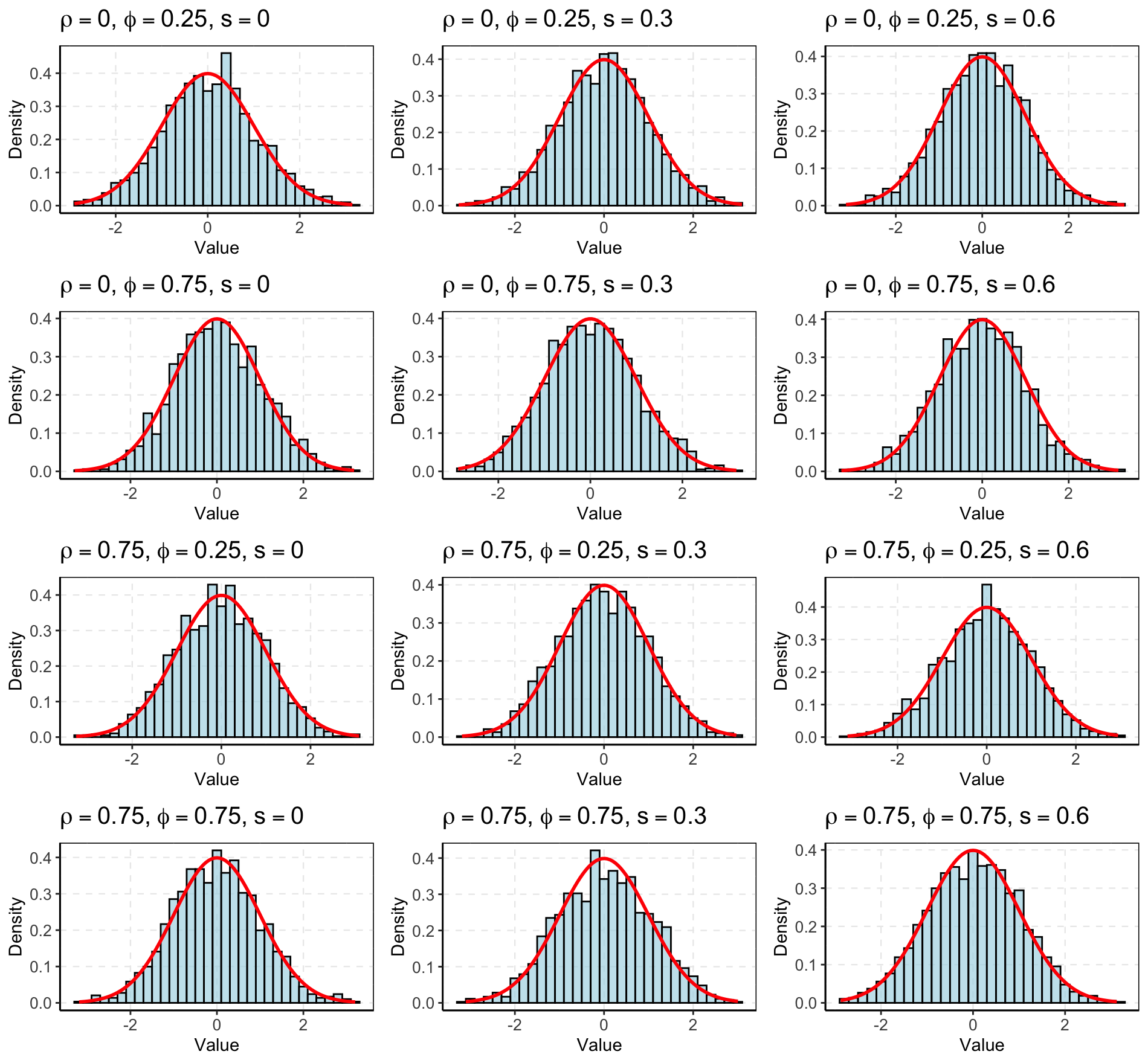}}
\caption{The histograms of $\{\hat{w}_{z_1,1} \hat \tau_{z_1,1}^{-1}(\hb_1)\} \sqrt{n}\, \hb^{\T}_{1}\{\hat{\ab}_{z_1,1} - \text{sign}(\ab_{1,1}^{\T}\hat{\ab}_{z_1,1} ) \cdot \ab_{1,1} - \hat{\bvartheta}_{z_1,1} \}$  based on 2000 repetitions. The sample size $n = 400$. The red curve plots the density of $\mathcal{N}(0,1)$.}
\label{fig:normality-h1-hist-iter-est}
\end{figure}

\begin{figure}[htbp]
\centerline{\includegraphics[width= 12cm]{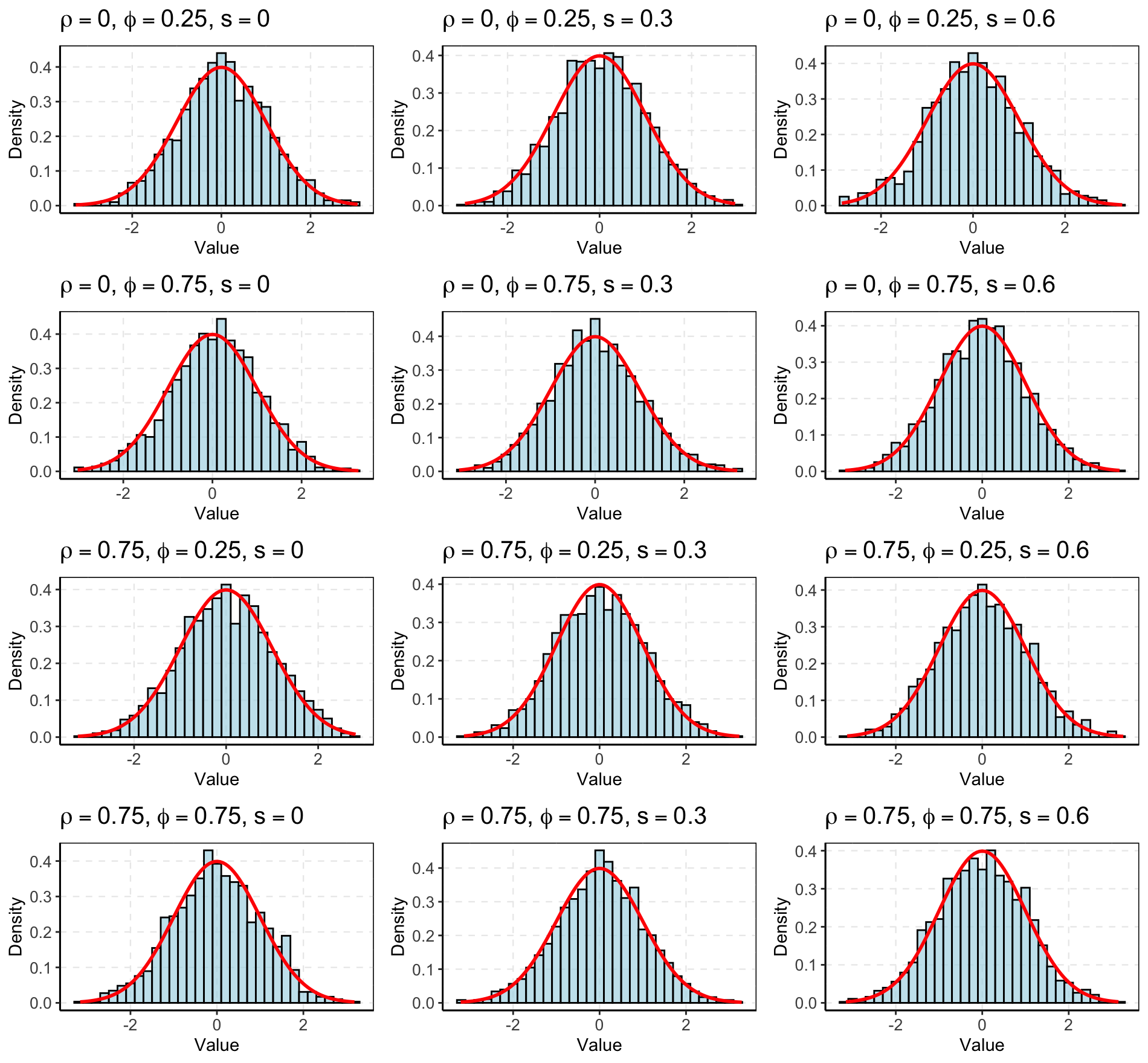}}
\caption{The histograms of $\{\hat{w}_{z_1,1} \hat \tau_{z_1,1}^{-1}(\hb_2)\} \sqrt{n}\, \hb^{\T}_{2}\{\hat{\ab}_{z_1,1} - \text{sign}(\ab_{1,1}^{\T}\hat{\ab}_{z_1,1} ) \cdot \ab_{1,1} - \hat{\bvartheta}_{z_1,1} \}$  based on 2000 repetitions. The sample size $n = 400$. The red curve plots the density of $\mathcal{N}(0,1)$.}
\label{fig:normality-h2-hist-iter-est}
\end{figure}

We finally evaluate the computational speed and cost of the proposed iterative  method initialized with our one-pass estimator. We fix $n=400$ and vary $(d_1,d_2) \in \{(20,20),(40,40),$ $(60,60),(80,80)\}$. 
Across multiple scenarios, we benchmark the proposed Pro.iter (initialized with Pro.init) against HOPE and CC-ISO in terms of runtime and peak RAM, accounting for both initialization and iterative-phase costs.
Figure~\ref{fig:runtime} plots the averages of runtime with standard deviation bands (shaded regions) for the proposed method, HOPE, and CC-ISO across different scenarios, based on 100 replications. Results for peak RAM are similar and can be found in Figure~\ref{fig:peak RAM} in the supplementary material. When the target tensor is low dimensional, the runtime and peak RAM of Pro.iter are comparable to HOPE and CC-ISO. As dimensionality increases, Pro.iter attains markedly shorter runtime and lower peak RAM usage than HOPE and CC-ISO.

\begin{figure}[htbp]
\centerline{\includegraphics[width= 12cm]{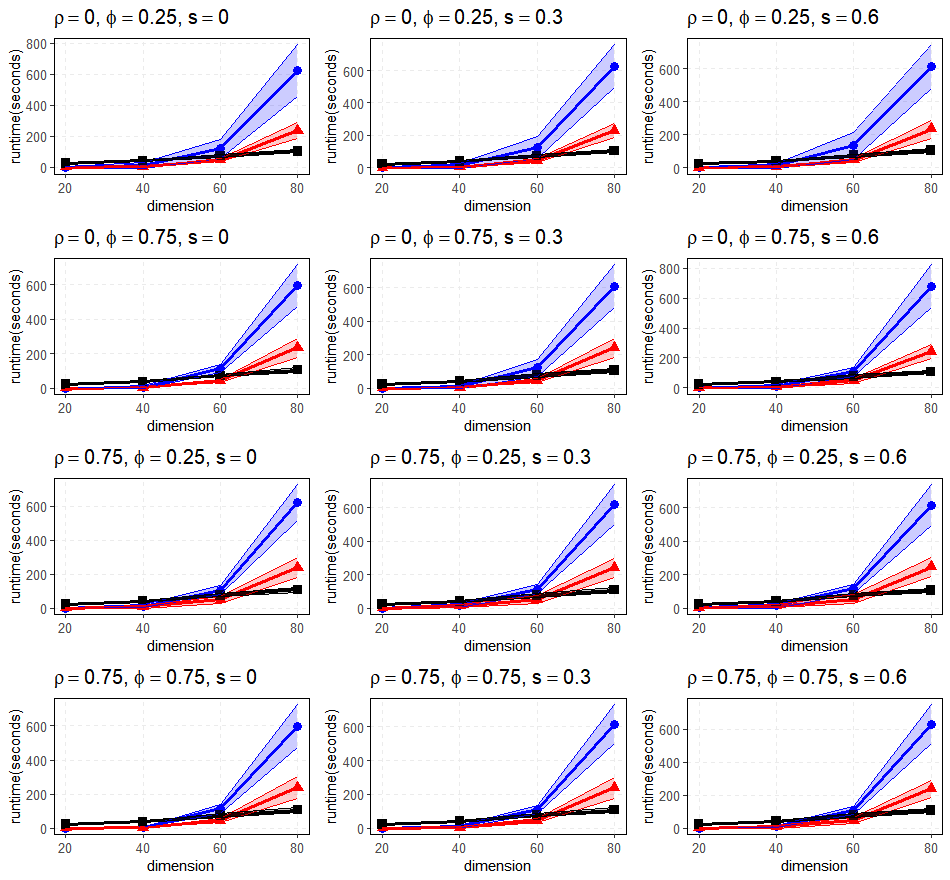}}
\caption{The lineplots for the averages and standard deviations (shaded region) of runtime based on 100 repetitions. The sample size $n = 400$. The legend is defined as follows: (i) Pro.iter initialized with Pro.init (\color{black}{$-$\scalebox{0.75}{$\blacksquare$}$-$}), (ii) HOPE ($\color{red}{-\blacktriangle-}$), and (iii) CC-ISO ($\color{blue}{-\bullet-}$).}
\label{fig:runtime}
\end{figure}



\section{Real data analysis: Air pollution data}\label{sec:application}
In this section, we analyze the spatio-temporal structure of air pollution in Beijing using a multi-dimensional representation of the monitoring data. This dataset contains six hourly air-pollution variables ($\textup{PM}_{2.5},\textup{PM}_{10},\textup{SO}_2,\textup{NO}_2,\textup{CO}$ and $\textup{O}_3$) from 12 nationally controlled air-quality monitoring stations in Beijing, which can be downloaded from \url{https://archive.ics.uci.edu/dataset/501/beijing+multi+site+air+quality+data}. 
The observation period spans from March 1, 2013, to February 28, 2017.

The dataset contains some missing and extreme values, which were handled through interpolation. Since our focus is on the spatio-temporal structure of air-pollution variations, we apply differencing to the hourly pollutant observations for each monitoring site and pollutant type. All series are standardized to remove the impact of different measurement scales. The aforementioned procedures result  in a tensor time series $\mathcal{Y}_t = (y_{t,\ell_1,\ell_2,\ell_3})_{12 \times 6 \times 24}$ for $t \in [1461]$ (i.e. $m=3,d_1 = 12, d_2 =6, d_3 = 24, n = 1461$), where $y_{t,\ell_1,\ell_2,\ell_3}$ records the concentration change of pollutant $\ell_2$ at station $\ell_1$ during the $\ell_3$-th hour of day $t$.  Figure \ref{fig:app-air-timeseries} in the supplementary material shows the time series plots of $\{\mathcal{Y}_t\}_{t=1}^n$.  This tensor representation enables the exploration of multi-way dependencies in Beijing’s air quality data, revealing how pollution intensity co-varies across space, time, and pollutant dimensions.



We use the tensor CP-factor model \eqref{model cp} to fit $\{\mathcal{Y}_t\}_{t=1}^n$, where $\ab_{i,j}$ represents the factor loading vector of the $i$-th factor in the $j$-th mode. To estimate the factor loading vectors based on our proposed method, we set the tuning parameters following Section~\ref{sec:tuning}. Using the proposed log-ER method, we obtain $\tilde{r} = 2$, indicating the presence of two latent factors. Initialized with the one-pass estimator introduced in Section~\ref{sec: initial}, the proposed iterative estimator in Algorithm~\ref{alg1} converges successfully.  

Table \ref{table:app-loading-a2} presents the estimations of the factor loadings $\ab_{i,2} \in \mathbb{R}^6$  based on Pro.iter, which reveal two main patterns of pollutant variation. The first loading vector ($i = 1$) has a very high value for O$_3$ (0.953) but small values for other pollutants, indicating that this factor mainly reflects changes in ozone concentration, which vary differently from other pollutants. 
The second loading vector ($i = 2$) has positive values for PM$_{2.5}$, PM$_{10}$, SO$_2$, NO$_2$, and CO, suggesting a common pollution pattern where several pollutants increase or decrease together. Therefore, we refer to the first estimated factor as the \textit{ozone-related factor}, 
which mainly captures variations driven by O$_3$, 
and the second as the \textit{general pollution factor}, 
representing the joint fluctuation of multiple pollutants.

\begin{table}[htbp]
 
\centering
\renewcommand{\arraystretch}{1.15}
\setlength{\tabcolsep}{6pt}

\caption{Estimations of  the loading vectors $\ab_{i,2}\in\mathbb{R}^6$ for the pollution-variable mode based on Pro.iter. Standard errors reported in parentheses  are calculated based on the asymptotic variance estimation $\hat{w}_{i,j}^{-2}\hat \tau_{i,j}^{2}(\hb)$. $^{*}$, $^{**}$, and $^{***}$ indicate significance at the levels 5\%, 1\%, and 1\textperthousand, respectively, based on two-sided $t$-tests. }
\label{table:app-loading-a2}
\begin{tabular}{c|cc}
\hline\hline
Pollutant & $i=1$ & $i=2$ \\
\hline
PM$_{2.5}$ & 0.008 (0.015) & 0.659$^{***}$ (0.035) \\
PM$_{10}$  & $-$0.021 (0.013) & 0.430$^{***}$ (0.025) \\
SO$_2$     & 0.049$^{**}$ (0.017) & 0.304$^{***}$ (0.036) \\
NO$_2$     & $-$0.236$^{***}$ (0.016) & 0.289$^{***}$ (0.053) \\
CO         & 0.182$^{***}$ (0.012) & 0.452$^{***}$ (0.031) \\
O$_3$      & 0.953$^{***}$ (0.002) & 0.009 (0.083) \\
\hline\hline
\end{tabular}
\end{table}
 




\begin{figure}[htbp]
    \centering
\subfigure[ozone-related factor $(i = 1)$]{\includegraphics[width=0.45\textwidth]{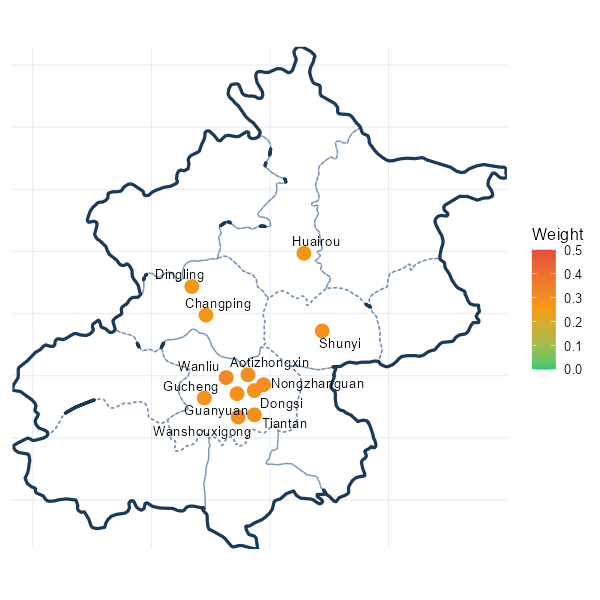}}
\subfigure[general pollution factor $(i = 2)$]{\includegraphics[width=0.45\textwidth]{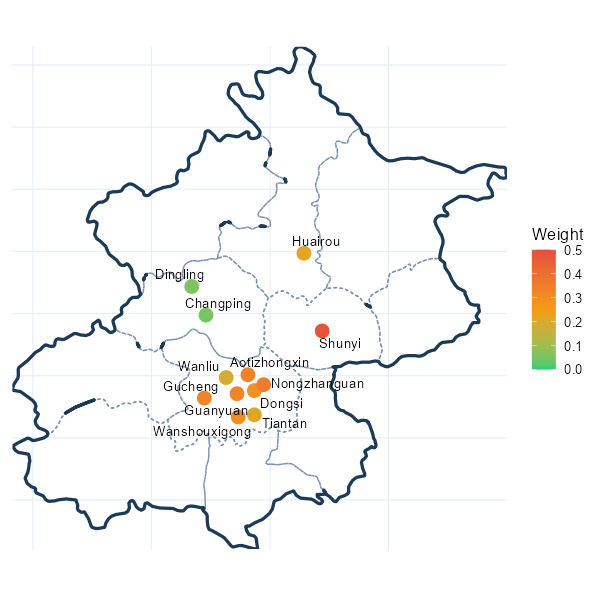}}
    \caption{Estimations of the  loading vectors $\ab_{i,1} \in \mathbb{R}^{12}$ for the monitoring-station mode based on Pro.iter.}
    \label{fig:beijing-station}
\end{figure}

Figure \ref{fig:beijing-station} illustrates the estimations of the factor loadings $\ab_{i,1} \in \mathbb{R}^{12}$ based on Pro.iter for the monitoring-station mode. Figure \ref{fig:beijing-station}(a) shows the factor loadings of the \textit{ozone-related factor} to be nearly uniform across stations, indicating the dominance of regional photochemistry and synoptic meteorology rather than local emissions. Figure \ref{fig:beijing-station}(b) reveals pronounced spatial heterogeneity for the \textit{general pollution factor}: Dingling and Changping (northern mountains) have small loadings due to sparse population, limited sources, and effective ventilation, whereas Shunyi (eastern plain, downwind of the urban core and influenced by airport, traffic, and industry) exhibits the largest loading, consistent with higher emissions and advective transport. Moreover, as reported in Table \ref{table:app-loading-a1-all} in the supplementary material, all estimated loadings based on Pro.iter for the \textit{ozone-related factor} are statistically significant, while for the \textit{general pollution factor}, all estimated loadings are statistically significant except those corresponding to Dingling and Changping. These results provide additional inferential support for the corresponding spatial interpretation. 

\begin{figure}[htbp]
 
    \centering
\subfigure[ozone-related factor $(i = 1)$]{\includegraphics[width=0.45\textwidth]{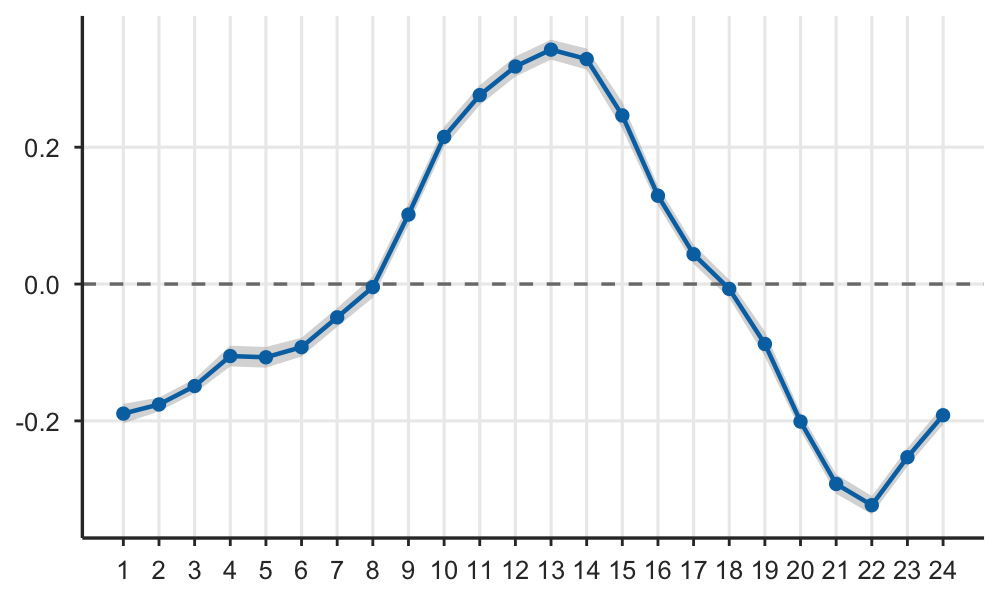}}
\subfigure[general pollution factor $(i = 2)$]{\includegraphics[width=0.45\textwidth]{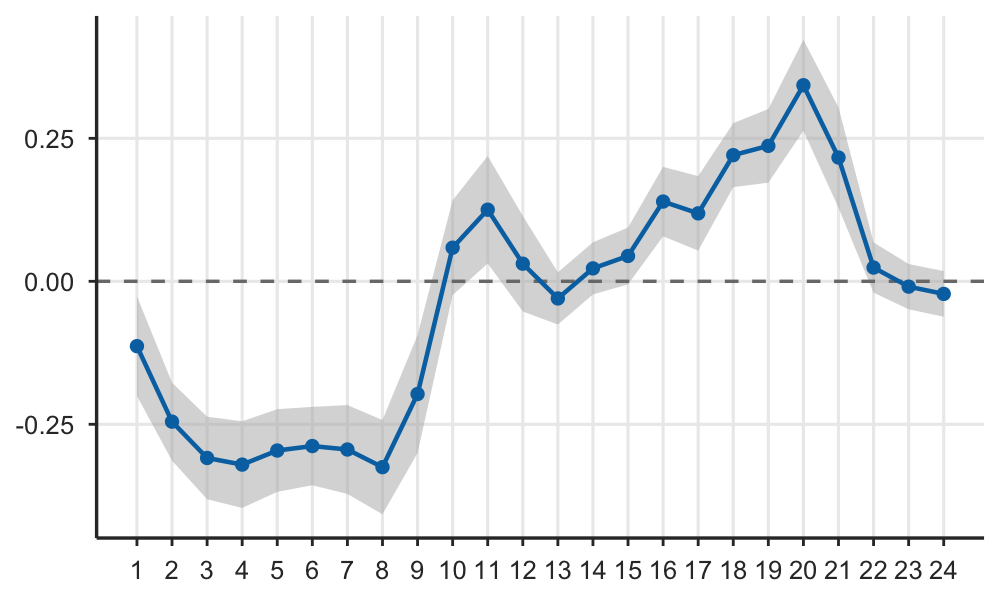}}
    \caption{Estimations of the  loading vectors $\ab_{i,3} \in \mathbb{R}^{24}$ for the diurnal mode based on Pro.iter. The gray shaded region represents the pointwise 95\% confidence interval for the estimated loadings.  Standard errors are calculated based on the asymptotic variance estimation  $\hat{w}_{i,j}^{-2}\hat \tau_{i,j}^{2}(\hb)$.} 
    \label{fig:beijing-hourly}

\end{figure}

Figure \ref{fig:beijing-hourly} illustrates the estimations of the factor loadings $\ab_{i,3}\in\mathbb{R}^{24}$ based on Pro.iter for the diurnal mode. The loading vector  of the \textit{ozone-related factor} rises after sunrise, peaks around 13:00–14:00, and turns negative at night, tracking the canonical photochemical cycle of daytime production and nocturnal loss via deposition \citep{li2015diurnal}. By contrast, the \textit{general pollution factor} is distinctly bimodal—minimal before dawn, a first peak near 10:00–11:00, and a higher evening peak (around 19:00–20:00).  This bimodal profile accords with established diurnal emission and mixing cycles: a morning peak from traffic and industrial start-up, and a higher evening peak from rush-hour emissions combined with boundary-layer stabilization that suppresses dispersion. Analogous morning–evening bimodality for urban aerosols (e.g., $\textup{PM}_{2.5}$ and $\textup{PM}_{10}$ in Beijing) is well documented \citep{Liu2015}. Overall, the two factors capture complementary diurnal dynamics—one driven by photochemical reactions (ozone-related) and the other by human emission activities (general pollution).

\begin{figure}[htbp]
    \centering
\subfigure[ozone-related factor $(i = 1)$]{\includegraphics[width=0.45\textwidth]{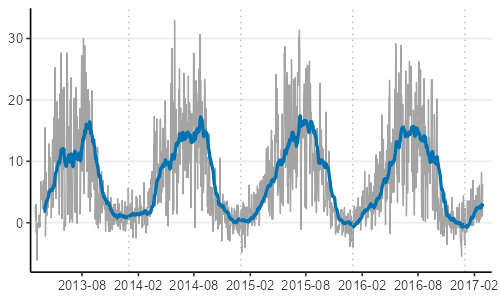}}
\subfigure[general pollution factor $(i = 2)$]{\includegraphics[width=0.45\textwidth]{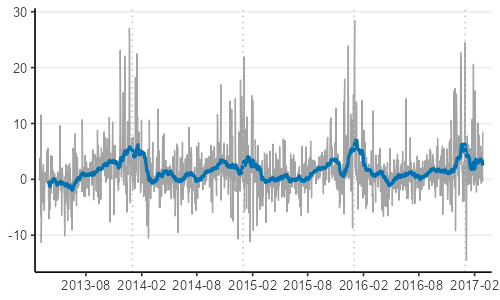}}
    \caption{The time series plots  of the two estimated latent factors based on Pro.iter. The dark blue solid line represents the 30-day one-sided simple moving average.}
    \label{fig:beijing-seasonal}
\end{figure}

Figure \ref{fig:beijing-seasonal} presents the time series of the two estimated latent factors based on Pro.iter; see Section \ref{sec:factor and cp estimation} in the supplementary material for details on the factor estimation procedure. Figure \ref{fig:beijing-seasonal}(a) shows the \textit{ozone-related factor}, which exhibits a distinct seasonal cycle with higher values in summer and lower values in winter. This pattern aligns with ozone's photochemical formation mechanism, which is strongly dependent on solar radiation and temperature. Figure \ref{fig:beijing-seasonal}(b) depicts the \textit{general pollution factor}, displaying the opposite seasonal pattern—higher values in cold seasons and lower values in warm seasons. This arises because cold months see increased coal/fuel combustion for heating (releasing more fine particles and gases) and stable atmospheric conditions; in contrast, summer features stronger air movement, higher wind speeds, and frequent rainfall, which stir, disperse, and scavenge pollutants to improve air quality.

Overall, by representing the multi-site, multi-pollutant, and hourly observations as a tensor time series, the proposed CP-factor estimation method successfully identifies two dominant latent components that capture distinct physical mechanisms underlying Beijing’s air pollution. The \textit{ozone-related factor} reflects photochemical processes, whereas the \textit{general pollution factor} represents anthropogenic emission activities. Together, these two factors provide a concise and interpretable characterization of the complementary seasonal dynamics of photochemically and emission-driven pollution in Beijing.
Section~\ref{sec: addtional empirical results in main paper} in the supplementary material reports the results of the real data analysis based on the other three methods (Pro.init, HOPE, and CC-ISO). The findings suggest  that the estimates of Pro.iter are more interpretable and more consistent with well-established pollution mechanisms than those produced by Pro.init, HOPE, and CC-ISO.  To assess the reliability of the above empirical results, we also show in Section~\ref{sec: addtional empirical results in main paper} in the supplementary material that these results remain essentially unchanged when using the winsorized data, indicating the robustness of our  conclusions.

Section \ref{sec: app-famafrench} in the supplementary material provides another real data analysis for financial data. It is well known that financial data exhibit strong cross-sectional dependence, and often involve highly correlated latent factors.  Table \ref{table:app-forecast} in the supplementary material reports the average forecasting errors for financial returns based on different methods. It can be observed that the tensor CP-factor methods without uncorrelated factor assumption consistently outperform the methods that rely on this assumption, which provides further evidence for the applicability of our proposed methods in practice.

\section{Theoretical analysis}\label{sec: theoretical}
\subsection{Assumptions}\label{sec: asmp}

 We first present some technical assumptions for our theoretical analysis. Assumptions \ref{tail}--\ref{gap new} are imposed to guarantee the consistency of the one-pass estimator introduced in Section \ref{sec: initial}. 
Given a general consistent initial estimator, Assumptions \ref{tail}--\ref{sparsity} and \ref{cross} are required to establish the theoretical guarantees for the associated iterative estimator introduced in Section \ref{sec: Double projection iterations}. 
    
	\begin{assumption}\label{tail}
		There exist universal constants $C_1>1$, $C_2>0$, and  $c_1\in(0,2]$  such that $\max_{i\in[r]}\max_{t\in[n]}\mathbb{P}(|f_{t,i}|>x)\le C_1\exp(-C_2x^{c_1})$, $\max_{j\in[m]}\max_{p_j\in[d_j]}\max_{t\in[n]}\mathbb{P}(|[\mathcal{E}_t]_{p_1,\ldots,p_m}|>x)\le C_1\exp(-C_2x^{c_1})$, and $\max_{t\in[n]}\mathbb{P}(|\xi_{t}|>x)\le  C_1\exp(-C_2x^{c_1})$
		for  any $x>0$.  There also exists a universal constant $C_3>0$ such that  $w_1\ge\cdots \ge w_r\ge C_3$. 
	\end{assumption}
	
	\begin{assumption}\label{mixing}
		Define the $\alpha$-mixing coefficients associated with the factors and error processes $\{f_{t,1},\ldots,f_{t,r},\mathcal{E}_t\}_{t\geq1}$ as
		$$\alpha(k)=\sup_t\sup_{A\in \mathcal{F}_{-\infty}^t,\,B\in\mathcal{F}_{t+k}^\infty}|\mathbb{P}(A\cap B)-\mathbb{P}(A)\mathbb{P}(B)|\,, ~~~k \ge 1\,,$$ 
		where $\mathcal{F}_{t_1}^{t_2}$ is the $\sigma$-field generated by $\{f_{t,1},\ldots,f_{t,r},\mathcal{E}_t\}_{t=t_1}^{t_2}$.  There exist some universal constants $C_4>1$, $C_5> 0$ and $c_2\in(0,1]$ such that $\alpha(k)\le C_4\exp(-C_5 k^{c_2})$ for any $k\ge 1$. 
	\end{assumption}

\begin{assumption}\label{sparsity}
		There exists a universal constant $C_6>1$ such that $C_6^{-1}\le \sigma_r(\Ab_j)\le \sigma_1(\Ab_j)\le C_6$ for any $j\in[m]$.  It holds that  
		$ \max_{i\in[r]}|\ab_{i,j}|_0 \le s_{j}$ for any $j \in [m]$. 
	\end{assumption}
    
	  As pointed out in \cite{chang2023modelling,chang2021central}, Assumptions \ref{tail} and \ref{mixing} are standard in the literature on high-dimensional data analysis, which are satisfied for a wide range of time series models. These assumptions ensure exponential-type upper bounds for the tail probabilities of the statistics concerned. Assumption \ref{tail} focuses on the cases where all the factors, idiosyncratic errors, and the linear combination $\xi_t$ have exponentially decaying tails. Assumption \ref{mixing} is a standard $\alpha$-mixing condition allowing the data to be serially dependent but not necessarily stationary.  As shown in \eqref{Xi and Sigma Y}, our procedure is based on lag-$k$ auto-covariances averaged over the sampling periods, rather than on a fixed stationary auto-covariance structure. Therefore, the factor process need not be stationary where $\mathbb{E}(f_{t,i}^2)$ is allowed to vary with $t$.  
     Assumption \ref{sparsity} can accommodate both sparse and dense loadings within a unified formulation. For example, if $s_j \ll d_j$, Assumption \ref{sparsity} corresponds to the sparse case; if $s_j=d_j$, it covers the dense case. Sparse loadings arise naturally in several important settings in the recent factor model literature \citep{uematsu2022estimation}.
     Properly handling the sparsity via thresholding can improve estimation efficiency. We only require that $\sigma_1(\Ab_j), \ldots, \sigma_r(\Ab_j)$ are uniformly bounded away from $0$ and $\infty$ for $j \in [m]$, which is weaker than the requirement $\max_{j \in [m]}\|\Ab_j^{\T} \Ab_j - \mathbf{I}_r\|_2 < 1$ imposed in \cite{han2024cp} and \cite{chen2026estimation}.

	\begin{assumption}\label{eigenvalue}
		Define  $\ubar\sigma_{\xi}^2=\min_{i\in[r],j\in[m]}\sigma_{i}(\mathbf{M}_{j})$ and $\bar\sigma_{\xi}^2=\max_{i\in[r],j\in[m]}\sigma_{i}(\mathbf{M}_{j})$. Assume that $\ubar\sigma_{\xi}\ge C_7$ for some universal constant $C_7>0$ and $ n^{-1/2}\bar\sigma_{\xi}w_1\ll\ubar\sigma_{\xi}^2$. 
	\end{assumption}
	
We can regard $\xi_{t}$ as a projection of $\text{vec}(\mathcal{Y}_t)$ to a lower dimension. Certainly, we expect that the projection can retain the signal of the factor process, which is guaranteed by Assumption \ref{eigenvalue}. Specifically, by Theorem 7 of \cite{horn2026positivity}, we have 
$$\ubar\sigma_{\xi}^2 \ge \bigg(\min_{i \in [r]} \sum_{k = 1}^K |g_{k,i,\xi}|^2 \bigg) \bigg\{\min_{j \in [m]} \frac{\sigma_{r}^2(\Ab_j)\sigma_{r}^2(\Bb_j)}{\sigma_{1}^{2}(\Ab_j)}\bigg\}\,.$$
Therefore,  for each $i \in [r]$, if there exists some $k \in [K]$ such that $|g_{k,i,\xi}|$ is bounded away from $0$, then it follows from Assumption \ref{sparsity} that $\ubar\sigma_{\xi}$ is bounded away from $0$. 
According to the definition of  $\mathbf{M}_{j}$, we have $\ubar\sigma_{\xi} \lesssim \max_{k\in[K],i\in[r]}|g_{k,i,\xi}|\lesssim  \bar\sigma_{\xi} \lesssim  w_1$. If $w_1,\ldots,w_r$ are fixed constants, Assumption \ref{eigenvalue} holds when both $\bar\sigma_{\xi}$ and $\ubar\sigma_{\xi}$ are uniformly bounded away from $0$ and $\infty$. 
    For the toy example mentioned in Section \ref{sec: model} with $\beta\ne 0$, if we select $\xi_t=(\prod_{j=1}^{m}d_j)^{-1}\sum_{h_1 = 1 }^{d_1}\cdots \sum_{h_m = 1 }^{d_m} [\mathcal{Y}_t]_{h_1,\ldots,h_m}$, then $\bar\sigma_{\xi}\asymp \ubar\sigma_{\xi}\asymp w_1\asymp (\prod_{j=1}^{m}d_j)^{1/2}$ provided that $|(n-k)^{-1}\sum_{t = k+1}^n\mathbb{E}[\{f_t - \mathbb{E}(\bar{f})\}\{f_{t-k} - \mathbb{E}(\bar{f})\}]|$ with $\bar{f} = n^{-1}\sum_{t=1}^n f_t$ is uniformly bounded away from $0$ and $\infty$, which implies that Assumption \ref{eigenvalue} holds automatically. For more general scenarios,  as long as $\xi_{t}$ is properly selected, it is expected that $\ubar\sigma_{\xi}\asymp w_r$ and $\bar\sigma_{\xi}\asymp w_1$. Then, Assumption \ref{eigenvalue} will hold when $w_1^2\ll w_r^2 \sqrt{n}$, which is a requirement on the relative strength of the factors. A similar condition also appears in Theorem 2 of \cite{han2024cp}. 

    \begin{assumption}\label{gap new}
		All the eigenvalues $\bar\lambda_1,\ldots,\bar\lambda_r$  in \eqref{generalized eigenequation} are   uniformly bounded away from $0$ and $\infty$. Moreover, $\min_{i \neq \ell}|\bar\lambda_i - \bar\lambda_{\ell}| \ge C_8$ and $\min_{i\in[r]}|g_{2,i,\xi}|\ge C_8\ubar\sigma_{\xi}$ for some universal constant $C_8>0$, where $g_{2,i,\xi}$ is defined in \eqref{gki}.
	\end{assumption}
	To identify the $r$ eigenvectors of $\Kb_{1,2,j}$ defined in \eqref{Kbj} corresponding to its $r$ nonzero eigenvalues, it is crucial that these eigenvalues are distinct. This is a typical assumption in the literature on eigen-analysis; see also Condition 5 in \cite{chang2023modelling}.

	\begin{assumption}\label{cross} For any deterministic vector $\bbeta\in\mathbb{R}^{\prod_{j=1}^{m}d_j}$, it holds that
		\[
		\max_{t\in[n]}\mathbb{P}\{|\bbeta^\T{\rm vec}(\mathcal{E}_t)|>x|\bbeta|_2\}\le C_1\exp(-C_2x^{c_1})
		\] 
		for any $x > 0$, where $C_1$, $C_2$, and $c_1$ are the same constants as those in Assumption \textup{\ref{tail}}.
	\end{assumption}
	Assumption \ref{cross} provides a tail bound for any linear combination  of the idiosyncratic error tensor. It allows for cross-sectional dependence among the errors. This assumption can hold under very general scenarios, e.g., when $\textup{vec}(\mathcal{E}_t)$ follows a multivariate Gaussian distribution with a covariance matrix bounded in spectral norm; see also Assumption 1 in \cite{han2024cp}.
    
	\subsection{Theoretical guarantees of the proposed methods}\label{sec: iterative theorem}
	
	Let $D_n=\prod_{j=1}^{m}d_j$ and  $S_{n}=\prod_{j=1}^{m}s_{j}$, and further define 
 \begin{equation*} 
 \Pi_{n} = \frac{\bar\sigma_{\xi} }{\ubar\sigma_{\xi}^2}\bigg(\frac{S_n \log D_n}{n}\bigg)^{1/2} \,.
 \end{equation*}
Set the threshold level $\delta_1=C_* (n^{-1} \log D_n)^{1/2}$ in \eqref{hat Sigma kj} for some constant $C_*>0$.  
    Theorem \ref{thm: aij} shows the consistency (up to the reflection and permutation indeterminacy) of the one-pass estimator $\{\tilde\ab_{i,j}\}_{i\in[\tilde{r}],j\in[m]}$ 
    introduced in Section \ref{sec: initial}.

	\begin{theorem}\label{thm: aij}
		Under Assumptions \textup{\ref{error}}--\textup{\ref{gap new}}, if $\Pi_{n} \ll 1$ and $\tilde r=r$,  there  exists a permutation of $[r]$, denoted by $\{z_1,\ldots,z_r\}$, such that
		$$|\tilde\ab_{z_i,j}- \tilde\kappa_{i,j} \ab_{i,j}|_2  =  \ubar\sigma_{\xi}\,\bar\sigma_{\xi}^{-1}\, O_{\rm p}(\Pi_{n})$$ for any  $i\in[r]$ and $j\in[m]$ with some constants $ \tilde\kappa_{i,j} \in\{-1,1\}$, provided that $\log D_n\ll n^c$ for some constant $c\in(0,1)$ depending only on $c_1$ and $c_2$ specified in Assumptions {\rm\ref{tail}} and {\rm\ref{mixing}}.   
	\end{theorem}

 	In Algorithm \ref{alg1}, the estimated factor series $\{\check f_{t,i}^{(\textit{v},j)}\}_{t=1}^n$ can be regarded as linear combinations of $\{\text{vec}(\mathcal{Y}_t)\}_{t=1}^n$ with some plug-in estimators of the coefficients $(\ab_{i,m}^{\MP}\otimes \cdots\otimes\ab_{i,1}^{\MP})$, where $(\ab_{1,j}^{\MP},\ldots,\ab_{ r,j}^{\MP})^{\T} = (\Ab_j^\T\Ab_j)^{-1}{\Ab}_j^{\T}$. To study the statistical error of the iterative estimator, we define the oracle linear combinations as $\xi_{t,i}=(\ab_{i,m}^{\MP}\otimes \cdots\otimes\ab_{i,1}^{\MP})^\T \text{vec}(\mathcal{Y}_t)=w_if_{t,i}+(\ab_{i,m}^{\MP}\otimes \cdots\otimes\ab_{i,1}^{\MP})^\T \text{vec}(\mathcal{E}_t)$, and write  $$\xi_{t,i}^{\textup{s}}=\bigg[\frac{1}{n}\sum_{s=1}^n\mathbb{E}\{(\xi_{s,i}-\bar\xi_i)^2\}\bigg]^{-1/2}(\xi_{t,i}-\bar\xi_i)$$ with $\bar \xi_i=n^{-1}\sum_{t=1}^n \xi_{t,i}$. If $r=1$, let $\xi_{t,i}^{\textup{sp}}=\xi_{t,i}^{\textup{s}}$ for $t\in[n]$ and $i\in[r]$. If $r\ge 2$, let $\bxi_i^{\textup{s}} = (\xi_{1,i}^{\textup{s}},\ldots,\xi_{n-1,i}^{\textup{s}})^{\T}$, and $\Fb_{\xi,\mminus i}^{\textup{s}}$ be a $(n-1)\times (r-1)$ matrix of which the columns are composed of $(\xi_{2,\ell}^{\textup{s}},\ldots,\xi_{n,\ell}^{\textup{s}})^{\T}$ for $\ell \ne i$. Then, following the double projection step, we define
\begin{equation}\label{xi it sp}
    	(\xi_{1,i}^{\textup{sp}},\ldots,\xi_{n-1,i}^{\textup{sp}})^{\T}=\bxi_i^{\textup{s}}-\Fb_{\xi,\mminus i}^{\textup{s}}[\mathbb{E}\{(\Fb_{\xi,\mminus i}^{\textup{s}})^{\T}\Fb_{\xi,\mminus i}^{\textup{s}}\}]^{-1}\mathbb{E}\{(\Fb_{\xi,\mminus i}^{\textup{s}})^{\T}\bxi_i^{\textup{s}}\}\,.
\end{equation}
Let $\bar\bvarphi_i=(\bar\varphi_{i,1},\ldots,\bar\varphi_{i,r})^{\T}$ be the $r$-dimensional vector with the $i$-th entry equal to 1, while the remaining $r-1$ entries form the vector $-[\mathbb{E}\{(\Fb_{\xi,\mminus i}^{\textup{s}})^{\T}\Fb_{\xi,\mminus i}^{\textup{s}}\}]^{-1}
    \mathbb{E}\{(\Fb_{\xi,\mminus i}^{\textup{s}})^{\T}\bxi_i^{\textup{s}}\}$. Then, $\xi_{t,i}^{\textup{sp}}=\xi_{t,i}^{\textup{s}}+\sum_{\ell\ne i}\bar\varphi_{i,\ell}\xi_{t+1,\ell}^{\textup{s}}$ for $t\in[n-1]$.




Set the threshold levels $\delta_{2,j}=\tilde C_*(n^{-1}\log d_j)^{1/2}$ for $j\in[m]$ in Algorithm \ref{alg1} with some sufficiently large constant $\tilde C_*>0$. Write     
	\[
\Phi_{n,j}=\frac{1 }{w_r}\sqrt{\frac{s_j\log d_j}{n}}~~\textrm{and}~~\bUpsilon_k = (\Upsilon_{k,i,\ell})_{r \times r}\,,
	\]
    where $\Upsilon_{k,i,\ell}= (n-k)^{-1} \sum_{t=k+1}^n\mathbb{E}\{(f_{t,i}-\bar f_i)(f_{t-k,\ell}-\bar f_{\ell})\}$ for $i,\ell\in[r]$ and $k\in \{0,1\}$. Let   
\begin{equation}\label{auto cross f xi}
\begin{split}
    &\gamma_{\max}=\max_{i\ne \ell}|\Upsilon_{1,i,\ell}|\,, ~~ \sigma_{f_i,\xi_i}  =  \mathbb{E}\bigg\{\frac{1}{n-1}\sum_{t=2}^n(f_{t,i}-\bar f_i)\xi_{t-1,i}^{\textup{sp}}\bigg\}\,,\\
   & ~~~~~~\textup{and}~~   L_n= \bigg(\frac{\sum_{j=1}^m d_j\log d_j}{n}\bigg)^{1/2}+\frac{(\sum_{j=1}^m d_j)^{1/\tilde c}}{n} \,,
\end{split} 
		\end{equation} 
        where $\tilde c^{-1}=1+2c_1^{-1}+c_2^{-1}$ for $c_1$ and $c_2$ in Assumptions \ref{tail} and \ref{mixing}. Theorem \ref{thm: iterative} gives the convergence rate of the iterative estimator obtained by Algorithm \ref{alg1}.
	\begin{theorem}\label{thm: iterative}
		Let Assumptions \textup{\ref{error}}--\textup{\ref{sparsity}} and \textup{\ref{cross}} hold. Assume that $\tilde r=r$ and the initial estimates in Algorithm \textup{\ref{alg1}} satisfy $\max_{i\in[r],j\in[m]}w_r^{-1}w_1|\tilde\ab_{z_i,j}^{(0)}-  \tilde\kappa_{i,j}  \ab_{i,j}|_2 = o_{\rm p}(1)$ for some permutation $\{z_1,\ldots,z_r\}$ of $[r]$ and some  constants $\tilde\kappa_{i,j}  \in \{-1,1\}$. If  $D_n\rightarrow \infty$ as $n\rightarrow \infty$, $C_{9}^{-1}\le\sigma_r(\bUpsilon_0)\le \sigma_1(\bUpsilon_0)\le C_{9}$ for some universal constant $C_{9}>1$,
               		 \begin{equation}\label{strong factor condition}
			\frac{w_1}{w_r^2} \bigg( \frac{\gamma_{\max} }{w_r}+\frac{1}{w_r\sqrt{n}}+ L_n\bigg) \ll 1   ~~\text{and}~~|\sigma_{f_i,\xi_i}| \ge C_{9}^{-1} 
		\end{equation}
 for all $i\in[r]$,  then we have
		\[
		\max_{i\in[r],j\in[m]}|  \hat\ab_{z_i,j}- \kappa_{i,j} \ab_{i,j}|_2 =  O_{\rm p}\bigg(\max_{j\in[m]}\Phi_{n,j}+\frac{\gamma_{\max}}{w_r^2}\bigg)\,
		\]
       for some constants $\kappa_{i,j} \in \{-1,1\}$, provided that  
        the number of iterations satisfies 
        $\textit{v}_{\max}\gtrsim-\log(\max_{j\in[m]}\Phi_{n,j}+\gamma_{\max} w_r^{-2})$ and $\max_{j \in [m]}\log d_j \ll n^c$ for some constant $c\in(0,1)$ depending only on $c_1$ and $c_2$ specified in Assumptions {\rm\ref{tail}} and {\rm\ref{mixing}}.
	\end{theorem}
        The requirement on the convergence rate of $\{\tilde\ab_{i,j}^{(0)}\}_{i\in[\tilde{r}],j\in[m]}$ can be easily satisfied if taking the one-pass estimator introduced in Section \ref{sec: initial} as the initial estimator of Algorithm \ref{alg1}. As discussed below Assumption \ref{eigenvalue}, when the initial linear combination is properly selected, it is expected that $\ubar{\sigma}_{\xi}\asymp w_r$ and $\bar{\sigma}_{\xi}\asymp w_1$. Theorem \ref{thm: aij} implies that the convergence rate of the one-pass estimator satisfies this requirement automatically. The first part of condition \eqref{strong factor condition} is mainly to control the plug-in error of $\{\tilde{\ab}^{(\textit{v})}_{i,j}\}_{i\in[\tilde{r}],j\in[m]}$ in the iterations. 
         The requirement $w_1w_r^{-2}n^{-1}(\sum_{j=1}^m d_j)^{1/\tilde c} \ll 1$ originates from the serial dependence of the error process $\{\mathcal{E}_t\}_{t\ge 1}$, and is unnecessary if  $\{\mathcal{E}_t\}_{t\ge 1}$ are serially independent sub-Gaussian tensors, and are also independent of the factor process $\{\fb_t\}_{t \ge 1}$.
        The second part of  condition \eqref{strong factor condition} is similar to Assumption \ref{eigenvalue}, which requires the lag-one cross-correlation between $\{\xi_{t,i}^{\textup{sp}}\}_{t=1}^n$ and $\{f_{t,i}\}_{t=1}^n$ to be non-vanishing. The convergence rate of the iterative estimator includes two parts. The first part depends on $\Phi_{n,j}$, which is a typical rate under sparsity. The second part depends on the lag-one cross-correlations of the factors $\gamma_{\max}$ and the factor strength $w_r$, which is mainly from the estimation error of the  factors when we decorrelate them in the double projection step. Under Assumption \ref{eigenvalue}, we can show that $\ubar{\sigma}_{\xi}\lesssim w_r$. When $\mathbb{E}(f_{t,i})=0$ for all $t,i$ and $\mathbb{E}(f_{t,i}f_{t-k,j})=0$ for all $i\ne j, k\ge 1$ as assumed in \cite{han2024cp}, we have $\gamma_{\max}=O(n^{-1})$ and the convergence rate in Theorem \ref{thm: iterative} can be simplified as $O_{\rm p}(\max_{j\in[m]}\Phi_{n,j})$, which implies the iterative estimator in this scenario is more accurate in comparison to the one-pass estimator introduced in Section \ref{sec: initial}.

	  Let  $\bar\eb_{i,j}=n^{-1}\sum_{t=1}^n\eb_{t,i,j}$  with $\eb_{t,i,j}$ defined in \eqref{ytij}. With $\xi_{t,i}^{\textup{s}}$ defined above \eqref{xi it sp}, write
	\begin{equation}\label{sigma ei xi}
\begin{split}
			\tilde\bSigma_{\eb_{\ell,j},\xi_i}(1)&=\frac{1}{n-1}\sum_{t=2}^n(\eb_{t,\ell,j}-\bar \eb_{\ell,j})\xi_{t-1,i}^{\textup{s}}\,,\\
	\tilde\bSigma_{\eb_{\ell,j},\xi_i}(0)&=\frac{1}{n-1}\sum_{t=2}^n[(\eb_{t,\ell,j}-\bar \eb_{\ell,j})\xi_{t,i}^{\textup{s}}-\mathbb{E}\{(\eb_{t,\ell,j}-\bar \eb_{\ell,j})\xi_{t,i}^{\textup{s}}\}] 
\end{split}
	\end{equation}	
 for $j \in [m]$ and any $i,\ell \in [r]$. Theorem \ref{thm: debias iterative} provides a limiting representation for the iterative estimator $\hat\ab_{i,j}$. 
	 
	  \begin{theorem}\label{thm: debias iterative}
	 	Let $\hb \in\mathbb{R}^{d_j}$ be any non-random  vector satisfying  $|\hb|_2 = 1$, and the  conditions in Theorem \textup{\ref{thm: iterative}} hold. For any $i\in[r]$ and $j\in[m]$, it holds that 
\begin{align*}
    	 	\hb^{\T} ( \hat\ab_{z_i,j} & -  \kappa_{i,j} \ab_{i,j}-  \hat{\bvartheta}_{z_i,j})\\
            =\,& \frac{ \kappa_{i,j}}{w_{i}\sigma_{f_{i},\xi_{i}}}\hb^{\T}(\Ib_{d_j}-\ab_{{i},j}\ab_{{i},j}^{\T})\bigg\{\tilde\bSigma_{\eb_{i,j},\xi_{i}}(1)+\sum_{\ell \neq i} \bar\varphi_{i,\ell} \tilde\bSigma_{\eb_{i,j},\xi_\ell}(0)\bigg\}\\
    	 	&+ O_{\rm p}\bigg(\frac{\gamma_{\max}}{w_{i}w_r}+\frac{1}{\sqrt{n}w_{i}w_r}+\frac{w_1}{w_{i}w_r}  L_n \max_{j\in[m]}\Phi_{n,j}\bigg)
         +o_{\rm p}\bigg(\frac{1}{w_i\sqrt{n}}\bigg)\,,
\end{align*}	 
    where $z_i$ and $\kappa_{i,j}$ are specified in Theorem \textup{\ref{thm: iterative}}, and $\sigma_{f_{i},\xi_{i}}$ and  $\bar\varphi_{i,\ell}$ are defined, respectively, in \eqref{auto cross f xi}  and below \eqref{xi it sp}. Furthermore, if 
\begin{equation}\label{iterative not degenerate}
		 	\lim_{n\rightarrow\infty}{\rm Var}\bigg[\frac{\sqrt{n}}{\sigma_{f_{i},\xi_{i}}}\hb^{\T}(\Ib_{d_j}-\ab_{i,j}\ab_{i,j}^{\T})\bigg\{\tilde\bSigma_{\eb_{i,j},\xi_{i}}(1)+\sum_{\ell\ne i}\bar\varphi_{i,\ell}\tilde\bSigma_{\eb_{{i},j},\xi_\ell}(0)\bigg\}\bigg] = \bar{\tau}^2_{i,j}(\hb) 
\end{equation}
	 	for some deterministic positive number $\bar{\tau}^2_{i,j}(\hb)$ and 
	 	\begin{align}
	 		&\frac{\gamma_{\max}}{w_r}+\frac{1}{\sqrt{n}w_r}+ \frac{w_1}{w_r}  L_n \max_{j\in[m]} \Phi_{n,j} \ll n^{-1/2}\,,\label{negligible error 2}
	 	\end{align}
	then 
	\[
\sqrt{n}\{w_{i}\bar{\tau}^{-1}_{i,j}(\hb)\} \hb^{\T}(\hat\ab_{z_i,j}-\kappa_{i,j}\ab_{i,j}- \hat{\bvartheta}_{z_i,j})\overset{{\rm d}}{\rightarrow} \mathcal{N}(0,1)\,.
	\]
	 \end{theorem}
Condition \eqref{iterative not degenerate} is to ensure that the asymptotic variance is not degenerate. 
Condition \eqref{negligible error 2} is to control the estimation error of the factors and  the plug-in error of $\{\tilde{\ab}^{(\textit{v})}_{i,j}\}_{i\in[\tilde{r}],j\in[m]}$ in the iterations. If all the factors are strong factors such that $ w_1\asymp w_r\asymp \sqrt{D_n}$, condition \eqref{negligible error 2} holds provided that $\gamma^2_{\max}\ll D_n/n$, $D_n\gg 1$, $(\max_{j\in[m]} s_j\log d_j)(\sum_{j=1}^md_j\log d_j)\ll nD_n$, and $(\max_{j\in[m]} s_j\log d_j)(\sum_{j=1}^md_j)^{2/\tilde c}\ll n^2D_n$. 

\section{Discussion}\label{sec: discuss}

In this paper, we develop new estimation methods for tensor CP-factor models that explicitly exploit the tensor structure and allow for correlated factors and loadings, thereby providing useful tools for analyzing high-dimensional tensor-valued data. Several assumptions adopted in this paper can be further relaxed.  Assumption \ref{error}, which requires the error process $\{\mathcal{E}_t\}_{t \ge 1}$ in \eqref{model cp} to be serially uncorrelated, is a key condition for the validity of our auto-covariance-based procedures. Once serial correlation is present in the error process, extending the proposed methods is challenging. In particular, the key identity $\Kb_{1,2,j} = \Ab_j\Gb_{1,\xi}\Gb_{2,\xi}^{-1}(\Ab_j^\T\Ab_j)^{-1}\Ab_j^\T$ with $\Ab_j=(\ab_{1,j},\ldots,\ab_{r,j})$ used to identify the factor loading vectors $\ab_{1,j},\ldots,\ab_{r,j}$ does not hold. How to identify and estimate  $\{\ab_{i,j}\}_{j\in[m],i \in [r]}$ in the setting with serially correlated error process deserves further investigation. Section \ref{sec: relax error serial dependence} in the supplementary material provides some further discussion for this. We also discuss in Section \ref{sec: relax tail condition} in the supplementary material that our theoretical results can be extended from the exponential-decay assumptions in Assumptions \ref{tail} and \ref{mixing} to polynomial-decay conditions.
 Assumption \ref{mixing} requires weak serial dependence among the observed tensor process $\{\mathcal{Y}_t\}_{t \ge 1}$ which does not cover the cases with unit-root tensor process. It would be interesting to extend the proposed methods to handle unit-root tensor process $\{\mathcal{Y}_t\}_{t \ge 1}$. We plan to investigate it in our future research.

\bibliographystyle{jasa}


\begingroup
\setlength{\bibsep}{2pt}      
\linespread{0.9}\selectfont 
\bibliography{Ref-abbreviation}

\begin{thebibliography}{17}
\newcommand{\enquote}[1]{``#1''}
\expandafter\ifx\csname natexlab\endcsname\relax\def\natexlab#1{#1}\fi
\expandafter\ifx\csname url\endcsname\relax
  \def\url#1{{\tt #1}}\fi
\expandafter\ifx\csname urlprefix\endcsname\relax\def\urlprefix{URL }\fi

\bibitem[{Andrews(1991)}]{andrews1991heteroskedasticity-app}
Andrews, D.~W. (1991).
\newblock Heteroskedasticity and autocorrelation consistent covariance matrix
  estimation.
\newblock {\em Econometrica\/}, 59, 817--858.

\bibitem[{Athreya and Lahiri(2006)}]{athreya2006measure-app}
Athreya, K.~B. and Lahiri, S.~N. (2006).
\newblock {\em Measure Theory and Probability Theory\/}.
\newblock Springer.

\bibitem[{Bai(2003)}]{bai2003inferential-app}
Bai, J. (2003).
\newblock Inferential theory for factor models of large dimensions.
\newblock {\em Econometrica\/}, 71, 135--171.

\bibitem[{Chang et~al.(2024)Chang, Chen, and Wu}]{chang2021central-app}
Chang, J., Chen, X., and Wu, M. (2024).
\newblock Central limit theorems for high dimensional dependent data.
\newblock {\em Bernoulli\/}, 30, 712--742.

\bibitem[{Chang et~al.(2026)Chang, Du, Huang, and Yao}]{chang2024unified-app}
Chang, J., Du, Y., Huang, G., and Yao, Q. (2026).
\newblock Identification and estimation for matrix time series CP-factor
  models.
\newblock {\em Ann. Stat.\/}, 54, 1372--1397.

\bibitem[{Chang et~al.(2018{\natexlab{a}})Chang, Guo, and
  Yao}]{chang2018principal-app}
Chang, J., Guo, B., and Yao, Q. (2018{\natexlab{a}}).
\newblock Principal component analysis for second-order stationary vector time
  series.
\newblock {\em Ann. Stat.\/}, 46, 2094--2124.

\bibitem[{Chang et~al.(2023)Chang, He, Yang, and Yao}]{chang2023modelling-app}
Chang, J., He, J., Yang, L., and Yao, Q. (2023).
\newblock Modelling matrix time series via a tensor CP-decomposition.
\newblock {\em J. R. Stat. Soc. Ser. B Stat. Methodol.\/}, 85, 127--148.

\bibitem[{Chang et~al.(2018{\natexlab{b}})Chang, Qiu, Yao, and
  Zou}]{chang2018confidence-app}
Chang, J., Qiu, Y., Yao, Q., and Zou, T. (2018{\natexlab{b}}).
\newblock Confidence regions for entries of a large precision matrix.
\newblock {\em J. Econom.\/}, 206, 57--82.

\bibitem[{Chen et~al.(2026)Chen, Han, and Yu}]{chen2026estimation-app}
Chen, B., Han, Y., and Yu, Q. (2026).
\newblock Estimation and inference for CP tensor factor models.
\newblock {\em J. Econom.\/}, 253, 106167.

\bibitem[{Chen et~al.(2021)Chen, Xiao, and Yang}]{chen2021autoregressive-app}
Chen, R., Xiao, H., and Yang, D. (2021).
\newblock Autoregressive models for matrix-valued time series.
\newblock {\em J. Econom.\/}, 222, 539--560.

\bibitem[{Ekstr{\"o}m(2014)}]{ekstrom2014general-app}
Ekstr{\"o}m, M. (2014).
\newblock A general central limit theorem for strong mixing sequences.
\newblock {\em Stat. Probab. Lett.\/}, 94, 236--238.

\bibitem[{Han et~al.(2024{\natexlab{a}})Han, Chen, Yang, and
  Zhang}]{han2024tensor-app}
Han, Y., Chen, R., Yang, D., and Zhang, C.-H. (2024{\natexlab{a}}).
\newblock Tensor factor model estimation by iterative projection.
\newblock {\em Ann. Stat.\/}, 52, 2641--2667.

\bibitem[{Han et~al.(2022)Han, Chen, and Zhang}]{han2022rank-app}
Han, Y., Chen, R., and Zhang, C.-H. (2022).
\newblock Rank determination in tensor factor model.
\newblock {\em Electron. J. Stat.\/}, 16, 1726--1803.

\bibitem[{Han et~al.(2024{\natexlab{b}})Han, Yang, Zhang, and
  Chen}]{han2024cp-app}
Han, Y., Yang, D., Zhang, C.-H., and Chen, R. (2024{\natexlab{b}}).
\newblock CP factor model for dynamic tensors.
\newblock {\em J. R. Stat. Soc. Ser. B Stat. Methodol.\/}, 86, 1383--1413.

\bibitem[{Rao and Rao(1998)}]{rao1998matrix-app}
Rao, C.~R. and Rao, M.~B. (1998).
\newblock {\em Matrix Algebra and Its Applications to Statistics and
  Econometrics\/}.
\newblock World Scientific.

\bibitem[{Trapani(2016)}]{trapani2016testing-app}
Trapani, L. (2016).
\newblock Testing for (in) finite moments.
\newblock {\em J. Econom.\/}, 191, 57--68.

\bibitem[{Wang et~al.(2019)Wang, Liu, and Chen}]{wang2019factor-app}
Wang, D., Liu, X., and Chen, R. (2019).
\newblock Factor models for matrix-valued high-dimensional time series.
\newblock {\em J. Econom.\/}, 208, 231--248.

\end{thebibliography}
\endgroup



\clearpage
\newpage 
\appendix

\spacingset{1.7}\selectfont
\setlength{\abovedisplayskip}{0.2\baselineskip}
\setlength{\belowdisplayskip}{0.2\baselineskip}
\setlength{\abovedisplayshortskip}{0.2\baselineskip}
\setlength{\belowdisplayshortskip}{0.2\baselineskip}
\begin{center}
	{\noindent \bf \Large Supplementary Material for ``CP-Factorization for High Dimensional Tensor Time Series and Double Projection Iterations" by Chang, Huang, Yao and  Yu}\\
\end{center}
\bigskip

\setcounter{page}{1}
\setcounter{section}{0}
\renewcommand\thesection{\Alph{section}}
\setcounter{lemma}{0}
\renewcommand{\thelemma}{L\arabic{lemma}}
\setcounter{theorem}{0}
\renewcommand{\thetheorem}{T\arabic{theorem}}
\renewcommand{\theHtheorem}{T\arabic{theorem}}

\setcounter{equation}{0}
\renewcommand{\theequation}{S.\arabic{equation}}

\setcounter{table}{0}
\setcounter{figure}{0}
\renewcommand{\thefigure}{F\arabic{figure}}
\renewcommand{\thetable}{T\arabic{table}}
\renewcommand{\thepage}{S\arabic{page}} 

 
 Section~\ref{sec: variance} discusses the estimation procedure for the asymptotic variance involved in Theorem \ref{thm: debias iterative}. Section~\ref{sec: supp-addition simulation} presents additional simulation and empirical results. Section~\ref{sec:factor and cp estimation} considers the estimation of the factors and common components. Section~\ref{sec: factor number} provides further discussion on estimating the number of factors.  Section~\ref{sec: supp main proof} contains the proofs of all theorems stated in both the main paper and the supplementary material. Section~\ref{sec:supp-error bounds} provides the proofs of auxiliary lemmas. Section~\ref{sec:relaxation assumptions} discusses possible relaxations of the technical assumptions.

 
\section{Estimation of the asymptotic variance}\label{sec: variance}
\subsection{Explicit form of the asymptotic variance in Theorem \ref{thm: debias iterative}}\label{sec: app-variance}
Under some regularity conditions, the quantity $\bar{\tau}^2_{i,j}(\hb)$ specified in \eqref{iterative not degenerate} can be derived explicitly. 
Specifically, let 
$$f_{t,i}^{\textup{s}} = \bigg[\frac{1}{n}\sum_{s=1}^n \mathbb{E}\{(f_{s,i}-\bar f_i)^2\}\bigg]^{-1/2}(f_{t,i}-\bar f_i)$$
be the standardized version of $f_{t,i}$, and
\begin{equation}\label{f ti sp}
    f^{\textup{sp}}_{t-1,i} = f_{t-1,i}^{\textup{s}} + \sum_{\ell \neq i}\bar\varphi_{i,\ell}f_{t,\ell}^{\textup{s}} 
\end{equation}
be the projection version of $f_{t-1,i}^{\textup{s}}$, where $\bar\varphi_{i,\ell}$ is specified below \eqref{xi it sp}. 
Let
\begin{equation}\label{def: bbeta ij}
    \bbeta_{i,j}(\hb) =  \bbb_{i,j}^{\MP}  \otimes \{\hb^{\T}(\Ib_{d_j}-\ab_{i,j}\ab_{i,j}^{\T})\}^{\T} \,.
\end{equation}
Without loss of generality and for notational simplicity, we ignore the reflection and permutation indeterminacy of $\hat \ab_{i,j}$. Following the limiting representation established in Theorem \ref{thm: debias iterative} and the arguments used in the proof of Theorem \ref{thm: estimation of iteration variance} (see \eqref{simplified expansion for hat a ij} in Section \ref{sec:thm: estimation of iteration variance}),
we have
\begin{equation*}
       	 w_i\hb^{\T}(\hat \ab_{i,j} - \ab_{i,j} - \hat{\bvartheta}_{i,j}) =  \frac{1}{\sigma_{f_i,\xi_i}} \frac{1}{n-1}\sum_{t=2}^n \bbeta_{i,j}(\hb)^{\T} \textup{vec}(\Eb_{t,j}) f_{t-1,i}^{\textup{sp}}  + o_{\rm p}(n^{-1/2})\,.
\end{equation*}
Let $\varsigma_{t,i,j}(\hb)
    =
    f_{t-1,i}^{\textup{sp}} \bbeta_{i,j}(\hb)^{\T}\textup{vec}(\Eb_{t,j})
    $ and write $|\cdot|_{+}=\max(\cdot,0)$. Under the conditions of Theorem \ref{thm: debias iterative}, $\{\varsigma_{t,i,j}(\hb)\}_{t=2}^n$ is an $\alpha$-mixing process with zero mean and mixing coefficients $\{\alpha(|\ell-1|_+)\}_{\ell\ge 1}$.
Hence, 
$\bar\tau_{i,j}^2(\hb)$ is determined by the long-run variance of
$\{\varsigma_{t,i,j}(\hb)\}_{t=2}^n$, i.e., 
\begin{align}\label{eq: long run covariance}
      \bar\tau_{i,j}^2(\hb)
    =
   \lim_{n\to\infty} \sigma_{f_i,\xi_i}^{-2} \textup{Var}\bigg\{\frac{1}{\sqrt{n-1}}
    \sum_{t=2}^n \varsigma_{t,i,j}(\hb)
    \bigg\} \,.
\end{align}
Based on Theorem \ref{thm: debias iterative}, the asymptotic variance of the iterative estimator is $w_i^{-2}\bar\tau_{i,j}^2(\hb)$.

\subsection{Details of the asymptotic variance estimation}\label{sec:variance iter}
Without loss of generality and for notational simplicity, we ignore the permutation indeterminacy of the estimator $\hat{\ab}_{i,j}$ in the introduction of the asymptotic variance estimation. Since $w_i$ and $\sigma_{f_i,\xi_i}$ are not separately identifiable, in order to estimate the asymptotic variance $w_i^{-2}\bar\tau_{i,j}^2(\hb)$, we first estimate $w_i|\sigma_{f_i,\xi_i}|$ and $|\sigma_{f_i,\xi_i}|\bar\tau_{i,j}(\hb)$, respectively.  Notice that  $T_{\delta_{2,j}}\{\tilde\bSigma^{(\textit{v}_{\max},j)}_{\tilde\yb_{i,j},\tilde\xi_{i}}(1)\}$ specified in Algorithm \ref{alg1} serves as an estimator for $(w_i\sigma_{f_i,\xi_i})\ab_{i,j}$. Therefore, we can estimate $w_i |\sigma_{f_i,\xi_i}|$ by
\begin{equation}\label{eq:wijdef}
   \hat w_{i,j}=|\hat\ab_{i,j}^\T T_{\delta_{2,j}}\{\tilde\bSigma^{(\textit{v}_{\max},j)}_{\tilde\yb_{i,j},\tilde\xi_{i}}(1)\}|\,. 
\end{equation}
By \eqref{eq: long run covariance}, estimating $|\sigma_{f_i,\xi_i}|\bar\tau_{i,j}(\hb)$ can be solved via estimating the long-run variance of the process $\{\varsigma_{t,i,j}(\hb)\}_{t=2}^n$. Since   $\bbeta_{i,j}(\hb)^{\T}\textup{vec}(\Eb_{t,j})=\bbeta_{i,j}(\hb)^{\T}\textup{vec}(\Yb_{t,j})$, and $\tilde\xi_{t-1,i}^{(\textit{v}_{\max},j)}$ defined in Algorithm \ref{alg1} provides an approximation of $f_{t-1,i}^{\textup{sp}}$, we define $ \hat{\varsigma}_{t,i,j}(\hb)
    =
    \tilde\xi_{t-1,i}^{(\textit{v}_{\max},j)} \hat{\bbeta}_{i,j}(\hb)^{\T}\textup{vec}(\Yb_{t,j})
    $
as an approximation of $\varsigma_{t,i,j}(\hb)$, where $\hat{\bbeta}_{i,j}(\hb)=(\tilde{\bbb}_{i,j}^{(\textit{v}_{\max})})^{\MP}\otimes \{\hb^\T (\Ib_{d_j}-\hat\ab_{i,j}\hat\ab_{i,j}^\T)\}^\T$ with $(\tilde{\bbb}_{i,j}^{(\textit{v}_{\max})})^{\MP}$ specified in Algorithm \ref{alg1}.  
To estimate $\sigma_{f_i,\xi_i}^2\bar\tau_{i,j}^2(\hb)$, we suggest the kernel-type estimator \citepS{andrews1991heteroskedasticity-app} as follows:
\begin{equation}\label{eq:kernelest}
        \tilde\tau^2_{i,j}(\hb)
    = 
    \sum_{s=-n+2}^{n-2}
    \mathcal{K}\bigg(\frac{s}{b_{n,i,j}}\bigg)
    \hat H_{s,i,j}(\hb)\,,
\end{equation}
where $\mathcal{K}(\cdot)$ is a symmetric kernel function that is continuous at $0$ with $\mathcal{K}(0)=1$, $b_{n,i,j}>0$ is the bandwidth, and
\[
\hat H_{s,i,j}(\hb) = \frac{1}{n-1}\sum_{t = \max(1,-s+1)}^{\min(n-1,n-1-s)}  \mathring{\hat{\varsigma}}_{t+1+s,i,j}(\hb) \mathring{\hat{\varsigma}}_{t+1,i,j}(\hb) 
\]
with $\mathring{\hat{\varsigma}}_{t,i,j}(\hb) = \hat{\varsigma}_{t,i,j}(\hb) - (n-1)^{-1}\sum_{s=2}^n \hat{\varsigma}_{s,i,j}(\hb)$.
When $\{\varsigma_{t,i,j}(\hb)\}_{t=2}^n$ are observed, \citeS{andrews1991heteroskedasticity-app} establishes the consistency of such kernel-type estimator (with replacing $\hat{\varsigma}_{t,i,j}(\hb)$ in \eqref{eq:kernelest} by $\varsigma_{t,i,j}(\hb)$) for long-run variance. When $\{\varsigma_{t,i,j}(\hb)\}_{t=2}^n$ are unobserved, the consistency of the kernel-type estimator \eqref{eq:kernelest} can still be established with some more tedious calculation. See, for example, the proof of Theorem 2 in  \citeS{chang2018confidence-app}.
In practice, the kernel function can be selected as the quadratic spectral kernel
\[
\mathcal{K}_{QS}(x)=\frac{25}{12\pi^2x^2}\bigg\{\frac{\sin(6\pi x/5)}{6\pi x/5}-\cos(6\pi x/5)\bigg\}\,,
\]
and the bandwidth $b_{n,i,j}$ can be selected by the data-driven rule suggested in Section 6 of \citeS{andrews1991heteroskedasticity-app}, i.e., $b_{n,i,j} = 1.3211\{\hat{v}_{i,j}(n-1)\}^{1/5}$  and 
$ \hat v_{i,j}
    =4\hat\rho_{i,j}^2 
    (1-\hat\rho_{i,j})^{-4}$
with $\hat\rho_{i,j}$  being the estimated autoregressive coefficient from fitting an AR(1) model to the time series
$\{\hat{\varsigma}_{t,i,j}(\hb)\}_{t=2}^{n}$.  Combining with $\hat{w}_{i,j}$ specified in \eqref{eq:wijdef}, $\hat{w}_{i,j}^{-2} \tilde\tau_{i,j}^2(\hb)$ provides an estimator for the asymptotic variance $w_{i}^{-2} \bar\tau_{i,j}^2(\hb)$.

Furthermore, if the error process $\{\mathcal{E}_t\}_{t\ge 1}$ is independent of the factor process $\{\mathbf{f}_t\}_{t\ge 1}$, the long-run variance \eqref{eq: long run covariance} admits a simple form. 
 Together with Assumption \ref{error} and the definition of 
 $f_{t-1,i}^{\textup{sp}}$, we have
\begin{align*}
    \bar\tau_{i,j}^2(\hb)
    &=
   \lim_{n\to\infty} \sigma_{f_i,\xi_i}^{-2}
   \textup{Var}\bigg\{
   \frac{1}{\sqrt{n-1}}
    \sum_{t=2}^n \varsigma_{t,i,j}(\hb)
    \bigg\}  \\
     &=   
    \lim_{n\to\infty} 
    \frac{\sigma_{f_i,\xi_i}^{-2}}{n-1}
    \sum_{t=2}^n\sum_{s=2}^n
    \textup{Cov}\{\varsigma_{t,i,j}(\hb),\varsigma_{s,i,j}(\hb)\}  \\
    &=
    \lim_{n\to\infty} 
    \frac{\sigma_{f_i,\xi_i}^{-2}}{n-1}
    \sum_{t=2}^n\sum_{s=2}^n
    \textup{Cov}[
    \bbeta_{i,j}(\hb)^{\T}\textup{vec}(\Eb_{t,j})f_{t-1,i}^{\textup{sp}},
    \bbeta_{i,j}(\hb)^{\T}\textup{vec}(\Eb_{s,j})f_{s-1,i}^{\textup{sp}}
    ] \\
    & =
   \lim_{n\to\infty}  
    \frac{\sigma_{f_i,\xi_i}^{-2} }{n-1}
    \sum_{t=2}^n\sum_{s=2}^n
    \mathbb{E}[
    f_{t-1,i}^{\textup{sp}}f_{s-1,i}^{\textup{sp}}
    \{\bbeta_{i,j}(\hb)^{\T}\textup{vec}(\Eb_{t,j})\}
    \{\bbeta_{i,j}(\hb)^{\T}\textup{vec}(\Eb_{s,j})\}
    ] \\
    &= \lim_{n\to\infty}
    \frac{\sigma_{f_i,\xi_i}^{-2}}{n-1}
    \sum_{t=2}^n
    \mathbb{E}[
    \{f_{t-1,i}^{\textup{sp}}\}^2
    \{\bbeta_{i,j}(\hb)^{\T}\textup{vec}(\Eb_{t,j})\}^2
    ]\,.
\end{align*}
Recall that $\bbeta_{i,j}(\hb)^{\T}\textup{vec}(\Eb_{t,j})=\bbeta_{i,j}(\hb)^{\T}\textup{vec}(\Yb_{t,j})$, and $\tilde\xi_{t-1,i}^{(\textit{v}_{\max},j)}$ defined in Algorithm \ref{alg1} is an approximation of $f_{t-1,i}^{\textup{sp}}$.  Let 
\begin{equation*}\label{hat sigma ij h}
\hat \tau_{i,j}(\hb)
=
\bigg|
\frac{1}{n-1}
\sum_{t=2}^n
(\tilde\xi_{t-1,i}^{(\textit{v}_{\max},j)})^2
\{\hat{\bbeta}_{i,j}(\hb)^{\T}\textup{vec}(\Yb_{t,j})\}^2
\bigg|^{1/2}\,,
\end{equation*}
which serves as a plug-in estimator of $|\sigma_{f_i,\xi_i}|\bar\tau_{i,j}(\hb)$. Theorem \ref{thm: estimation of iteration variance} establishes that
$\hat w_{i,j}^{-1} \hat\tau_{i,j}(\hb)$ is a consistent estimator of 
$w_i^{-1} \bar{\tau}_{i,j}(\hb)$ under the scenario where $\{\mathcal{E}_t\}_{t \ge 1}$ is independent of $\{\mathbf{f}_t\}_{t \ge 1}$. Therefore, we can estimate the asymptotic variance of the iterative estimator by $\hat w_{i,j}^{-2} \hat\tau_{i,j}^2(\hb)$.

\begin{theorem}\label{thm: estimation of iteration variance}
    Under the same assumptions as in Theorem \textup{\ref{thm: debias iterative}} with $(z_1,\ldots,z_r) = (1,\ldots,r)$, if further $\{\mathcal{E}_t\}_{t \ge 1}$ is independent of $\{\mathbf{f}_t\}_{t \ge 1}$, and
  \begin{equation}\label{condition for variance estimator}
   w_1\bigg(\max_{j\in[m]}\Phi_{n,j}+\frac{\gamma_{\max}}{w_r^2}\bigg)\ll 1~~\textrm{and}~~\frac{(\sum_{j=1}^m d_j)^{1/\tilde c_1}}{n}\bigg(\max_{j\in[m]}\Phi_{n,j}+\frac{\gamma_{\max}}{w_r^2}\bigg)\ll 1\,,
    \end{equation}
    where $\tilde c_1^{-1}=1+4c_1^{-1}+c_2^{-1}$ for $c_1$ and $c_2$ specified in Assumptions \ref{tail} and \ref{mixing}. Then
    \[
    \frac{\hat w_{i,j}^{-1}\hat \tau_{i,j}(\hb)}{w_i^{-1} \bar{\tau}_{i,j}(\hb) }-1=o_{\rm p}(1)\,.
    \]
\end{theorem}
Condition \eqref{condition for variance estimator} is imposed to control the plug-in error of $\hat\ab_{i,j}$ in the construction of $\hat w_{i,j}$ and $\hat\tau_{i,j}(\hb)$. If all the factors are strong in the sense that $ w_1\asymp w_r\asymp \sqrt{D_n}$, the first part of condition \eqref{condition for variance estimator} holds provided that $\max_{j\in[m]} s_j\log d_j \ll n$, and $\gamma_{\max}\ll \sqrt{D_n}$. The second part of  \eqref{condition for variance estimator} is mainly used to bound the error term $\{\hat\bbeta_{i,j}(\hb)-\bbeta_{i,j}(\hb)\}^{\T}\textup{vec}(\Eb_{t,j})$ in the proof of Theorem \ref{thm: estimation of iteration variance}. Under high-dimensional settings where $d_j\gg n$, one can bound this error based on the covering number argument similarly to (G.3) of Lemma G.1 in \citeS{han2024tensor-app}. Specifically, for any  $i\in[r]$ and $j\in[m]$, we can obtain 
\[
\sup_{t\in[n]}|\{\bbeta_{i,j}(\hb)-\hat\bbeta_{i,j}(\hb)\}^{\T} \textup{vec}(\Eb_{t,j})|=O_{\rm p}\bigg\{\bigg(\max_{j\in[m]}\Phi_{n,j}+\frac{\gamma_{\max}}{w_r^2}\bigg)\bigg(\sum_{j=1}^m d_j \log n \bigg)^{1/c_1}\bigg\}\,.
\]
With this bound, we can show that Theorem \ref{thm: estimation of iteration variance}  still holds if we replace the second part of \eqref{condition for variance estimator} with the following condition $(\max_{j\in[m]} s_j\log d_j)^{1/2}(\log n)^{2/c_1}\lesssim \sum_{j=1}^m d_j$.  Theorem \ref{thm: estimation of iteration variance} actually ignores the permutation indeterminacy by taking $(z_1,\ldots,z_r)=(1,\ldots,r)$. If we consider the permutation indeterminacy, the conclusion in Theorem \ref{thm: estimation of iteration variance} can be modified as
\[
    \frac{\hat w_{z_i,j}^{-1}\hat \tau_{z_i,j}(\hb)}
    {w_i^{-1}\bar{\tau}_{i,j}(\hb)}
    -1=o_{\rm p}(1)
\]
with $z_i$ specified in Theorem \ref{thm: iterative}.

Next, we evaluate the finite-sample properties of the variance estimators $ \hat w_{i,j}^{-2} \tilde \tau_{i,j}^2(\hb)$ and $ \hat w_{i,j}^{-2} \hat \tau_{i,j}^{2}(\hb)$ via simulation studies. The bandwidth $b_{n,i,j}$ and kernel function $\mathcal{K}(\cdot)$ are selected following the procedure described above. The data generation process follows the setup in Section \ref{sec:numerical}, with the sample size $n \in \{400,800,1600,3200\}$. It should be
noted that there exists the permutation indeterminacy between the estimates  and the true loading vectors  in practice. We set $(i,j) = (1,1)$ and $z_1 = \arg\min_{i \in [\tilde r]}\{1 - |\ab_{1,1}^{\T}\hat{\ab}_{i,1}|^2\}$ herein, which eliminates the permutation indeterminacy between the estimates and the true loading vector $\ab_{1,1}$.
Furthermore, we exclude cases that satisfy either $\min_{i \in [\tilde{r}]}\{1 - |\ab_{1,1}^{\T}\hat{\ab}_{i,1}|^2\} > 0.3$ or $\tilde{r} \neq r$. 
The estimation errors of $\hat w_{z_1,1}^{-2}\tilde \tau_{z_1,1}^{2}(\hb)$ and 
$\hat w_{z_1,1}^{-2}\hat \tau_{z_1,1}^{2}(\hb)$  are measured, respectively, by
\begin{equation}\label{eq:estimation error-iter-var}
\begin{split}
 D_{(\textup{LR})}(\hb)  
 &=  
 \{ \hat w_{z_1,1}^{-2} \tilde \tau_{z_1,1}^{2}(\hb) 
 - w_1^{-2} \bar{\tau}^{2}_{1,1}(\hb)  \}^2 \,,\\
 D_{(\textup{PI})}(\hb)  
 &=  
  \{ \hat w_{z_1,1}^{-2} \hat \tau_{z_1,1}^{2}(\hb)
 - w_1^{-2} \bar{\tau}^{2}_{1,1}(\hb) \}^2 \,.
\end{split}
\end{equation}
 As shown in Table \ref{table:var-iter-all}, the average estimation errors for both proposed estimators decrease as $n$ increases across all scenarios, which supports the consistency of the proposed variance estimators.

\begin{table}
\footnotesize
\caption{
The square root of the averages of the estimation errors \eqref{eq:estimation error-iter-var} for the estimators $\hat{w}_{z_1,1}^{-2}\hat \tau_{z_1,1}^{2}(\hb)$ and $\hat{w}_{z_1,1}^{-2}\tilde\tau_{z_1,1}^{2}(\hb)$ of the asymptotic variance based on 2000 repetitions. All numbers reported below are multiplied by $1000$.}

\centering
\renewcommand\tabcolsep{4pt}
\label{table:var-iter-all}

\resizebox{\textwidth}{!}{
\begin{tabular}{
>{\centering\arraybackslash}m{0.8cm}|
>{\centering\arraybackslash}m{0.8cm}|
>{\centering\arraybackslash}m{0.7cm}|
>{\centering\arraybackslash}m{0.8cm}|
cccc|cccc}
\hline\hline
\multirow{2}{*}{$\rho$} 
& \multirow{2}{*}{$\phi$} 
& \multirow{2}{*}{$s$} 
& \multirow{2}{*}{$\hb$} 
& \multicolumn{4}{c|}{$ D_{(\textup{PI})}(\hb)$} 
& \multicolumn{4}{c}{$ D_{(\textup{LR})}(\hb)$} \\
& & & 
& $n=400$ & $n=800$ & $n=1600$ & $n=3200$
& $n=400$ & $n=800$ & $n=1600$ & $n=3200$ \\ 
\hline

\multirow{12}{*}{\centering 0}    
& \multirow{6}{*}{\centering 0.25} 
& \multirow{2}{*}{\centering 0}   
& $\hb_1$ & 0.40 & 0.26 & 0.18 & 0.13 & 0.56 & 0.35 & 0.25 & 0.19 \\
& & & $\hb_2$ & 0.40 & 0.27 & 0.18 & 0.13 & 0.57 & 0.37 & 0.25 & 0.18 \\
\cline{3-12}

& & \multirow{2}{*}{\centering 0.3} 
& $\hb_1$ & 0.39 & 0.26 & 0.17 & 0.12 & 0.53 & 0.36 & 0.24 & 0.17 \\
& & & $\hb_2$ & 0.38 & 0.26 & 0.18 & 0.12 & 0.53 & 0.34 & 0.25 & 0.17 \\
\cline{3-12}

& & \multirow{2}{*}{\centering 0.6} 
& $\hb_1$ & 0.39 & 0.25 & 0.18 & 0.12 & 0.53 & 0.36 & 0.24 & 0.16 \\
& & & $\hb_2$ & 0.39 & 0.25 & 0.17 & 0.12 & 0.50 & 0.35 & 0.24 & 0.17 \\
\cline{2-12}

& \multirow{6}{*}{\centering 0.75} 
& \multirow{2}{*}{\centering 0}   
& $\hb_1$ & 0.59 & 0.42 & 0.29 & 0.20 & 0.83 & 0.58 & 0.40 & 0.28 \\
& & & $\hb_2$ & 0.61 & 0.42 & 0.29 & 0.21 & 0.84 & 0.58 & 0.40 & 0.29 \\
\cline{3-12}

& & \multirow{2}{*}{\centering 0.3} 
& $\hb_1$ & 0.46 & 0.32 & 0.21 & 0.15 & 0.64 & 0.43 & 0.30 & 0.20 \\
& & & $\hb_2$ & 0.46 & 0.32 & 0.21 & 0.14 & 0.62 & 0.43 & 0.31 & 0.20 \\
\cline{3-12}

& & \multirow{2}{*}{\centering 0.6} 
& $\hb_1$ & 0.43 & 0.28 & 0.19 & 0.13 & 0.58 & 0.38 & 0.26 & 0.18 \\
& & & $\hb_2$ & 0.40 & 0.28 & 0.19 & 0.13 & 0.55 & 0.39 & 0.26 & 0.18 \\
\hline

\multirow{12}{*}{\centering 0.75}    
& \multirow{6}{*}{\centering 0.25} 
& \multirow{2}{*}{\centering 0}   
& $\hb_1$ & 3.42 & 1.49 & 0.98 & 0.65 & 4.03 & 1.85 & 1.18 & 0.79 \\
& & & $\hb_2$ & 3.34 & 1.47 & 0.96 & 0.69 & 3.78 & 1.83 & 1.16 & 0.83 \\
\cline{3-12}

& & \multirow{2}{*}{\centering 0.3} 
& $\hb_1$ & 2.71 & 1.45 & 0.92 & 0.61 & 3.31 & 1.73 & 1.13 & 0.74 \\
& & & $\hb_2$ & 2.69 & 1.37 & 0.91 & 0.61 & 3.13 & 1.68 & 1.13 & 0.76 \\
\cline{3-12}

& & \multirow{2}{*}{\centering 0.6} 
& $\hb_1$ & 2.59 & 1.45 & 0.90 & 0.61 & 3.08 & 1.73 & 1.10 & 0.75 \\
& & & $\hb_2$ & 2.60 & 1.38 & 0.89 & 0.60 & 2.92 & 1.67 & 1.05 & 0.72 \\
\cline{2-12}

& \multirow{6}{*}{\centering 0.75} 
& \multirow{2}{*}{\centering 0}   
& $\hb_1$ & 4.36 & 2.50 & 1.52 & 1.11 & 5.40 & 3.03 & 1.81 & 1.27 \\
& & & $\hb_2$ & 4.25 & 2.46 & 1.49 & 1.12 & 5.39 & 2.99 & 1.77 & 1.42 \\
\cline{3-12}

& & \multirow{2}{*}{\centering 0.3} 
& $\hb_1$ & 3.39 & 1.78 & 1.14 & 0.74 & 3.99 & 2.11 & 1.38 & 0.91 \\
& & & $\hb_2$ & 3.36 & 1.84 & 1.10 & 0.75 & 3.79 & 2.15 & 1.33 & 0.90 \\
\cline{3-12}

& & \multirow{2}{*}{\centering 0.6} 
& $\hb_1$ & 2.88 & 1.55 & 0.99 & 0.66 & 3.25 & 1.89 & 1.18 & 0.81 \\
& & & $\hb_2$ & 2.82 & 1.52 & 0.98 & 0.64 & 3.23 & 1.77 & 1.20 & 0.79 \\
\hline\hline
\end{tabular}
}
\end{table}
\begin{figure}[htbp]
\centerline{\includegraphics[width= 12cm]{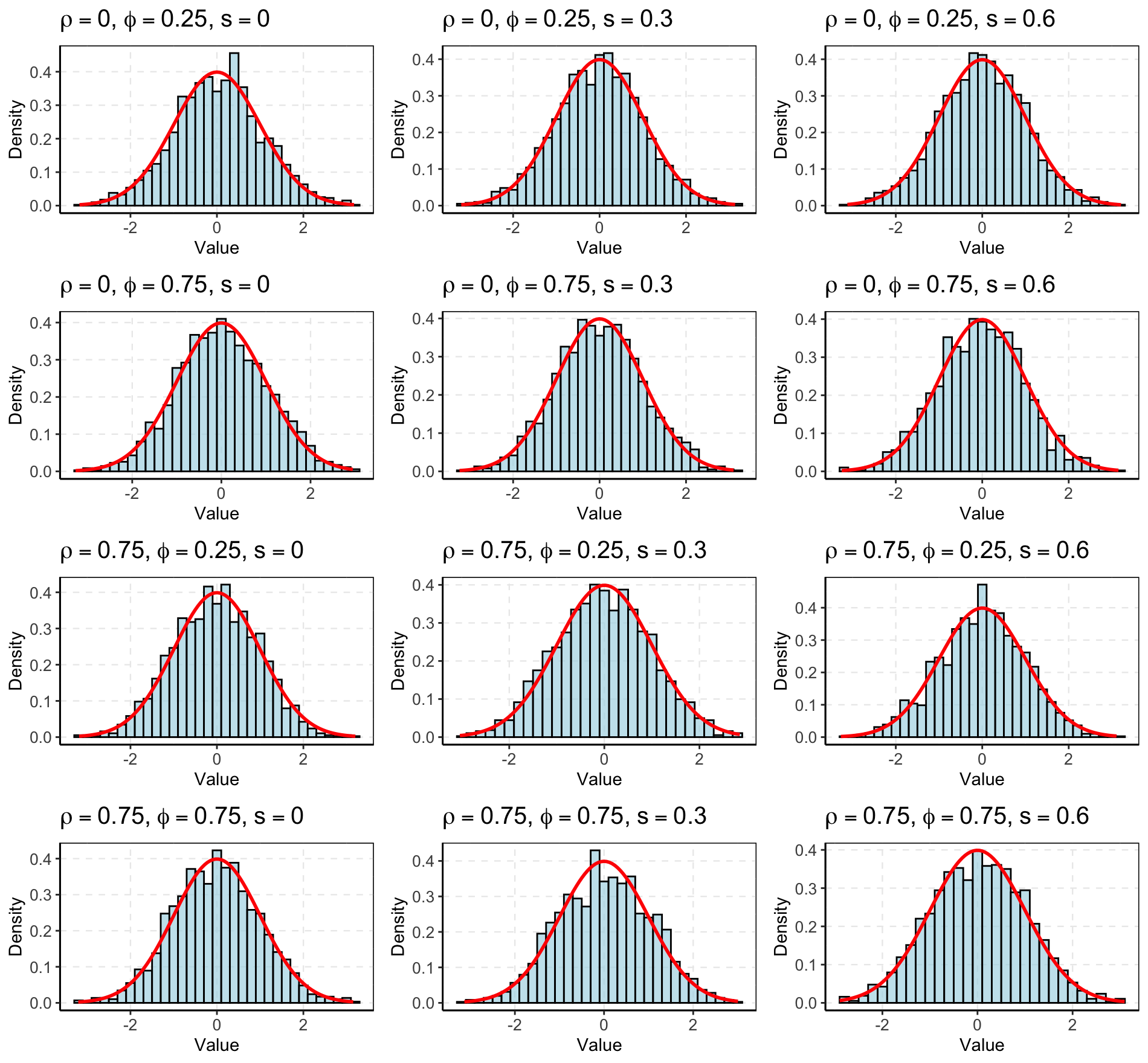}}
\caption{The histograms of $\{\hat{w}_{z_1,1} \tilde \tau_{z_1,1}^{-1}(\hb_1)\} \sqrt{n}\, \hb^{\T}_{1}\{\hat{\ab}_{z_1,1} - \text{sign}(\ab_{1,1}^{\T}\hat{\ab}_{z_1,1} ) \cdot \ab_{1,1} - \hat{\bvartheta}_{z_1,1} \}$  based on 2000 repetitions. The sample size $n = 400$. The red curve plots the density of $\mathcal{N}(0,1)$.}
\label{fig:normality-h1-hist-iter-est-longrun}
\end{figure}

\begin{figure}[htbp]
\centerline{\includegraphics[width= 12cm]{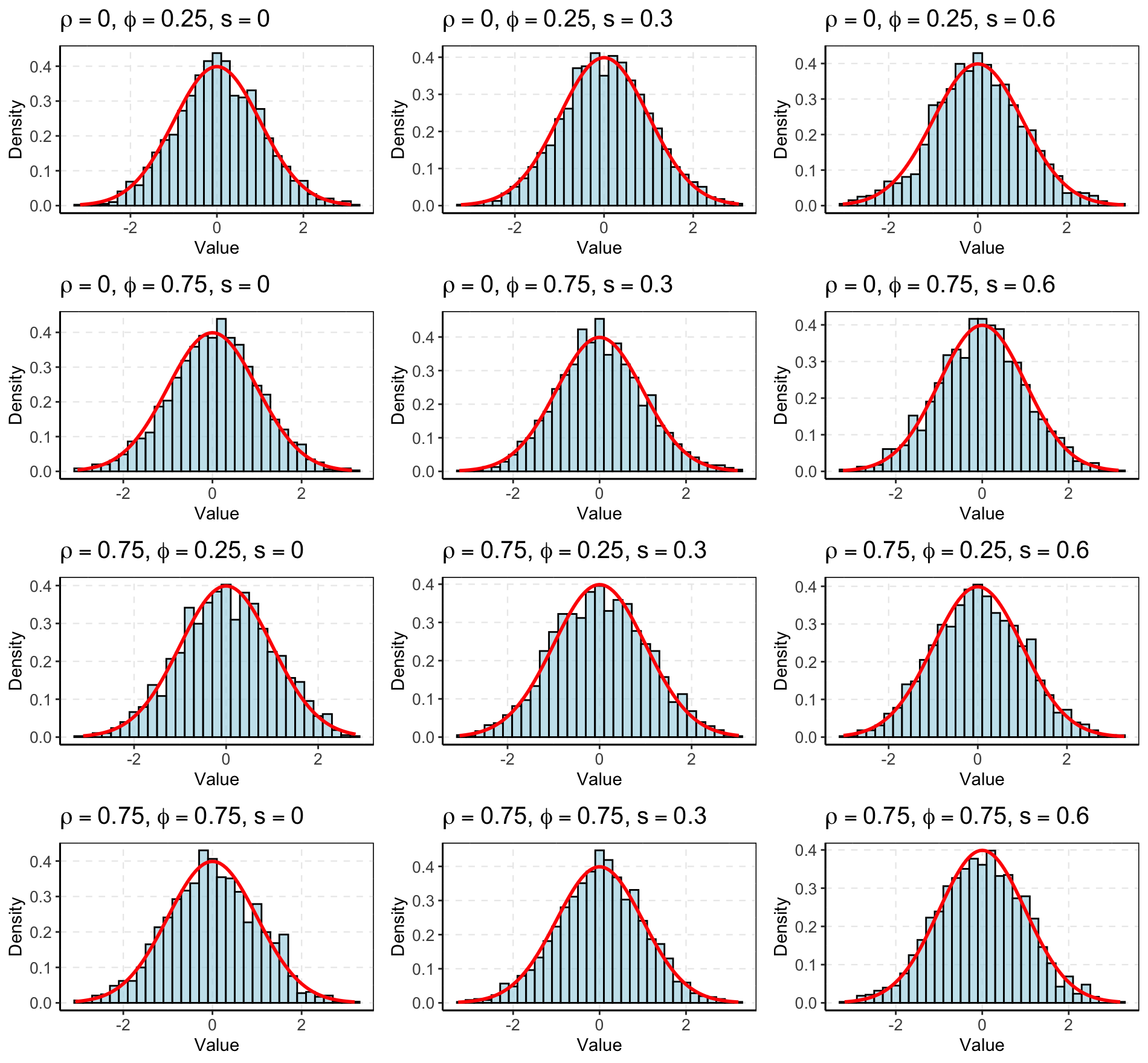}}
\caption{The histograms of $\{\hat{w}_{z_1,1} \tilde \tau_{z_1,1}^{-1}(\hb_2)\} \sqrt{n}\, \hb^{\T}_{2}\{\hat{\ab}_{z_1,1} - \text{sign}(\ab_{1,1}^{\T}\hat{\ab}_{z_1,1} ) \cdot \ab_{1,1} - \hat{\bvartheta}_{z_1,1} \}$  based on 2000 repetitions. The sample size $n = 400$. The red curve plots the density of $\mathcal{N}(0,1)$.}
\label{fig:normality-h2-hist-iter-est-longrun}
\end{figure}

\section{Additional simulation and empirical results}\label{sec: supp-addition simulation}

\subsection{Additional simulation results}\label{sec: addtional simulation results in main paper}

This section presents additional results for the numerical studies  discussed in Section \ref{sec:numerical}. 
 Figures \ref{fig:normality-h1-hist-iter-est-longrun} and \ref{fig:normality-h2-hist-iter-est-longrun} 
present  the histograms of $\{\hat w_{z_1,1}\tilde \tau_{z_1,1}^{-1}(\hb_k)\}\sqrt{n}\,\hb_{k}^{\T}\{\hat{\ab}_{z_1,1}-\text{sign}(\ab_{1,1}^{\T}\hat{\ab}_{z_1,1})\cdot \ab_{1,1}-\hat{\bvartheta}_{z_1,1}\}$ for $k \in \{1,2\}$ and $n = 400$ based on 2000 repetitions, which verify the asymptotic normality of our iterative estimator based on the asymptotic variance estimation $\hat w_{z_1,1}^{-2} \tilde\tau_{z_1,1}^{2}(\hb)$.  
Figure \ref{fig:peak RAM} presents line plots of the average peak RAM values, with the shaded region  representing the corresponding standard deviations, based on 100 repetitions. 
Figures \ref{fig:Krobust-acc} and \ref{fig:Krobust-error} summarize the sensitivity analysis with respect to the tuning parameter $K$. Specifically, Figure \ref{fig:Krobust-acc} shows the relative frequency estimates of correctly selecting the number of factors by the log-ER estimator, and Figure \ref{fig:Krobust-error} reports the averages of  the estimation errors of Pro.iter and Pro.init.



\begin{figure}[htbp]
\centerline{\includegraphics[width= 12cm]{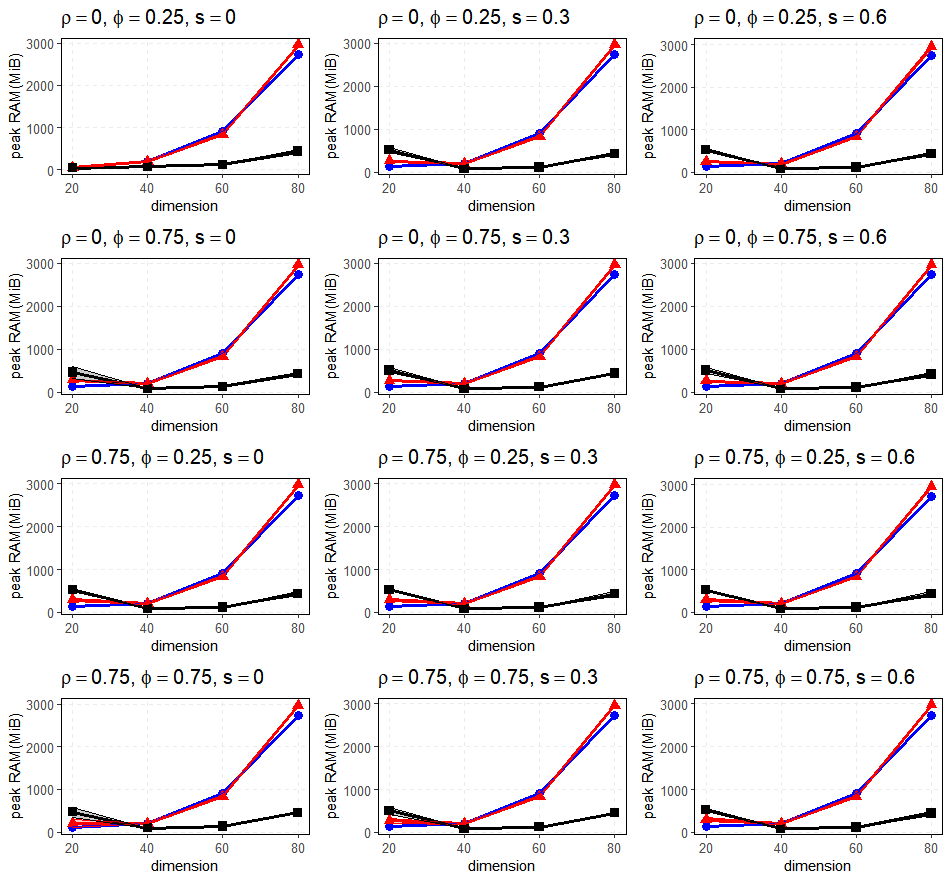}}
\caption{The lineplots for the averages and standard deviations (shaded region) of peak RAM based on 100 repetitions. The sample size $n = 400$. The legend is defined as follows: (i) Pro.iter initialized with Pro.init (\color{black}{$-$\scalebox{0.75}{$\blacksquare$}$-$}), (ii) HOPE ($\color{red}{-\blacktriangle-}$), and (iii) CC-ISO ($\color{blue}{-\bullet-}$).}
\label{fig:peak RAM}
\end{figure}

\begin{figure}[htbp]
\centerline{\includegraphics[width= 10cm]{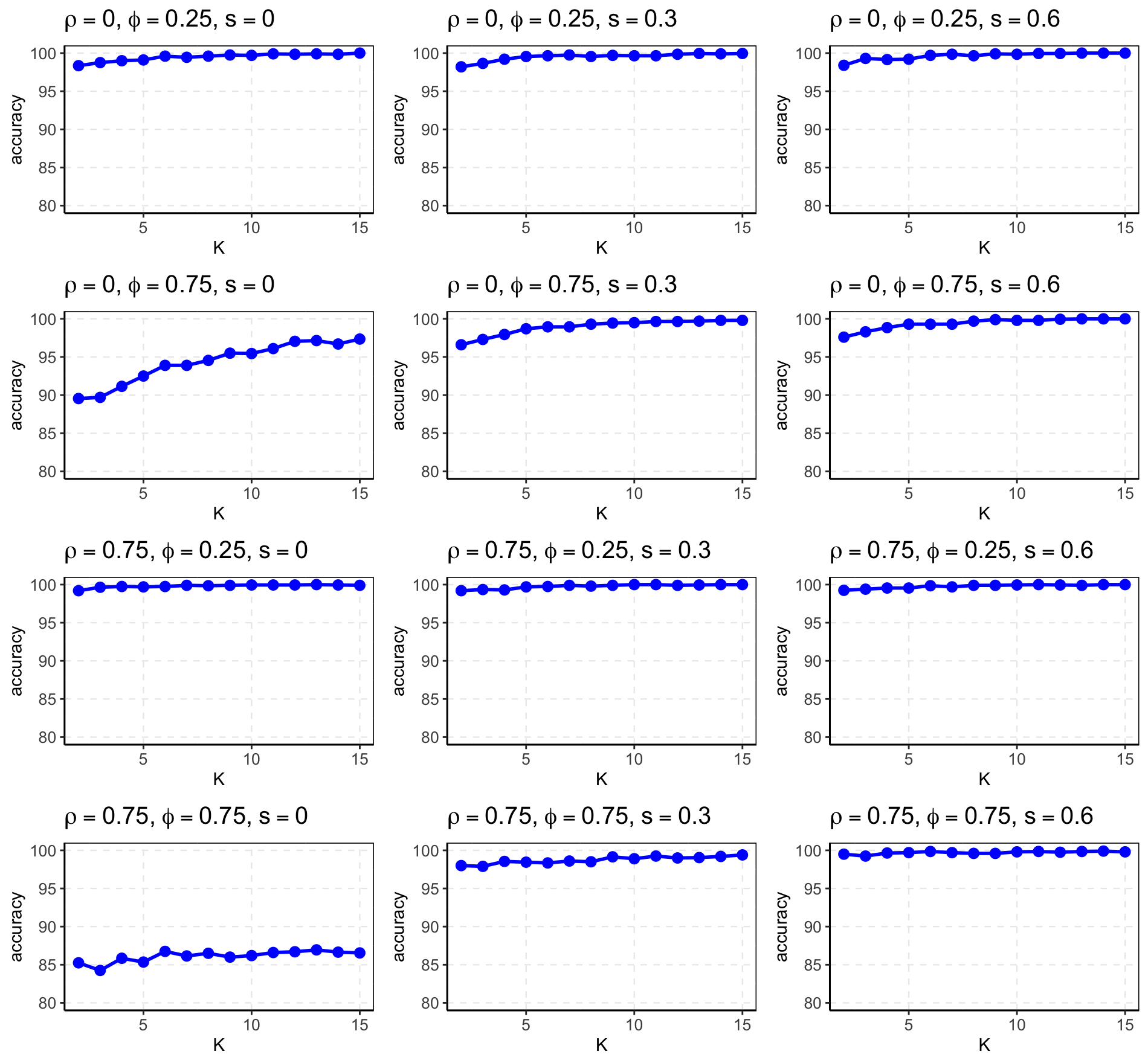}}
\caption{The lineplots for the relative frequency estimates of $\mathbb{P}(\tilde{r} = r)$ with respect to $K \in \{2,\ldots,15\}$ based on 2000 repetitions, where $\tilde{r}$ is determined by the log-ER estimator. The sample size $n = 400$.}
\label{fig:Krobust-acc}
\end{figure}

\begin{figure}[htbp]
\centerline{\includegraphics[width= 10cm]{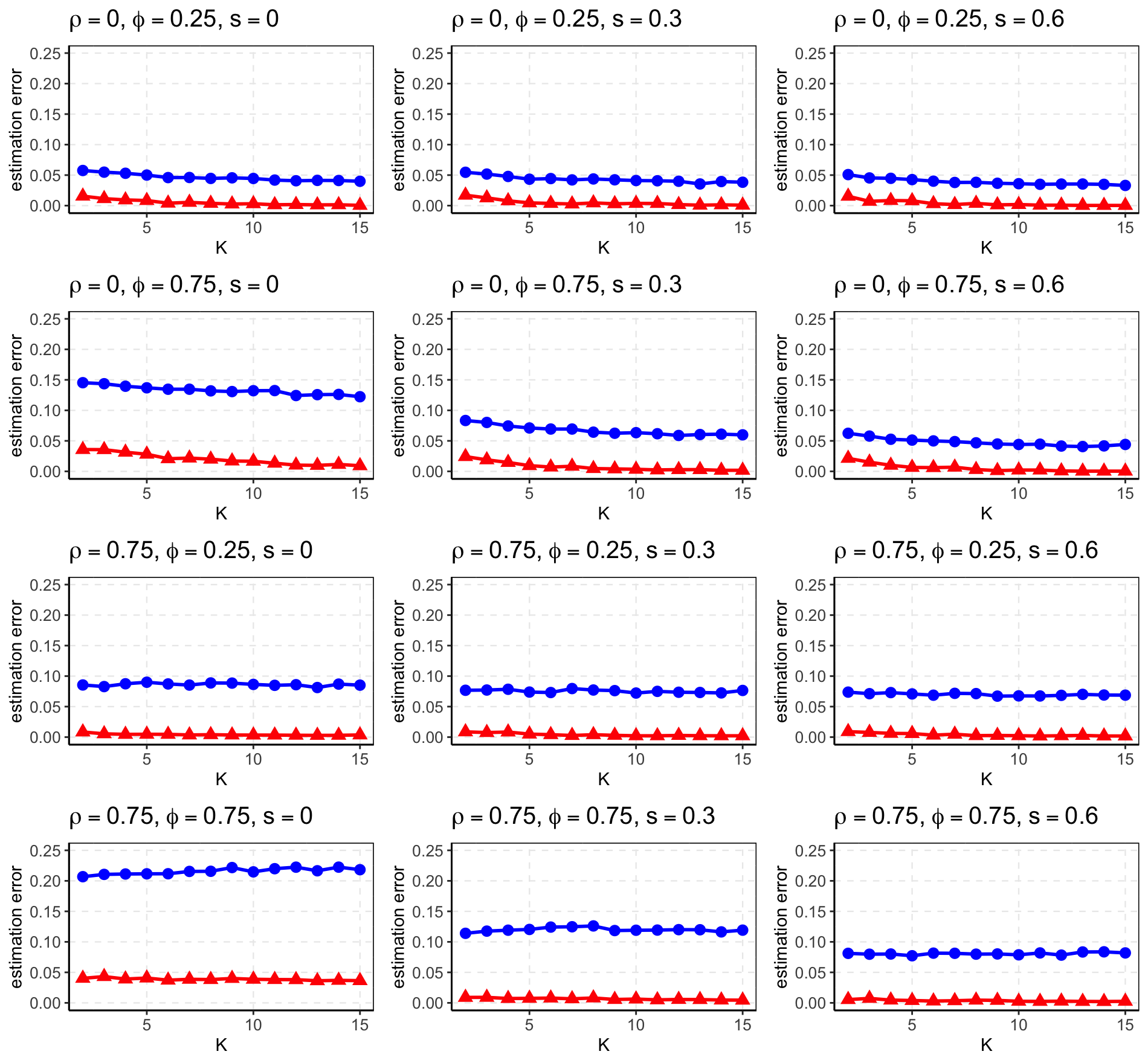}}
\caption{The lineplots for the averages of estimation errors \eqref{eq:estimation error} with respect to $K \in \{2,\ldots,15\}$ based on 2000 repetitions. The sample size $n = 400$. The legend is defined as follows: Pro.iter ($\color{red}{-\blacktriangle-}$) and Pro.init ($\color{blue}{-\bullet-}$).}
\label{fig:Krobust-error}
\end{figure}

\clearpage
\newpage
\subsection{Additional results for the analysis of air pollution data}\label{sec: addtional empirical results in main paper}

This section provides some additional results for the real data analysis in Section~\ref{sec:application}. Figure~\ref{fig:app-air-timeseries} shows the time series plots of pollutant concentration changes used in the analysis.  

\begin{figure}[htbp]
\centerline{\includegraphics[width= 17 cm]{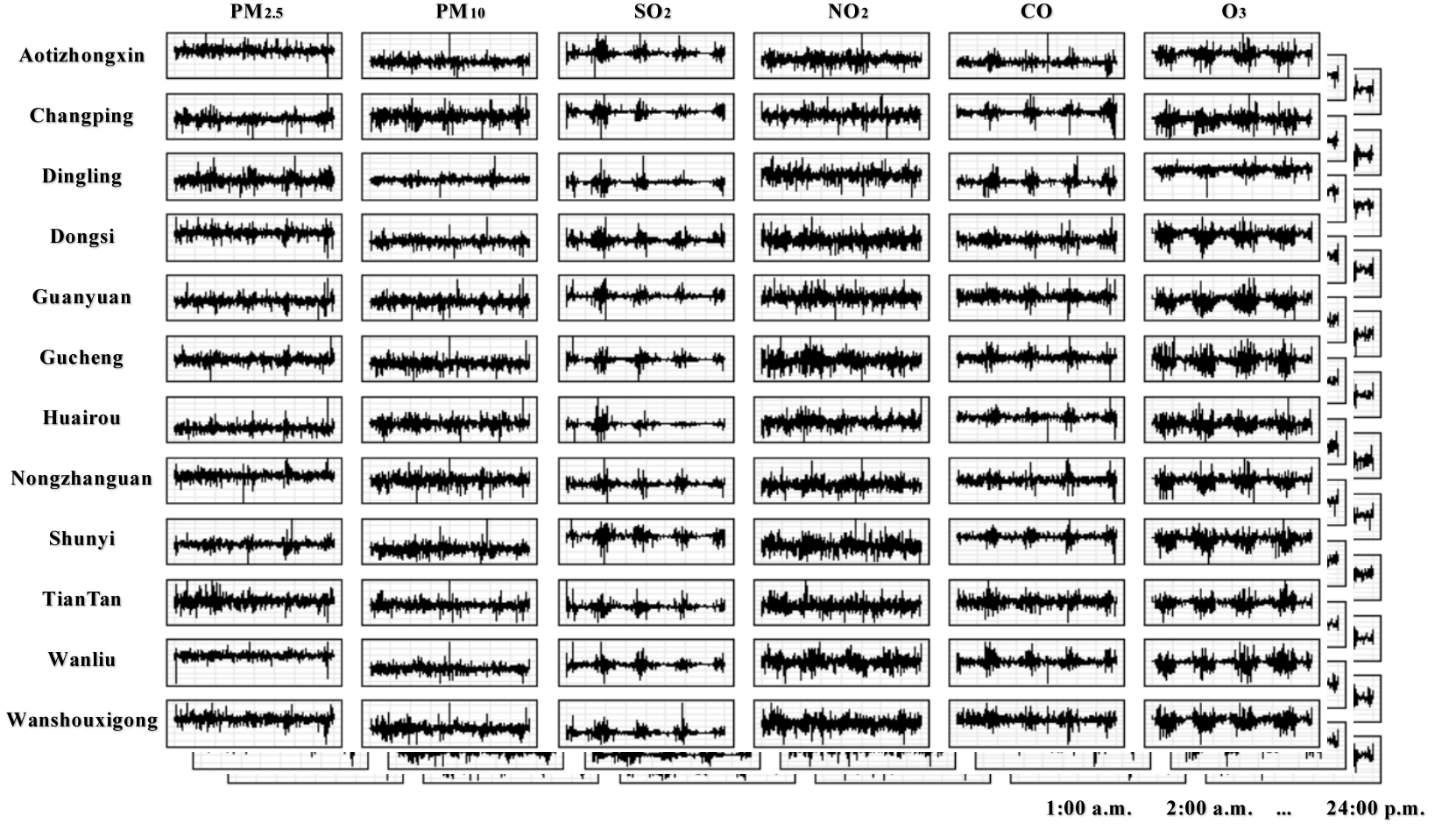}}
\caption{The time series plots of the concentration change of pollutants.}
\label{fig:app-air-timeseries}
\end{figure}

We first compare the results obtained based on different estimation methods.
Since all methods (Pro.iter, Pro.init, HOPE, and CC-ISO) are identifiable only up to the reflection and permutation indeterminacy, we need to apply suitable reordering and sign adjustments for each method  to make the results comparable. 
Based on the analysis in Section~\ref{sec:application}, the proposed iterative method (Pro.iter) identifies two interpretable pollution factors, namely the \textit{ozone-related factor} and the \textit{general pollution factor}. For each of the other three methods (Pro.init, HOPE, and CC-ISO), we align the associated estimates with these two benchmark patterns. 
More specifically, for each of these three methods,  to resolve the permutation indeterminacy, we reorder the estimated factors and estimated loading vectors so that the first factor corresponds to the \textit{ozone-related factor} ($i=1$) and the second factor corresponds to the \textit{general pollution factor} ($i=2$). To resolve the reflection indeterminacy, we consider the following sign conventions.  For the pollution-variable mode ($j=2$), we require the estimated loading of the \textit{ozone-related factor} ($i=1$) on O$_3$ and the estimated loading of the \textit{general pollution factor} ($i=2$) on PM$_{2.5}$ to be positive. For the monitoring-station mode ($j=1$), we require the estimated loadings of both factors to be positive at the first monitoring station (Aotizhongxin). For the diurnal mode ($j=3$), we require the first elements of the estimated loading vectors for both factors to be negative.

\begin{table}[htbp]
\footnotesize
\centering
\renewcommand{\arraystretch}{1.15}
\setlength{\tabcolsep}{4pt}
\caption{Estimations of the loading vectors $\ab_{i,2}\in\mathbb{R}^{6}$ for the pollution-variable mode based on four methods: Pro.iter, Pro.init, HOPE, and CC-ISO. For Pro.iter, standard errors reported in parentheses are calculated based on the asymptotic variance estimation $\hat{w}_{i,j}^{-2}\hat \tau_{i,j}^{2}(\hb)$. $^{*}$, $^{**}$, and $^{***}$ indicate significance at the levels 5\%, 1\%, and 1\textperthousand, respectively, based on two-sided $t$-tests. }
\label{table:app-loading-a2-all}
\begin{tabular}{c|cccc|cccc}
\hline\hline
& \multicolumn{4}{c|}{$i=1$} & \multicolumn{4}{c}{$i=2$} \\
Pollutant
& Pro.iter & Pro.init & HOPE & CC-ISO
& Pro.iter & Pro.init & HOPE & CC-ISO \\
\hline
PM$_{2.5}$
& 0.008 (0.015)
& $-$0.102
& $-$0.030
& 0.015
& 0.659$^{***}$ (0.035)
& 0.625
& 0.485
& 0.511
\\
PM$_{10}$
& $-$0.021 (0.013)
& $-$0.154
& $-$0.071
& $-$0.045
& 0.430$^{***}$ (0.025)
& 0.423
& 0.400
& 0.452
\\
SO$_2$
& 0.049$^{**}$ (0.017)
& $-$0.219
& 0.028
& 0.064
& 0.304$^{***}$ (0.036)
& 0.288
& 0.263
& 0.276
\\
NO$_2$
& $-$0.236$^{***}$ (0.016)
& $-$0.339
& $-$0.317
& $-$0.176
& 0.289$^{***}$ (0.053)
& 0.324
& 0.542
& 0.488
\\
CO
& 0.182$^{***}$ (0.012)
& $-$0.008
& $-$0.132
& $-$0.080
& 0.452$^{***}$ (0.031)
& 0.492
& 0.467
& 0.459
\\
O$_3$
& 0.953$^{***}$ (0.002)
& 0.896
& 0.936
& 0.978
& 0.009 (0.083)
& $-$0.008
& $-$0.151
& $-$0.102
\\
\hline\hline
\end{tabular}
\end{table}

As shown in Table \ref{table:app-loading-a2-all}, the three  methods (Pro.init, HOPE, and CC-ISO) also identify two interpretable pollution patterns for the pollution-variable mode ($j=2$), namely the \textit{ozone-related factor} ($i=1$) and the \textit{general pollution factor} ($i=2$). 
The estimations of the loading vectors $\ab_{i,2}\in\mathbb{R}^6$ for the pollution-variable mode based on different methods (Pro.iter, Pro.init, HOPE, and CC-ISO) have similar patterns. 
The main differences among the four methods arise in the estimated factor loadings for the monitoring-station mode ($j=1$) and the diurnal mode ($j=3$), as well as in the estimated factor process.
 
\begin{table}[htbp]
\centering
\footnotesize
\caption{Estimations of the loading vectors $\ab_{i,1}\in\mathbb{R}^{12}$ for the monitoring-station mode based on four methods: Pro.iter, Pro.init, HOPE, and CC-ISO. For Pro.iter, standard errors reported in parentheses are calculated based on the asymptotic variance estimation $\hat{w}_{i,j}^{-2}\hat \tau_{i,j}^{2}(\hb)$. $^{*}$, $^{**}$, and $^{***}$ indicate significance at the levels 5\%, 1\%, and 1\textperthousand, respectively, based on two-sided $t$-tests. }
\label{table:app-loading-a1-all}
\renewcommand{\arraystretch}{1.15}
\setlength{\tabcolsep}{3pt}
\begin{tabular}{c|cccc|cccc}
\hline\hline
& \multicolumn{4}{c|}{$i=1$} & \multicolumn{4}{c}{$i=2$} \\
Station
& Pro.iter & Pro.init & HOPE & CC-ISO
& Pro.iter & Pro.init & HOPE & CC-ISO \\
\hline
Aotizhongxin
& 0.300$^{***}$ (0.004)
& 0.130
& 0.290
& 0.292
& 0.341$^{***}$ (0.030)
& 0.275
& 0.190
& 0.361
\\
Changping
& 0.265$^{***}$ (0.005)
& 0.392
& 0.274
& 0.276
& 0.049 (0.048)
& 0.234
& $-$0.504
& 0.088
\\
Dingling
& 0.263$^{***}$ (0.006)
& 0.390
& 0.280
& 0.277
& 0.059 (0.043)
& 0.234
& $-$0.142
& 0.044
\\
Dongsi
& 0.278$^{***}$ (0.005)
& 0.039
& 0.293
& 0.290
& 0.288$^{***}$ (0.039)
& 0.308
& 0.166
& 0.395
\\
Guanyuan
& 0.289$^{***}$ (0.005)
& 0.019
& 0.292
& 0.296
& 0.333$^{***}$ (0.032)
& 0.314
& 0.055
& 0.394
\\
Gucheng
& 0.294$^{***}$ (0.005)
& 0.202
& 0.279
& 0.277
& 0.329$^{***}$ (0.042)
& 0.293
& $-$0.300
& 0.174
\\
Huairou
& 0.267$^{***}$ (0.007)
& 0.371
& 0.294
& 0.291
& 0.215$^{***}$ (0.050)
& 0.291
& $-$0.649
& 0.058
\\
Nongzhanguan
& 0.299$^{***}$ (0.005)
& 0.113
& 0.297
& 0.298
& 0.362$^{***}$ (0.034)
& 0.308
& 0.298
& 0.390
\\
Shunyi
& 0.287$^{***}$ (0.006)
& 0.276
& 0.276
& 0.277
& 0.476$^{***}$ (0.059)
& 0.265
& $-$0.088
& 0.187
\\
Tiantan
& 0.289$^{***}$ (0.005)
& 0.380
& 0.292
& 0.293
& 0.211$^{***}$ (0.040)
& 0.317
& 0.175
& 0.357
\\
Wanliu
& 0.325$^{***}$ (0.006)
& 0.509
& 0.296
& 0.295
& 0.194$^{***}$ (0.033)
& 0.272
& $-$0.108
& 0.268
\\
Wanshouxigong
& 0.301$^{***}$ (0.004)
& 0.069
& 0.300
& 0.300
& 0.298$^{***}$ (0.041)
& 0.334
& 0.096
& 0.360
\\
\hline\hline
\end{tabular}
\end{table}

 For the monitoring-station mode ($j=1$), Table~\ref{table:app-loading-a1-all} shows that, for the \textit{general pollution factor}, Pro.iter and CC-ISO produce broadly similar spatial patterns: all estimated loadings are positive, and both methods assign relatively small estimated loadings to suburban stations such as Dingling and Changping. This is consistent with the interpretation in Section \ref{sec:application} that these stations are located in less polluted areas.  
    By contrast, Pro.init yields a more homogeneous pattern for the \textit{general pollution factor}, with the estimated loadings of similar magnitude across stations. The HOPE estimates are less regular and include several negative estimated loadings, making the spatial pattern harder to interpret. For the \textit{ozone-related factor}, all four methods give positive estimated loading vectors.

  For the diurnal mode ($j=3$), Figure~\ref{fig:beijing-hourly-all} shows that the four methods give broadly similar estimated loadings for the \textit{ozone-related factor}.  The difference is more pronounced for the \textit{general pollution factor}. The estimates of Pro.iter exhibit a clear bimodal pattern, with a morning peak and a higher evening peak, which is consistent with the daily cycle of human activity discussed in Section \ref{sec:application}. By contrast, the factor loadings estimated by Pro.init, HOPE, and CC-ISO exhibit less plausible bimodal patterns for the \textit{general pollution factor}. In particular, the estimated factor loadings of these three methods begin to rise as early as around 3:00 a.m., when human activity is still very limited. Moreover, they do not capture the daytime accumulation pattern of general air pollutants; instead, their estimated loadings decline noticeably in the afternoon, which is less consistent with typical daily pollution patterns in urban environments. By comparison, the estimates of Pro.iter are more closely aligned with the known daily patterns of human activity.

\begin{figure}[htbp]
 
    \centering
\subfigure[ozone-related factor $(i = 1)$]{\includegraphics[width=0.45\textwidth]{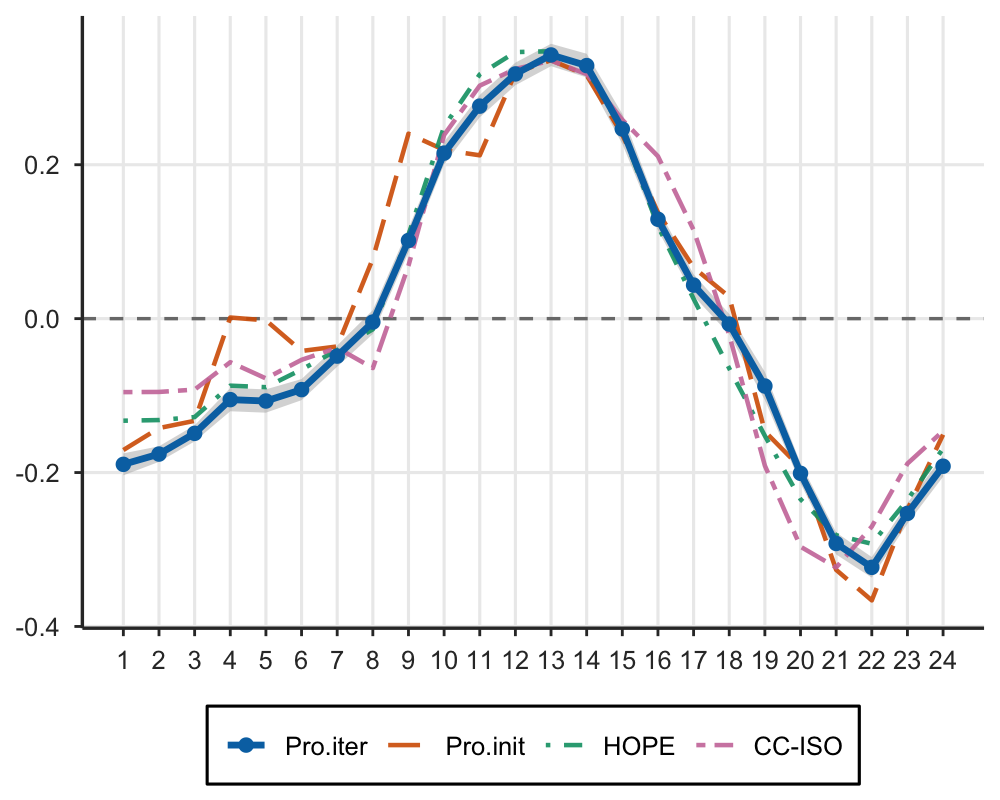}}
\subfigure[general pollution factor $(i = 2)$]{\includegraphics[width=0.45\textwidth]{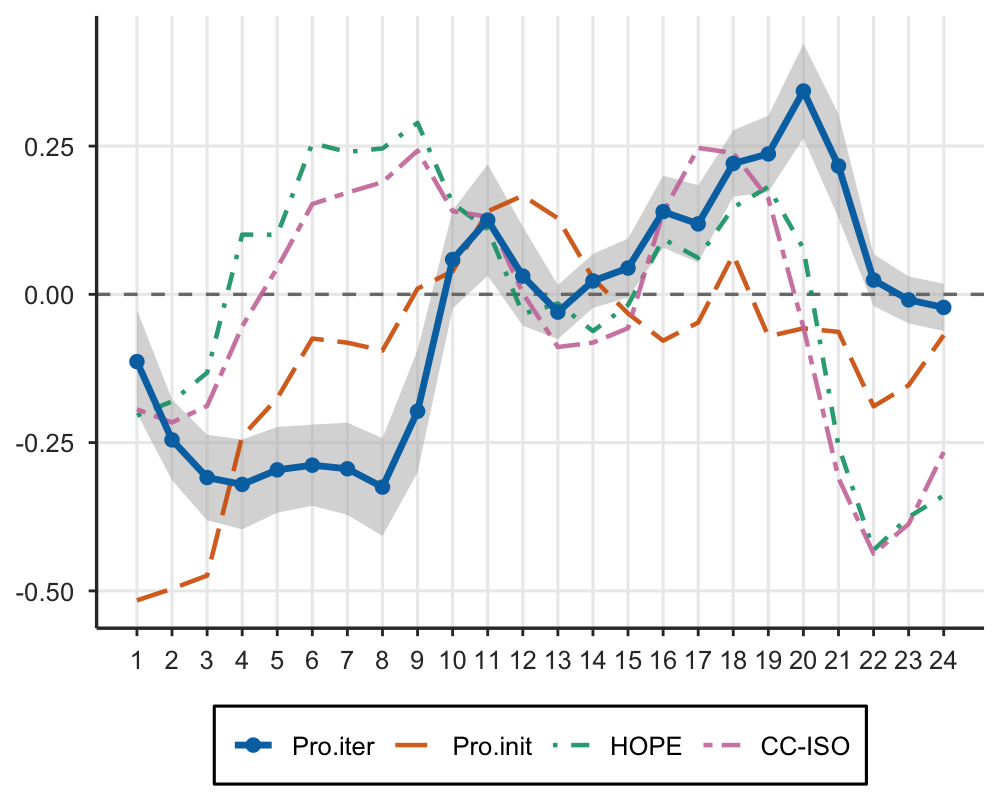}}
   
  \caption{Estimations of the loading vectors $\ab_{i,3}\in\mathbb{R}^{24}$ for the diurnal mode based on four methods: Pro.iter, Pro.init, HOPE, and CC-ISO. The gray shaded region represents the pointwise 95\% confidence interval for Pro.iter. Standard errors are calculated based on the asymptotic variance estimation  $\hat{w}_{i,j}^{-2}\hat \tau_{i,j}^{2}(\hb)$.}
    \label{fig:beijing-hourly-all}

\end{figure}

\begin{figure}[htbp]
    \centering
\subfigure[ozone-related factor $(i = 1)$]{\includegraphics[width=0.48\textwidth]{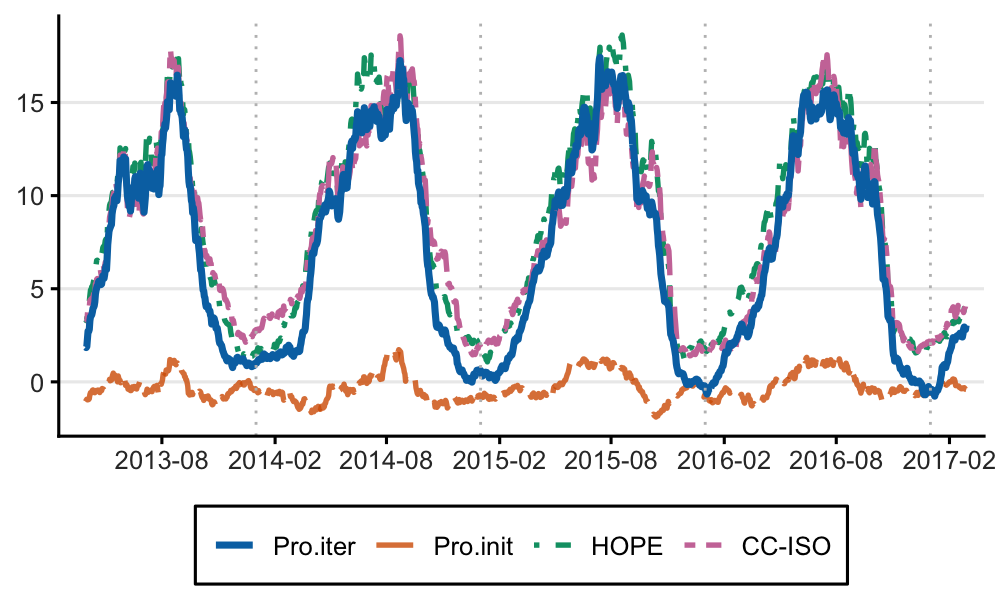}}
\subfigure[general pollution factor $(i = 2)$]{\includegraphics[width=0.48\textwidth]{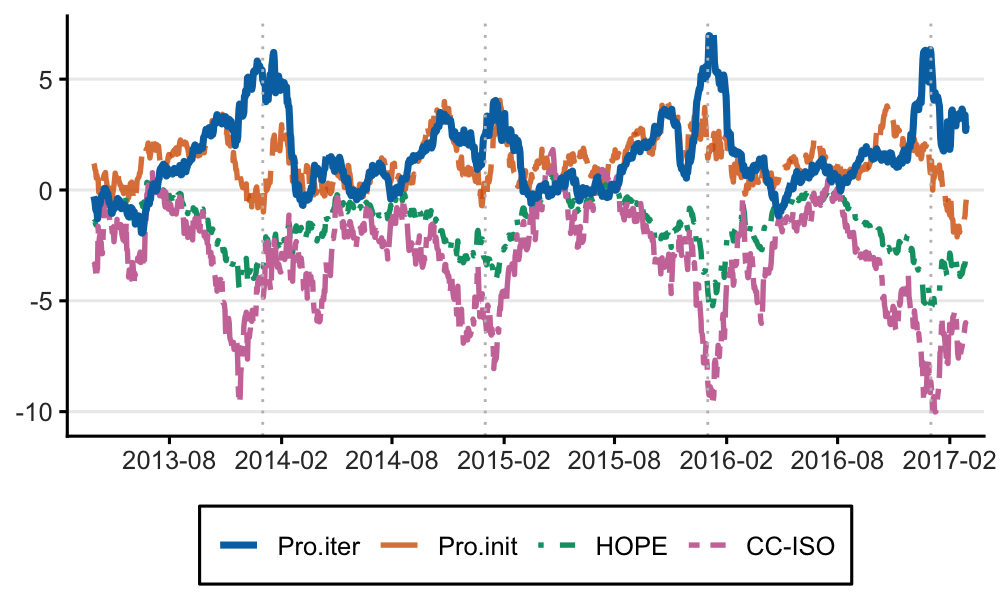}}
    \caption{The time series plots of 30-day one-sided simple moving averages of the two estimated latent factors based on four methods: Pro.iter, Pro.init, HOPE, and CC-ISO.}
    \label{fig:beijing-seasonal-all}
\end{figure}

 As shown in Figure~\ref{fig:beijing-seasonal-all}, the estimations of the \textit{ozone-related factor} based on Pro.iter, HOPE, and CC-ISO  exhibit a clear seasonal cycle,  while this seasonal pattern is much weaker in  that of Pro.init. For the \textit{general pollution factor}, the factor processes estimated by HOPE and CC-ISO exhibit patterns that are nearly opposite to that of Pro.iter. In particular, both methods produce relatively low values in winter, which would imply lower concentrations of general pollutants (PM$_{2.5}$, PM$_{10}$, SO$_2$, NO$_2$, and CO) during the cold season. This is difficult to reconcile with the well-known winter pollution pattern in Beijing and is therefore less plausible.

Above comparisons show that, although  the four methods can recover the same two pollution patterns, the results based on Pro.iter are more interpretable in terms of their spatial, diurnal, and seasonal structures. This provides further empirical support for the effectiveness of our proposed method.

In addition, we apply the empirical moment test of \citeS{trapani2016testing-app} to each marginal series $\{y_{t,\ell_1,\ell_2,\ell_3}\}_{t=1}^n$ ($\ell_1 \in [12], \ell_2 \in [6], \ell_3 \in [24]$).  More specifically, for each marginal series, we proceed sequentially as follows: we first test the existence of the 8-th moment; if it is supported by the data, we stop; otherwise, we test the 6-th moment, then finally the 4-th moment. This allows us to determine, for each marginal series, the highest empirically supported finite moment order.  We also repeat the same analysis for the 1\% winsorized data,  obtained by winsorizing each marginal series at the 5\textperthousand\ lower tail and the 5\textperthousand\ upper tail. The results reported in Table \ref{table:moment test} show that the raw data display some heavy-tailedness. To assess the reliability of our main empirical results, we rerun the real data analysis using the 1\% winsorized data.  The detailed results are reported in Table \ref{table:app-loading-a2-robust} and Figures \ref{fig:beijing-station-robust}--\ref{fig:beijing-seasonal-robust}. In comparison with the results in Table \ref{table:app-loading-a2} and Figures \ref{fig:beijing-station}--\ref{fig:beijing-seasonal}, we can conclude that the resulting estimated factor loadings, factors, and substantive interpretations remain essentially unchanged, indicating that our main empirical conclusions are robust.

\begin{table}[htbp]
\caption{Proportions of marginal series $\{y_{t,\ell_1,\ell_2,\ell_3}\}_{t=1}^n$ ($\ell_1 \in [12], \ell_2 \in [6], \ell_3 \in [24]$) whose highest empirically supported finite moment is at most $k$-th order, based on the test of \protect\citeS{trapani2016testing-app} at the 5\% significance level.}
\centering
\renewcommand\tabcolsep{8pt}
\label{table:moment test}
\begin{tabular}{c|cccc}
\hline\hline
            & $k < 4$ & $k = 4$ & $k = 6$ & $k = 8$  \\ \hline
raw data       & 28.76             & 54.40             & 12.96             & 3.88              \\
1\% winsorized data & 1.68              & 44.04             & 26.91             & 27.37            \\ \hline\hline
\end{tabular}
\end{table}

\begin{table}[htbp]
\centering
 
\renewcommand{\arraystretch}{1.15}
\setlength{\tabcolsep}{6pt}
\caption{Estimations of the loading vectors $\ab_{i,2}\in\mathbb{R}^6$ for the pollution-variable mode based on Pro.iter using the 1\% winsorized data. Standard errors reported in parentheses are calculated based on the asymptotic variance estimation $\hat{w}_{i,j}^{-2}\hat \tau_{i,j}^{2}(\hb)$. $^{*}$, $^{**}$, and $^{***}$ indicate significance at the levels 5\%, 1\%, and 1\textperthousand, respectively, based on two-sided $t$-tests. }  
\label{table:app-loading-a2-robust}
\begin{tabular}{c|cc}
\hline\hline
Pollutant & $i=1$ & $i=2$ \\
\hline
PM$_{2.5}$ & 0.014 (0.015) & 0.655$^{***}$ (0.034) \\
PM$_{10}$  & $-$0.019 (0.013) & 0.436$^{***}$ (0.024) \\
SO$_2$     & 0.048$^{**}$ (0.016) & 0.284$^{***}$ (0.036) \\
NO$_2$     & $-$0.230$^{***}$ (0.015) & 0.295$^{***}$ (0.049) \\
CO         & 0.187$^{***}$ (0.012) & 0.461$^{***}$ (0.029) \\
O$_3$      & 0.954$^{***}$ (0.002) & 0.015 (0.078) \\
\hline\hline
\end{tabular}
\end{table}

\begin{figure}[htbp]
    \centering
\subfigure[ozone-related factor $(i = 1)$]{\includegraphics[width=0.45\textwidth]{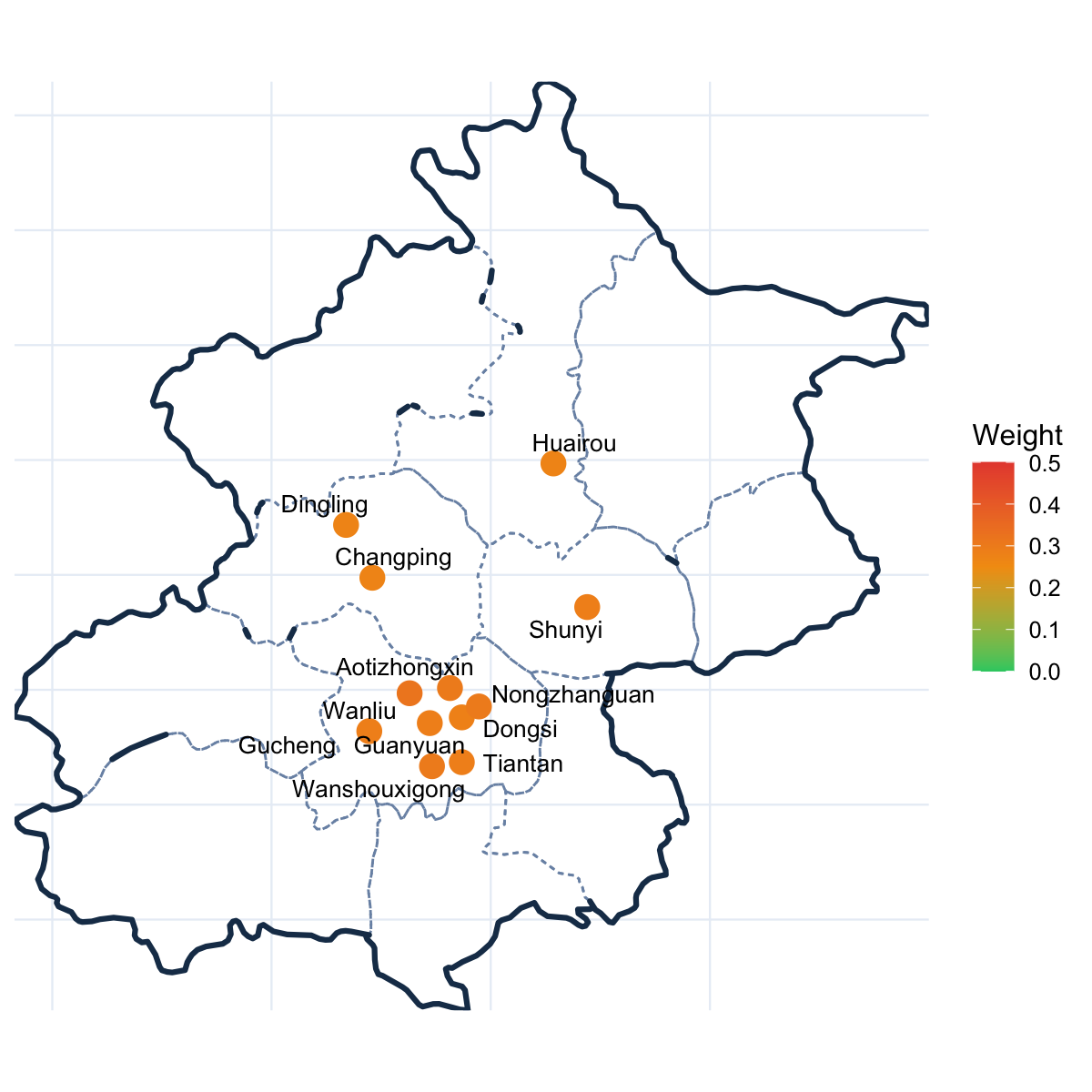}}
\subfigure[general pollution factor $(i = 2)$]{\includegraphics[width=0.45\textwidth]{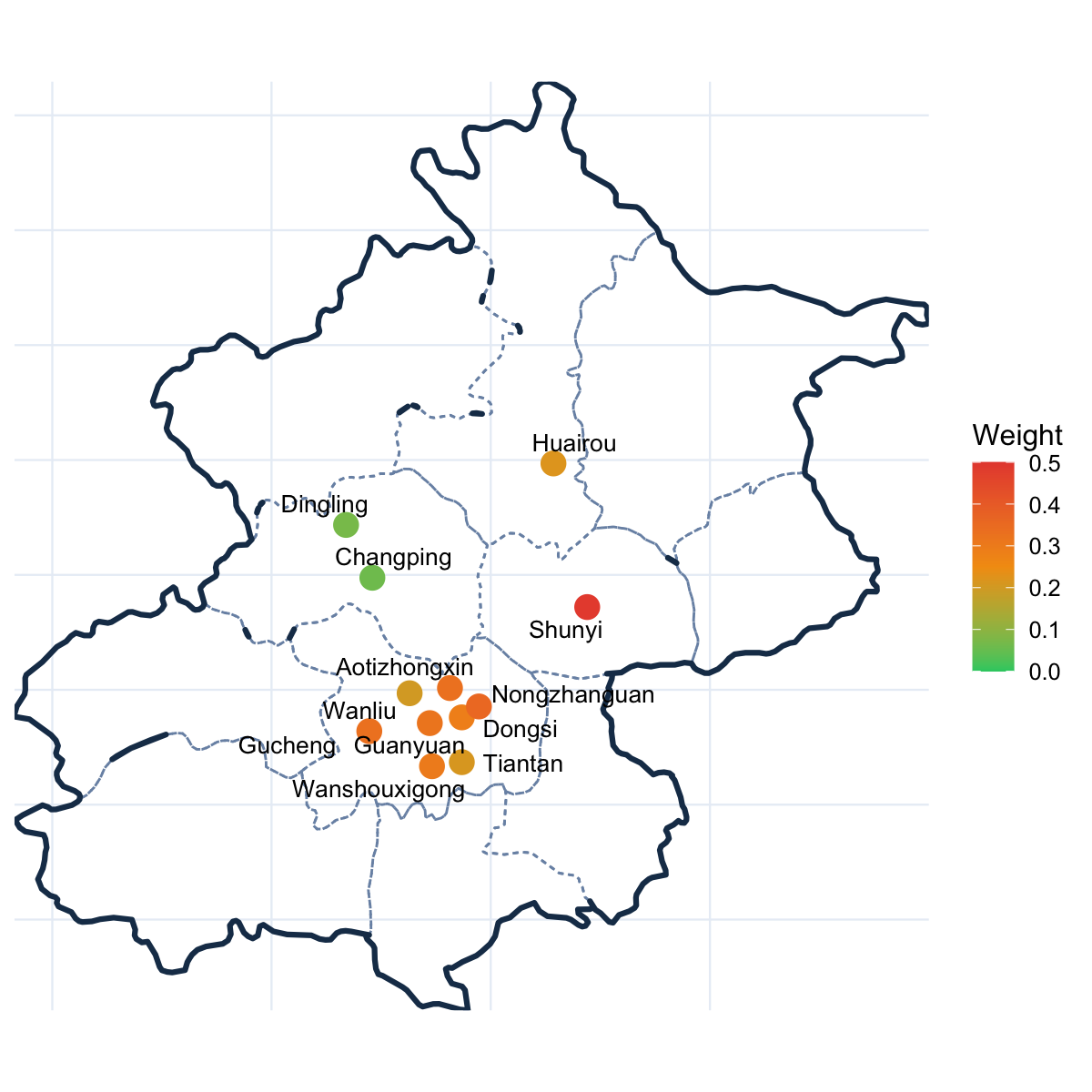}}
    \caption{Estimations of the loading vectors $\ab_{i,1} \in \mathbb{R}^{12}$ for the   monitoring-station mode based on Pro.iter using the 1\% winsorized data.}
    \label{fig:beijing-station-robust}
\end{figure}

\begin{figure}[htbp]
 
    \centering
\subfigure[ozone-related factor $(i = 1)$]{\includegraphics[width=0.45\textwidth]{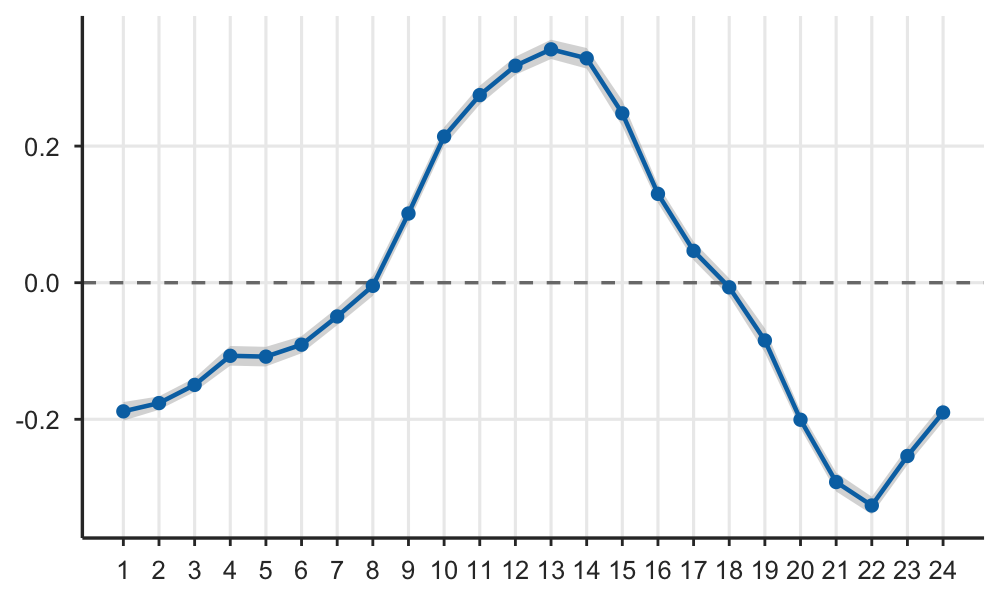}}
\subfigure[general pollution factor $(i = 2)$]{\includegraphics[width=0.45\textwidth]{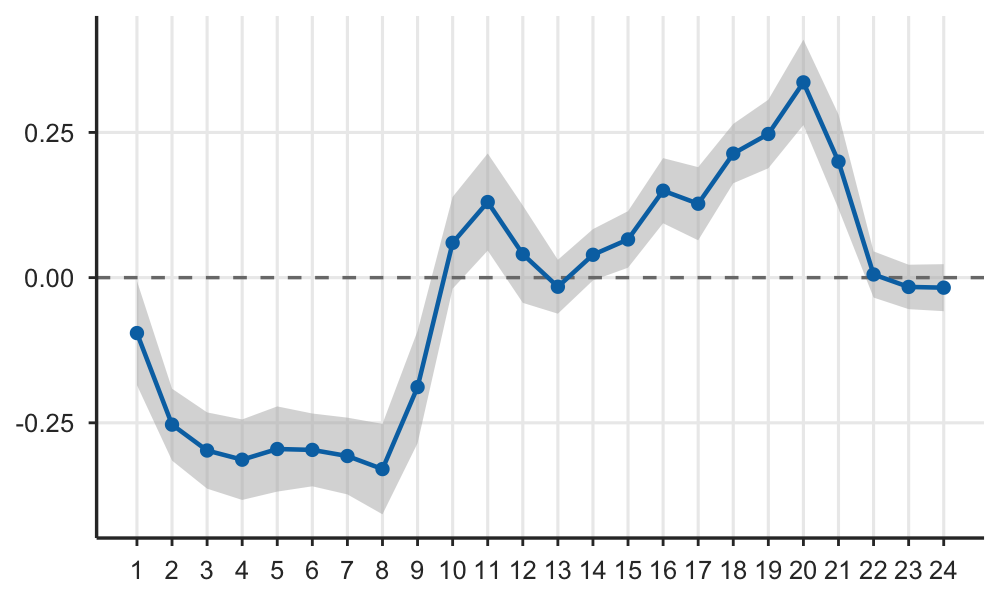}}
    \caption{Estimations of the loading vectors  $\ab_{i,3} \in \mathbb{R}^{24}$ for the diurnal mode based on Pro.iter using the 1\% winsorized data. The gray shaded region represents the pointwise 95\% confidence interval for the estimated loadings.  Standard errors are calculated based on the asymptotic variance estimation  $\hat{w}_{i,j}^{-2}\hat \tau_{i,j}^{2}(\hb)$.}
    \label{fig:beijing-hourly-robust}
 
\end{figure}

\begin{figure}[htbp]
    \centering
\subfigure[ozone-related factor $(i = 1)$]{\includegraphics[width=0.45\textwidth]{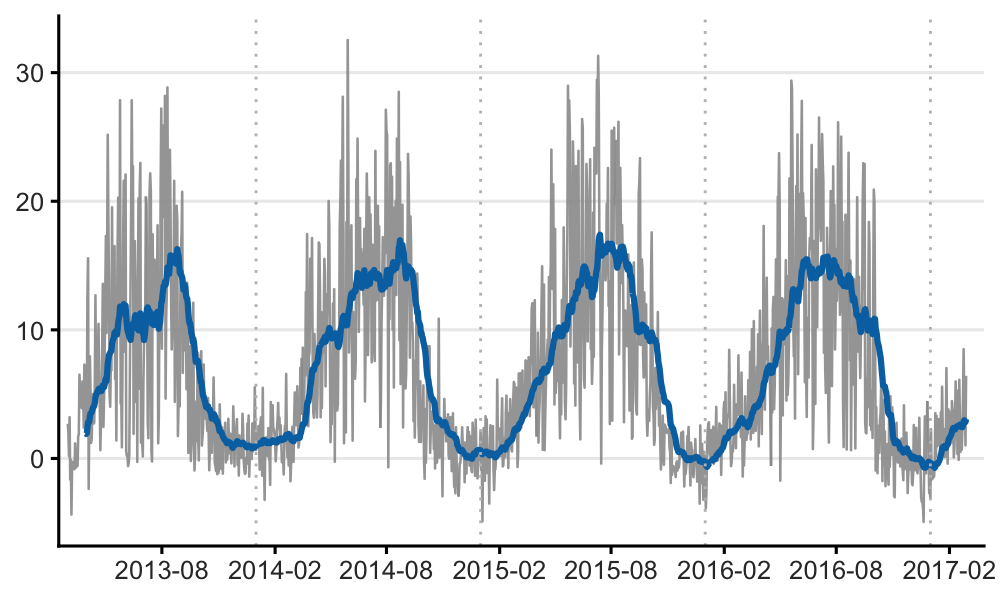}}
\subfigure[general pollution factor $(i = 2)$]{\includegraphics[width=0.45\textwidth]{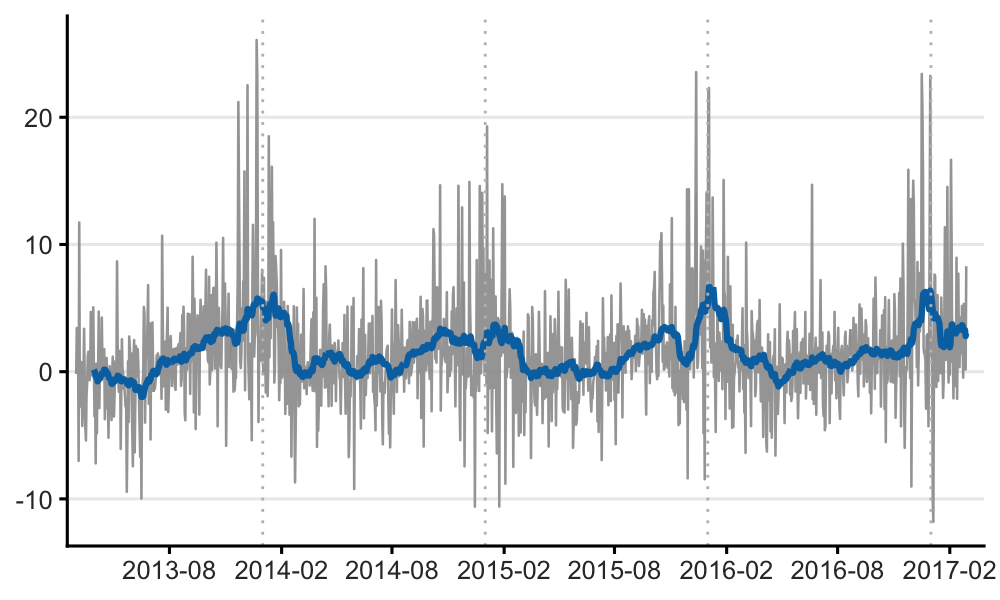}}
    \caption{The time series plots of the two
estimated latent factors based on Pro.iter using the 1\% winsorized data. The dark blue solid line represents the 30-day one-sided simple moving average.}
    \label{fig:beijing-seasonal-robust}
\end{figure}

\newpage

\subsection{Real data analysis: Fama–French 100 return data}\label{sec: app-famafrench}
In this section, we illustrate the proposed methods for the tensor CP-factor model \eqref{model cp} using the Fama–French 100 return series. We collect monthly returns from January 1964 to December 2021, yielding 69,600 observations over a total of 696 months. The dataset is obtained from \url{http://mba.tuck.dartmouth.edu/pages/faculty/ken.french/data_library.html}. The portfolios are constructed from the intersections of 10 size levels, denoted by $({\rm S}_{1}, \ldots, {\rm S}_{10})$, and 10 levels of the book-to-market equity ratio (BE), denoted by $({\rm BE}_{1}, \ldots, {\rm BE}_{10})$. The dataset contains a small number of missing values in the early years, which we set to zero.
Since all 100 series are clearly related to overall market conditions, following \citeS{wang2019factor-app}, we remove the influence of market effects prior to the empirical analysis by subtracting the corresponding monthly excess market return from each series. The market return data are obtained from the same source.

The 100 market-adjusted return series can be represented as a tensor time series $\mathcal{Y}_{t} = (y_{t,i,j})_{10\times 10}$ for $t\in[696]$ (i.e. $m=2$, $d_1=d_2=10$, $n=696$), where $y_{t,i,j}$ is the market-adjusted return at the $i$-th level of size ${\rm S}_{i}$ and the $j$-th level of the BE-ratio ${\rm BE}_{j}$ at time $t$. Figure \ref{fig:app-timeseries} shows the time series plots of the market-adjusted return series $\{y_{t,i,j}\}_{t=1}^n$ for $i,j \in[10]$. The rows in Figure \ref{fig:app-timeseries} correspond to the ten levels of size and the columns correspond to the ten levels of the BE-ratio.

\begin{figure}[htbp]
\centerline{\includegraphics[width= 16cm]{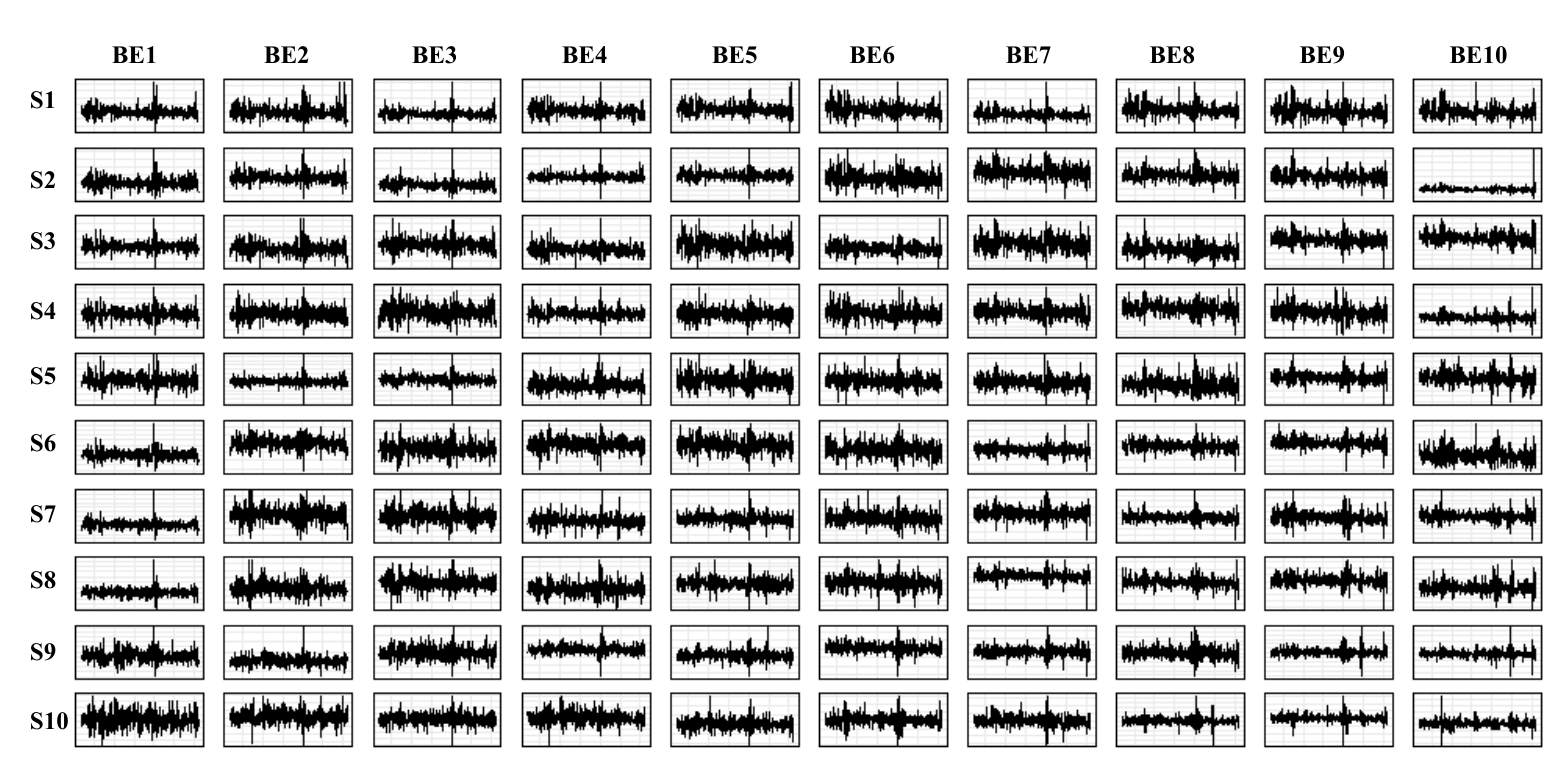}}
\caption{The time series plots of 100 market-adjusted returns formed on different levels of size (by rows) and book equity to market equity ratio (by columns). The horizontal axis represents time and the vertical axis represents the monthly returns.}
\label{fig:app-timeseries}
\end{figure}

We illustrate the usefulness of our methods by performing one- and two-step-ahead rolling forecasts for the 240 monthly observations in the last twenty years (2002--2021). To estimate the number of factors, following the selection of the tuning parameters in Section \ref{sec:tuning}, the log-ER method suggests $\tilde{r} = 1$ based on the data $\{\mathcal{Y}_t\}_{t=1}^{456}$, and we use $\tilde{r} = 1$ throughout the rolling forecasts. For each $s \in [240]$, let $\hat{\Ab}_j = \hat{\ab}_{1,j}$ for $j \in \{1,2\}$ be the estimated loading vectors based on the data $\{\mathcal{Y}_t\}_{t=s}^{455+s}$. We can then obtain the estimated factor series by $\hat f_{t} = (\hat\ab_{1,2}^{\MP} \otimes \hat\ab_{1,1}^{\MP})^{\T}\textup{vec}(\mathcal{Y}_t)$ for  $t \in \{s,\ldots,455+s\}$, where $\hat\ab_{1,j}^{\MP} 
 = \hat\Ab_j(\hat\Ab_j^\T \hat\Ab_j)^{-1}$ for $j \in \{1,2\}$. To produce the one-step-ahead forecast of $\mathcal{Y}_{456+s}$, denoted by 
$\hat{\mathcal{Y}}^{(1)}_{456+s} = (\hat{y}^{(1)}_{456+s,i,j})_{10\times 10}$, 
we model the factor process  $\{\hat{f}_t\}_{t=s}^{455+s}$.  Specifically, for each $s \in [240]$, we fit an AR model for  the factor process $\{\hat{f}_t\}_{t=s}^{455+s}$  with the order selected by the Akaike Information Criterion (AIC). For the two-step-ahead forecast, we repeat the above procedure using  $\{\mathcal{Y}_t\}_{t=s}^{454+s}$, and the forecast  $\hat{\mathcal{Y}}^{(2)}_{456+s} = (\hat{y}^{(2)}_{456+s,i,j})_{10\times 10}$  is then obtained by plugging the one-step-ahead forecasts into the fitted models.  Therefore, for each $s \in [240]$, we can obtain one- and two-step-ahead forecasts of $\mathcal{Y}_{456+s}$  based on the proposed one-pass method (Pro.init) and the iterative method (Pro.iter), respectively. 
For comparison, we can also fit $\{\mathcal{Y}_{t}\}_{t = s}^{455+s}$ and $\{\mathcal{Y}_{t}\}_{t = s}^{454+s}$ by the following methods and obtain the associated one-step  and two-step ahead forecasts:

\begin{itemize}
  \item (cPCA, HOPE) The composite PCA and High-Order Projection Estimator  in \citeS{han2024cp-app} with the recommended tuning parameter $h = 1$ therein. We adopt the rank parameter $\tilde{r} = 1$ based on the result of our proposed method, and fix $\tilde{r}=1$ in the rolling forecasts. We fit the obtained univariate time series by an AR model with the order determined by the AIC.

  \item (RP-PCA, CC-ISO) The Randomized Projection PCA and Contemporary Covariance-based Iterative Simultaneous Orthogonalization in \citeS{chen2026estimation-app}. We estimate the rank parameter $\tilde{r} = 1$ based on the data $\{\mathcal{Y}_{t}\}_{t = 1}^{456}$ through the unfolded eigenvalue ratio method therein, and fix $\tilde{r}=1$ in the rolling forecasts. The obtained univariate time series is fitted by an AR model with the order determined by the AIC.

   \item (RCP) The Refined CP method of \citeS{chang2023modelling-app} with the pre-determined parameter $K = 10$ therein. The associated rank in this method is estimated as $\hat{d} = 1$ based on $\{\mathcal{Y}_t\}_{t=1}^{456}$ and then fixed in the rolling forecasts. We fit the obtained univariate time series by an AR model with the order determined by the AIC. The method is implemented using  the R package \texttt{HDTSA}.

     \item (UCP) The Unified CP method of \citeS{chang2024unified-app} with the pre-determined parameters $K = 20$ and $\tilde{K} = 10$ therein. The associated ranks in this method are estimated as $(\hat{d},\hat{d}_1,\hat{d}_2) = (2,2,1)$ based on $\{\mathcal{Y}_t\}_{t=1}^{456}$ and then fixed in the rolling forecasts. We fit the obtained 2-dimensional time series by a VAR model with the order determined by the AIC. The method is implemented using  the R package \texttt{HDTSA}.

  \item (FAC) The matrix Tucker-factor model with the FAC method proposed by \citeS{wang2019factor-app}  with the pre-determined parameter $h_0 = 1$ as suggested therein. The associated ranks in this model are estimated as $(\hat{k}_1,\hat{k}_2) = (1,1)$ by the ratio estimators suggested therein based on $\{\mathcal{Y}_t\}_{t = 1}^{456}$, and are fixed in the rolling forecasts. We fit the obtained univariate time series by an AR model with the order determined by the AIC.

  \item (TOPUP, TIPUP) The Time series Outer-Product Unfolding Procedure and the Time series Inner-Product Unfolding Procedure proposed by \citeS{han2024tensor-app} for the matrix Tucker-factor model.  The associated ranks in this model are estimated as $(\hat{k}_1,\hat{k}_2) = (2,2)$  by the information criterion considered in \citeS{han2022rank-app} based on $\{\mathcal{Y}_t\}_{t = 1}^{456}$, and are fixed in the rolling forecasts. We fit the obtained 4-dimensional time series by a VAR model with the order determined by the AIC. The methods are implemented using  the R package \texttt{tensorTS}.
 
  \item (MAR) The matrix-AR(1) model of \citeS{chen2021autoregressive-app}.

 \item (TS-PCA) Apply the principal component analysis for time series proposed by \citeS{chang2018principal-app} to the 100-dimensional time series $\{\textup{vec}({\mathcal{Y}}_{t})\}_{t = s}^{455+s}$ and $\{\textup{vec}({\mathcal{Y}}_{t})\}_{t = s}^{454+s}$, respectively, to obtain the associated one-step and two-step ahead forecasts. The method is implemented using the R package \texttt{HDTSA}. For the obtained univariate time series, we fit it by an AR model with the order determined by the AIC. For the obtained multivariate time series, we fit it by a VAR model with the order determined by the AIC. 
 
   \item (UniAR) Fit each of 100 component time series by an AR model with the order determined by the AIC.
\end{itemize}

For each $s \in [240]$, the one-step ahead forecasting performance  is evaluated by the $\textup{rRMSE}(s)$ and $\textup{rMAE}(s)$ defined as
\begin{gather*}
    \text{rRMSE}(s)  = \bigg\{ \frac{1}{100} \sum_{i = 1}^{10}\sum_{j = 1}^{10}|\hat{y}^{(1)}_{456+s,i,j} - y_{456+s,i,j}|^2  \bigg\}^{1/2}\,,\\
    \text{rMAE}(s)   =  \frac{1}{100} \sum_{i = 1}^{10}\sum_{j = 1}^{10} |\hat{y}^{(1)}_{456+s,i,j} - y_{456+s,i,j}| \,.
\end{gather*}
For the two-step-ahead forecast, the corresponding $\textup{rRMSE}(s)$ and $\textup{rMAE}(s)$ are defined analogously. Table \ref{table:app-forecast} reports the averages of $\{\textup{rRMSE}(s)\}_{s=1}^{240}$ and $\{\textup{rMAE}(s)\}_{s=1}^{240}$, denoted by $\textup{rRMSE}$ and $\textup{rMAE}$, respectively. The standard deviations of $\{\textup{rRMSE}(s)\}_{s=1}^{240}$ and $\{\textup{rMAE}(s)\}_{s=1}^{240}$ are reported in parentheses. The results show that our proposed methods are  promising for forecasting financial returns. In particular, Pro.iter achieves the smallest forecasting errors in the one-step-ahead forecasts, and Pro.init performs very competitively to Pro.iter. For the two-step-ahead forecasts, although the best performance is attained by UCP, both Pro.iter and Pro.init perform  very closely to UCP.  More broadly, the tensor CP-factor methods without  uncorrelated factor assumption (Pro.iter, Pro.init, RCP, and UCP) outperform the methods that rely on this assumption (cPCA and HOPE). It is well known that financial data exhibit strong cross-sectional dependence, and often involve highly correlated latent factors. 
The results of Table \ref{table:app-forecast} suggest that allowing correlated factors is important for capturing the underlying dependence structure in practice, which provides further evidence on the applicability of our proposed methods.

\begin{table}[htbp]
\centering
\scriptsize
\renewcommand{\arraystretch}{1.15}
\setlength{\tabcolsep}{2.5pt}
\caption{Average forecasting errors and standard deviations (in parentheses) based on different methods. Bold numbers indicate the smallest average forecasting error among all methods.}
\label{table:app-forecast}
\begin{tabular}{cc|cccc}
\hline \hline 
\multicolumn{2}{c|}{Method}                                                                                                                                & one-step rRMSE                            & one-step rMAE                             & two-step rRMSE                            & two-step rMAE                             \\ \hline
\multicolumn{1}{c|}{\multirow{6}{*}{\begin{tabular}[c]{@{}c@{}}Tensor CP-factor methods\\ without uncorrelated factor assumption\end{tabular}}} & Pro.iter & \textbf{3.4847} (1.6003) & \textbf{2.6623} (1.1455) & 3.4874 (1.5905)                           & 2.6683 (1.1336)                           \\
\multicolumn{1}{c|}{}                                                                                                                           & Pro.init & 3.4890 (1.5960)                           & 2.6646 (1.1409)                           & 3.4977 (1.5846)                           & 2.6735 (1.1305)                           \\
\multicolumn{1}{c|}{}                                                                                                                           & RCP      & 3.5146 (1.5829)                           & 2.6910 (1.1288)                           & 3.5142 (1.5863)                           & 2.6922 (1.1358)                           \\
\multicolumn{1}{c|}{}                                                                                                                           & UCP      & 3.4905 (1.5698)                           & 2.6676 (1.1133)                           & \textbf{3.4869} (1.5685) &  \textbf{2.6674} (1.1170) \\
\multicolumn{1}{c|}{}                                                                                                                           & RP-PCA   & 3.5209 (1.5976)                           & 2.6951 (1.1448)                           & 3.5193 (1.6001)                           & 2.6954 (1.1490)                           \\
\multicolumn{1}{c|}{}                                                                                                                           & CC-ISO   & 3.5223 (1.5972)                           & 2.6968 (1.1439)                           & 3.5189 (1.5996)                           & 2.6951 (1.1482)                           \\ \hline
\multicolumn{1}{c|}{\multirow{2}{*}{\begin{tabular}[c]{@{}c@{}}Tensor CP-factor methods\\ with uncorrelated factor assumption\end{tabular}}}    & cPCA     & 3.5293 (1.5883)                           & 2.7047 (1.1333)                           & 3.5239 (1.5926)                           & 2.7013 (1.1408)                           \\
\multicolumn{1}{c|}{}                                                                                                                           & HOPE     & 3.5255 (1.5876)                           & 2.7008 (1.1327)                           & 3.5208 (1.5943)                           & 2.6982 (1.1422)                           \\ \hline
\multicolumn{1}{c|}{\multirow{3}{*}{Tensor Tucker-factor methods}}                                                                              & FAC      & 3.5057 (1.5952)                           & 2.6834 (1.1402)                           & 3.5018 (1.5954)                           & 2.6805 (1.1403)                           \\
\multicolumn{1}{c|}{}                                                                                                                           & TOPUP    & 3.5268 (1.5826)                           & 2.7022 (1.1283)                           & 3.5269 (1.5899)                           & 2.7036 (1.1367)                           \\
\multicolumn{1}{c|}{}                                                                                                                           & TIPUP    & 3.5303 (1.5891)                           & 2.7036 (1.1319)                           & 3.5294 (1.5961)                           & 2.7038 (1.1406)                           \\ \hline
\multicolumn{1}{c|}{\multirow{3}{*}{Other benchmark methods}}                                                                                   & MAR      & 3.5154 (1.6093)                           & 2.6923 (1.1597)                           & 3.4959 (1.6057)                           & 2.6765 (1.1539)                           \\
\multicolumn{1}{c|}{}                                                                                                                           & TS-PCA   & 3.5244 (1.6005)                           & 2.6999 (1.1586)                           & 3.5124 (1.5881)                           & 2.6919 (1.1483)                           \\
\multicolumn{1}{c|}{}                                                                                                                           & UniAR    & 3.5470 (1.5789)                           & 2.7143 (1.1250)                           & 3.5413 (1.5817)                           & 2.7130 (1.1323)                           \\ \hline\hline
\end{tabular}
\end{table}

\section{Estimation of factors and common components}\label{sec:factor and cp estimation}
The factors and common components may also be of interest in certain scenarios, and they can be estimated by plugging in the estimated factor loading vectors. Specifically, since $w_i$ and $f_{t,i}$ in model \eqref{model cp} cannot be identified separately, we directly estimate their product $w_i f_{t,i}$ for each $i\in[r]$ and $t\in[n]$. Given the iterative estimator $\{\hat\ab_{i,j}\}_{i\in[\tilde{r}],\,j\in[m]}$, we define $\hat f_{t,i}=(\hat\ab_{i,m}^{\MP} \otimes \cdots \otimes \hat\ab_{i,1}^{\MP})^{\T}\textup{vec}(\mathcal{Y}_t)$ for $i \in [\tilde{r}]$ and $t \in [n]$, where $(\hat\ab_{1,j}^{\MP},\ldots,\hat\ab_{\tilde r,j}^{\MP})^{\T}=(\hat\Ab_j^\T\hat\Ab_j)^{-1}\hat{\Ab}_j^{\T}$ with $\hat\Ab_j=(\hat\ab_{1,j},\ldots,\hat\ab_{\tilde r,j})$. Therefore,  the factors $(w_1 f_{t,1},\ldots,w_r f_{t,r})$ are then estimated by $(\hat{f}_{t,1},\ldots,\hat{f}_{t,\tilde{r}})$.
In model \eqref{model cp}, we write $\mathcal{C}_t=\sum_{i=1}^r w_i{f}_{t,i}\,\ab_{i,1} \circ \ab_{i,2} \circ \cdots \circ \ab_{i,m}$ and estimate it by $\hat{\mathcal{C}}_t=\sum_{i=1}^{\tilde{r}}\hat{f}_{t,i}\,\hat\ab_{i,1} \circ \hat\ab_{i,2} \circ \cdots \circ \hat\ab_{i,m}$.  Write $\Phi_{n} =  \max_{j\in[m]}\Phi_{n,j} + \gamma_{\max} w_r^{-2}$. Theorem \ref{lemma: factor iterative} establishes the consistency of these estimators.

\begin{theorem}\label{lemma: factor iterative}
    Let the conditions of Theorem \textup{\ref{thm: debias iterative}} hold. For each fixed $i\in[r]$ and $t\in[n]$, it holds that
    \begin{align*}
        \frac{1}{w_i} \bigg|\bigg(\prod_{j = 1}^m \kappa_{i,j}\bigg) \hat{f}_{t,z_i}- w_i f_{t,i}\bigg| &=O_{\rm p}\bigg(\frac{1}{w_i} + \Phi_{n}\bigg)\,,\\
 \frac{1}{D_n} | \textup{vec}(\hat{\mathcal{C}}_t)-\textup{vec}(\mathcal{C}_t) |^2_2 &= O_{\rm p}\bigg( \frac{1}{D_n}  +\frac{w_1^2}{D_n} \Phi_{n}^2\bigg)\,,
    \end{align*}
    where $z_{i}$ and $\kappa_{i,j}$ are specified in Theorem \textup{\ref{thm: debias iterative}}.
\end{theorem}

Theorem \ref{lemma: factor iterative} indicates that the convergence rate of the factor estimator comprises two components. 
The first component, $O_{\rm p}(w_i^{-1})$, originates from the noise $\mathcal{E}_t$ and also appears in \citeS{han2024cp-app}. The second component arises from the plug-in errors of $\{\hat\ab_{i,j}\}_{i\in[\tilde{r}],j\in[m]}$. When $w_i \max_{j\in[m]}(s_j\log d_j)^{1/2}\ll w_r\sqrt{n}$ and $\gamma_{\max}w_i \ll w_r^2$, the rate $O_{\rm p}(w_i^{-1})$ dominates. 
The convergence rate of the estimated common component tensor depends on both the estimation errors of the loadings and the factors.

Next, we conduct simulation studies to compare the finite-sample performance of the methods discussed in the paper (Pro.iter, HOPE, CC-ISO, Pro.init, cPCA, RP-PCA, and RCP) in estimating  the common components. 
The data-generating process follows the setup in Section~\ref{sec:numerical}. The estimation error between the estimated common components $\{\check{\mathcal{C}}_{t}\}_{t\in [n]}$ and the true common components $\{\mathcal{C}_{t}\}_{t\in [n]}$ is measured by
\begin{equation}\label{eq: cp estimation error}
  \psi^2_{\textup{cp}} (\{\check{\mathcal{C}}_{t}\}_{t\in [n]} ,  \{\mathcal{C}_{t}\}_{t\in [n]}  )
  =
  \bigg\{\frac{1}{n D_n} \sum_{t = 1}^n |\textup{vec}(\check{\mathcal{C}}_t)-\textup{vec}(\mathcal{C}_t) |_2^2\bigg\}^{1/2}\,.
\end{equation}
As shown in Table \ref{table:cp-all}, when $\rho = 0$, Pro.iter performs comparably with CC-ISO and HOPE, and significantly outperforms the other methods. When $\rho = 0.75$, Pro.iter outperforms both CC-ISO and HOPE. Moreover, Pro.init outperforms all other one-pass estimators across all scenarios.  These results confirm that the proposed methods also have good performance in estimating the common components in finite samples.

\begin{table}[htbp]
\centering
\scriptsize
\renewcommand{\arraystretch}{1.15}
\setlength{\tabcolsep}{3pt}
\caption{
The averages and standard deviations (in parentheses) of the estimation errors \eqref{eq: cp estimation error} of common components for different methods based on 2000 repetitions. Bold numbers indicate the smallest average estimation error among all competing methods.}
\label{table:cp-all}
\begin{tabular}{c|c|c|c|ccc|cccc}
\hline\hline
\multirow{2}{*}{\textbf{$\rho$}} & \multirow{2}{*}{\textbf{$\phi$}} & \multirow{2}{*}{\textbf{$s$}} & \multirow{2}{*}{\textbf{$n$}} & \multicolumn{3}{c|}{\textbf{Iterative estimates}} & \multicolumn{4}{c}{\textbf{One-pass estimates}} \\
\cline{5-11}
& & & & \textbf{Pro.iter} & \textbf{HOPE} & \textbf{CC-ISO} & \textbf{Pro.init} & \textbf{cPCA} & \textbf{RP-PCA} & \textbf{RCP} \\
\hline
\multirow{12}{*}{0}
& \multirow{6}{*}{0.25}
& \multirow{2}{*}{0}
& 400 & \textbf{0.12} (0.05) & 0.12 (0.09) & 0.12 (0.11) & 0.39 (1.14) & 0.87 (0.64) & 0.96 (0.67) & 0.62 (0.49) \\
& & & 800 & \textbf{0.11} (0.07) & 0.12 (0.08) & 0.11 (0.09) & 0.26 (0.23) & 0.76 (0.56) & 0.86 (0.60) & 0.54 (0.46) \\
\cline{3-11}
& & \multirow{2}{*}{0.3}
& 400 & \textbf{0.11} (0.06) & 0.12 (0.11) & 0.11 (0.08) & 0.35 (0.19) & 0.79 (0.63) & 0.86 (0.67) & 0.58 (0.47) \\
& & & 800 & \textbf{0.11} (0.05) & 0.11 (0.07) & 0.11 (0.09) & 0.26 (0.75) & 0.68 (0.58) & 0.75 (0.62) & 0.52 (0.45) \\
\cline{3-11}
& & \multirow{2}{*}{0.6}
& 400 & \textbf{0.10} (0.03) & 0.11 (0.11) & 0.11 (0.10) & 0.33 (0.18) & 0.73 (0.67) & 0.79 (0.69) & 0.57 (0.46) \\
& & & 800 & \textbf{0.10} (0.03) & 0.10 (0.05) & 0.11 (0.09) & 0.23 (0.48) & 0.58 (0.58) & 0.65 (0.63) & 0.51 (0.45) \\
\cline{2-11}
& \multirow{6}{*}{0.75}
& \multirow{2}{*}{0}
& 400 & 0.25 (0.15) & 0.26 (0.19) & \textbf{0.23} (0.12) & 0.80 (1.11) & 1.35 (0.35) & 1.42 (0.38) & 1.05 (0.47) \\
& & & 800 & \textbf{0.23} (0.09) & 0.25 (0.17) & 0.24 (0.15) & 0.59 (2.19) & 1.40 (0.32) & 1.45 (0.37) & 1.07 (0.48) \\
\cline{3-11}
& & \multirow{2}{*}{0.3}
& 400 & \textbf{0.14} (0.06) & 0.15 (0.10) & 0.15 (0.10) & 0.49 (0.72) & 1.18 (0.53) & 1.28 (0.51) & 0.76 (0.53) \\
& & & 800 & \textbf{0.14} (0.05) & 0.15 (0.11) & 0.15 (0.12) & 0.32 (0.44) & 1.21 (0.49) & 1.34 (0.48) & 0.79 (0.55) \\
\cline{3-11}
& & \multirow{2}{*}{0.6}
& 400 & \textbf{0.12} (0.04) & 0.12 (0.10) & 0.12 (0.08) & 0.36 (0.31) & 0.88 (0.64) & 0.94 (0.66) & 0.62 (0.49) \\
& & & 800 & \textbf{0.11} (0.05) & 0.12 (0.08) & 0.12 (0.07) & 0.24 (0.25) & 0.77 (0.55) & 0.88 (0.61) & 0.56 (0.47) \\
\hline
\multirow{12}{*}{0.75}
& \multirow{6}{*}{0.25}
& \multirow{2}{*}{0}
& 400 & \textbf{0.15} (0.03) & 0.32 (0.31) & 0.33 (0.31) & 0.39 (0.42) & 1.81 (0.72) & 1.84 (0.73) & 0.42 (0.42) \\
& & & 800 & \textbf{0.13} (0.01) & 0.30 (0.30) & 0.31 (0.30) & 0.24 (0.28) & 1.79 (0.66) & 1.81 (0.67) & 0.40 (0.41) \\
\cline{3-11}
& & \multirow{2}{*}{0.3}
& 400 & \textbf{0.14} (0.01) & 0.32 (0.31) & 0.34 (0.32) & 0.35 (0.26) & 1.87 (0.79) & 1.90 (0.79) & 0.43 (0.44) \\
& & & 800 & \textbf{0.12} (0.01) & 0.34 (0.32) & 0.34 (0.32) & 0.22 (0.21) & 1.89 (0.75) & 1.91 (0.75) & 0.39 (0.40) \\
\cline{3-11}
& & \multirow{2}{*}{0.6}
& 400 & \textbf{0.13} (0.03) & 0.34 (0.33) & 0.35 (0.33) & 0.34 (0.24) & 1.96 (0.88) & 1.99 (0.87) & 0.43 (0.95) \\
& & & 800 & \textbf{0.11} (0.02) & 0.36 (0.33) & 0.36 (0.33) & 0.21 (0.15) & 1.97 (0.89) & 1.99 (0.86) & 0.40 (0.41) \\
\cline{2-11}
& \multirow{6}{*}{0.75}
& \multirow{2}{*}{0}
& 400 & 0.30 (0.13) & \textbf{0.29} (0.19) & 0.68 (0.34) & 0.71 (0.76) & 1.40 (0.38) & 1.09 (0.37) & 0.51 (0.41) \\
& & & 800 & \textbf{0.24} (0.06) & 0.26 (0.17) & 0.48 (0.33) & 0.48 (0.66) & 1.41 (0.34) & 1.21 (0.41) & 0.42 (0.40) \\
\cline{3-11}
& & \multirow{2}{*}{0.3}
& 400 & \textbf{0.18} (0.07) & 0.26 (0.25) & 0.32 (0.32) & 0.48 (0.58) & 1.58 (0.55) & 1.55 (0.59) & 0.41 (0.39) \\
& & & 800 & \textbf{0.15} (0.02) & 0.25 (0.24) & 0.27 (0.25) & 0.32 (0.97) & 1.56 (0.54) & 1.56 (0.55) & 0.39 (0.41) \\
\cline{3-11}
& & \multirow{2}{*}{0.6}
& 400 & \textbf{0.14} (0.04) & 0.30 (0.30) & 0.31 (0.31) & 0.43 (1.44) & 1.76 (0.69) & 1.77 (0.72) & 0.42 (0.43) \\
& & & 800 & \textbf{0.12} (0.02) & 0.32 (0.31) & 0.32 (0.31) & 0.23 (0.26) & 1.77 (0.71) & 1.78 (0.70) & 0.38 (0.40) \\
\hline\hline
\end{tabular}
\end{table}

\section{Further discussion on the estimated number of factors}\label{sec: factor number}

\subsection{Consistency of the estimators of the number of factors}\label{sec:factor number consistency}

Theorem \ref{thm: factor number consistency} shows that the ER and log-ER estimators specified in \eqref{hat rj} and \eqref{hat rj-log} are consistent estimators for the number of factors $r$.

 \begin{theorem}\label{thm: factor number consistency}
     Set the threshold level $\delta_1=C_* (n^{-1} \log D_n)^{1/2}$ in \eqref{hat Sigma kj} for some constant $C_*>0$. Under Assumptions \textup{\ref{error}}--\textup{\ref{gap new}}, if   $\log D_n\ll n^c$ for some constant $c\in(0,1)$ depending only on $c_1$ and $c_2$ specified in Assumptions {\rm\ref{tail}} and {\rm\ref{mixing}}, as $n \to \infty$, the following two assertions hold.
     
     \textup{(i)} Let $\ubar\sigma_{\xi}^2\Pi_n  \ll c_n \ll  \bar\sigma_{\xi}^{-2}\ubar\sigma_{\xi}^4$ for the $c_n$ in \eqref{hat rj}. Then 
     \[
     \mathbb{P}\Big\{\max_{j\in[m]}\tilde r^{(\textup{er})}_j(\delta_1)=r\Big\}\rightarrow 1 \,.
     \]
   
    \textup{(ii)} Let $\log(1+\ubar\sigma_{\xi}^2\Pi_n)
    \ll
    c_n
    \ll \{\log(1+\bar\sigma_{\xi}^2)\}^{-1}\{\log(1+\ubar\sigma_{\xi}^2)\}^2$
for the $c_n$ in \eqref{hat rj-log}. Then
\[
   \mathbb{P}\Big\{\max_{j\in[m]}\tilde r^{(\textup{log})}_j(\delta_1)=r\Big\} \rightarrow 1\,.
\] 
 \end{theorem}
  The conditions imposed on $c_n$ ensure a proper separation between the signal part and the noise part. In particular, the requirement $\ubar\sigma_{\xi}^2\Pi_n  \ll c_n \ll  \bar\sigma_{\xi}^{-2}\ubar\sigma_{\xi}^4$ (and its counterpart for the log-ER criterion) guarantees that the estimation error is asymptotically negligible relative to the eigen-gap. 

\subsection{Effects of misspecifying the number of factors}\label{sec:misspecify r}
As shown in Theorem \ref{thm: factor number consistency}, our proposed estimators for the number of factors are consistent. However, the estimated number of factors  may still deviate from the true value in finite samples.  It is therefore critical to examine the robustness of our proposed estimation procedures against such misspecification.

  Let $\tilde{r}$ be the estimate of $r$ involved in our estimation procedures. Recall $r = 3$ in our simulation studies considered in Section \ref{sec:numerical}.  To mimic the misspecification issue of $r$, we vary $\tilde r$ from $1$ to $10$, and continue to evaluate the estimation error defined in \eqref{eq:estimation error} in Section \ref{sec:numerical}. The corresponding results based on  Pro.iter and Pro.init are reported in Figure \ref{fig:WrongR-error}.   It can be observed that (i) when $\tilde r<r$, both methods perform poorly because some true factors are omitted; (ii) when $\tilde r=r$, both methods achieve their best performance; and (iii) when $\tilde r>r$, the error of the one-pass method (Pro.init) increases, whereas the iterative method (Pro.iter) is almost unaffected. 
Note that, even when $\tilde r>r$, the measure \eqref{eq:estimation error}  still remains small as long as the  estimated loading vectors contain accurate estimates of the true loading vectors.  The results in Figure~\ref{fig:WrongR-error} suggest that, although $\tilde r>r$ leads to an over-fitted factor structure, our iterative estimator (Pro.iter)  can still recover all the true loading vectors well, which indicates that our proposed iterative estimator is reasonably robust to the issue with an overestimated number of factors.

As shown in Table \ref{table: rank} in Section \ref{sec:numerical}, the proposed log-ER estimator may underestimate $r$ in finite samples, although the frequency is small. 
In practice, to reduce the chance that the estimated number of factors is smaller than the true value $r$, we can apply a two-stage procedure.  In the first stage, we apply the log-ER estimator and the Pro.iter method to $\{\mathcal{Y}_t\}_{t=1}^n$ to get the first-stage estimated number of factors $\tilde{r}^{(1)}$ and the associated estimation of the loading vectors, and then obtain the estimated idiosyncratic error tensor sequence $\{\hat{\mathcal{E}}_t\}_{t=1}^n$. In the second stage, we apply the log-ER estimator to $\{\hat{\mathcal{E}}_t\}_{t=1}^n$ to obtain the second-stage estimated number of factors $\tilde r^{(2)}$.  Based on these two stages, we can select $\tilde r=\tilde r^{(1)}+\tilde r^{(2)}$ as the estimate of  $r$. Table \ref{table: rank-misspecify} reports the performance of  one-stage procedure (log-ER estimator) and two-stage procedure introduced above, respectively.  The results show that the proposed two-stage procedure rarely underestimates $r$ in finite samples. Combining these findings with the results in Figure~\ref{fig:WrongR-error}, we conclude that the proposed iterative method, with $\tilde r$ obtained by the two-stage procedure, is reasonably robust to misspecification of $r$.

\begin{figure}[htbp]
\centerline{\includegraphics[width= 12cm]{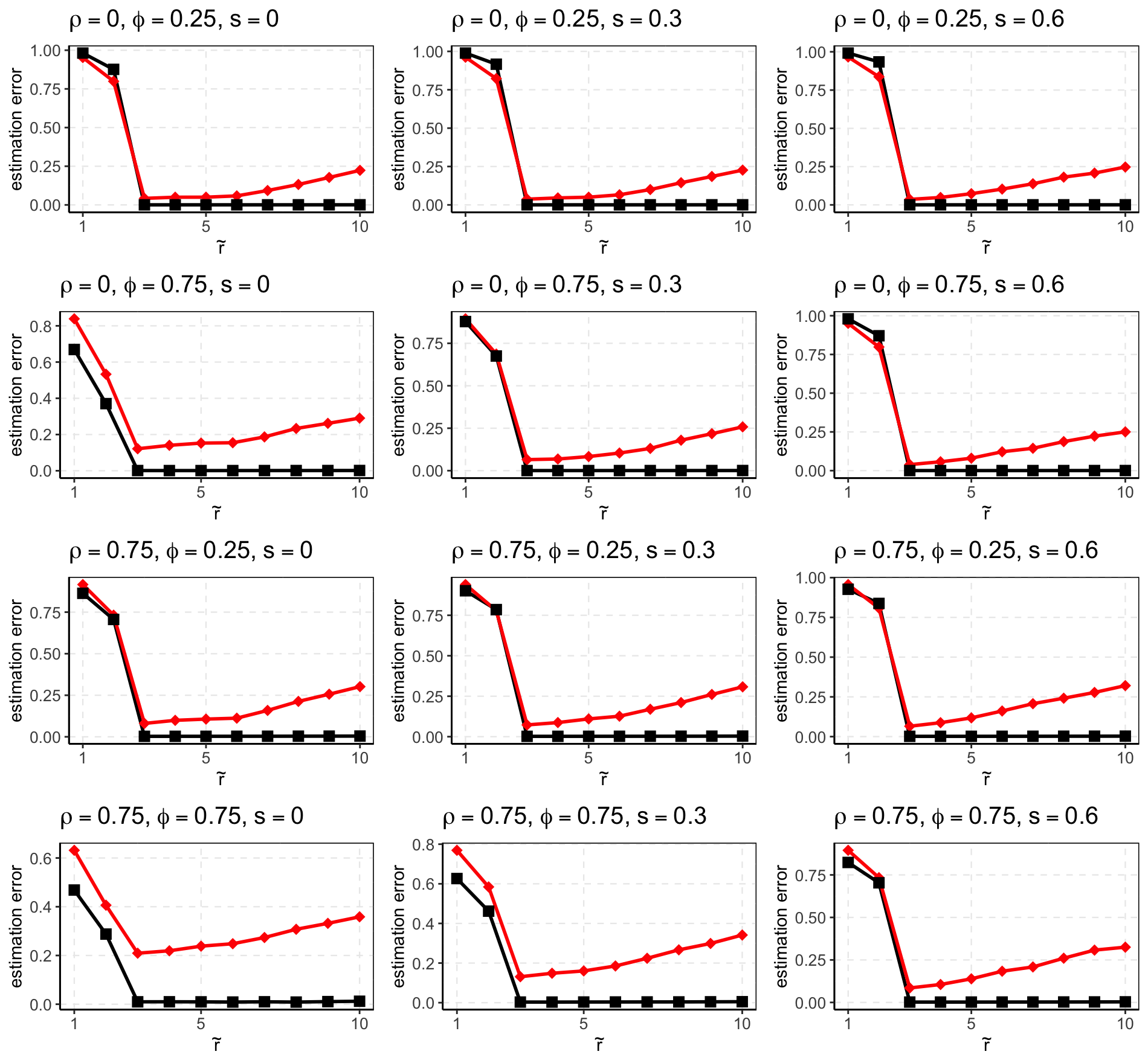}}
\caption{The lineplots for the averages of estimation errors \eqref{eq:estimation error} with respect to $\tilde{r} \in \{1,\ldots,10\}$ based on 2000 repetitions. The sample size $n = 400$. The legend is defined as follows: Pro.iter ($\color{black}{-\blacksquare-}$) and Pro.init ($\color{red}{-\blacklozenge-}$).}
\label{fig:WrongR-error}
\end{figure}
 
\begin{table}[htbp]
\caption{
Relative frequency estimates of $\mathbb{P}(\tilde{r} < r)$ and $\mathbb{P}(\tilde{r} \ge r)$ based on 2000 repetitions, where $\tilde{r}$ is obtained from the one-stage and two-stage procedures. All numbers reported below are multiplied by 100.
}
\renewcommand\tabcolsep{6pt}
\label{table: rank-misspecify}
\footnotesize
\centering
\begin{tabular}{c|c|c|c|cc|cc}
\hline\hline
\multirow{2}{*}{$\rho$} & \multirow{2}{*}{$\phi$} & \multirow{2}{*}{$s$} & \multirow{2}{*}{$n$} & \multicolumn{2}{c|}{\textbf{One-stage procedure}}           & \multicolumn{2}{c}{\textbf{Two-stage procedure}}            \\ \cline{5-8} 
                        &                         &                      &                      & $\mathbb{P}(\tilde{r} < r)$ & $\mathbb{P}(\tilde{r} \ge r)$ & $\mathbb{P}(\tilde{r} < r)$ & $\mathbb{P}(\tilde{r} \ge r)$ \\ \hline
\multirow{12}{*}{0}     & \multirow{6}{*}{0.25}   & \multirow{2}{*}{0}   & 400                  & 0.25                        & 99.75                         & 0.00                        & 100.00                        \\
                        &                         &                      & 800                  & 0.65                        & 99.35                         & 0.00                        & 100.00                        \\ \cline{3-8} 
                        &                         & \multirow{2}{*}{0.3} & 400                  & 0.40                        & 99.60                         & 0.00                        & 100.00                        \\
                        &                         &                      & 800                  & 0.40                        & 99.60                         & 0.00                        & 100.00                        \\ \cline{3-8} 
                        &                         & \multirow{2}{*}{0.6} & 400                  & 0.10                        & 99.90                         & 0.00                        & 100.00                        \\
                        &                         &                      & 800                  & 0.55                        & 99.45                         & 0.00                        & 100.00                        \\ \cline{2-8} 
                        & \multirow{6}{*}{0.75}   & \multirow{2}{*}{0}   & 400                  & 4.55                        & 95.45                         & 0.00                        & 100.00                        \\
                        &                         &                      & 800                  & 1.60                        & 98.40                         & 0.00                        & 100.00                        \\ \cline{3-8} 
                        &                         & \multirow{2}{*}{0.3} & 400                  & 0.60                        & 99.40                         & 0.00                        & 100.00                        \\
                        &                         &                      & 800                  & 0.35                        & 99.65                         & 0.00                        & 100.00                        \\ \cline{3-8} 
                        &                         & \multirow{2}{*}{0.6} & 400                  & 0.10                        & 99.90                         & 0.00                        & 100.00                        \\
                        &                         &                      & 800                  & 0.20                        & 99.80                         & 0.00                        & 100.00                        \\ \hline
\multirow{12}{*}{0.75}  & \multirow{6}{*}{0.25}   & \multirow{2}{*}{0}   & 400                  & 0.10                        & 99.90                         & 0.00                        & 100.00                        \\
                        &                         &                      & 800                  & 0.00                        & 100.00                        & 0.00                        & 100.00                        \\ \cline{3-8} 
                        &                         & \multirow{2}{*}{0.3} & 400                  & 0.00                        & 100.00                        & 0.00                        & 100.00                        \\
                        &                         &                      & 800                  & 0.00                        & 100.00                        & 0.00                        & 100.00                        \\ \cline{3-8} 
                        &                         & \multirow{2}{*}{0.6} & 400                  & 0.15                        & 99.85                         & 0.00                        & 100.00                        \\
                        &                         &                      & 800                  & 0.00                        & 100.00                        & 0.00                        & 100.00                        \\ \cline{2-8} 
                        & \multirow{6}{*}{0.75}   & \multirow{2}{*}{0}   & 400                  & 13.80                       & 86.20                         & 0.10                        & 99.90                         \\
                        &                         &                      & 800                  & 0.55                        & 99.45                         & 0.00                        & 100.00                        \\ \cline{3-8} 
                        &                         & \multirow{2}{*}{0.3} & 400                  & 0.65                        & 99.35                         & 0.00                        & 100.00                        \\
                        &                         &                      & 800                  & 0.00                        & 100.00                        & 0.00                        & 100.00                        \\ \cline{3-8} 
                        &                         & \multirow{2}{*}{0.6} & 400                  & 0.15                        & 99.85                         & 0.00                        & 100.00                        \\
                        &                         &                      & 800                  & 0.00                        & 100.00                        & 0.00                        & 100.00                        \\ \hline\hline
\end{tabular}
\end{table}

\section{Proofs of Theorems \ref{thm: aij}--\ref{thm: debias iterative}, and  Theorems \ref{thm: estimation of iteration variance}--\ref{thm: factor number consistency}}\label{sec: supp main proof}
 To simplify the notation, we use $\mathbf{o}_{\rm p}(\cdot)$ and $\mathbf{O}_{\rm p}(\cdot)$ to denote matrices (or vectors) whose spectral norms are $o_{\rm p}(\cdot)$ and $O_{\rm p}(\cdot)$, respectively. For a matrix $\Hb$, denote by $\mathcal{M}(\Hb)$ the linear space spanned by the columns of $\Hb$.

\subsection{Proof of Theorem \ref{thm: aij}}
 The proof is divided into two steps. Step 1 shows  the consistency of $\tilde{\Qb}_j$ specified in Section \ref{sec: initial}. Step 2 completes the proof of the consistency of $\tilde{\ab}_{i,j}$.

\subsubsection{Step 1: Consistency of \texorpdfstring{$\tilde{\Qb}_j$}{Qj}}
			
			We start with the truncated matrix $\tilde\bSigma_{k,j}$ defined in \eqref{hat Sigma kj}. By definition, $\tilde\bSigma_{k,j}=T_{\delta_1}\{\tilde\bSigma_{\Yb_j,\xi}(k)\}$, where $\tilde\bSigma_{\Yb_j,\xi}(k)=\tilde\bSigma_{\Cb_j,\xi}(k)+\tilde\bSigma_{\Eb_j,\xi}(k)$ with 
\begin{align}
    \tilde\bSigma_{\Cb_j,\xi}(k)&= \frac{1}{n-k}\sum_{t=k+1}^n(\Cb_{t,j}-\bar\Cb_j)(\xi_{t-k}-\bar\xi)\,,~~ \bar\Cb_j=\frac{1}{n}\sum_{t=1}^n\Cb_{t,j}\,, \label{sigma C}\\
    \tilde\bSigma_{\Eb_j,\xi}(k)&= \frac{1}{n-k}\sum_{t=k+1}^n(\Eb_{t,j}-\bar\Eb_j)(\xi_{t-k}-\bar\xi)\,,~~\bar\Eb_j=\frac{1}{n}\sum_{t=1}^n\Eb_{t,j}\,,  \label{tilde Sigma Ej} 
\end{align}            
			and $\Cb_{t,j}$, $\Eb_{t,j}$ are  specified in \eqref{unfold X}. Let $\tilde\Gb_{k,\xi}=\textup{diag}(\tilde g_{k,1,\xi},\ldots,\tilde g_{k,r,\xi})$ be the sample version of $\Gb_{k,\xi}$ defined above \eqref{gki}, where
	\begin{equation}\label{hat g ki}
		\tilde g_{k,i,\xi}=\frac{1}{n-k}\sum_{t=k+1}^nw_i(f_{t,i}-\bar f_i)(\xi_{t-k}-\bar\xi)\,.
	\end{equation}
			  Then, it follows that $	\tilde\bSigma_{\Cb_j,\xi}(k)=\Ab_j\tilde\Gb_{k,\xi}\Bb_j^\T$. To
prove the consistency of $\tilde{\Qb}_j$, we need Lemmas \ref{pro:rank-B}--\ref{lem: l2 truncated covariance}.  The proofs of Lemmas \ref{pro:rank-B}--\ref{lem: l2 truncated covariance} are given in
Sections \ref{sec:pro:rank-B}--\ref{sec:lem: l2 truncated covariance}, respectively.

 \begin{lemma}\label{pro:rank-B}
 If $\textup{rank}(\Ab_j) = r$ for $j \in [m]$, it holds that $\textup{rank}(\Bb_j) = r$ for $j \in [m]$. Moreover, under Assumption \textup{\ref{sparsity}}, it holds that $\sigma_1(\Bb_j),\ldots,\sigma_r(\Bb_j)$ are uniformly bounded away from $0$ and $\infty$.
\end{lemma}

               \begin{lemma}\label{lem: gkixi}
               Under Assumptions \textup{\ref{tail}} and \textup{\ref{mixing}}, for any $i \in [r]$, it holds that
        \begin{equation}\label{gkji}
				\max_{k \in [K]}|\tilde g_{k,i,\xi}-g_{k,i,\xi}| =  O_{\rm p}(w_i n^{-1/2})\,,
			\end{equation}
           provided that $ \log D_n \ll n^c$ for some constant $c\in(0,1)$ depending only on $c_1$ and $c_2$ specified in Assumptions {\rm\ref{tail}} and {\rm\ref{mixing}}, where $g_{k,i,\xi}$ and  $\tilde g_{k,i,\xi}$  are defined in \eqref{gki} and \eqref{hat g ki}, respectively.
           \end{lemma}

    \begin{lemma}\label{lem: l2 truncated covariance}
		Let Assumptions \textup{\ref{error}}--\textup{\ref{eigenvalue}}  hold and  $\delta_1=C_* (n^{-1} \log D_n)^{1/2}$  for some sufficiently large constant $C_*>0$. It holds that  
			\begin{equation*}
			  		\|T_{\delta_1}\{\tilde\bSigma_{\Yb_j,\xi}(k)\}-\tilde\bSigma_{\Cb_j,\xi}(k)\|_2  =  \ubar\sigma_{\xi}^2\bar\sigma_{\xi}^{-1} O_{\rm p}(\Pi_{n})  
			\end{equation*}
			for any $k\in[K]$,	provided that $ \log D_n \ll n^c$ for some constant $c\in(0,1)$ depending only on $c_1$ and $c_2$ specified in Assumptions {\rm\ref{tail}} and {\rm\ref{mixing}}, where $\tilde\bSigma_{\Cb_j,\xi}(k)$ is defined in \eqref{sigma C}. 
		\end{lemma}
           Since $\tilde\bSigma_{\Cb_j,\xi}(k)=\Ab_j\tilde\Gb_{k,\xi}\Bb_j^\T$, following Lemmas \ref{pro:rank-B} and \ref{lem: gkixi}, and Assumptions \ref{sparsity} and \ref{eigenvalue}, by Triangle inequality and Cauchy--Schwarz inequality, we have 
			\begin{equation}\label{Sigma C upper bound}
            \|\tilde\bSigma_{\Cb_j,\xi}(k)\|_2 \lesssim   \max_{i\in[r]}|\tilde g_{k,i,\xi}| \lesssim \bar\sigma_{\xi}\{1 + o_{\rm p}(1)\}\,.
			\end{equation}
Following Lemma \ref{lem: l2 truncated covariance} and \eqref{Sigma C upper bound}, by the definition of $\tilde{\mathbf{M}}_{j}$, Triangle inequality, and  Cauchy--Schwarz inequality, and using the fact that $(\ubar\sigma_{\xi}^2\bar\sigma_{\xi}^{-1}\Pi_n)^2\lesssim \ubar\sigma_{\xi}^2\Pi_n$ under the condition $\Pi_n\ll1$, it holds that
			\begin{equation}\label{tilde M and tilde C}
				\bigg\|\tilde{\mathbf{M}}_{j}-\sum_{k=1}^K \tilde\bSigma_{\Cb_j,\xi}(k)^\T \tilde\bSigma_{\Cb_j,\xi}(k) \bigg\|_2 = \ubar\sigma_{\xi}^2\, O_{\rm p}(\Pi_{n})\,.
			\end{equation}
			On the other hand, by Lemma \ref{lem: gkixi} and Triangle inequality,  we also have
			\begin{equation}\label{tilde C and Sigma Y}
			\bigg\|\sum_{k=1}^K \tilde\bSigma_{\Cb_j,\xi}(k)^\T \tilde\bSigma_{\Cb_j,\xi}(k) - \sum_{k=1}^K \bSigma_{\Yb_j,\xi}(k)^\T \bSigma_{\Yb_j,\xi}(k) \bigg\|_2 = O_{\rm p}(w_1 \bar\sigma_{\xi} n^{-1/2})\,.
			\end{equation}
            Under Assumption \ref{eigenvalue}, we have $w_1 \bar\sigma_{\xi} n^{-1/2} \ll \ubar\sigma_{\xi}^2 $. Meanwhile, all the nonzero eigenvalues of $\sum_{k=1}^K \bSigma_{\Yb_j,\xi}(k)^\T \bSigma_{\Yb_j,\xi}(k)$ are lower bounded by $\ubar\sigma_{\xi}^2$.
			 Therefore, by  Weyl's theorem,  we conclude that the minimum nonzero eigenvalue of $	\sum_{k=1}^K \tilde\bSigma_{\Cb_j,\xi}(k)^\T \tilde\bSigma_{\Cb_j,\xi}(k)$ is larger than $\ubar\sigma_{\xi}^2\{1-
            o_{\rm p}(1)\}$. Furthermore, note that 
			\[
			\sum_{k=1}^K \tilde\bSigma_{\Cb_j,\xi}(k)^\T \tilde\bSigma_{\Cb_j,\xi}(k) = \Qb_j\Vb_j\bigg(\sum_{k=1}^K\tilde\Gb_{k,\xi}\Ab_j^\T\Ab_j\tilde\Gb_{k,\xi}\bigg)\Vb_j^\T\Qb_j^\T\,,
			\]
			whose  $r$ leading eigenvectors are in $\mathcal{M}(\Qb_j)$. Recall $\tilde{r} = r$. Therefore, by \eqref{tilde M and tilde C} and standard results in perturbation theory, see e.g. Lemma 1 of \citeS{chang2018principal-app}, we conclude that
            \begin{equation}\label{thm: initial}
                			\|\tilde\Qb_j-\Qb_j\tilde{\mathcal{H}}_j\|_2 = O_{\rm p}(\Pi_{n})
            \end{equation}
			for some orthogonal matrix $\tilde{\mathcal{H}}_j$, provided that $\log D_n \ll n^c$ for some constant $c\in(0,1)$ depending only on $c_1$ and $c_2$ specified in Assumptions {\rm\ref{tail}} and {\rm\ref{mixing}}.

\subsubsection{Step 2: Consistency of \texorpdfstring{$\tilde{\ab}_{i,j}$}{aij}}
By definition, $\tilde\ab_{i,j}$ is an eigenvector of $\tilde\Kb_{1,2,j}$, while $\ab_{i,j}$ is an associated eigenvector of $\Kb_{1,2,j}$. Following the perturbation theory for eigen-analysis, such as Lemma 4 in \citeS{chang2023modelling-app}, to prove the consistency of $\tilde\ab_{i,j}$, we should start with the consistency of $\tilde\Kb_{1,2,j}$ under spectral norm.
			
			We will not bound the error $\|\tilde\Kb_{1,2,j}-\Kb_{1,2,j}\|_2$ directly. Instead, similarly to the definition of $\Kb_{1,2,j}$ in \eqref{Kbj}, we define $\hat\Kb_{1,2,j}$  by replacing $\bSigma_{\Yb_j,\xi}(k)$ with $\tilde\bSigma_{\Cb_j,\xi}(k)=\Ab_j\tilde\Gb_{k,\xi}\Bb_j^\T$  for $k\in \{1,2\}$, where the diagonal entries of $\tilde\Gb_{k,\xi}$ are defined in \eqref{hat g ki}. Similarly to $\Kb_{1,2,j} = \Ab_j\Gb_{1,\xi}\Gb_{2,\xi}^{-1}(\Ab_j^\T\Ab_j)^{-1}\Ab_j^\T$, we have
			\begin{equation}\label{define hat K21j}
				\hat\Kb_{1,2,j}=\Ab_j\tilde\Gb_{1,\xi}\tilde\Gb_{2,\xi}^{-1}(\Ab_j^\T\Ab_j)^{-1}\Ab_j^\T\,.
			\end{equation}
			Hence, $\ab_{i,j}$ is also an eigenvector of $\hat\Kb_{1,2,j}$ with the associated eigenvalue $\hat\lambda_i=\tilde g_{2,i,\xi}^{-1}\tilde g_{1,i,\xi}$.  To complete the proof of   Theorem \ref{thm: aij}, we need the following lemma with its proof given in Section \ref{sec:lem:gap K21j}.

            	\begin{lemma}\label{pro: singular K21j}
	 		Under Assumptions \textup{\ref{error}}--\textup{\ref{gap new}}, there exist a $d_j\times(d_j-1)$ matrix $\Ob_{j,\mminus i}$ and a universal constant $C_{10}>0$ such that $(\ab_{i,j},\Ob_{j,\mminus i})$ is an orthogonal matrix and
		\begin{equation*}\label{gap K21j}
	\begin{split}
	    		\sigma_{\min}\{\Ob_{j,\mminus i}^\T(\Kb_{1,2,j}-\bar\lambda_i\Ib_{d_j})\Ob_{j,\mminus i}\} &\ge C_{10}\,,\\
                \sigma_{\min}\{\Ob_{j,\mminus i}^\T(\hat\Kb_{1,2,j}-\hat\lambda_i\Ib_{d_j})\Ob_{j,\mminus i}\}&\ge  C_{10}\{1-o_{\rm p}(1)\}\,,
	\end{split}
            \end{equation*}
            where $\hat{\Kb}_{1,2,j}$ is defined in \eqref{define hat K21j} with the associated  eigenvalues $\hat\lambda_i=\tilde g_{2,i,\xi}^{-1}\tilde g_{1,i,\xi}$ for $i \in [r]$. In addition, $\hat\lambda_i=\bar\lambda_i + o_{\rm p}(1)$ for $i \in [r]$.
	\end{lemma}

            By Lemma \ref{pro: singular K21j}, following the perturbation theory from Lemma 4 of \citeS{chang2023modelling-app}, it remains to bound the error $\|\tilde\Kb_{1,2,j}-\hat\Kb_{1,2,j}\|_2$. According to the definition of $\tilde\Kb_{1,2,j}$, we first investigate the minimum eigenvalue of $\tilde\Qb_j^\T\tilde\bSigma_{2,j}^\T\tilde\bSigma_{2,j}\tilde\Qb_j$. By  Lemma \ref{lem: l2 truncated covariance} and \eqref{Sigma C upper bound}, we have
		\begin{equation}\label{eq: QSigmaSigmaQ-QSigmaCSigmaCQ}       \|\tilde\Qb_j^\T\tilde\bSigma_{2,j}^\T\tilde\bSigma_{2,j}\tilde\Qb_j-\tilde\Qb_j^\T\tilde\bSigma_{\Cb_j,\xi}(2)^\T\tilde\bSigma_{\Cb_j,\xi}(2)\tilde\Qb_j\|_2 = \ubar\sigma_{\xi}^2 O_{\rm p}(\Pi_{n})\,.
		\end{equation}
			Therefore, we should consider the minimum eigenvalue of $\tilde\Qb_j^\T\tilde\bSigma_{\Cb_j,\xi}(2)^\T\tilde\bSigma_{\Cb_j,\xi}(2)\tilde\Qb_j$. By \eqref{thm: initial}, $\tilde\Qb_j$ is a consistent estimator of $\Qb_j\tilde{\mathcal{H}}_j$. Then, because  $\sigma_{\min}(\tilde{\mathcal{H}}_j^\T\Qb_j^\T\Bb_j) \ge C$ for some universal constant $C>0$, we can conclude that  $\sigma_{\min}(\tilde{\Qb}_j^\T\Bb_j) \ge C\{1-o_{\rm p}(1)\}$.  It follows that
			\begin{equation}\label{middle Qj population}
   \begin{split}
       	\sigma_{\min}\{\tilde\Qb_j^\T\tilde\bSigma_{\Cb_j,\xi}(2)^\T\tilde\bSigma_{\Cb_j,\xi}(2)\tilde\Qb_j\}&=\sigma_{\min}(\tilde\Qb_j^\T\Bb_j\tilde\Gb_{2,\xi}\Ab_j^\T\Ab_j\tilde\Gb_{2,\xi}\Bb_j^\T\tilde\Qb_j)\\
        &\gtrsim \ubar\sigma_{\xi}^2\{1-o_{\rm p}(1)\}\,.
   \end{split}			
			\end{equation}
			Then, by \eqref{eq: QSigmaSigmaQ-QSigmaCSigmaCQ}, \eqref{middle Qj population} and Weyl's Theorem, we can further conclude that
			\begin{equation}\label{middle Qj}
				\sigma_{\min}(\tilde\Qb_j^\T\tilde\bSigma_{2,j}^\T\tilde\bSigma_{2,j}\tilde\Qb_j)\gtrsim \ubar\sigma_{\xi}^2\{1-o_{\rm p}(1)\}\,.
			\end{equation}
			 Next, for any $p \times p$ invertible matrices $\Wb_1$ and $\Wb_2$, it holds that 
			\begin{equation}\label{matrix inverse}
				\Wb_1^{-1}=\Wb_2^{-1}-\Wb_1^{-1}(\Wb_1-\Wb_2)\Wb_2^{-1}=\Wb_2^{-1}\{\Ib_p+(\Wb_1-\Wb_2)\Wb_2^{-1}\}^{-1} \,.
			\end{equation}
			Combining \eqref{eq: QSigmaSigmaQ-QSigmaCSigmaCQ} with \eqref{matrix inverse}, we have
			\[
			\begin{split}
				&\tilde\bSigma_{\Cb_j,\xi}(1)\tilde\Qb_j(\tilde\Qb_j^\T\tilde\bSigma_{2,j}^\T\tilde\bSigma_{2,j}\tilde\Qb_j)^{-1}\\
				& ~~~~~~~~ = \tilde\bSigma_{\Cb_j,\xi}(1)\tilde\Qb_j\{\tilde\Qb_j^\T\tilde\bSigma_{\Cb_j,\xi}(2)^\T\tilde\bSigma_{\Cb_j,\xi}(2)\tilde\Qb_j\}^{-1}\{\Ib_r+\mathbf{o}_{\rm p}(1)\}^{-1}\,.
			\end{split}
			\]
			 Then,  by direct calculation, it holds that
			\[
			\begin{split}
				\|\tilde\bSigma_{\Cb_j,\xi}(1)\tilde\Qb_j(\tilde\Qb_j^\T\tilde\bSigma_{2,j}^\T\tilde\bSigma_{2,j}\tilde\Qb_j)^{-1}\|_2
			 = O_{\rm p}(\ubar\sigma_{\xi}^{-1})\,. 
			\end{split}
			\]
			Further by Lemma \ref{lem: l2 truncated covariance}, \eqref{middle Qj}, and Triangle inequality, we also have
			\begin{equation}\label{tilde K 123}
				\|\tilde\bSigma_{k,j}\tilde\Qb_j(\tilde\Qb_j^\T\tilde\bSigma_{2,j}^\T\tilde\bSigma_{2,j}\tilde\Qb_j)^{-1}\|_2  = O_{\rm p}(\ubar\sigma_{\xi}^{-1})\,, \quad k \in \{1,2\}\,.
			\end{equation}
Now, by  Lemma \ref{lem: l2 truncated covariance} and \eqref{tilde K 123}, it follows that
			\begin{equation}\label{tilde K 123 plus error}
				\begin{split}
					\tilde\Kb_{1,2,j}=\tilde\bSigma_{1,j}\tilde\Qb_j(\tilde\Qb_j^\T\tilde\bSigma_{2,j}^\T\tilde\bSigma_{2,j}\tilde\Qb_j)^{-1}\tilde\Qb_j^\T\tilde\bSigma_{\Cb_j,\xi}(2)^\T+ \frac{\ubar\sigma_{\xi}}{\bar\sigma_{\xi}}\mathbf{O}_{\rm p}(\Pi_{n})\,.
				\end{split}
			\end{equation}
			Using \eqref{matrix inverse} once again, we obtain
\begin{align*}
    & (\tilde\Qb_j^\T \tilde\bSigma_{2,j}^\T\tilde\bSigma_{2,j}\tilde\Qb_j)^{-1}\tilde\Qb_j^\T\tilde\bSigma_{\Cb_j,\xi}(2)^\T\\
				&~~~~=\{\tilde\Qb_j^\T\tilde\bSigma_{\Cb_j,\xi}(2)^\T\tilde\bSigma_{\Cb_j,\xi}(2)\tilde\Qb_j\}^{-1}\tilde\Qb_j^\T\tilde\bSigma_{\Cb_j,\xi}(2)^\T\\
				&~~~~\quad-(\tilde\Qb_j^\T\tilde\bSigma_{2,j}^\T\tilde\bSigma_{2,j}\tilde\Qb_j)^{-1}\tilde\Qb_j^\T\{\tilde\bSigma_{2,j}^\T\tilde\bSigma_{2,j}-\tilde\bSigma_{\Cb_j,\xi}(2)^\T\tilde\bSigma_{\Cb_j,\xi}(2)\}\tilde\Qb_j\\
				&~~~~\quad \quad  \cdot  \{\tilde\Qb_j^\T\tilde\bSigma_{\Cb_j,\xi}(2)^\T\tilde\bSigma_{\Cb_j,\xi}(2)\tilde\Qb_j\}^{-1}\tilde\Qb_j^\T\tilde\bSigma_{\Cb_j,\xi}(2)^\T\\
				&~~~~= (\Bb_j^\T\tilde\Qb_j)^{-1}\tilde\Gb_{2,\xi}^{-1}(\Ab_j^\T\Ab_j)^{-1}\Ab_j^\T\\
				&~~~~\quad-(\tilde\Qb_j^\T\tilde\bSigma_{2,j}^\T\tilde\bSigma_{2,j}\tilde\Qb_j)^{-1}\tilde\Qb_j^\T\{\tilde\bSigma_{2,j}-\tilde\bSigma_{\Cb_j,\xi}(2)\}^\T\Ab_j(\Ab_j^\T\Ab_j)^{-1}\Ab_j^\T\\
				&~~~~\quad-(\tilde\Qb_j^\T\tilde\bSigma_{2,j}^\T\tilde\bSigma_{2,j}\tilde\Qb_j)^{-1}\tilde\Qb_j^\T\tilde\bSigma_{2,j}^\T\{\tilde\bSigma_{2,j}-\tilde\bSigma_{\Cb_j,\xi}(2)\}\tilde\Qb_j(\Bb_j^\T\tilde\Qb_j)^{-1}\tilde\Gb_{2,\xi}^{-1}(\Ab_j^\T\Ab_j)^{-1}\Ab_j^\T\,.
\end{align*}
			Left-multiplying both sides by $\tilde\bSigma_{1,j}\tilde\Qb_j$, and combining with Lemma \ref{lem: l2 truncated covariance}, \eqref{tilde K 123} and \eqref{tilde K 123 plus error}, we can conclude that 
			\[
			\tilde\Kb_{1,2,j}- \frac{\ubar\sigma_{\xi}}{\bar\sigma_{\xi}} \mathbf{O}_{\rm p}(\Pi_{n})=\hat\Kb_{1,2,j}+ \frac{\ubar\sigma_{\xi}}{\bar\sigma_{\xi}}\mathbf{O}_{\rm p}(\Pi_{n})+\tilde\Kb_{1,2,j} \cdot \frac{\ubar\sigma_{\xi}}{\bar\sigma_{\xi}}\mathbf{O}_{\rm p}(\Pi_{n})\,,
			\]
			where  we use the fact  $\hat\Kb_{1,2,j}=\tilde\bSigma_{\Cb_j,\xi}(1)\tilde\Qb_j\{\tilde\Qb_j^\T\tilde\bSigma_{\Cb_j,\xi}(2)^\T\tilde\bSigma_{\Cb_j,\xi}(2)\tilde\Qb_j\}^{-1}\tilde\Qb_j^\T\tilde\bSigma_{\Cb_j,\xi}(2)^\T$. Consequently, we have
			\begin{equation}\label{tilde K hat K error}
			\|	\tilde\Kb_{1,2,j} - \hat\Kb_{1,2,j}\|_2= \frac{\ubar\sigma_{\xi}}{\bar\sigma_{\xi}} O_{\rm p}(\Pi_{n})\,.
			\end{equation}
			Then, Theorem \ref{thm: aij} follows from Lemma \ref{pro: singular K21j}, \eqref{tilde K hat K error}, and standard perturbation theory; see Lemma 4 of \citeS{chang2023modelling-app}.  
$\hfill\Box$

		\subsection{Proof of Theorem \ref{thm: iterative}}
        The proof is divided into three steps. Step 1 constructs an event to control some random quantities. Step 2 shows how the iterations improve the convergence rate in each round. Step 3 completes the proof of Theorem \ref{thm: iterative} by letting the number of iterations grow gradually.  For notational simplicity, the proofs of Steps 1--3 in Sections \ref{sec:proofT2-step1}--\ref{sec:proofT2-step3} ignore the reflection and permutation indeterminacy and focus on the case $r \ge 2$. Section \ref{sec:proofT2-step4} further discusses the impact of the reflection and permutation indeterminacy, and shows that the proof for the case $r = 1$ is trivial and follows directly as a specialization of the argument for $r\ge 2$.

\subsubsection{Step 1: Construct an event to bound some random quantities}\label{sec:proofT2-step1}

         Recall that $\tilde{r} = r$. Define $\theta_j^{(\textit{v})}= \max_{i \in [r]}|\tilde\ab_{i,j}^{(\textit{v})}-\ab_{i,j}|_2$, which measures the estimation error of $\tilde\Ab_j^{(\textit{v})}= (\tilde\ab_{1,j}^{(\textit{v})},\ldots,\tilde\ab_{r,j}^{(\textit{v})})$ obtained  after  the $\textit{v}$-th round of the iteration in Algorithm \ref{alg1}. Write
\begin{equation}\label{theta jv}
    	\bar\theta_j^{(\textit{v})}=\max(\theta_1^{(\textit{v})},\ldots,\theta_{j-1}^{(\textit{v})},\theta_{j}^{(\textit{v}-1)},\ldots,\theta_{m}^{(\textit{v}-1)})\,,\quad \textit{v} \ge 1,\, j\in[m]\,.
\end{equation}
     It is important to construct the relationship between $\theta_j^{(\textit{v})}$ and $\bar\theta_j^{(\textit{v})}$. To do this, for some constant $\tilde C > 0$ that specified later, we first define a series of events 
            \begin{align}\label{Omega 1n}
          \nonumber       \Xi_{1,n}(\tilde C)&=\bigg\{\bigg|\frac{1}{n-k-k_1-k_2}\sum_{t=k+k_1+1}^{n-k_2}(f_{t,i}-\bar f_i)(f_{t-k,\ell} - \bar f_\ell) - \Upsilon_{k,i,\ell} \bigg| < \tilde C^{-1}  \\ \nonumber
         &\qquad \qquad\qquad \qquad\qquad\qquad   \textup{for any} ~ i,\ell \in[r]~\textup{and}~ k,k_1,k_2 \in \{0,1\}\bigg\}\,,\\ \nonumber
					\Xi_{2,n}(\tilde C)&=\bigg\{\bigg|\frac{1}{n-k-k_1-k_2}\sum_{t=k+k_1+1}^{n-k_2}[(f_{t,i}-\bar f_i)\xi_{t-k,\ell}^{\textup{s}}-\mathbb{E}\{(f_{t,i}-\bar f_i)\xi_{t-k,\ell}^{\textup{s}}\}]\bigg|< \tilde C^{-1}  \\ \nonumber
                    &\qquad \qquad\qquad \qquad\qquad\qquad   \textup{for any} ~i,\ell \in[r] ~\textup{and}~ k,k_1,k_2 \in \{0,1\} \bigg\}\,,\\ \nonumber
					\Xi_{3,n}(\tilde C)&=\bigg\{\frac{1}{n-1}\|(\Fb_{\xi,\mminus i}^{\textup{s}})^{\T}\Fb_{\xi,\mminus i}^{\textup{s}}-\mathbb{E}\{(\Fb_{\xi,\mminus i}^{\textup{s}})^{\T}\Fb_{\xi,\mminus i}^{\textup{s}}\}\|_2 < \tilde C^{-1} ~\textup{for any}~ i\in[r]\bigg\}\,,\\ \nonumber
					\Xi_{4,n}(\tilde C)&=\bigg\{\frac{1}{n-1}|(\Fb_{\xi,\mminus i}^{\textup{s}})^{\T}\bxi_i^{\textup{s}}-\mathbb{E}\{(\Fb_{\xi,\mminus i}^{\textup{s}})^{\T}\bxi_i^{\textup{s}}\}|_2  < \tilde C^{-1} ~\textup{for any}~ i\in[r]\bigg\}\,,\\  
					\Xi_{5,n}(\tilde C)&=\bigg\{\bigg|\frac{1}{n-k_1-k_2}\sum_{t=k_1+1}^{n-k_2}[(\xi_{t,i}-\bar\xi_i)^2-\mathbb{E}\{(\xi_{t,i}-\bar\xi_i)^2\}]\bigg| <  w_i^2 \tilde  C^{-1} \\ \nonumber
                      &\qquad \qquad\qquad \qquad\qquad\qquad   \textup{for any} ~ i \in[r]~\textup{and}~  k_1,k_2 \in \{0,1\}\bigg\}\,. 
            \end{align}
          For $L_n$ specified in \eqref{auto cross f xi}, let $\Xi_{6,n}(\tilde C)$ be the event such that
			\begin{equation}\label{Omega 6n}
				\begin{split}
				    &\max_{\substack{|\bbeta_j|_2=1,\, \bbeta_j\in\mathbb{R}^{d_j},\\\forall j\in[m]}}\bigg|\frac{1}{n-k-k_1-k_2}\sum_{t=k+k_1+1}^{n-k_2}(f_{t-k,i} - \bar f_i)(\mathcal{E}_t\times_{j=1}^m\bbeta_j^{\T})\bigg|\,
                    \le \tilde C L_n\,
				\end{split}
			\end{equation}
            for any $i\in[r]$ and $k,k_1,k_2\in\{0,1\}$, 
   where $\times_{j=1}^m\bbeta_j^{\T}$ is shorthand for $\times_1\bbeta_1^{\T}\times_2\cdots\times_m\bbeta_m^{\T}$ with $\times_j$ defined as the $j$-mode product of a tensor and a matrix. In detail, for a tensor $\mathcal{A}\in\mathbb{R}^{d_1\times \cdots\times d_m}$ and a matrix $\bbeta\in\mathbb{R}^{d_0\times d_j}$, $\mathcal{A}\times_j\bbeta$ is still an $m$-mode tensor of size $d_1\times \cdots \times d_{j-1}\times d_0\times d_{j+1}\times \cdots \times d_m$ with the entries being
   \[
   [\mathcal{A}\times_j\bbeta]_{i_1,\ldots,i_m}=\sum_{h=1}^{d_j}[\mathcal{A}]_{i_1,\ldots,i_{j-1},h,i_{j+1},\ldots,i_m}[\bbeta]_{i_j,h}\,.
   \]
			Let $\Xi_{7,n}(\tilde C)$ be the event such that
			\begin{equation}\label{Omega 7n}
				\begin{split}
					&\max_{\substack{|\tilde\bbeta_j|_2=1 = |\bbeta_j|_2,\\
                    \tilde\bbeta_j,\, \bbeta_j \in\mathbb{R}^{d_j},\forall j\in[m]}}\bigg|\frac{1}{n-1}\sum_{t=2}^n(\mathcal{E}_{t-1}\times_{j=1}^m\tilde\bbeta_j^{\T})(\mathcal{E}_t\times_{j=1}^m\bbeta_j^{\T})\bigg|
					\le  \tilde C L_n\,.
				\end{split}
			\end{equation}
			Let $\Xi_{8,n}(\tilde C)$ be the event such that
			\begin{equation}\label{Omega 8n}
				\max_{|\bbeta_j|_2=1,\,\bbeta_j \in\mathbb{R}^{d_j},\forall j\in[m]}\bigg|\frac{1}{n-k_1-k_2}\sum_{t=k_1+1}^{n-k_2}\mathcal{E}_t\times_{j=1}^m\bbeta_j^{\T}\bigg|\le  \tilde C L_n 
			\end{equation}
			for $ k_1,k_2 \in \{0,1\}$. Let $\Xi_{9,n}(\tilde C)$ be the event such that
            \begin{align}
                	&\max_{\substack{|\tilde\bbeta_j|_2=1 = |\bbeta_j|_2,\\
                    \tilde\bbeta_j,\, \bbeta_j \in\mathbb{R}^{d_j},\forall j\in[m]}}\bigg|\frac{1}{n-k_1-k_2}\sum_{t=k_1 +1}^{n-k_2}(\otimes^{j=1}_m\tilde\bbeta_j)^{\T}[\text{vec}(\mathcal{E}_t)\text{vec}(\mathcal{E}_t)^{\T} \nonumber\\
                    &\qquad \qquad \qquad \qquad \qquad \qquad \qquad -\mathbb{E}\{\text{vec}(\mathcal{E}_t)\text{vec}(\mathcal{E}_t)^{\T}\}](\otimes^{j=1}_m\bbeta_j)\bigg|  \le \tilde C L_n \label{Omega 9n}
            \end{align}
            for $ k_1,k_2 \in \{0,1\}$,  where $\otimes^{j=1}_m\bbeta_j$ is shorthand for $\bbeta_m \otimes \cdots \otimes \bbeta_1$.
   
			 Let $\Xi_n(\tilde C)=\cap_{h=1}^9 \Xi_{h,n}(\tilde C)$. 
			Lemma \ref{lem: Omega n} shows that $\mathbb{P}\{\Xi_n(\tilde C)\}\rightarrow 1$ as $n\rightarrow \infty$ for some sufficiently large constant $\tilde C > 0$.   The proof of Lemma \ref{lem: Omega n} is given in Section \ref{sec:lem: Omega n}. 
       
        \begin{lemma}\label{lem: Omega n}
			Under the assumptions of Theorem \textup{\ref{thm: iterative}}, there exists a sufficiently large constant $C_0>0$ such that $\mathbb{P}\{\Xi_n(\tilde C)\}\rightarrow 1$ as $n\rightarrow \infty$ for any constant $\tilde C>C_0$.
		\end{lemma}

\subsubsection{Step 2: Error reduction across iterations}\label{sec:proofT2-step2}

Now, for the $\textit{v}$-th iteration round and the $j$-th mode, let $\bar{\check f}_{i}^{(\textit{v},j)} = n^{-1}\sum_{t=1}^n \check f_{t,i}^{(\textit{v},j)}$, where $\check f_{t,i}^{(\textit{v},j)}$ is defined in Algorithm \ref{alg1}. Further let
   \[
   (\tilde{\sigma}^{(\textit{v},j)}_{\check f,i})^2 = \frac{1}{n-1}\sum_{t=1}^n(\check f_{t,i}^{(\textit{v},j)} - \bar{\check f}_{i}^{(\textit{v},j)})^2\,.
   \]
   Define $\tilde\Fb_{\mminus i}^{(\textit{v},j)}$ and $\tilde\fb_{i}^{(\textit{v},j)}$ by similar steps above \eqref{eq: tilde xi} but replacing $\check f_{t,i}$ with $\check f_{t,i}^{(\textit{v},j)}$. Furthermore,  let $\tilde\bvarphi_{i}^{(\textit{v},j)}=(\tilde\varphi_{i,1}^{(\textit{v},j)},\ldots,\tilde\varphi_{i,r}^{(\textit{v},j)})^{\T}$ be the $r$-dimensional vector with the $i$-th entry equal to 1, and the remaining $r-1$ entries form the vector $-\{(\tilde\Fb_{\mminus i}^{(\textit{v},j)})^{\T}\tilde\Fb_{\mminus i}^{(\textit{v},j)}\}^{-1}(\tilde\Fb_{\mminus i}^{(\textit{v},j)})^{\T}\tilde\fb_{i}^{(\textit{v},j)}$.
			 Let $\bLambda_{\xi} = (\Lambda_{\xi,i,j})_{r\times r}$ be the $r\times r$ diagonal matrix with the $i$-th diagonal entry being $n^{-1}\sum_{t = 1}^n\mathbb{E}\{(\xi_{t,i}-\bar\xi_i)^2\}$.  Define
             \begin{align*}
                 	\bzeta_1^{(i,j,\textit{v})}&=\bigg\{\frac{w_i}{n-1}\sum_{t=2}^n(f_{t,i}-\bar f_i)\tilde\xi_{t-1,i}^{(\textit{v},j)}\bigg\}\{\bbb_{i,j}^{\T}(\tilde\bbb_{i,j}^{(\textit{v})})^{\MP}\}\ab_{i,j}\,,\\
			 	\bzeta_2^{(i,j,\textit{v})}&= \sum_{\ell\ne i}\bigg\{\frac{w_{\ell}}{n-1}\sum_{t=2}^n(f_{t,\ell}-\bar f_{\ell})\tilde\xi_{t-1,i}^{(\textit{v},j)}\bigg\}\{\bbb_{\ell,j}^{\T}(\tilde\bbb_{i,j}^{(\textit{v})})^{\MP}\}\ab_{\ell,j}\,,\\
			 	\bzeta_3^{(i,j,\textit{v})}&= \frac{\Lambda_{\xi,i,i}^{1/2}}{\tilde{\sigma}^{(\textit{v},j)}_{\check f,i} }\tilde\bSigma_{\eb_{i,j},\xi_i}(1) +\sum_{\ell\ne i}\frac{\tilde\varphi_{i,\ell}^{(\textit{v},j)}\Lambda_{\xi,\ell,\ell}^{1/2}}{\tilde{\sigma}_{\check f, \ell}^{(\textit{v},j)}}\tilde\bSigma_{\eb_{i,j},\xi_{\ell}}(0)\,,\\
			 	\bzeta_4^{(i,j,\textit{v})}&= \bigg\{\frac{1}{n-1}\sum_{t=2}^n\tilde\xi_{t-1,i}^{(\textit{v},j)}(\Eb_{t,j}-\bar\Eb_j)(\tilde\bbb_{i,j}^{(\textit{v})})^{\MP}\bigg\}-\bzeta_3^{(i,j,\textit{v})}\,.
             \end{align*}
		Then, for each $i\in[r]$ and $j\in[m]$,  we have
			\begin{equation}\label{iterative sigma decompose}
				\begin{split}
					\tilde\bSigma^{(\textit{v},j)}_{\tilde\yb_{i,j},\tilde\xi_{i}}(1)
					=\bzeta_1^{(i,j,\textit{v})}+\bzeta_2^{(i,j,\textit{v})}+\bzeta_3^{(i,j,\textit{v})}+\bzeta_4^{(i,j,\textit{v})}\,,
				\end{split}
			\end{equation}
			where $	\tilde\bSigma^{(\textit{v},j)}_{\tilde\yb_{i,j},\tilde\xi_{i}}(1)$ is defined in Algorithm \ref{alg1}. Recall the estimator 
            \[
            \tilde\ab_{i,j}^{(\textit{v})}= T_{\delta_{2,j}}\{\tilde\bSigma^{(\textit{v},j)}_{\tilde\yb_{i,j},\tilde\xi_{i}}(1)\}/ |T_{\delta_{2,j}}\{\tilde\bSigma^{(\textit{v},j)}_{\tilde\yb_{i,j},\tilde\xi_{i}}(1)\}|_2\,.
            \]
            To calculate the estimation error of $\tilde\ab_{i,j}^{(\textit{v})}$, we need the following  results shown in Lemmas \ref{lemma: E varphi}--\ref{lem: zeta_2 and zeta_4}, where the involved constant $\tilde C_0 > 0$ is independent of $\tilde C$ and $(i,j,\textit{v})$ that may vary in different lemmas.  The proofs of Lemmas  \ref{lemma: E varphi}--\ref{lem: zeta_2 and zeta_4} are given in Sections \ref{sec:pro:varphi}--\ref{sec:lem: zeta_2 and zeta_4}, respectively.

\begin{lemma}\label{lemma: E varphi}
    Under the assumptions of Theorem \textup{\ref{thm: iterative}} and $r \ge 2$, for any $i\ne \ell$, as $n\rightarrow\infty$, we have $\bar\varphi_{i,\ell}=O(\gamma_{\max}+n^{-1}w_r^{-2})$, where $\bar\varphi_{i,\ell}$ is defined below \eqref{xi it sp}.
\end{lemma}

            \begin{lemma}\label{bound for Lambda and phi}
                Under the assumptions of Theorem \textup{\ref{thm: iterative}} and the event $\Xi_n(\tilde C)=\cap_{h=1}^9\Xi_{h,n}(\tilde C)$ for  a sufficiently large constant   $\tilde C > 0$, given $j \in [m]$ and $\textit{v} \ge 1$, if $w_r^{-1}w_1\bar\theta_j^{(\textit{v})}\le \tilde C^{-2}$ and $r\ge 2$, it holds for all $i,\ell\in[r]$ that
                \[
                \bigg|\frac{\Lambda_{\xi,i,i}^{1/2}}{\tilde{\sigma}^{(\textit{v},j)}_{\check f,i} }-1\bigg|\le \tilde C_0\tilde C^{-1}~~\textup{and}~~|\tilde\varphi_{i,\ell}^{(\textit{v},j)}-\bar\varphi_{i,\ell}|\le \tilde C_0\tilde C^{-1} 
                \]
                as $n \to \infty$, where $\tilde C_0>0$ is a constant independent of $\tilde C$ and $(i,j,\textit{v})$, and $\bar\varphi_{i,\ell}$ is defined below \eqref{xi it sp}.
            \end{lemma}
       
        \begin{lemma}\label{lem: zeta_1 lower bound}
		Under the same assumptions as in Lemma \textup{\ref{bound for Lambda and phi}}, it holds for all $i\in[r]$ that
        \[
        \bigg|\frac{1}{w_i\sigma_{f_i,\xi_i}}\bzeta_1^{(i,j,\textit{v})}-\ab_{i,j}\bigg|_2 \le \tilde C_0\tilde C^{-1/2} 
        \]
as $n \to \infty$, where $\tilde C_0>0$ is a constant independent of $\tilde C$ and $(i,j,\textit{v})$, and  $\sigma_{f_i,\xi_i}$ is defined in \eqref{auto cross f xi}.
		\end{lemma}	
               \begin{lemma}\label{lem: zeta_2 and zeta_4}
		Under the same assumptions as in Lemma \textup{\ref{bound for Lambda and phi}}, it holds for all $i\in[r]$ that
  \begin{equation}\label{eq: zeta2+zeta4}
   \begin{split}
				      |\bzeta_2^{(i,j,\textit{v})}|_2 + |\bzeta_4^{(i,j,\textit{v})}|_2 
                 &\le  \tilde C_0\tilde C\bigg( \frac{\gamma_{\max}+\Delta_{0,n}}{w_r} +\frac{1}{n w_r^3}\bigg) \\
               &\quad  +\tilde C_0\tilde C\bigg(\frac{w_1}{w_i}\bar\theta_j^{(\textit{v})}+\frac{w_1}{w_iw_r}L_n \bigg)w_i\bar\theta_j^{(\textit{v})} 
\end{split}
\end{equation}
as $n \to \infty$, where $\tilde C_0>0$ is a constant independent of $\tilde C$ and $(i,j,\textit{v})$, $L_n$ is defined in \eqref{auto cross f xi}, and 
\[
\begin{split}
    \Delta_{0,n}&=\max_{i\in[r]}\frac{1}{n-1}|(\Fb_{\xi,\mminus i}^{\textup{s}})^{\T}\bxi_i^{\textup{s}}-\mathbb{E}\{(\Fb_{\xi,\mminus i}^{\textup{s}})^{\T}\bxi_i^{\textup{s}}\}|_2 \\
            & \quad + \max_{i,\ell \in [r]} \bigg| \frac{1}{n-1}\sum_{t=2}^n (f_{t,i}-\bar f_i)(f_{t-1,\ell}-\bar f_{\ell})  - \frac{1}{n-1}\sum_{t=2}^n \mathbb{E}\{ (f_{t,i}-\bar f_i)(f_{t-1,\ell}-\bar f_{\ell})\} \bigg|\,.
\end{split}
            \]
		\end{lemma}

 The above lemmas serve as preliminaries for deriving the upper bound of $\theta_j^{(\textit{v})}$.  To proceed, we need to further define several random quantities.  Recall the vectors $\tilde\bSigma_{\eb_{\ell,j},\xi_i}(1)$ and $\tilde\bSigma_{\eb_{\ell,j},\xi_i}(0)$ defined in \eqref{sigma ei xi}. Let
	\begin{equation}\label{Delta 123 nj}
		\begin{split}
			\Delta_{1,n,j}&=\max_{i,\ell  \in[r]} \{|\tilde\bSigma_{\eb_{\ell,j},\xi_i}(1)|_{\max}+|\tilde\bSigma_{\eb_{\ell,j},\xi_i}(0)|_{\max} \}\,,\\
			\Delta_{2,n,j}(x)&=\max_{i,\ell  \in[r]}\sum_{p = 1}^{d_j}I\{|[\tilde\bSigma_{\eb_{\ell,j},\xi_i}(1)]_p|+|[\tilde\bSigma_{\eb_{\ell,j},\xi_i}(0)]_p|>x(n^{-1} \log d_j)^{1/2} \}
		\end{split}
	\end{equation}
	for $x>0$, where $[\cdot]_p$ stands for the $p$-th entry of a vector. Lemma \ref{lem: truncated projected covariance} below provides an error bound for $ |T_{\delta_{2,j}}\{\tilde\bSigma^{(\textit{v},j)}_{\tilde\yb_{i,j},\tilde\xi_{i}}(1)\}-\bzeta_1^{(i,j,\textit{v})}|_2$, whose proof is given in Section \ref{sec:lem: truncated projected covariance} and relies on the previous results in Lemmas \ref{lemma: E varphi}--\ref{lem: zeta_2 and zeta_4}.
    \begin{lemma}\label{lem: truncated projected covariance}
    		Under the same assumptions in Lemma \textup{\ref{bound for Lambda and phi}}, by setting the threshold level in  Algorithm \textup{\ref{alg1}} as $\delta_{2,j}=\tilde C_*(n^{-1}\log d_j)^{1/2}$ for some sufficiently large constant $\tilde C_*>0$, it holds for all $i\in[r]$ that
         		\[
			\begin{split}
				|T_{\delta_{2,j}}\{\tilde\bSigma^{(\textit{v},j)}_{\tilde\yb_{i,j},\tilde\xi_{i}}(1)\}-\bzeta_1^{(i,j,\textit{v})}|_2 &\le \tilde C_0 \tilde C\bigg\{w_r\Delta_{n,j}+ \frac{1}{w_r}\bigg(\gamma_{\max} + \Delta_{0,n} + \frac{1}{n w_r^2} \bigg)\bigg\}\\
                &\quad+\tilde C_0\tilde C\bigg(\frac{w_1}{w_i}\bar\theta_j^{(\textit{v})}+\frac{w_1}{w_iw_r}L_n\bigg)w_i\bar\theta_j^{(\textit{v})} 
			\end{split}
			\]
            as $n \to \infty$, where $\tilde C_0>0$ is a constant independent of $\tilde C$ and $(i,j,\textit{v})$, $L_n$ is defined in \eqref{auto cross f xi},  and 
            \[
            		  \Delta_{n,j}=w_r^{-1}\{\delta_{2,j}s_{j}^{1/2}+ \Delta_{1,n,j}s_j^{1/2}+\Delta_{1,n,j}\Delta_{2,n,j}(0.4\tilde C_*) \} \,.
            \]
    \end{lemma}

   Based on Lemma \ref{lem: truncated projected covariance}, we now construct an upper bound for $\theta_j^{(\textit{v})}$. On the one hand, if 
   \begin{equation}\label{opposite case}
          \Delta_{n,j}+ \frac{1}{w_r^2}\bigg(\gamma_{\max} + \Delta_{0,n} + \frac{1}{n w_r^2} \bigg)\ge \frac{\min_{i \in [r]}|\sigma_{f_i,\xi_i}|}{4\tilde C_0\tilde C}\,,
   \end{equation}
   then by the fact that $\theta_j^{(\textit{v})}\le 2$,  we directly have 
   \[
   \theta_j^{(\textit{v})}\le \frac{8\tilde C_0\tilde C}{\min_{i \in [r]}|\sigma_{f_i,\xi_i}|}\bigg\{\Delta_{n,j}+ \frac{1}{w_r^2}\bigg(\gamma_{\max} + \Delta_{0,n} + \frac{1}{n w_r^2} \bigg)\bigg\}+\tilde C^{-1}\bar  \theta_j^{(\textit{v})}\,.
   \]
   On the other hand, if \eqref{opposite case} does not hold,  by Lemma \ref{lem: truncated projected covariance} and condition \eqref{strong factor condition}, we have
   \[
   |T_{\delta_{2,j}}\{\tilde\bSigma^{(\textit{v},j)}_{\tilde\yb_{i,j},\tilde\xi_{i}}(1)\}-\bzeta_1^{(i,j,\textit{v})}|_2\le 0.3w_i|\sigma_{f_i,\xi_i}|\,,
   \]
   because $w_r^{-1}w_1\bar\theta_j^{(\textit{v})}$ is assumed to be sufficiently small, $(w_iw_r)^{-1}w_1L_n\ll 1$ by condition \eqref{strong factor condition}, and $|\sigma_{f_i,\xi_i}|$ is lower bounded by some constant in condition \eqref{strong factor condition}. 
   Moreover, for sufficiently large $\tilde C$,   Lemma \ref{lem: zeta_1 lower bound} implies that $|\bzeta_1^{(i,j,\textit{v})}|_2\ge 0.8w_i|\sigma_{f_i,\xi_i}|$. Therefore, by Triangle inequality, we have $|T_{\delta_{2,j}}\{\tilde\bSigma^{(\textit{v},j)}_{\tilde\yb_{i,j},\tilde\xi_{i}}(1)\}|_2\ge 0.5 w_i |\sigma_{f_i,\xi_i}|$.  By the definition of the estimator $\tilde\ab_{i,j}^{(\textit{v})}$, we have
			\[
			\begin{split}
				\tilde\ab_{i,j}^{(\textit{v})}&=\frac{T_{\delta_{2,j}}\{\tilde\bSigma^{(\textit{v},j)}_{\tilde\yb_{i,j},\tilde\xi_{i}}(1)\}}{|T_{\delta_{2,j}}\{\tilde\bSigma^{(\textit{v},j)}_{\tilde\yb_{i,j},\tilde\xi_{i}}(1)\}|_2} \\
                &=\frac{\bzeta_1^{(i,j,\textit{v})}}{|T_{\delta_{2,j}}\{\tilde\bSigma^{(\textit{v},j)}_{\tilde\yb_{i,j},\tilde\xi_{i}}(1)\}|_2}+\frac{T_{\delta_{2,j}}\{\tilde\bSigma^{(\textit{v},j)}_{\tilde\yb_{i,j},\tilde\xi_{i}}(1)\}-\bzeta_1^{(i,j,\textit{v})}}{|T_{\delta_{2,j}}\{\tilde\bSigma^{(\textit{v},j)}_{\tilde\yb_{i,j},\tilde\xi_{i}}(1)\}|_2}\\
				&= \frac{\bzeta_1^{(i,j,\textit{v})}}{|\bzeta_1^{(i,j,\textit{v})}|_2}-\frac{\bzeta_1^{(i,j,\textit{v})}}{|\bzeta_1^{(i,j,\textit{v})}|_2} \cdot \frac{|T_{\delta_{2,j}}\{\tilde\bSigma^{(\textit{v},j)}_{\tilde\yb_{i,j},\tilde\xi_{i}}(1)\}|_2-|\bzeta_1^{(i,j,\textit{v})}|_2}{|T_{\delta_{2,j}}\{\tilde\bSigma^{(\textit{v},j)}_{\tilde\yb_{i,j},\tilde\xi_{i}}(1)\}|_2}+\frac{T_{\delta_{2,j}}\{\tilde\bSigma^{(\textit{v},j)}_{\tilde\yb_{i,j},\tilde\xi_{i}}(1)\}-\bzeta_1^{(i,j,\textit{v})}}{|T_{\delta_{2,j}}\{\tilde\bSigma^{(\textit{v},j)}_{\tilde\yb_{i,j},\tilde\xi_{i}}(1)\}|_2}\,.
			\end{split}
			\]
			Notice that $\bzeta_1^{(i,j,\textit{v})}/|\bzeta_1^{(i,j,\textit{v})}|_2\in\{\ab_{i,j},-\ab_{i,j}\}$. Combining Lemma \ref{lem: truncated projected covariance} with condition \eqref{strong factor condition}, and ignoring the reflection and permutation indeterminacy, we can conclude that 
	\begin{equation}\label{iterative relationship}
	    	   \theta_j^{(\textit{v})}\le \tilde C_0\tilde C\bigg\{\Delta_{n,j}+ \frac{1}{w_r^2}\bigg(\gamma_{\max} + \Delta_{0,n} + \frac{1}{n w_r^2} \bigg)\bigg\}+\tilde C^{-1}\bar  \theta_j^{(\textit{v})} 
	\end{equation}
            as $n \to \infty$, where $\tilde C_0>0$ is a constant independent of $\tilde C$ and $(i,j,\textit{v})$.

\subsubsection{Step 3: The convergence rate after sufficient iterations}\label{sec:proofT2-step3}

Because $  \max_{j\in[m]}w_r^{-1}w_1\theta_j^{(0)}=o_{\rm p}(1)$,  we have $w_r^{-1}w_1\bar\theta_1^{(1)}\le \tilde C^{-3}$ with probability approaching one for the large constant $\tilde C$ in Lemmas
\ref{bound for Lambda and phi}--\ref{lem: truncated projected covariance}.
For now, let us take this event and $\Xi_n(\tilde C)$ as given. To ease notation, let
\[
    \tilde\Delta_{n}
    =
    \max_{j\in[m]} \Delta_{n,j}
    +
    \frac{1}{w_r^2}
    \bigg(
    \gamma_{\max}
    +
    \Delta_{0,n}
    +
    \frac{1}{n w_r^2}
    \bigg)\,.
\]
Then, if $w_r^{-1}w_1\bar\theta_j^{(\textit{v})}\le \tilde C^{-2}$, 
\eqref{iterative relationship} implies that
\begin{equation}\label{eq:basic iterative recursion}
    \theta_j^{(\textit{v})}
    \le
    A\tilde\Delta_n+\alpha\bar\theta_j^{(\textit{v})}\,,
\end{equation}
where $A=\tilde C_0\tilde C$, and $\alpha=\tilde C^{-1}\in(0,1)$ since $\tilde C$ is chosen sufficiently large. We next control the magnitude of $\tilde\Delta_n$. Lemma \ref{lem: Delta} below is used to
bound $\Delta_{0,n}$, $\Delta_{1,n,j}$, and $\Delta_{2,n,j}(x)$, whose proof is given in Section \ref{sec:lem: Delta}.

\begin{lemma}\label{lem: Delta}
Under the assumptions of Theorem \textup{\ref{thm: iterative}} and $r \ge 2$, it holds that
\[
    \Delta_{0,n}=O_{\rm p}\bigg(\frac{1}{\sqrt n}\bigg) 
    ~~\textup{and}~~
    \Delta_{1,n,j}=O_{\rm p}\bigg(\sqrt{\frac{\log d_j}{n}}\bigg)
\]
for each $j\in[m]$, and there exists a large constant $\tilde C_*>0$ such that
$\Delta_{2,n,j}(0.4\tilde C_*)=O_{\rm p}(1)$, where $\Delta_{0,n}$ is defined
in Lemma \textup{\ref{lem: zeta_2 and zeta_4}}, and
$\Delta_{1,n,j}$, $\Delta_{2,n,j}(x)$ are defined in \eqref{Delta 123 nj}.
\end{lemma}

By Lemma \ref{lem: Delta} and the definition of $\Delta_{n,j}$ in Lemma
\ref{lem: truncated projected covariance}, we have
\[
    \Delta_{n,j}=O_{\rm p}(\Phi_{n,j})
\]
for each $j\in[m]$, as long as the constant $\tilde C_*$
is sufficiently large. Furthermore, under condition \eqref{strong factor condition}
and the fact that $s_j\le d_j$, we have
\[
    \max_{j\in[m]}\frac{w_1}{w_r}\Delta_{n,j}=o_{\rm p}(1) 
    ~~\textup{and}~~
    \frac{w_1}{w_r^3}
    \bigg(
    \gamma_{\max}
    +
    \Delta_{0,n}
    +
    \frac{1}{n w_r^2}
    \bigg)
    =
    o_{\rm p}(1)\,.
\]
Therefore, with probability approaching one, it holds that
\begin{equation}\label{eq:small Delta event}
    \frac{w_1}{w_r}A\tilde\Delta_n
    \le
    \tilde C^{-4}\,.
\end{equation}
We also take this event as given.

We now prove the result recursively. Fix $\textit{v}=1$ and let
$j$ grow from $1$ to $m$. When $j=1$, by \eqref{eq:basic iterative recursion}, it holds that
\[
    \theta_1^{(1)}
    \le
    A\tilde\Delta_n+\alpha\bar\theta_1^{(1)}.
\]
Next, let $j=2$. By the definition of $\bar\theta_2^{(1)}$, we have
\begin{equation}\label{eq: bar theta 21}
\begin{split}
    \bar\theta_2^{(1)}
     =
    \max(\theta_1^{(1)},\theta_2^{(0)},\ldots,\theta_m^{(0)}) 
     \le
    \max(A\tilde\Delta_n+\alpha\bar\theta_1^{(1)},
    \bar\theta_1^{(1)})                                      
     \le
    A\tilde\Delta_n+\bar\theta_1^{(1)}\,.
\end{split}
\end{equation}
Thus, by \eqref{eq:small Delta event} and
$w_r^{-1}w_1\bar\theta_1^{(1)}\le \tilde C^{-3}$, it holds that
\[
    \frac{w_1}{w_r}\bar\theta_2^{(1)}
    \le
    \frac{w_1}{w_r}A\tilde\Delta_n
    +
    \frac{w_1}{w_r}\bar\theta_1^{(1)}
    \le
    \tilde C^{-4}+\tilde C^{-3}
    \le
    \tilde C^{-2}
\]
as long as $\tilde C$ is sufficiently large. Therefore,
\eqref{eq:basic iterative recursion} can be applied to $j=2$, and
\[
    \theta_2^{(1)}
    \le
    A\tilde\Delta_n+\alpha\bar\theta_2^{(1)}
    \le
    A(1+\alpha)\tilde\Delta_n+\alpha\bar\theta_1^{(1)}\,.
\]
By induction over $j \in [m]$, we obtain
\[
\begin{split}
    \bar\theta_j^{(1)}
    &\le
    \max(\theta_{j-1}^{(1)},\bar\theta_{j-1}^{(1)})
    \le
    \frac{w_r}{w_1}\tilde C^{-2}\,,\\
    \theta_j^{(1)}
    &\le
    A(1+\alpha+\cdots+\alpha^{j-1})\tilde\Delta_n
    +
    \alpha\bar\theta_1^{(1)}
    \le
    \frac{A}{1-\alpha}\tilde\Delta_n
    +
    \alpha\bar\theta_1^{(1)}\,.
\end{split}
\]
The first inequality above verifies that the condition required for
\eqref{eq:basic iterative recursion} remains valid throughout the first round.

Now let $\textit{v}=2$. By definition,
\[
    \bar\theta_1^{(2)}
    =
    \max_{j\in[m]}\theta_j^{(1)}
    \le
    \frac{A}{1-\alpha}\tilde\Delta_n
    +
    \alpha\bar\theta_1^{(1)}\,.
\]
Using \eqref{eq:small Delta event} again, we have
\[
    \frac{w_1}{w_r}\bar\theta_1^{(2)}
    \le
    \frac{\tilde C^{-4}}{1-\alpha}
    +
    \alpha\tilde C^{-3}
    \le
    \tilde C^{-2}
\]
for sufficiently large $\tilde C$. Hence, 
\eqref{eq:basic iterative recursion} is applicable in the second round.
Following the same induction over $j \in [m]$, it holds that
\[
\begin{split}
    \theta_j^{(2)}
     & \le
    A(1+\alpha+\cdots+\alpha^{j-1})\tilde\Delta_n
    +
    \alpha\bar\theta_1^{(2)}  \\
    &\le
    \frac{A}{1-\alpha}\tilde\Delta_n
    +
    \alpha\bar\theta_1^{(2)}\\
    &\le
    \frac{A}{1-\alpha}(1+\alpha)\tilde\Delta_n
    +
    \alpha^2\bar\theta_1^{(1)}\,.
\end{split}
\]
Repeating the above argument over the iteration index $\textit{v}$ yields, for any
$\textit{v}\ge1$,
\begin{equation}\label{eq: theta iterative final recursion}
    \theta_j^{(\textit{v})}
    \le
    \frac{A}{1-\alpha}
    (1+\alpha+\cdots+\alpha^{\textit{v}-1})\tilde\Delta_n
    +
    \alpha^{\textit{v}}\bar\theta_1^{(1)}
    \le
    \frac{A}{(1-\alpha)^2}\tilde\Delta_n
    +
    \alpha^{\textit{v}}\bar\theta_1^{(1)}\,.
\end{equation}
Moreover, the same induction verifies that
\[
    \frac{w_1}{w_r}\bar\theta_j^{(\textit{v})}\le \tilde C^{-2}
\]
for each $j\in[m]$ and each iteration $\textit{v}$ considered above, so all
applications of \eqref{eq:basic iterative recursion} are justified.

Next, by Lemma \ref{lem: Delta}, we have
\[
    \tilde\Delta_n
    =
    O_{\rm p}\bigg(
    \max_{j\in[m]}\Phi_{n,j}
    +
    \frac{\gamma_{\max}}{w_r^2}
    +
    \frac{1}{w_r^2\sqrt n}
    +
    \frac{1}{n w_r^4}
    \bigg)\,.
\]
Since $D_n\rightarrow\infty$, $m$ is fixed, and $w_r$ is bounded away from $0$,
the last two terms are dominated by $\max_{j\in[m]}\Phi_{n,j}$. Therefore,
\begin{equation}\label{eq: tilde Delta final rate}
    \tilde\Delta_n
    =
    O_{\rm p}\bigg(
    \max_{j\in[m]}\Phi_{n,j}
    +
    \frac{\gamma_{\max}}{w_r^2}
    \bigg)\,.
\end{equation}
Finally, take
\[
    \textit{v}_{\max}
    \gtrsim
    -\log\bigg(
    \max_{j\in[m]}\Phi_{n,j}
    +
    \frac{\gamma_{\max}}{w_r^2}
    \bigg)\,.
\]
Since $\alpha\in(0,1)$ and $\bar\theta_1^{(1)}\le2$, we have
\[
    \alpha^{\textit{v}_{\max}}\bar\theta_1^{(1)}
    =
    O_{\rm p}\bigg(
    \max_{j\in[m]}\Phi_{n,j}
    +
    \frac{\gamma_{\max}}{w_r^2}
    \bigg).
\]
Combining this with \eqref{eq: theta iterative final recursion} and
\eqref{eq: tilde Delta final rate}, we obtain
\[
    \max_{i\in[r],j\in[m]}
    |\hat\ab_{i,j}-\ab_{i,j}|_2
    =
    O_{\rm p}\bigg(
    \max_{j\in[m]}\Phi_{n,j}
    +
    \frac{\gamma_{\max}}{w_r^2}
    \bigg)\,.
\]

\subsubsection{Some further discussion}\label{sec:proofT2-step4}
When $r=1$, we set $\gamma_{\max}=0$. In this case,
$\tilde\Fb_{\mminus i}^{(\textit{v},j)}$ and
$\Fb_{\xi,\mminus i}^{\textup{s}}$ disappear, and
$\tilde\xi_{t,i}^{(\textit{v},j)}=\tilde f_{t,i}^{(\textit{v},j)}$,
$\xi_{t,i}^{\textup{sp}}=\xi_{t,i}^{\textup{s}}$. Moreover,
$\bzeta_2^{(i,j,\textit{v})}$ and the second term of
$\bzeta_3^{(i,j,\textit{v})}$ defined above \eqref{iterative sigma decompose}
vanish. Therefore, the same proof applies with the terms involving
$\gamma_{\max}$ and $\bar\varphi_{i,\ell}$ removed. Therefore, the same convergence rate follows immediately.

Next, we discuss the impact of the reflection and permutation indeterminacy on the results.  On the one hand, if $z_i \neq i$, we can simply permute the estimators to match $z_i = i$.
On the other hand, for some $\tilde\kappa_{i,j} \in \{-1,1\}$, we can absorb these constants by relabelling the factors and factor loadings, i.e., rewrite the tensor CP-factor model \eqref{model cp} as follows
\begin{equation}\label{eq: rewrite cp model} 
\begin{split}
    		\mathcal{Y}_t &= \sum_{i=1}^r \bigg( \prod_{j=1}^m \tilde\kappa_{i,j}\bigg)  w_i f_{t,i}\, (\tilde\kappa_{i,1}\ab_{i,1}) \circ (\tilde\kappa_{i,2}\ab_{i,2}) \circ \cdots \circ (\tilde\kappa_{i,m}\ab_{i,m}) +\mathcal{E}_t \\
            &=  \sum_{i=1}^r w_i \ddot{f}_{t,i} \,\ddot{\ab}_{i,1} \circ \ddot{\ab}_{i,2} \circ \cdots \circ \ddot{\ab}_{i,m} +\mathcal{E}_t \,,     
            \end{split}
\end{equation}
 where the factor $\ddot{f}_{t,i}  = (\textstyle\prod\nolimits_{j=1}^m\tilde\kappa_{i,j}) f_{t,i}$ and the factor loading $\ddot{\ab}_{i,j} = \tilde\kappa_{i,j}\ab_{i,j}$ still satisfy Assumptions \ref{error}--\ref{cross}. We define the oracle linear combinations $\ddot{\xi}_{t,i}$ with  $(\ddot{f}_{t,i},\ddot{\ab}_{i,j})$ by similar steps above \eqref{xi it sp}. Then, all the assumptions required in Theorem \ref{thm: iterative} still hold under the new model \eqref{eq: rewrite cp model}.  With the reflection indeterminacy, we define the estimation error by ${\theta}_j^{(\textit{v})}=\max_{i\in[r]}\min(|\tilde\ab_{i,j}^{(\textit{v})}-{\ab}_{i,j}|_2,|\tilde\ab_{i,j}^{(\textit{v})}+{\ab}_{i,j}|_2)$. Then, $\bar\theta_1^{(1)}=\max_{i\in[r]}|\tilde\ab_{i,1}^{(0)}-{\ddot\ab}_{i,1}|_2$. By Lemma \ref{lem: zeta_1 lower bound} and following the same proof strategy as that for \eqref{iterative relationship}, we have
\begin{equation}\label{iteration relationship reflection}
\begin{split}
        \theta_1^{(1)} &= \max_{i\in[r]} |\tilde\ab_{i,1}^{(1)}-\textup{sgn}(\ddot\sigma_{\ddot f_i,\ddot\xi_i}) \cdot {\ddot\ab}_{i,1}|_2 \\
        &\le \tilde C_0\tilde C\bigg\{\Delta_{n,1}+ \frac{1}{w_r^2}\bigg(\gamma_{\max} + \Delta_{0,n} + \frac{1}{n w_r^2} \bigg)\bigg\}+\tilde C^{-1}\bar  \theta_1^{(1)}\,,
\end{split}
\end{equation}
 where $\textup{sgn}(x)= 2 I(x\ge 0) - 1$ is the sign function, and $\ddot\sigma_{\ddot f_i,\ddot\xi_i}$ is defined in the same manner as $\sigma_{ f_i,\xi_i}$ but with replacing $(f_{t,i},\ab_{i,j},\xi_{t,i})$ by $(\ddot{f}_{t,i},\ddot{\ab}_{i,j},\ddot{\xi}_{t,i})$.  In fact, it follows directly from their definitions that $\ddot\sigma_{\ddot f_i,\ddot\xi_i}=\sigma_{f_i,\xi_i}$.
Using the same iteration argument as in \eqref{iterative relationship}, and arguing as in the derivation of \eqref{eq: theta iterative final recursion}, we can conclude that
\[
{\theta}_j^{(\textit{v}_{\max})} \le  \frac{A}{(1-\alpha)^2}\tilde\Delta_n
    + 2\alpha^{\textit{v}_{\max}}=O_{\rm p}\bigg(\max_{j\in[m]}\Phi_{n,j}+\frac{\gamma_{\max}}{w_r^2}\bigg) 
\]
for $\textit{v}_{\max} \gtrsim  -\log (
    \max_{j\in[m]}\Phi_{n,j}
    +
    \gamma_{\max} w_r^{-2}
     )$. This completes the proof of Theorem \ref{thm: iterative}. 
     $\hfill\Box$      

	\subsection{Proof of Theorem \ref{thm: debias iterative}}\label{sec: prove thm debias iterative}

    

The proof focuses on the case $r\ge 2$, since the case $r=1$ is straightforward and can be handled similarly. For notational simplicity, the reflection and permutation indeterminacy is temporarily ignored by taking $\kappa_{i,j}=1$ and $z_i=i$, which will be revisited later. We take the event $ \Xi_n(\tilde C)$ as given. Moreover, by the convergence rate in Theorem \ref{thm: iterative} and condition \eqref{strong factor condition}, we have $w_r^{-1}w_1\bar\theta_j^{(\textit{v}_{\max})}\le \tilde C^{-2}$ with probability approaching one. We also condition on this event throughout the proof.	 By the definition of $\hat{\bvartheta}_{i,j}$, it holds that
    \[
\hb^\T(\hat\ab_{i,j}-\hat{\bvartheta}_{i,j})=\frac{\hb^\T\tilde\bSigma^{(\textit{v}_{\max},j)}_{\tilde\yb_{i,j},\tilde\xi_{i}}(1)}{\hat\ab_{i,j}^{\T} \tilde\bSigma^{(\textit{v}_{\max},j)}_{\tilde\yb_{i,j},\tilde\xi_{i}}(1)}\,.
    \]
     We first derive the asymptotic representation of $\hat\ab_{i,j}^{\T} \tilde\bSigma^{(\textit{v}_{\max},j)}_{\tilde\yb_{i,j},\tilde\xi_{i}}(1)$. Following the decomposition in (\ref{iterative sigma decompose}), write
     \[
\hat\ab_{i,j}^{\T} \tilde\bSigma^{(\textit{v}_{\max},j)}_{\tilde\yb_{i,j},\tilde\xi_{i}}(1)=\sum_{\ell=1}^4 \hat\ab_{i,j}^\T  \bzeta_\ell^{(i,j,\textit{v}_{\max})}\,.  
\]
For $\bzeta_1^{(i,j,\textit{v}_{\max})}$, notice that the constant $\tilde C$ in Lemma \ref{lem: zeta_1 lower bound} can be sufficiently large. Thus,
\[
\ab_{i,j}^\T \bzeta_1^{(i,j,\textit{v}_{\max})}=w_i\sigma_{f_i,\xi_i}\{1+o_{\rm p}(1)\}\,.
\]
Since $\bzeta_1^{(i,j,\textit{v}_{\max})}$ is  proportional to $\ab_{i,j}$, we have $\bzeta_1^{(i,j,\textit{v}_{\max})}=(\ab_{i,j}^\T \bzeta_1^{(i,j,\textit{v}_{\max})})\ab_{i,j}$ and
		\begin{equation}\label{aij zeta1}
		\begin{split}
			|\hat\ab_{i,j}^{\T}\bzeta_1^{(i,j,\textit{v}_{\max})}-\ab_{i,j}^{\T}\bzeta_1^{(i,j,\textit{v}_{\max})}| &= |\ab_{i,j}^{\T}\bzeta_1^{(i,j,\textit{v}_{\max})}   \ab_{i,j}^{\T}(\hat\ab_{i,j}-\ab_{i,j})|\\
            &= |w_i\sigma_{f_i,\xi_i}| \cdot O_{\rm p}(|\hat\ab_{i,j}-\ab_{i,j}|_2^2)\\
            &=o_{\rm p}( w_i\sigma_{f_i,\xi_i} )\,.
		\end{split}
		\end{equation}
        Therefore, $\hat\ab_{i,j}^\T \bzeta_1^{(i,j,\textit{v}_{\max})}=w_i\sigma_{f_i,\xi_i}\{1+o_{\rm p}(1)\}$. 
        
	For $\bzeta_2^{(i,j,\textit{v}_{\max})}$ and $\bzeta_4^{(i,j,\textit{v}_{\max})}$, following the proof of Lemma \ref{lem: zeta_2 and zeta_4}, we have 
\begin{equation}\label{zeta 2 and zeta 4 consistency}
    \begin{split}	&|\bzeta_2^{(i,j,\textit{v}_{\max})}|_2+|\bzeta_4^{(i,j,\textit{v}_{\max})}|_2\\
&~~~~~~= O_{\rm p}\bigg\{\frac{\gamma_{\max}+ \Delta_{0,n}}{w_r}+ \frac{1}{n w_r^3} +\bigg(w_1\bar\theta_j^{(\textit{v}_{\max})}+\frac{w_1}{w_r}L_n\bigg)\bar\theta_j^{(\textit{v}_{\max})}\bigg\}\,\\
&~~~~~~= O_{\rm p}\bigg(\frac{\gamma_{\max}+ n^{-1/2}}{w_r}+\frac{w_1}{w_r}L_n \max_{j\in[m]}\Phi_{n,j}  \bigg)\,,
\end{split}
\end{equation}
        where the second equality follows from condition \eqref{strong factor condition}, Lemma \ref{lem: Delta}, and the convergence rate of $\bar\theta_j^{(\textit{v}_{\max})}$ in Theorem \ref{thm: iterative}. Hence, we have $|\hat\ab_{i,j}^\T(\bzeta_2^{(i,j,\textit{v}_{\max})}+\bzeta_4^{(i,j,\textit{v}_{\max})})|=o_{\rm p}(w_i)$. 

		For $\bzeta_3^{(i,j,\textit{v}_{\max})}$, write
\begin{equation*}
	\hat\ab_{i,j}^{\T}\bzeta_3^{(i,j,\textit{v}_{\max})}= \ab_{i,j}^{\T}\bzeta_3^{(i,j,\textit{v}_{\max})} + (\hat\ab_{i,j}-\ab_{i,j})^{\T}\bzeta_3^{(i,j,\textit{v}_{\max})}\,.
		\end{equation*}
        On the one hand, following the proof of Lemma \ref{bound for Lambda and phi}, we can conclude that
        \[
                        \frac{\Lambda_{\xi,i,i}^{1/2}}{\tilde{\sigma}^{(\textit{v}_{\max},j)}_{\check f,i} }=1+o_{\rm p}(1) ~~\textup{and}~~
                \tilde\varphi_{i,\ell}^{(\textit{v}_{\max},j)}=\bar\varphi_{i,\ell}+o_{\rm p}(1)\,.
        \]
        Then, by the definition of $\bzeta_3^{(i,j,\textit{v}_{\max})}$ above (\ref{iterative sigma decompose}), we have
		\begin{equation}\label{aij zeta3}
		\begin{split}
		 \ab_{i,j}^{\T}\bzeta_3^{(i,j,\textit{v}_{\max})}&=\{1+o_{\rm p}(1)\}\ab_{i,j}^{\T}\bigg[\tilde\bSigma_{\eb_{i,j},\xi_i}(1)+\sum_{\ell\ne i}\{\bar\varphi_{i,\ell}+o_{\rm p}(1)\}\tilde\bSigma_{\eb_{i,j},\xi_\ell}(0)\bigg] =  O_{\rm p}(n^{-1/2})\,,
		\end{split}
		\end{equation}
        where we use the fact that $\hb^\T \tilde\bSigma_{\eb_{i,j},\xi_\ell}(k)=O_{\rm p}(n^{-1/2})$ for any deterministic unit vector $\hb$, $i,\ell\in[r]$, and $k\in\{0,1\}$.
		On the other hand, by Cauchy–Schwarz inequality,  we have
		\[
\begin{split}
    		|(\hat\ab_{i,j}-\ab_{i,j})^{\T}\bzeta_3^{(i,j,\textit{v}_{\max})}|& \lesssim \bar\theta^{(\textit{v}_{\max})}_j   \bigg\{|\tilde\bSigma_{\eb_{i,j},\xi_i}(1)|_2^2+\sum_{\ell\ne i}|\tilde\bSigma_{\eb_{i,j},\xi_\ell}(0)|_2^2\bigg\}^{1/2} \\
            &=  O_{\rm p}\bigg(\max_{j\in[m]}\Phi_{n,j}+\frac{\gamma_{\max}}{w_r^2}\bigg) \cdot \sqrt{\frac{d_j}{n}} \\
            &= o_{\rm p}(w_i)\,,
\end{split}		\]
where the last equality holds by condition \eqref{strong factor condition}.
Combining the above bounds, we obtain $|\hat\ab_{i,j}^{\T}\bzeta_3^{(i,j,\textit{v}_{\max})}|_2= o_{\rm p}(w_i)$. It follows that 
\[
\hat\ab_{i,j}^{\T} \tilde\bSigma^{(\textit{v}_{\max},j)}_{\tilde\yb_{i,j},\tilde\xi_{i}}(1)=w_i\sigma_{f_i,\xi_i}\{1+o_{\rm p}(1)\}\,.
\]
For notational convenience, define 
        \[
        \tilde \Phi_{n}=\frac{\gamma_{\max}+ n^{-1/2}}{w_r}+\frac{w_1}{w_r}L_n\max_{j\in[m]}\Phi_{n,j}\,.
        \]
        Now, by (\ref{iterative sigma decompose}) and \eqref{zeta 2 and zeta 4 consistency}, we have
		\[
\begin{split}
		\hb^\T(\hat\ab_{i,j}-\hat{\bvartheta}_{i,j})
        &=\frac{\hb^{\T}(\bzeta_1^{(i,j,\textit{v}_{\max})}+\bzeta_3^{(i,j,\textit{v}_{\max})})}{\hat\ab_{i,j}^{\T} \tilde\bSigma^{(\textit{v}_{\max},j)}_{\tilde\yb_{i,j},\tilde\xi_{i}}(1)}+\frac{\hb^{\T}(\bzeta_2^{(i,j,\textit{v}_{\max})}+\bzeta_4^{(i,j,\textit{v}_{\max})})}{\hat\ab_{i,j}^{\T} \tilde\bSigma^{(\textit{v}_{\max},j)}_{\tilde\yb_{i,j},\tilde\xi_{i}}(1)} \\
        &= \frac{\hb^{\T}(\bzeta_1^{(i,j,\textit{v}_{\max})}+\bzeta_3^{(i,j,\textit{v}_{\max})})}{\hat\ab_{i,j}^{\T} \tilde\bSigma^{(\textit{v}_{\max},j)}_{\tilde\yb_{i,j},\tilde\xi_{i}}(1)}+O_{\rm p}(w_i^{-1}\tilde \Phi_{n})\,.
\end{split}
		\]
		On the one hand, by \eqref{aij zeta1}, condition \eqref{zeta 2 and zeta 4 consistency}, and \eqref{negligible error 2}, we can conclude that
\begin{equation}\label{eq: h transpose zeta 1}
    \begin{split}
	\frac{\hb^{\T}\bzeta_1^{(i,j,\textit{v}_{\max})}}{\hat\ab_{i,j}^{\T} \tilde\bSigma^{(\textit{v}_{\max},j)}_{\tilde\yb_{i,j},\tilde\xi_{i}}(1)}&=\hb^{\T}\ab_{i,j}-\frac{\hb^{\T}\ab_{i,j}\{\hat\ab_{i,j}^{\T} \tilde\bSigma^{(\textit{v}_{\max},j)}_{\tilde\yb_{i,j},\tilde\xi_{i}}(1)-\ab_{i,j}^{\T}\bzeta_1^{(i,j,\textit{v}_{\max})}\}}{\hat\ab_{i,j}^{\T} \tilde\bSigma^{(\textit{v}_{\max},j)}_{\tilde\yb_{i,j},\tilde\xi_{i}}(1)}\\
	&= \hb^{\T}\ab_{i,j}-\frac{1+o_{\rm p}(1)}{w_i\sigma_{f_i,\xi_i}}\hb^{\T}\ab_{i,j}\ab_{i,j}^{\T}\bzeta_3^{(i,j,\textit{v}_{\max})} \\
    &\quad +O_{\rm p}(w_i^{-1}\tilde\Phi_{n})+o_{\rm p}(w_i^{-1}n^{-1/2})\,.
\end{split}
\end{equation}
		On the other hand, following the arguments used to derive \eqref{aij zeta3}, we obtain
\begin{equation}\label{eq: h transpose zeta3}
    \begin{split}
	\frac{\hb^{\T}\bzeta_3^{(i,j,\textit{v}_{\max})}}{\hat\ab_{i,j}^{\T} \tilde\bSigma^{(\textit{v}_{\max},j)}_{\tilde\yb_{i,j},\tilde\xi_{i}}(1)}&=\frac{1+o_{\rm p}(1)}{w_i\sigma_{f_i,\xi_i}}\hb^{\T}\bigg\{\tilde\bSigma_{\eb_{i,j},\xi_i}(1)+\sum_{\ell\ne i}\bar\varphi_{i,\ell}\tilde\bSigma_{\eb_{i,j},\xi_\ell}(0)\bigg\}+o_{\rm p}(w_i^{-1}n^{-1/2})\,.
\end{split}
\end{equation}
Consequently, we have
\[
\begin{split}
	\hb^\T(\hat\ab_{i,j}-\ab_{i,j}-\hat{\bvartheta}_{i,j})
	&=\frac{1}{w_i\sigma_{f_i,\xi_i}}\hb^{\T}(\Ib_{d_j}-\ab_{i,j}\ab_{i,j}^{\T})\bigg\{\tilde\bSigma_{\eb_{i,j},\xi_i}(1)+\sum_{\ell\ne i}\bar\varphi_{i,\ell}\tilde\bSigma_{\eb_{i,j},\xi_\ell}(0)\bigg\}\\
	&\quad+O_{\rm p}(w_i^{-1}\tilde\Phi_{n})+o_{\rm p}(w_i^{-1}n^{-1/2})\,,
\end{split}
\]
which establishes the desired limiting representation.

The asymptotic distribution follows from Slutsky's theorem and the central
limit theorem for $\alpha$-mixing triangular arrays from Theorem 1 in
\citeS{ekstrom2014general-app}. Recall $\xi_{t,i}^{\textup{sp}}
    =
    \xi_{t,i}^{\textup{s}}
    +
    \sum_{\ell\ne i}\bar\varphi_{i,\ell}\xi_{t+1,\ell}^{\textup{s}}$. With $\bbeta_{i,j}(\hb) =  \bbb_{i,j}^{\MP}  \otimes \{\hb^{\T}(\Ib_{d_j}-\ab_{i,j}\ab_{i,j}^{\T})\}^{\T}$ specified in \eqref{def: bbeta ij}, by \eqref{sigma ei xi} and the
definition of $\xi_{t,i}^{\textup{sp}}$, it follows that
\[
\begin{split}
    &\hb^{\T}(\Ib_{d_j}-\ab_{i,j}\ab_{i,j}^{\T})
    \bigg\{
    \tilde\bSigma_{\eb_{i,j},\xi_i}(1)
    +
    \sum_{\ell\ne i}
    \bar\varphi_{i,\ell}
    \tilde\bSigma_{\eb_{i,j},\xi_\ell}(0)
    \bigg\}  \\
    &\quad=
    \frac{1}{n-1}\sum_{t=2}^n
    \bbeta_{i,j}(\hb)^{\T}\textup{vec}(\Eb_{t,j})\xi_{t-1,i}^{\textup{sp}}
    -
    \frac{1}{n-1}\sum_{t=2}^n
    \mathbb{E}\{
    \bbeta_{i,j}(\hb)^{\T}\textup{vec}(\Eb_{t,j})\xi_{t-1,i}^{\textup{sp}}
    \}
    +
    o_{\rm p}(n^{-1/2})\,.
\end{split}
\]
Under condition \eqref{negligible error 2}, the remainder term
$O_{\rm p}(\tilde\Phi_{n})$ in the limiting representation is
$o_{\rm p}(n^{-1/2})$. Therefore,
\[
\begin{split}
    &(w_i\sigma_{f_i,\xi_i})
    \hb^\T(\hat\ab_{i,j}-\ab_{i,j}-\hat{\bvartheta}_{i,j})  \\
    &\quad=
    \frac{1}{n-1}\sum_{t=2}^n
     [
    \bbeta_{i,j}(\hb)^{\T}\textup{vec}(\Eb_{t,j})\xi_{t-1,i}^{\textup{sp}}
    -
    \mathbb{E}\{
    \bbeta_{i,j}(\hb)^{\T}\textup{vec}(\Eb_{t,j})\xi_{t-1,i}^{\textup{sp}}
    \}
     ]
    +
    o_{\rm p}(n^{-1/2})\,.
\end{split}
\]
We remark that, because the dimensions $d_1,\ldots,d_m$ may grow with the sample size $n$, the leading term on the right-hand side should be regarded as a sample mean from a triangular array.
Therefore, an application of Theorem 1 in \citeS{ekstrom2014general-app} reduces the proof to verifying that there exist constants $\nu>0$ and $C>0$ such that 
			\[
			\mathbb{E}\big[ |\bbeta_{i,j}(\hb)^{\T} \textup{vec}(\Eb_{t,j})\xi_{t-1,i}^{\textup{sp}} |^{2+\nu}\big]<C ~~\textup{and}~~\sum_{\ell=0}^\infty(\ell+1)^2\alpha_n^{\nu/(4+\nu)}(\ell)<C \,,
			\]
			 where  $\alpha_n(\ell)=\alpha(|\ell-1|_+)$ is the $\alpha$-mixing coefficient for  $\{\bbeta_{i,j}(\hb)^{\T} \textup{vec}(\Eb_{t,j})\xi_{t-1,i}^{\textup{sp}}\}_{t=2}^{n}$. On the one hand, by the tail probabilities in Assumptions \ref{tail} and \ref{cross},  the moment condition holds for any constant $\nu>0$. On the other hand, the required condition on the $\alpha$-mixing coefficients follows  from Assumption \ref{mixing}. The result then follows from Theorem 1 in \citeS{ekstrom2014general-app}.

We now address the identification issue arising from the reflection and permutation indeterminacy. On the one hand, if $z_i \neq i$, we can simply permute the estimators to match $z_i = i$. 
On the other hand, for the reflection indeterminacy, notice that the quantities in \eqref{eq: h transpose zeta 1} and \eqref{eq: h transpose zeta3} involve estimators from two adjacent iterations. More specifically, $\hat\ab_{i,j}$ corresponds to the estimator obtained after the $\textit{v}_{\max}$-th iteration, while $\tilde\bSigma^{(\textit{v}_{\max},j)}_{\tilde\yb_{i,j},\tilde\xi_{i}}(1)$ involves the plug-in estimators from the $(\textit{v}_{\max}-1)$-th iteration. Hence, their reflection signs should be matched carefully. Without loss of generality, suppose that the reflection sign associated with the plug-in estimator from the $(\textit{v}_{\max}-1)$-th iteration is $\tilde\kappa_{i,j}$. We have
\[
\max_{i\in[r],j\in[m]}
\big|
\tilde\ab_{i,j}^{(\textit{v}_{\max}-1)}
-
\tilde\kappa_{i,j}\ab_{i,j}
\big|_2
=
O_{\rm p}\bigg(
\max_{j\in[m]}\Phi_{n,j}
+
\frac{\gamma_{\max}}{w_r^2}
\bigg)\,.
\]
By the similar arguments used to derive \eqref{iteration relationship reflection}, the reflection sign in the next iteration changes by $\textup{sgn}(\sigma_{f_i,\xi_i})$. Therefore, the reflection sign of $\hat\ab_{i,j}$ satisfies $ \kappa_{i,j} = \textup{sgn}(\sigma_{f_i,\xi_i}) \cdot \tilde\kappa_{i,j}$.  Equivalently,  
\[
\max_{i\in[r],j\in[m]}
\big|
\tilde\ab_{i,j}^{(\textit{v}_{\max}-1)}
-
\textup{sgn}(\sigma_{f_i,\xi_i}) \cdot \kappa_{i,j}\ab_{i,j}
\big|_2
=
O_{\rm p}\bigg(
\max_{j\in[m]}\Phi_{n,j}
+
\frac{\gamma_{\max}}{w_r^2}
\bigg)\,.
\]
This relation matches the reflection indeterminacy of the plug-in estimators in \eqref{eq: h transpose zeta 1} and \eqref{eq: h transpose zeta3} with that of $\hat\ab_{i,j}$. Then, adopting the notation in the model \eqref{eq: rewrite cp model}, similarly to \eqref{eq: h transpose zeta 1} and \eqref{eq: h transpose zeta3}, we have
\[
    \begin{split}
	\frac{\hb^{\T}\bzeta_1^{(i,j,\textit{v}_{\max})}}{\hat\ab_{i,j}^{\T} \tilde\bSigma^{(\textit{v}_{\max},j)}_{\tilde\yb_{i,j},\tilde\xi_{i}}(1)}
	&= \textup{sgn}(\sigma_{f_i,\xi_i})\cdot\hb^{\T}\ddot\ab_{i,j} \\
    & \quad -\frac{1+o_{\rm p}(1)}{\textup{sgn}(\sigma_{f_i,\xi_i})\cdot w_i\sigma_{ f_i,\xi_i}}\hb^{\T}\ddot\ab_{i,j}\ddot\ab_{i,j}^{\T}\bigg\{\tilde\bSigma_{\ddot\eb_{i,j},\ddot\xi_i}(1)+\sum_{\ell\ne i}\ddot{\bar\varphi}_{i,\ell}\tilde\bSigma_{\ddot\eb_{i,j},\ddot\xi_\ell}(0)\bigg\}\,\\
    &\quad+O_{\rm p}(w_i^{-1}\tilde\Phi_{n})+o_{\rm p}(w_i^{-1}n^{-1/2})\,,\\
	\frac{\hb^{\T}\bzeta_3^{(i,j,\textit{v}_{\max})}}{\hat\ab_{i,j}^{\T} \tilde\bSigma^{(\textit{v}_{\max},j)}_{\tilde\yb_{i,j},\tilde\xi_{i}}(1)}&=\frac{1+o_{\rm p}(1)}{\textup{sgn}(\sigma_{f_i,\xi_i})\cdot w_i\sigma_{f_i,\xi_i}}\hb^{\T}\bigg\{\tilde\bSigma_{\ddot\eb_{i,j},\ddot\xi_i}(1)+\sum_{\ell\ne i}\ddot{\bar\varphi}_{i,\ell}\tilde\bSigma_{\ddot\eb_{i,j},\ddot\xi_\ell}(0)\bigg\}+o_{\rm p}(w_i^{-1}n^{-1/2})\,,
\end{split}
\]
where  $\tilde\bSigma_{\ddot{\eb}_{i,j},\ddot{\xi}_i}(1)$, $\ddot{\bar{\varphi}}_{i,\ell}$, and $\tilde\bSigma_{\ddot{\eb}_{i,j},\ddot{\xi}_\ell}(0)$ are defined in the same manner as $\tilde\bSigma_{\eb_{i,j},\xi_i}(1)$, $\bar\varphi_{i,\ell}$, and $\tilde\bSigma_{\eb_{i,j},\xi_\ell}(0)$, respectively, but with $(f_{t,i},\ab_{i,j},\eb_{t,i,j})$ replaced by $(\ddot{f}_{t,i},\ddot{\ab}_{i,j},\ddot{\eb}_{t,i,j})$. By their definitions, we have
\[
\begin{split}
     & \tilde\bSigma_{\ddot{\eb}_{i,j},\ddot{\xi}_i}(1) = \tilde\kappa_{i,j} \tilde\bSigma_{\eb_{i,j},\xi_i}(1)\,,~~ \ddot{\bar{\varphi}}_{i,\ell} = \bigg(\prod_{j=1}^m \tilde\kappa_{i,j}\bigg)\bigg(\prod_{j=1}^m \tilde\kappa_{\ell,j}\bigg)\bar\varphi_{i,\ell}\,,\\ 
   &~~\textup{and}~~\tilde\bSigma_{\ddot{\eb}_{i,j},\ddot{\xi}_\ell}(0)=\bigg(\prod_{j' \neq j}^m \tilde\kappa_{i,j'}\bigg)\bigg(\prod_{j=1}^m \tilde\kappa_{\ell,j}\bigg)\tilde\bSigma_{\eb_{i,j},\xi_\ell}(0)\, .
\end{split}
\]
Combining with the fact that $\textup{sgn}(\sigma_{f_i,\xi_i})\cdot \ddot\ab_{i,j}=\textup{sgn}(\sigma_{f_i,\xi_i})\cdot \tilde\kappa_{i,j}\ab_{i,j}=\kappa_{i,j}\ab_{i,j}$, we can conclude that
\[
\begin{split}
	\hb^\T(\hat\ab_{i,j}-\kappa_{i,j}{\ab}_{i,j}-\hat{\bvartheta}_{i,j})
	&=\frac{\kappa_{i,j}}{w_i{\sigma}_{{f}_i,{\xi}_i}}\hb^{\T}(\Ib_{d_j}-{\ab}_{i,j}{\ab}_{i,j}^{\T})\bigg\{\tilde\bSigma_{{\eb}_{i,j},{\xi}_i}(1)+\sum_{\ell\ne i}{\bar{\varphi}}_{i,\ell}\tilde\bSigma_{{\eb}_{i,j},{\xi}_\ell}(0)\bigg\}\\
	&\quad+O_{\rm p}(w_i^{-1}\tilde\Phi_{n})+o_{\rm p}(w_i^{-1}n^{-1/2})\,,
\end{split}
\]
which completes the proof of Theorem \ref{thm: debias iterative}.
$\hfill\Box$

\subsection{Proof of Theorem  \ref{thm: estimation of iteration variance}	 }\label{sec:thm: estimation of iteration variance}

Without loss of generality and for notational simplicity, we ignore the reflection and permutation indeterminacy, and take $z_i=i$ and $\kappa_{i,j}=1$, where $z_i$ and $\kappa_{i,j}$ are specified in Theorem \ref{thm: iterative}. Otherwise, we refer to the same technique used in the proof of Theorem \ref{thm: iterative} to handle the reflection and permutation indeterminacy.

Firstly, we aim to prove that 
\begin{equation}\label{wi and sigma}
    \frac{\hat w_{i,j}}{w_i|\sigma_{f_i,\xi_i}|}=1+o_{\rm p}(1)\,.
\end{equation}
Recall the definition of $\bar\theta_j^{(\textit{v})}$ in \eqref{theta jv}. By Theorem \ref{thm: iterative} and condition \eqref{strong factor condition}, we can conclude that  $w_r^{-1}w_1\bar\theta_j^{(\textit{v}_{\max})}=o_{\rm p}(1)$. Therefore, the event $w_r^{-1}w_1\bar\theta_j^{(\textit{v}_{\max})}\le \tilde C^{-2}$ holds  with probability approaching one for sufficiently large constant $\tilde C>0$, which is the condition required by Lemma \ref{lem: zeta_1 lower bound} and Lemma \ref{lem: truncated projected covariance}. Then, following the proof of these two lemmas, and by Lemma \ref{lem: Delta}, we can conclude that
\[
\bigg|\frac{1}{w_i\sigma_{f_i,\xi_i}}\bzeta_1^{(i,j,\textit{v}_{\max})}-\ab_{i,j}\bigg|_2=o_{\rm p}(1)~~\textup{and}~~ \frac{1}{w_i}|T_{\delta_{2,j}}\{\tilde\bSigma^{(\textit{v}_{\max},j)}_{\tilde\yb_{i,j},\tilde\xi_{i}}(1)\}-\bzeta_1^{(i,j,\textit{v}_{\max})}|_2=o_{\rm p}(1)\,.
\]
Therefore, by Triangle inequality, it holds that
\[
\begin{split}
    \bigg|\frac{\hat w_{i,j}}{w_i|\sigma_{f_i,\xi_i}|}-1\bigg| 
    & \le \bigg|\frac{\hat\ab_{i,j}^{\T}[T_{\delta_{2,j}}\{\tilde\bSigma^{(\textit{v}_{\max},j)}_{\tilde\yb_{i,j},\tilde\xi_{i}}(1)\}-\bzeta_1^{(i,j,\textit{v}_{\max})}]}{w_i|\sigma_{f_i,\xi_i}|}\bigg|\\
    &~~~+\bigg|\frac{(\hat\ab_{i,j} - \ab_{i,j})^{\T}\bzeta_1^{(i,j,\textit{v}_{\max})}}{w_i|\sigma_{f_i,\xi_i}|} \bigg|+ \bigg|\frac{\ab_{i,j}^{\T}\bzeta_1^{(i,j,\textit{v}_{\max})} - w_i\sigma_{f_i,\xi_i} }{w_i|\sigma_{f_i,\xi_i}|} \bigg| \\
    &=o_{\rm p}(1)\,,
\end{split}
\]
which implies \eqref{wi and sigma}.

Next, we show $\hat{\tau}_{i,j}(\hb)$ is a consistent estimator for $|\sigma_{f_i,\xi_i}|\bar\tau_{i,j}(\hb)$. 
For a deterministic vector $\hb \in \mathbb{R}^{d_j}$, we have
\begin{align*}
    &\hb^{\T} \bigg\{\tilde\bSigma_{\eb_{i,j},\xi_{i}}(1)+\sum_{\ell \neq i}\bar\varphi_{i,\ell}\tilde\bSigma_{\eb_{i,j},\xi_\ell}(0)\bigg\} \\
    &~~~~~~~~~~~~=   \frac{1}{n-1}  \sum_{t=2}^n\hb^{\T} (\eb_{t,i,j}-\bar \eb_{i,j})\bigg\{\xi_{t-1,i}^{\textup{s}} + \sum_{\ell \neq i}\bar\varphi_{i,\ell} \xi_{t,\ell}^{\textup{s}}\bigg\} \\
&~~~~~~ ~~~~~~ \quad - \sum_{\ell \neq i}\bar\varphi_{i,\ell} \frac{1}{n-1}  \sum_{t=2}^n \hb^{\T} \mathbb{E}\{(\eb_{t,i,j}-\bar \eb_{i,j})\xi_{t,\ell}^{\textup{s}}\} \,.
\end{align*}
  Recall the definition of $f_{t,i}^{\textup{sp}}$ specified in \eqref{f ti sp}. Notice that 
\begin{equation}\label{xi ti and f ti}
    \begin{split}
    \xi_{t,i}^{\textup{s}} &= \underbrace{w_i \bigg\{\frac{1}{n}\sum_{s=1}^n\mathbb{E}\{(\xi_{s,i}-\bar \xi_i)^2\} \bigg\}^{-1/2}  \bigg\{\frac{1}{n}\sum_{s=1}^n\mathbb{E}\{(f_{s,i}-\bar f_i)^2\} \bigg\}^{1/2}}_{ = 1 + O(w_i^{-1})}  f^{\textup{s}}_{t,i} \\
    & \quad + \underbrace{\bigg\{\frac{1}{n}\sum_{s=1}^n\mathbb{E}\{(\xi_{s,i}-\bar \xi_i)^2\}\bigg\}^{-1/2}}_{ = O(w_i^{-1})}(u_{t,i} - \bar{u}_{i})\,,
    \end{split}
\end{equation}
where $u_{t,i} = (\ab_{i,m}^{\MP}\otimes \cdots\otimes\ab_{i,1}^{\MP})^\T \text{vec}(\mathcal{E}_t)$ and $\bar{u}_i = n^{-1}\sum_{t=1}^n u_{t,i}$.  Under the conditions of Theorem \ref{thm: estimation of iteration variance}, by Lemma \ref{lemma: E varphi}, we can conclude that
\begin{align*}
 &\frac{1}{n-1} \sum_{t=2}^n\hb^{\T} (\eb_{t,i,j}-\bar \eb_{i,j})\bigg(\xi_{t-1,i}^{\textup{s}} + \sum_{\ell \neq i}\bar\varphi_{i,\ell} \xi_{t,\ell}^{\textup{s}}\bigg) \\
 & ~~~~~~ =   \frac{1}{n-1} \sum_{t=2}^n\hb^{\T}(\eb_{t,i,j}-\bar \eb_{i,j}) \bigg(f_{t-1,i}^{\textup{s}} + \sum_{\ell \neq i}\bar\varphi_{i,\ell} f_{t,\ell}^{\textup{s}}\bigg)  + O_{\rm p}(n^{-1/2} w_r^{-1}) +  O_{\rm p}(\gamma_{\max} w_r^{-1})  \\
  & ~~~~~~ = \frac{1}{n-1} \sum_{t=2}^n\hb^{\T}(\eb_{t,i,j}-\bar \eb_{i,j}) f^{\textup{sp}}_{t-1,i} + O_{\rm p}(n^{-1/2} w_r^{-1}) + O_{\rm p}(\gamma_{\max} w_r^{-1})\\
  & ~~~~~~ = \frac{1}{n-1}  \sum_{t=2}^n \hb^{\T}(\eb_{t,i,j}-\bar \eb_{i,j}) f^{\textup{sp}}_{t-1,i} + o_{\rm p}(n^{-1/2})\,,
\end{align*}
 where the last line holds by condition \eqref{negligible error 2}. On the other hand, by Assumption \ref{error}, Lemma \ref{lemma: E varphi}, and condition \eqref{negligible error 2}, we have
\[
\sum_{\ell \neq i}\bar\varphi_{i,\ell} \frac{1}{n-1}  \sum_{t=2}^n \hb^{\T} \mathbb{E}\{(\eb_{t,i,j}-\bar \eb_{i,j})\xi_{t,\ell}^{\textup{s}}\}=O(n^{-1} w_r^{-3}) + O(\gamma_{\max} w_r^{-1})=o(n^{-1/2})\,.
\] 
Recall $\bbeta_{i,j}(\hb) =  \bbb_{i,j}^{\MP}  \otimes \{\hb^{\T}(\Ib_{d_j}-\ab_{i,j}\ab_{i,j}^{\T})\}^{\T}$ specified in \eqref{def: bbeta ij}. Hence, we have
\begin{equation}\label{simplified expansion for hat a ij}
       	 w_i\hb^{\T}(\hat \ab_{i,j} - \ab_{i,j} - \hat{\bvartheta}_{i,j}) =  \frac{1}{\sigma_{f_i,\xi_i}} \frac{1}{n-1}\sum_{t=2}^n \bbeta_{i,j}(\hb)^{\T} \textup{vec}(\Eb_{t,j}) f_{t-1,i}^{\textup{sp}}  + o_{\rm p}(n^{-1/2})\,.
\end{equation}
 Then, to verify Theorem \ref{thm: estimation of iteration variance}, it remains to show that
\[
\begin{split}
        &\bigg|\frac{1}{n-1}\sum_{t=2}^n(\tilde\xi_{t-1,i}^{(\textit{v}_{\max},j)})^2\{\hat{\bbeta}_{i,j}(\hb)^{\T} \textup{vec}(\Yb_{t,j})\}^2 \\
    &~~~~~~-\frac{1}{n-1}\sum_{t=2}^n\mathbb{E} [(f_{t-1,i}^{\textup{sp}})^2\{\bbeta_{i,j}(\hb)^{\T} \textup{vec}(\Eb_{t,j})\}^2 ]\bigg|=o_{\rm p}(1)\,.
\end{split}
\]
By Triangle inequality, it suffices to show the following \eqref{eqarray iterative 1}--\eqref{eqarray iterative 4}:
\begin{align}
    &\frac{1}{n-1}\sum_{t=2}^n\Big((f_{t-1,i}^{\textup{sp}})^2\{\bbeta_{i,j}(\hb)^{\T} \textup{vec}(\Eb_{t,j})\}^2-\mathbb{E}[(f_{t-1,i}^{\textup{sp}})^2\{\bbeta_{i,j}(\hb)^{\T} \textup{vec}(\Eb_{t,j})\}^2]\Big)=o_{\rm p}(1)\,,\label{eqarray iterative 1}\\
&\frac{1}{n-1}\sum_{t=2}^n(f_{t-1,i}^{\textup{sp}})^2 [\{\bbeta_{i,j}(\hb)^{\T} \textup{vec}(\Eb_{t,j})\}^2-\{\hat\bbeta_{i,j}(\hb)^{\T} \textup{vec}(\Eb_{t,j})\}^2 ]=o_{\rm p}(1)\,,\label{eqarray iterative 2}\\
&\frac{1}{n-1}\sum_{t=2}^n \{(f_{t-1,i}^{\textup{sp}})^2-(\tilde\xi_{t-1,i}^{(\textit{v}_{\max},j)})^2 \} \{\hat\bbeta_{i,j}(\hb)^{\T} \textup{vec}(\Eb_{t,j})\}^2=o_{\rm p}(1)\,,\label{eqarray iterative 3}\\
&\frac{1}{n-1}\sum_{t=2}^n(\tilde\xi_{t-1,i}^{(\textit{v}_{\max},j)})^2 [\{\hat\bbeta_{i,j}(\hb)^{\T} \textup{vec}(\Yb_{t,j})\}^2-\{\hat\bbeta_{i,j}(\hb)^{\T} \textup{vec}(\Eb_{t,j})\}^2 ]=o_{\rm p}(1)\,.\label{eqarray iterative 4}
\end{align}
In the following, we prove \eqref{eqarray iterative 1}--\eqref{eqarray iterative 4} one by one.

\underline{{\it Proof of \eqref{eqarray iterative 1}.}} Similarly to Lemma \ref{lem: gkixi}, under Assumption \ref{tail}, Assumption \ref{mixing}, and Assumption \ref{cross}, \eqref{eqarray iterative 1} is a direct concentration result of some $\alpha$-mixing process with exponential tail, and we omit the details.

\underline{{\it Proof of \eqref{eqarray iterative 2}.}}  
Under Assumptions \ref{tail}, \ref{mixing}, \ref{cross}, and the additional
independence condition in Theorem \ref{thm: estimation of iteration variance}, by 
the similar arguments as in the proof of Lemma \ref{lem: Omega n} for the event $\Xi_{9,n}(\tilde C)$, we have
\begin{equation}\label{eq:weightedversion-xi9}
    \begin{split}
&\max_{\substack{|\tilde\bbeta_{j}|_2=1 = |\bbeta_{j}|_2,\\
                    \tilde\bbeta_{j},\bbeta_{j} \in\mathbb{R}^{d_{j}},\forall j\in[m]}}
\bigg|
\frac{1}{n-1}\sum_{t=2}^n
(f_{t-1,i}^{\textup{sp}})^2
(\otimes^{j=1}_m\tilde\bbeta_{j})^{\T}
[
\textup{vec}(\mathcal{E}_t)\textup{vec}(\mathcal{E}_t)^\T\\
&~~~~~~~~~~~~~~~~~~~~~~~~~~~~ -
\mathbb{E}\{
\textup{vec}(\mathcal{E}_t)\textup{vec}(\mathcal{E}_t)^\T
\}
]
(\otimes^{j=1}_m\bbeta_{j})
\bigg|   =
O_{\rm p}(\tilde L_n)\,.
\end{split}
\end{equation}
where $\otimes_{m}^{j =1} \bbeta_{j}$ is shorthand for $\bbeta_m \otimes \cdots \otimes  \bbeta_{1}$, and
\[
\tilde L_n=\bigg(\frac{\sum_{j=1}^m d_j \log d_j}{n}\bigg)^{1/2}+\frac{(\sum_{j=1}^m d_j)^{1/\tilde c_1}}{n}
\]
with $\tilde c_1^{-1}=1+4c_1^{-1}+c_2^{-1}$.

Recall that $\bbeta_{i,j}(\hb)
    =
    \bbb_{i,j}^{\MP}\otimes
    \{\hb^\T(\Ib_{d_j}-\ab_{i,j}\ab_{i,j}^{\T})\}^{\T}$. Write $(\Bb_j^\T \Bb_j)^{-1} = (\varpi_{p,q})_{r\times r}$. We have $\bbb_{i,j}^{\MP} = \sum_{\ell = 1}^r \varpi_{\ell,i}(\otimes_{m}^{j^\prime \neq j} \ab_{\ell,j^\prime})$, where $\otimes_{m}^{j^\prime \neq j} \bbeta_{j^\prime}$ is shorthand for
$\bbeta_{m} \otimes \cdots \otimes \bbeta_{j+1} \otimes \bbeta_{j-1} \otimes \cdots \otimes \bbeta_{1}$. Therefore, we can write $\bbeta_{i,j}(\hb) = \sum_{\ell = 1}^r \varpi_{\ell,i} \ub_{\ell,i,j}(\hb)$, where $\ub_{\ell,i,j}(\hb) = (\otimes_{m}^{j^\prime \neq j} \ab_{\ell,j^\prime}) \otimes \{\hb^\T(\Ib_{d_j}-\ab_{i,j}\ab_{i,j}^{\T})\}^{\T} $. Similarly, we can write $\hat\bbeta_{i,j}(\hb) = \sum_{\ell = 1}^r \hat{\varpi}_{\ell,i} \hat{\ub}_{\ell,i,j}(\hb)$, where $\hat{\varpi}_{\ell,i}$ and $\hat{\ub}_{\ell,i,j}(\hb)$ are defined in the same manner as $\varpi_{\ell,i}$ and $\ub_{\ell,i,j}(\hb)$, respectively, but with replacing the true values by their associated plug-in estimators.  Notice that
    \[
    \hat\bbeta_{i,j}(\hb) - \bbeta_{i,j}(\hb) =  \sum_{\ell = 1}^r(\hat\varpi_{\ell,i} -\varpi_{\ell,i}) \hat\ub_{\ell,i,j}(\hb) +   \sum_{\ell = 1}^r \varpi_{\ell,i} \{\hat\ub_{\ell,i,j}(\hb) - \ub_{\ell,i,j}(\hb)\}\,,
    \]
and $\hat\ub_{\ell,i,j}(\hb) - \ub_{\ell,i,j}(\hb)$ can be written as the sum of $m$ Kronecker product vectors of the form
$(\otimes^{j' \neq j}_m \bbeta_{j'}) \otimes \bbeta_j$ with $\bbeta_{j'} \in \mathbb{R}^{d_{j'}}$ for ${j'} \in [m]$. Therefore, by the one-by-one replacement argument for the Kronecker products, $\hat\bbeta_{i,j}(\hb) - \bbeta_{i,j}(\hb)$ can be written as  a linear combination of finitely many Kronecker product vectors of the form $(\otimes^{j' \neq j}_m \bbeta_{j'}) \otimes \bbeta_j$ with $\bbeta_{j'} \in \mathbb{R}^{d_{j'}}$ for ${j'} \in [m]$. After normalization,  by Theorem \ref{thm: iterative}, we can take $|\bbeta_{j'}|_2 = 1$ for ${j'} \in [m]$, and the associated coefficients are uniformly bounded by
\[
    O_{\rm p}\bigg(
    \max_{j\in[m]}\Phi_{n,j}
    +
    \frac{\gamma_{\max}}{w_r^2}
    \bigg)\,.
\] 
Given any vectors $\bbeta_{j'} \in \mathbb{R}^{d_{j'}}$ for ${j'} \in [m]$, it holds that
\begin{equation}\label{eq: kronecker product equality}
    \mathcal{E}_t \times_{j' = 1}^m \bbeta_{j'}
    =
    (\bbeta_{m} \otimes \cdots \otimes \bbeta_1)^{\T} \textup{vec}(\mathcal{E}_t)
    =
    \{(\otimes_{m}^{j^\prime \neq j}\bbeta_{j^\prime}) \otimes \bbeta_j \}^{\T} \textup{vec}(\Eb_{t,j})\,.
\end{equation}
Let
\[
    \Sbb_{n,j}
    =
    \frac{1}{n-1}\sum_{t=2}^n
    (f_{t-1,i}^{\textup{sp}})^2
    [
    \textup{vec}(\Eb_{t,j})\textup{vec}(\Eb_{t,j})^\T
    -
    \mathbb{E}\{
    \textup{vec}(\Eb_{t,j})\textup{vec}(\Eb_{t,j})^\T
    \}
    ]\,.
\]
Together with \eqref{eq:weightedversion-xi9}, \eqref{eq: kronecker product equality} and the one-by-one replacement argument for the Kronecker product expansion of $\hat\bbeta_{i,j}(\hb) - \bbeta_{i,j}(\hb)$, we have
\[
\begin{split}
| \bbeta_{i,j}(\hb)^\T\Sbb_{n,j}\bbeta_{i,j}(\hb) -
\hat\bbeta_{i,j}&(\hb)^\T\Sbb_{n,j}\hat\bbeta_{i,j}(\hb)
| 
=
 O_{\rm p}\bigg(
    \max_{j\in[m]}\Phi_{n,j}
    +
    \frac{\gamma_{\max}}{w_r^2}
    \bigg) \cdot O_{\rm p}(\tilde L_n)\,.
\end{split}
\]
On the other hand, by Assumption \ref{cross}, $\|
    \mathbb{E}\{
    \textup{vec}(\Eb_{t,j})\textup{vec}(\Eb_{t,j})^\T
    \}
    \|_2 
    \le  C$ for some constant $C>0$ and any $t \in [n]$. Meanwhile, $ (n-1)^{-1}\sum_{t=2}^n(f_{t-1,i}^{\textup{sp}})^2=O_{\rm p}(1)$. Therefore,
\[
\begin{split}
&\bigg|
\frac{1}{n-1}\sum_{t=2}^n
(f_{t-1,i}^{\textup{sp}})^2  \bbeta_{i,j}(\hb)^\T[\mathbb{E}\{
    \textup{vec}(\Eb_{t,j})\textup{vec}(\Eb_{t,j})^\T
    \}]\bbeta_{i,j}(\hb)\\
    &\qquad\qquad-\frac{1}{n-1}\sum_{t=2}^n
(f_{t-1,i}^{\textup{sp}})^2  \hat\bbeta_{i,j}(\hb)^\T[\mathbb{E}\{
    \textup{vec}(\Eb_{t,j})\textup{vec}(\Eb_{t,j})^\T
    \}]\hat\bbeta_{i,j}(\hb)
\bigg|\\
&~~~~~~=  O_{\rm p}\bigg(
    \max_{j\in[m]}\Phi_{n,j}
    +
    \frac{\gamma_{\max}}{w_r^2}
    \bigg) \,.
\end{split}
\]
Then, by conditions \eqref{negligible error 2} and \eqref{condition for variance estimator}, we have
\begin{equation}\label{eqarray iterative 2 final}
    \begin{split}
    &\frac{1}{n-1}\sum_{t=2}^n(f_{t-1,i}^{\textup{sp}})^2 [\{\bbeta_{i,j}(\hb)^{\T} \textup{vec}(\Eb_{t,j})\}^2-\{\hat\bbeta_{i,j}(\hb)^{\T} \textup{vec}(\Eb_{t,j})\}^2 ]\, \\
&~~~~~~= O_{\rm p}\bigg(
    \max_{j\in[m]}\Phi_{n,j}
    +
    \frac{\gamma_{\max}}{w_r^2}
    \bigg) \cdot O_{\rm p}\big(\tilde L_n+1\big)\\
 &~~~~~~ =
O_{\rm p}\bigg\{
\bigg(
\max_{j\in[m]}\Phi_{n,j}
+
\frac{\gamma_{\max}}{w_r^2}
\bigg)
\frac{(\sum_{j=1}^m d_j)^{1/\tilde c_1}}{n}
\bigg\}+o_{\rm p}(1) \\
&~~~~~~=
o_{\rm p}(1)\,,
\end{split}
\end{equation}
which implies \eqref{eqarray iterative 2}.

 \underline{{\it Proof of \eqref{eqarray iterative 3}.}} Recall that $\tilde\xi_{t-1,i}^{(\textit{v}_{\max},j)}=\tilde f_{t-1,i}^{(\textit{v}_{\max},j)}+\sum_{\ell\ne i}^r\tilde{\varphi}_{i,\ell}^{(\textit{v}_{\max},j)}\tilde f_{t,\ell}^{(\textit{v}_{\max},j)}$. 
 By definition, for any $\textit{v} \ge 1$, we write
\begin{equation*} 
    \tilde f^{(\textit{v},j)}_{t,i}-\xi_{t,i}^{\textup{s}}=\frac{1}{\tilde{\sigma}^{(\textit{v},j)}_{\check f, i}}( \check f^{(\textit{v},j)}_{t,i}- \bar{\check f}^{(\textit{v},j)}_{i}-\xi_{t,i}+\bar\xi_i) +\bigg(\frac{\Lambda_{\xi,i,i}^{1/2}}{\tilde{\sigma}^{(\textit{v},j)}_{\check f, i}}-1\bigg)\xi_{t,i}^{\textup{s}}\,.
\end{equation*}
 Define $\check \ab_{i,j^\prime}=\hat\ab_{i,j^\prime}$ for $j^\prime <j$ and $\check \ab_{i,j^\prime}=\tilde\ab_{i,j^\prime}^{(\textit{v}_{\max}-1)}$ for $j^\prime\ge j$, and define $\check \ab_{i,j^\prime}^{\MP}$ in the same manner as $\hat\ab_{i,j^\prime}^{\MP}$ but with replacing $\hat\ab_{i,j^\prime}$ by $\check\ab_{i,j^\prime}$.  Then,
\[
      \begin{split}
  & \check f^{(\textit{v}_{\max},j)}_{t,i} - \bar{\check f}^{(\textit{v}_{\max},j)}_{i} -\xi_{t,i} + \bar\xi_i \\
   &~~~~~~= \bigg\{\prod_{j^\prime=1}^m(\ab_{i,j^\prime}^\T\check\ab_{i,j^\prime}^{\MP} )-1\bigg\} w_i(f_{t,i} - \bar f_i)+ \sum_{\ell\ne i} \bigg\{\prod_{j^\prime = 1}^m(\ab_{\ell,j^\prime}^\T\check\ab_{i,j^\prime}^{\MP} ) \bigg\} w_\ell (f_{t,\ell} - \bar f_\ell)\\
   &~~~~~~\quad + (\otimes_{m}^{j^\prime=1}\check\ab_{i,j^\prime}^{\MP}-\otimes_{m}^{j^\prime=1}\ab^{\MP}_{i,j^\prime})^\T \textup{vec}(\mathcal{E}_t - \bar{\mathcal{E}})\,.
   \end{split}
\]
On the one hand, under Assumptions \ref{tail}, \ref{mixing} and \ref{cross}, we can conclude that 
\[
\begin{split}
    &\frac{1}{n-1}\sum_{t=2}^n\{w_i(f_{t-1,i} - \bar f_i)\}^2\{\bbeta_{i,j}(\hb)^{\T} \textup{vec}(\Eb_{t,j})\}^2=O_{\rm p}(w_i^2)\,, \quad i\in[r]\,.
\end{split}
\]
On the other hand, similarly to \eqref{eqarray iterative 2 final} but with more tedious calculation, we can conclude the next lemma, whose proof is given in Section \ref{sec:lem: 4th order}.

\begin{lemma}\label{lemma: 4th order cross moment}
    Under the assumptions of Theorem \textup{\ref{thm: debias iterative}}, for any $i\in[r]$ and $j\in[m]$, we have 
    \[
\begin{split}
    &\frac{1}{n-1}\sum_{t=2}^n\{(\otimes_{m}^{j^\prime=1}\check\ab_{i,j^\prime}^{\MP}-\otimes_{m}^{j^\prime=1}\ab^{\MP}_{i,j^\prime})^\T \textup{vec}(\mathcal{E}_{t-1} - \bar{\mathcal{E}})\}^2\{\bbeta_{i,j}(\hb)^{\T} \textup{vec}(\Eb_{t,j})\}^2\,\\
    &~~~~~~=O_{\rm p} \Big\{ (\tilde L_n^2+\tilde L_n+1 ) \max_{j' \in [m]}| \check\ab_{i,j^\prime}^{\MP}- \ab^{\MP}_{i,j^\prime}|_2^2 \Big\}\,.
\end{split}
\]
\end{lemma}
By Lemma \ref{bound for Lambda and phi}, we have $\tilde{\sigma}^{(\textit{v}_{\max},j)}_{\check f, i}\asymp w_i$ with probability approaching one. Moreover, for a sufficiently large $\textit{v}_{\max}$, $\check\ab_{i,j^\prime}$ shares the same convergence rate as that of $\hat\ab_{i,j^\prime}$.  Therefore, under conditions \eqref{strong factor condition} and \eqref{condition for variance estimator}, and by the convergence rate in Theorem \ref{thm: iterative}, we can conclude that
\[
\frac{1}{n-1}\sum_{t=2}^n\bigg\{\frac{1}{\tilde{\sigma}^{(\textit{v}_{\max},j)}_{\check f, i}}( \check f^{(\textit{v}_{\max},j)}_{t-1,i} - \bar{\check f}^{(\textit{v}_{\max},j)}_{i} -\xi_{t-1,i} + \bar\xi_i )\bigg\}^2\{\bbeta_{i,j}(\hb)^{\T} \textup{vec}(\Eb_{t,j})\}^2 = o_{\rm p}(1)\,.
\]
Following the proof of  Lemma \ref{bound for Lambda and phi}, we have $\Lambda_{\xi,i,i}^{1/2}/\tilde{\sigma}^{(\textit{v}_{\max},j)}_{\check f, i}-1=o_{\rm p}(1)$. Then, we can conclude that
\[
\frac{1}{n-1}\sum_{t=2}^n\bigg\{\bigg( \frac{\Lambda_{\xi,i,i}^{1/2}}{\tilde{\sigma}^{(\textit{v}_{\max},j)}_{\check f, i}}-1\bigg)\xi_{t-1,i}^{\textup{s}}\bigg\}^2\{\bbeta_{i,j}(\hb)^{\T} \textup{vec}(\Eb_{t,j})\}^2=o_{\rm p}(1)\,.
\]
Hence, it follows that
\begin{equation*}\label{eqarray iterative 3232}
  \frac{1}{n-1}\sum_{t=2}^n(\tilde f^{(\textit{v}_{\max},j)}_{t-1,i}-\xi_{t-1,i}^{\textup{s}})^2\{\bbeta_{i,j}(\hb)^{\T} \textup{vec}(\Eb_{t,j})\}^2=o_{\rm p}(1)\,.  
\end{equation*}
Similarly to \eqref{eqarray iterative 2}, it holds that 
\[
\frac{1}{n-1}\sum_{t=2}^n(\tilde f^{(\textit{v}_{\max},j)}_{t-1,i}-\xi_{t-1,i}^{\textup{s}})^2[\{\bbeta_{i,j}(\hb)^{\T} \textup{vec}(\Eb_{t,j})\}^2-\{\hat\bbeta_{i,j}(\hb)^{\T} \textup{vec}(\Eb_{t,j})\}^2]=o_{\rm p}(1)\,, 
\]
which implies that 
\begin{equation}\label{eqarray iterative 3233}
  \frac{1}{n-1}\sum_{t=2}^n(\tilde f^{(\textit{v}_{\max},j)}_{t-1,i}-\xi_{t-1,i}^{\textup{s}})^2\{\hat\bbeta_{i,j}(\hb)^{\T} \textup{vec}(\Eb_{t,j})\}^2=o_{\rm p}(1)\,.  
\end{equation}
By similar but slightly more tedious argument, for $\ell\ne i$,  we can also conclude that
\begin{equation}\label{eqarray iterative 32}
    \frac{1}{n-1}\sum_{t=2}^n (\tilde f_{t,\ell}^{(\textit{v}_{\max},j)} - \xi_{t,\ell}^{\textup{s}})^2 \{\hat\bbeta_{i,j}(\hb)^{\T} \textup{vec}(\Eb_{t,j})\}^2=o_{\rm p}(1)\,.
\end{equation}
In fact, under condition \eqref{negligible error 2}, we have $w_r\rightarrow \infty$ as $n \to \infty$. Then, by the expansion in \eqref{xi ti and f ti} and a technique similar to that used to prove \eqref{eqarray iterative 2},  we can conclude that
\[
\begin{split}
     &\frac{1}{n-1}\sum_{t=2}^n (\xi_{t-1,i}^{\textup{s}} -  f_{t-1,i}^{\textup{s}})^2 \{\bbeta_{i,j}(\hb)^{\T} \textup{vec}(\Eb_{t,j})\}^2= o_{\rm p}(1)\,,\\
      &\frac{1}{n-1}\sum_{t=2}^n (\xi_{t-1,i}^{\textup{s}} - f_{t-1,i}^{\textup{s}})^2  [\{\hat\bbeta_{i,j}(\hb)^{\T} \textup{vec}(\Eb_{t,j})\}^2-\{\bbeta_{i,j}(\hb)^{\T} \textup{vec}(\Eb_{t,j})\}^2 ]=o_{\rm p}(1)\,, 
\end{split}
\]
which further implies that
\begin{equation}\label{eqarray iterative 311}
    \frac{1}{n-1}\sum_{t=2}^n (\xi_{t-1,i}^{\textup{s}} - f_{t-1,i}^{\textup{s}})^2 \{\hat\bbeta_{i,j}(\hb)^{\T} \textup{vec}(\Eb_{t,j})\}^2=o_{\rm p}(1)\,.
\end{equation}
By similar argument, for $\ell \neq i$, we can also conclude that
\begin{equation}\label{eqarray iterative 322} 
     \frac{1}{n-1}\sum_{t=2}^n (\xi_{t,\ell}^{\textup{s}} - f_{t,\ell}^{\textup{s}})^2 \{\hat\bbeta_{i,j}(\hb)^{\T} \textup{vec}(\Eb_{t,j})\}^2=o_{\rm p}(1)\,.
\end{equation} 
Therefore, by combining \eqref{eqarray iterative 3233}--\eqref{eqarray iterative 322}, and the fact that 
$\tilde{\varphi}_{i,\ell}^{(\textit{v}_{\max},j)}-\bar\varphi_{i,\ell}=o_{\rm p}(1)$ 
from the proof of Lemma \ref{bound for Lambda and phi}, using the 
Cauchy--Schwarz inequality and Triangle inequality, we can conclude that
\begin{equation}\label{eq:fspmxisp}
    \frac{1}{n-1}\sum_{t=2}^n (f_{t-1,i}^{\textup{sp}} - \tilde\xi_{t-1,i}^{(\textit{v}_{\max},j)})^2 \{\hat\bbeta_{i,j}(\hb)^{\T} \textup{vec}(\Eb_{t,j})\}^2=o_{\rm p}(1)\,.
\end{equation}
By Cauchy--Schwarz inequality, it holds that
\begin{align*}
    &\bigg|\frac{1}{n-1}\sum_{t=2}^n
 \{(f_{t-1,i}^{\textup{sp}})^2-
(\tilde\xi_{t-1,i}^{(\textit{v}_{\max},j)})^2 \}
\{\hat\bbeta_{i,j}(\hb)^{\T}\textup{vec}(\Eb_{t,j})\}^2\bigg|\\
&~~~~~~\le
\bigg[
\frac{1}{n-1}\sum_{t=2}^n
(f_{t-1,i}^{\textup{sp}} - \tilde\xi_{t-1,i}^{(\textit{v}_{\max},j)})^2
\{\hat\bbeta_{i,j}(\hb)^{\T}\textup{vec}(\Eb_{t,j})\}^2
\bigg]^{1/2}\\
&~~~~~~\quad\times
\bigg[
\frac{1}{n-1}\sum_{t=2}^n
(f_{t-1,i}^{\textup{sp}} + \tilde\xi_{t-1,i}^{(\textit{v}_{\max},j)})^2
\{\hat\bbeta_{i,j}(\hb)^{\T}\textup{vec}(\Eb_{t,j})\}^2
\bigg]^{1/2}\,.
\end{align*}
The first term on the right-hand side is $o_{\rm p}(1)$ by \eqref{eq:fspmxisp}. Moreover, since
\[
(f_{t-1,i}^{\textup{sp}} + \tilde\xi_{t-1,i}^{(\textit{v}_{\max},j)})^2
\lesssim
(f_{t-1,i}^{\textup{sp}})^2+
(f_{t-1,i}^{\textup{sp}} - \tilde\xi_{t-1,i}^{(\textit{v}_{\max},j)})^2\,,
\]
and
\[
\frac{1}{n-1}\sum_{t=2}^n
(f_{t-1,i}^{\textup{sp}})^2
\{\hat\bbeta_{i,j}(\hb)^{\T}\textup{vec}(\Eb_{t,j})\}^2=O_{\rm p}(1)\,,
\]
the second term on the right-hand side is $O_{\rm p}(1)$. Hence, we can conclude \eqref{eqarray iterative 3}.

 \underline{{\it Proof of \eqref{eqarray iterative 4}.}}  Recall that 
\[
\textup{vec}(\Yb_{t,j})-\textup{vec}(\Eb_{t,j})=\sum_{\ell=1}^r(w_\ell f_{t,\ell})(\bbb_{\ell,j}\otimes \ab_{\ell,j})\,.
\]
Then, by Triangle inequality,
\[
\begin{split}
    \textup{left side of \eqref{eqarray iterative 4}} &\lesssim \sum_{\ell=1}^r\frac{\{w_\ell\hat\bbeta_{i,j}(\hb)^{\T}(\bbb_{\ell,j}\otimes \ab_{\ell,j})\}^2}{n-1}\sum_{t=2}^n(\tilde\xi_{t-1,i}^{(\textit{v}_{\max},j)})^2 f_{t,\ell}^2\,.
\end{split}
\]
On the one hand, similarly to \eqref{eqarray iterative 3}, we have
\[
\frac{1}{n-1}\sum_{t=2}^n(\tilde\xi_{t-1,i}^{(\textit{v}_{\max},j)})^2 f_{t,\ell}^2=\frac{1}{n-1}\sum_{t=2}^n(f_{t-1,i}^{\textup{sp}})^2 f_{t,\ell}^2+o_{\rm p}(1)=O_{\rm p}(1)\,,\quad \forall \ell\in[r]\,.
\]
On the other hand,  we have
\[
\begin{split}
    |\hat\bbeta_{i,j}(\hb)^{\T}(\bbb_{\ell,j}\otimes \ab_{\ell,j})|&=| \bbb_{\ell,j}^\T (\tilde{\bbb}_{i,j}^{(\textit{v}_{\max})})^{\MP}| \cdot | \hb^\T (\Ib_{d_j}-\hat\ab_{i,j}\hat\ab_{i,j}^\T)  \ab_{\ell,j}|\,.
\end{split}
\]
When $\ell\ne i$, we have 
\[
|\bbb_{\ell,j}^\T (\tilde{\bbb}_{i,j}^{(\textit{v}_{\max})})^{\MP}|\lesssim |(\tilde{\bbb}_{i,j}^{(\textit{v}_{\max})})^{\MP}-\bbb_{i,j}^{\MP}|_2=O_{\rm p}\bigg(\max_{j\in[m]}\Phi_{n,j}+\frac{\gamma_{\max}}{w_r^2}\bigg)\,.
\]
When $\ell= i$, we have
\[
| \hb^\T (\Ib_{d_j}-\hat\ab_{i,j}\hat\ab_{i,j}^\T)  \ab_{\ell,j}|\lesssim |\hat \ab_{i,j}-\ab_{i,j}|_2+|\hat \ab_{i,j}-\ab_{i,j}|_2^2=O_{\rm p}\bigg(\max_{j\in[m]}\Phi_{n,j}+\frac{\gamma_{\max}}{w_r^2}\bigg)\,.
\]
Therefore, as long as $w_1(\max_{j\in[m]}\Phi_{n,j}+w_r^{-2}\gamma_{\max})=o(1)$, we can conclude \eqref{eqarray iterative 4}, and Theorem \ref{thm: estimation of iteration variance} holds.
$\hfill\Box$

\subsection{Proof of Theorem \ref{lemma: factor iterative}}

Without loss of generality and for notational simplicity, we ignore the reflection and permutation indeterminacy, and take $z_i=i$ and $\kappa_{i,j}=1$, where $z_i$ and $\kappa_{i,j}$ are specified in Theorem \ref{thm: iterative}. Otherwise, we refer to the same technique used in the proof of Theorem \ref{thm: iterative} to handle the reflection and permutation indeterminacy.

By model \eqref{model cp} and the definition of $\hat f_{t,i}$, we have
\[
        \begin{split}
        \hat f_{t,i}= \bigg\{\prod^{m}_{j = 1}(\ab_{i,j}^\T\hat\ab_{i,j}^{\MP})\bigg\} w_if_{t,i}+\sum_{\ell\ne i}\bigg\{\prod^{m}_{j = 1}(\ab_{\ell,j}^\T\hat\ab_{i,j}^{\MP})\bigg\}w_{\ell}f_{t,\ell}+(\hat\ab_{i,m}^{\MP} \otimes \cdots \otimes \hat\ab_{i,1}^{\MP})^{\T}\textup{vec}(\mathcal{E}_t)\,.
    \end{split}
\]
  Recall that $\Phi_n = \max_{j\in[m]}\Phi_{n,j}+w_r^{-2}\gamma_{\max}$. By direct calculation, $\hat\ab_{i,j}^{\MP}$ converges to $\ab_{i,j}^{\MP}$ at the same rate as that in Theorem \ref{thm: iterative} for any $i\in[r]$ and $j\in[m]$. Thus, for $\ell\ne i$, it holds that
\[
\prod^{m}_{j = 1}(\ab_{i,j}^\T\hat\ab_{i,j}^{\MP})-1=O_{\rm p}(\Phi_n)~~\textup{and}~~(\ab_{\ell,j}^\T\hat\ab_{i,j}^{\MP})=O_{\rm p}(\Phi_n)\,.
\]
Therefore, it remains to bound $(\hat\ab_{i,m}^{\MP} \otimes \cdots \otimes \hat\ab_{i,1}^{\MP})^{\T}\textup{vec}(\mathcal{E}_t)=\mathcal{E}_t\times_{j=1}^m\hat\ab_{i,j}^{\MP}$. In fact, by the similar arguments used in the proof of Lemma \ref{lem: Omega n} for the event $\Xi_{6,n}(\tilde{C})$,  we can conclude that
 \[
\max_{|\bbeta_j|_2=1,\bbeta_j \in \mathbb{R}^{d_j},\forall j\in[m]} |\mathcal{E}_t\times_{j=1}^m\bbeta_j^{\T} |=O_{\rm p} \bigg\{  \bigg(\sum_{j=1}^m d_j \bigg)^{1/c_1} \bigg\} 
\]
for $t\in[n]$, where the constant $c_1$ is specified in Assumption \ref{cross}. 
Hence, for any $t\in[n]$ and $i\in[r]$, it holds that
\[
 |\mathcal{E}_t\times_{j=1}^m\hat\ab_{i,j}^{\MP}-\mathcal{E}_t\times_{j=1}^m\ab_{i,j}^{\MP}|=O_{\rm p} \bigg\{  \bigg(\sum_{j=1}^m d_j \bigg)^{1/c_1} \max_{j\in[m]}|\hat\ab_{i,j}^{\MP}-\ab_{i,j}^{\MP}|_2 \bigg\}\,,
\]
which implies that
\[
|\mathcal{E}_t\times_{j=1}^m\hat\ab_{i,j}^{\MP}|=O_{\rm p}\bigg\{1+\Phi_n \bigg(\sum_{j=1}^m d_j\bigg)^{1/c_1} \bigg\}\,.
\]
Therefore, with $m\ge 2$, we conclude that
\[
\begin{split}
  \frac{1}{w_i}|\hat f_{t,i}-w_if_{t,i}| &= O_{\rm p}\bigg[\Phi_n+\frac{w_1}{w_r}\Phi_n^m + \frac{1}{w_i}\bigg\{1+\Phi_n \bigg(\sum_{j=1}^m d_j \bigg)^{1/c_1} \bigg\}\bigg]\,.
\end{split}
\]
Under conditions \eqref{strong factor condition} and \eqref{negligible error 2}, we have $w_r^{-1}w_1\Phi_{n} = o(1)$. Then, $w_r^{-1}w_1\Phi_{n}^m = o(\Phi_{n})$ for $m\ge 2$. Moreover, we have
\[
\frac{\gamma_{\max}}{w_r^2} \bigg(\sum_{j=1}^m d_j \bigg)^{1/c_1} \lesssim \frac{1}{w_r\sqrt{n}} \bigg(\sum_{j=1}^m d_j \bigg)^{1/c_1} \le \frac{1}{w_r} \bigg\{\frac{({\sum}_{j=1}^m d_j)^{1/\tilde c}}{n} \bigg\}^{1/2} \,,
\]
where we use the fact that $\tilde c^{-1} \ge 1+2c_1^{-1}$. If $({\sum}_{j=1}^m d_j)^{1/\tilde c}\ll n$, we have  $w_r^{-2}\gamma_{\max}({\sum}_{j=1}^m d_j)^{1/c_1}\ll 1$. If $({\sum}_{j=1}^m d_j)^{1/\tilde c}\gtrsim n$, then condition \eqref{strong factor condition} indicates that
\[
1\gg \frac{({\sum}_{j=1}^m d_j)^{1/\tilde c}}{nw_r}\gtrsim \frac{1}{w_r} \bigg\{\frac{({\sum}_{j=1}^m d_j)^{1/\tilde c}}{n} \bigg\}^{1/2} \,.
\]
Therefore, we always have $w_r^{-2}\gamma_{\max}({\sum}_{j=1}^m d_j)^{1/c_1}\ll 1$. Moreover, condition \eqref{negligible error 2} implies that 
\[
\frac{({\sum}_{j=1}^m d_j)^{1/\tilde c} (\max_{j\in[m]} s_j\log d_j)^{1/2}}{nw_r^2} \ll 1\,.
\]
Because $\max_{j\in[m]} s_j\log d_j\lesssim (\sum_{j=1}^m d_j)^2$ and $\tilde c^{-1}\ge 1+2c_1^{-1}$, we can conclude that
\[
\bigg\{\bigg(\sum_{j=1}^m d_j \bigg)^{1/c_1} \max_{j \in [m]}\Phi_{n,j}\bigg\}^2 =\frac{(\sum_{j=1}^m d_j)^{2/c_1}\max_{j\in[m]} s_j\log d_j}{nw_r^2} \ll 1 \,.
\]
Therefore, it holds that
\[
\begin{split}
    \frac{1}{w_i}|\hat f_{t,i}-w_if_{t,i}|= O_{\rm p}\bigg(\frac{1}{w_i}+ \Phi_n \bigg)\,.
\end{split}
\]
Furthermore, for the common component tensor, it holds that
\[
\begin{split}
     \frac{1}{D_n}|\textup{vec}(\hat{\mathcal{C}}_t)-\textup{vec}({\mathcal{C}}_t)|^2_2 &\lesssim  \frac{1}{D_n}\sum^{r}_{i=1} |\hat f_{t,i}(\hat\ab_{i,m} \otimes \cdots \otimes \hat\ab_{i,1})-w_if_{t,i}(\ab_{i,m} \otimes \cdots \otimes \ab_{i,1}) |^2_2\, \\
      &= O_{\rm p}\bigg( \frac{1}{D_n} +\frac{w_1^2}{D_n} \Phi_{n}^2\bigg)\,.
\end{split}
\]
This completes the proof of Theorem \ref{lemma: factor iterative}. 
 $\hfill\Box$

\subsection{Proof of Theorem \ref{thm: factor number consistency}}
 
We begin with the consistency of $\max_{j\in[m]}\tilde r^{(\textup{er})}_j(\delta_1)$. Note that the condition $\bar\sigma_{\xi}^2\Pi_n \ll \ubar\sigma_{\xi}^2$ implies $\Pi_n \ll 1$. 
    Recall that \eqref{tilde M and tilde C} provides an upper bound for the difference between $\tilde \Mb_j$ and  $\sum_{k=1}^K \tilde\bSigma_{\Cb_j,\xi}(k)^\T \tilde\bSigma_{\Cb_j,\xi}(k)$ for $j\in[m]$. By Weyl's theorem and \eqref{tilde M and tilde C}, we have
    \[
    \max_{r< i\le \lfloor 0.5d_{\min} \rfloor}\sigma_i(\tilde\Mb_j)= \ubar{\sigma}_{\xi}^2 O_{\rm p}(\Pi_n)\,,\quad j\in[m]\,.
    \]
    By \eqref{tilde C and Sigma Y} and  $\ubar{\sigma}_{\xi}^2\le \sigma_r(\Mb_j)\le\cdots\le \sigma_1(\Mb_j)\le \bar{\sigma}_{\xi}^2$,  we have
    \[
    \ubar{\sigma}_{\xi}^2  \lesssim \sigma_i(\tilde\Mb_j)\lesssim \bar{\sigma}_{\xi}^2  \,,\quad i\in[r]\,,\, j\in[m]\,
    \]
    with probability approaching one. 
    Then, if $\ubar\sigma_{\xi}^{2} \Pi_n \ll c_n \ll \bar\sigma_{\xi}^{-2}\ubar\sigma_{\xi}^4$ for the $c_n$ in \eqref{hat rj}, with probability approaching one, we can conclude that for each $j\in[m]$,
    \[
    \left\{\begin{aligned}
        &\min_{r< i\le \lfloor 0.5d_{\min} \rfloor}\frac{\sigma_{i+1}(\tilde\Mb_j)+c_n}{\sigma_{i}(\tilde\Mb_j)+c_n}\ge \frac{c_n}{\sigma_{r+1}(\tilde\Mb_j)+c_n}\gtrsim 1\,,\\
        &\min_{1\le i< r}\frac{\sigma_{i+1}(\tilde\Mb_j)+c_n}{\sigma_{i}(\tilde\Mb_j)+c_n}\ge \frac{\sigma_{r}(\tilde\Mb_j)}{\sigma_{1}(\tilde\Mb_j)+c_n}\gtrsim \bar{\sigma}_{\xi}^{-2}\ubar{\sigma}_{\xi}^2  \,,\\
        &\frac{\sigma_{r+1}(\tilde\Mb_j)+c_n}{\sigma_{r}(\tilde\Mb_j)+c_n}\lesssim \frac{c_n}{\sigma_r(\tilde \Mb_j)}\ll \bar{\sigma}_{\xi}^{-2}\ubar{\sigma}_{\xi}^2  \,,
    \end{aligned}\right.
    \]
which implies $\max_{j\in[m]}\tilde r^{(\textup{er})}_j(\delta_1)=r$ holds with probability approaching one.

Next, for $\tilde r^{(\textup{log})}_j(\delta_1)$, the condition 
$\log(1+\ubar{\sigma}_{\xi}^2\Pi_n)\ll \{\log(1+\bar\sigma_{\xi}^2)\}^{-1}\{\log(1+\ubar\sigma_{\xi}^2)\}^2$ also implies $\Pi_n\ll1$. Indeed, if $\Pi_n\gtrsim1$, because $\ubar\sigma_{\xi}\gtrsim1$, then
$\log(1+\ubar{\sigma}_{\xi}^2\Pi_n)\gtrsim \log(1+\ubar\sigma_{\xi}^2)$, which contradicts the stated condition because
\[
    \frac{\{\log(1+\ubar\sigma_{\xi}^2)\}^2}{\log(1+\bar\sigma_{\xi}^2)}
    \le
    \log(1+\ubar\sigma_{\xi}^2)\,.
\]
 Hence, $\Pi_n\ll1$. Then, repeating the argument used for the ER criterion and replacing eigenvalues by their logarithms yields
\[
    \max_{r< i\le \lfloor 0.5d_{\min} \rfloor}
    \log \{1+\sigma_i(\tilde\Mb_j)\}
    \le
    \log \{1+\sigma_{r+1}(\tilde\Mb_j)\}
    =
    O_{\rm p}\{\log (1+\ubar{\sigma}_{\xi}^2 \Pi_n)\}\,,
    \quad j\in[m]\,.
\]
Meanwhile, because $\ubar\sigma_{\xi}\gtrsim1$,  with probability approaching one, we have
\[
    \log(1+\ubar{\sigma}_{\xi}^2)
    \lesssim
    \log\{1+\sigma_i(\tilde\Mb_j)\}
    \lesssim
    \log(1+\bar{\sigma}_{\xi}^2)\,,
    \quad i\in[r],\,  j\in[m]\,.
\]
Therefore, if
$\log(1+\ubar\sigma_{\xi}^2\Pi_n)\ll c_n\ll \{\log(1+\bar\sigma_{\xi}^2)\}^{-1}\{\log(1+\ubar\sigma_{\xi}^2)\}^2$
for the $c_n$ in \eqref{hat rj-log}, with probability approaching one, we can conclude that for each $j\in[m]$,
\[
    \left\{\begin{aligned}
        &\min_{r< i\le \lfloor 0.5d_{\min} \rfloor}
        \frac{\log\{1+\sigma_{i+1}(\tilde\Mb_j)\}+c_n}
        {\log\{1+\sigma_{i}(\tilde\Mb_j)\}+c_n}
        \ge
        \frac{c_n}{\log\{1+\sigma_{r+1}(\tilde\Mb_j)\}+c_n}
        \gtrsim 1\,,\\
        &\min_{1\le i< r}
        \frac{\log\{1+\sigma_{i+1}(\tilde\Mb_j)\}+c_n}
        {\log\{1+\sigma_{i}(\tilde\Mb_j)\}+c_n}
        \ge
        \frac{\log\{1+\sigma_{r}(\tilde\Mb_j)\}}
        {\log\{1+\sigma_{1}(\tilde\Mb_j)\}+c_n}
        \gtrsim
        \frac{\log(1+\ubar{\sigma}_{\xi}^2)}
        {\log(1+\bar{\sigma}_{\xi}^2)}\,,\\
        &\frac{\log\{1+\sigma_{r+1}(\tilde\Mb_j)\}+c_n}
        {\log\{1+\sigma_{r}(\tilde\Mb_j)\}+c_n}
        \lesssim
        \frac{c_n}{\log(1+\ubar{\sigma}_{\xi}^2)}
        \ll
        \frac{\log(1+\ubar{\sigma}_{\xi}^2)}
        {\log(1+\bar{\sigma}_{\xi}^2)}\,.
    \end{aligned}\right.
\]
Therefore, with probability approaching one, the ratio corresponding to $i=r$ is smaller than all ratios corresponding to $i<r$ and $i>r$. Hence,  $\max_{j\in[m]}\tilde r^{(\textup{log})}_j(\delta_1)=r$ holds with probability approaching one. Theorem \ref{thm: factor number consistency} is verified.
 $\hfill\Box$

\section{Proofs of Auxiliary Lemmas}\label{sec:supp-error bounds}

\subsection{Proof of Lemma \ref{pro:rank-B}}\label{sec:pro:rank-B}
 
  For $j \in [m]$, to show that $\textup{rank}(\Bb_j) = r$, it suffices to show that $\Bb_j^\T \Bb_j$ is positive definite. Notice that
\begin{equation}\label{eq:BB rank}
    \Bb_j^\T \Bb_j = (\Ab_m^{\T}\Ab_m) \bullet \cdots \bullet (\Ab_{j+1}^{\T}\Ab_{j+1}) \bullet (\Ab_{j-1}^{\T}\Ab_{j-1}) \bullet \cdots \bullet  (\Ab_{1}^{\T}\Ab_{1})\,,
\end{equation}
where $\bullet$ stands for Hadamard product. Since $\textup{rank}(\Ab_j)=r$ for $j \in [m]$, each $\Ab_j$ has full column rank and hence $\Ab_j^{\T}\Ab_j$ is positive definite. By the Schur product theorem, the Hadamard product of positive definite matrices is positive definite, hence $\Bb_j^\T \Bb_j$ is positive definite for $j \in [m]$. Thus, $\textup{rank}(\Bb_j)=r$ follows.  Moreover, by \eqref{eq:BB rank} and Proposition 6.3.4 of \citeS{rao1998matrix-app}, for any $1 \le j' \neq j \le m$, it holds that
\[
 \sigma_r(\Ab_{j'}^\T \Ab_{j'}) \le  \sigma_i(\Bb_j^\T \Bb_j) \le \sigma_1(\Ab_{j'}^\T \Ab_{j'})\,, \quad i \in [r]\,.
\]
The second assertion follows directly from Assumption \ref{sparsity}. We complete the proof of Lemma \ref{pro:rank-B}. 
 $\hfill\Box$

 \subsection{Proof of Lemma \ref{lem: gkixi}}\label{sec:lem: gkixi}      
            By the definitions of $\tilde g_{k,i,\xi}$ and $g_{k,i,\xi}$, we have
			\begin{equation}\label{expansion sigma2-gik}
				\begin{split}
					\tilde g_{k,i,\xi}-g_{k,i,\xi} &= \frac{w_i}{n-k}\sum_{t=k+1}^n \{f_{t,i}\xi_{t-k} - \mathbb{E}(f_{t,i}\xi_{t-k})\}  -\frac{w_i}{n-k}\sum_{t=k+1}^n \{ f_{t,i} \bar\xi - \mathbb{E}(f_{t,i})\mathbb{E}(\bar\xi)\} \\
                    &\quad - \frac{w_i}{n-k}\sum_{t=k+1}^n\{ \bar f_{i} \xi_{t-k} - \mathbb{E}(\bar f_{i})\mathbb{E}(\xi_{t-k})\}  +  w_i \{\bar f_{i} \bar\xi - \mathbb{E}(\bar f_{i} )\mathbb{E}(\bar\xi)\}\,.
				\end{split}
			\end{equation}
			Recall $|\cdot|_{\MP} = \max(\,\cdot\,,0)$. Similarly to the proof of Lemma 1 in \citeS{chang2023modelling-app}, under Assumptions \ref{tail} and \ref{mixing}, $\{ f_{t,i}\xi_{t-k} - \mathbb{E}(f_{t,i}\xi_{t-k}) \}_{t = k+1}^n$ is an $\alpha$-mixing process with zero mean and mixing coefficients $\{\alpha(|\ell-k|_+)\}_{\ell\ge 1}$. Moreover, by Assumption \ref{tail}, it holds that
			\[
			\mathbb{P}(| f_{t,i}\xi_{t-k}|>x)\le \mathbb{P}(|f_{t,i}|>\sqrt{x})+\mathbb{P}(| \xi_{t-k}|>\sqrt{x})\lesssim \exp(-C_2x^{c_1/2})
			\]
			for any $x>0$. Lemma L5 in \citeS{chang2021central-app} states that there exists some constant $C>0$  such that
			\begin{equation*}\label{tail error-gik}
				\mathbb{P}\bigg(\bigg|\frac{1}{n-k}\sum_{t=k+1}^n \{f_{t,i}\xi_{t-k} - \mathbb{E}(f_{t,i}\xi_{t-k})\}\bigg|>x\bigg)\lesssim \exp(-Cnx^2)+\exp(-Cn^{\tilde c}x^{\tilde c})
			\end{equation*}
			for any $x > 0$, where  $\tilde c=(1+2c_1^{-1}+c_2^{-1})^{-1}$. Applying the same argument to the remaining terms in \eqref{expansion sigma2-gik}, we further obtain 
		\begin{equation*}\label{tail sigma2-gik}
				\mathbb{P}(| w_{i}^{-1}( \tilde g_{k,i,\xi}-g_{k,i,\xi} )|\ge x)\lesssim \exp(-Cnx^2) +\exp(-Cn^{\tilde c}x^{\tilde c}) +\exp(-Cn^{\check c}x^{\check c})
			\end{equation*}
			for any $x \in (0,1)$, where $\check c = (2+|c_1^{-1} -1|_{\MP} + c_2^{-1})^{-1}$. Since $K$ is fixed, as a result, for any $i \in [r]$,
			\begin{equation*} 
				\max_{k \in [K]}|\tilde g_{k,i,\xi}-g_{k,i,\xi}|= O_{\rm p}(w_i n^{-1/2})\,.
			\end{equation*}
            This completes the proof of Lemma \ref{lem: gkixi}. 
$\hfill\Box$

\subsection{Proof of Lemma \ref{lem: l2 truncated covariance}}\label{sec:lem: l2 truncated covariance} 
			
 Following the definitions in \eqref{sigma C} and \eqref{tilde Sigma Ej}, to simplify the notation, we denote, respectively, the $(p,q)$-th entries of $\tilde\bSigma_{\Yb_j,\xi}(k)$, $\tilde\bSigma_{\Cb_j,\xi}(k)$ and $\tilde\bSigma_{\Eb_j,\xi}(k)$ by $\tilde\sigma_{p,q}^{(j,k)}$, $\tilde \sigma_{1,p,q}^{(j,k)}$ and $\tilde \sigma_{2,p,q}^{(j,k)}$ for $p \in [d_j]$ and $q\in [d_{\mminus j}]$. Let $e_{t,j,p,q}$ be the $(p,q)$-th entry of $\Eb_{t,j}$ and $\bar e_{j,p,q}=n^{-1}\sum_{t=1}^n e_{t,j,p,q}$. By definition,
			\begin{equation}\label{expansion sigma2}
				\begin{split}
					\tilde\sigma_{2,p,q}^{(j,k)}&=\frac{1}{n-k}\sum_{t=k+1}^ne_{t,j,p,q}\xi_{t-k}-\frac{\bar\xi}{n-k}\sum_{t=k+1}^ne_{t,j,p,q} -\frac{\bar e_{j,p,q}}{n-k}\sum_{t=k+1}^n\xi_{t-k}+\bar e_{j,p,q}\bar\xi\,.
				\end{split}
			\end{equation}
Similarly to the proof of Lemma \ref{lem: gkixi}, under Assumptions \ref{error} and \ref{mixing}, $\{e_{t,j,p,q}\xi_{t-k}\}_{t = k+1}^n$ is an $\alpha$-mixing process with zero mean and mixing coefficients $\{\alpha(|\ell-k|_+)\}_{\ell\ge 1}$. Moreover, by Assumption \ref{tail}, it holds that
			\[
			\mathbb{P}(|e_{t,j,p,q}\xi_{t-k}|>x)\le \mathbb{P}(|e_{t,j,p,q}|>\sqrt{x})+\mathbb{P}(|\xi_{t-k}|>\sqrt{x})\lesssim \exp(-C_2x^{c_1/2})
			\]
			for any $x>0$. Lemma L5 in \citeS{chang2021central-app} states that there exists some constant $C>0$  such that
			\begin{equation}\label{tail error}
				\mathbb{P}\bigg(\bigg|\frac{1}{n-k}\sum_{t=k+1}^n e_{t,j,p,q}\xi_{t-k}\bigg|>x\bigg)\lesssim \exp(-Cnx^2)+\exp(-Cn^{\tilde c}x^{\tilde c})
			\end{equation}
			for any $x>0$, where  $\tilde c=(1+2c_1^{-1}+c_2^{-1})^{-1}$. Handling the remaining terms in \eqref{expansion sigma2} similarly, one can further conclude that
		\begin{equation}\label{tail sigma2}
				\mathbb{P}(| \tilde\sigma_{2,p,q}^{(j,k)}|\ge x)\lesssim \exp(-Cnx^2) +\exp(-Cn^{\tilde c}x^{\tilde c}) +\exp(-Cn^{\check c}x^{\check c})
			\end{equation}
			 for any $x \in (0,1)$, where $\check c = (2+|c_1^{-1} -1|_{\MP} + c_2^{-1})^{-1}$. As a result,
			\begin{equation}\label{sigma2}
				\max_{p\in[d_j],q\in[d_{\mminus j}]}|\tilde\sigma_{2,p,q}^{(j,k)}|=\max_{p\in[d_j],q\in[d_{\mminus j}]}|\tilde\sigma_{p,q}^{(j,k)}-\tilde\sigma_{1,p,q}^{(j,k)}| =  O_{\rm p}\{ (n^{-1}\log D_n)^{1/2}\}\,,
			\end{equation}
			provided that $\log D_n=o(n^c)$ for some $c\in(0,1)$ depending only on $c_1$ and $c_2$ specified in Assumptions {\rm\ref{tail}} and {\rm\ref{mixing}}.

			Now, by Triangle inequality,
			\begin{equation}\label{decomposition truncation}
				\begin{split}
					&\|T_{\delta_1}\{\tilde \bSigma_{\Yb_j,\xi}(k)\}-\tilde\bSigma_{\Cb_j,\xi}(k)\|_2\\
					&~~~~~~ \le   \underbrace{\|T_{\delta_1}\{\tilde \bSigma_{\Yb_j,\xi}(k)\}-T_{\delta_1}\{\tilde\bSigma_{\Cb_j,\xi}(k)\}\|_2}_{\mathcal{I}_1}+ \underbrace{\|T_{\delta_1}\{\tilde\bSigma_{\Cb_j,\xi}(k)\}-\tilde\bSigma_{\Cb_j,\xi}(k)\|_2}_{\mathcal{I}_2}\,.
				\end{split}
			\end{equation}
			We start with  $\mathcal{I}_2$.  Note that
			\begin{equation}\label{I2}
				\mathcal{I}_2^2\le \bigg[\max_{p\in[d_j]}\sum_{q=1}^{d_{\mminus j}}|\tilde\sigma_{1,p,q}^{(j,k)}|I\{|\tilde\sigma_{1,p,q}^{(j,k)}|<\delta_1\}\bigg]\bigg[\max_{q\in[d_{\mminus j}]}\sum_{p=1}^{d_j}|\tilde\sigma_{1,p,q}^{(j,k)}|I\{|\tilde\sigma_{1,p,q}^{(j,k)}|<\delta_1\}\bigg]\,.
			\end{equation}
			By definition,	$
			\tilde\sigma_{1,p,q}^{(j,k)}= \sum_{i=1}^r \tilde g_{k,i,\xi}a_{i,j,p}b_{i,j,q}$, where $a_{i,j,p}$ and $b_{i,j,q}$ are, respectively, the $p$-th entry of $\ab_{i,j}$ and the $q$-th entry of $\bbb_{i,j}$. 
            It then follows that 
			\begin{equation}\label{sigma1}
				\begin{split}
					&\max_{q\in[d_{\mminus j}]}\sum_{p=1}^{d_j}|\tilde\sigma_{1,p,q}^{(j,k)}| I\{ | \tilde\sigma_{1,p,q}^{(j,k)}|<\delta_1\} \\
                    &~~~~~~\le \delta_1 \max_{q\in[d_{\mminus j}]}\sum_{p=1}^{d_j} I\{|\tilde\sigma_{1,p,q}^{(j,k)}| \neq 0\}\\
					&~~~~~~\lesssim \delta_1   \sum_{i = 1}^r \sum_{p=1}^{d_j}  I\{|a_{i,j,p}| \neq 0\} \lesssim  \delta_1 s_{j} \,.
				\end{split}
			\end{equation}
			Parallelly, one can also show that
			\[
			\max_{p\in[d_j]}\sum_{q=1}^{d_{\mminus j}}|\tilde\sigma_{1,p,q}^{(j,k)}|I\{|\tilde\sigma_{1,p,q}^{(j,k)}|<\delta_1\} \lesssim   \delta_1  \prod_{j^\prime\ne j} s_{j^\prime} \,.
			\]
			As a result,
			\[
			\mathcal{I}_2  \lesssim \delta_1 S_n^{1/2} \,.
			\]
			
			We next consider $\mathcal{I}_1$. To simplify the notation, we suppress the dependence on $j,k$ and define three matrices $\hat\bSigma_{1,1}, \hat \bSigma_{1,2}$, $\hat\bSigma_{1,3}$, whose $(p,q)$-th entries are defined respectively as
			\[
			\begin{split}
				 [\hat\bSigma_{1,1}]_{p,q}& =(\tilde\sigma_{p,q}^{(j,k)}-\tilde\sigma_{1,p,q}^{(j,k)})I\{|\tilde\sigma_{p,q}^{(j,k)}|\ge\delta_1,|\tilde\sigma_{1,p,q}^{(j,k)}|\ge\delta_1\}\,,\\
				[\hat\bSigma_{1,2}]_{p,q} &=\tilde\sigma_{p,q}^{(j,k)}I\{|\tilde\sigma_{p,q}^{(j,k)}|\ge\delta_1,|\tilde\sigma_{1,p,q}^{(j,k)}|<\delta_1\}\,,\\
    [\hat\bSigma_{1,3}]_{p,q}&=\tilde\sigma_{1,p,q}^{(j,k)}I\{|\tilde\sigma_{p,q}^{(j,k)}|<\delta_1,|\tilde\sigma_{1,p,q}^{(j,k)}|\ge\delta_1\}\,.
			\end{split}
			\] 
			By the definition of $\mathcal{I}_1$ and Triangle inequality, we have
			\begin{equation}\label{I1}
				\mathcal{I}_1\le \|\hat\bSigma_{1,1}\|_2+\|\hat\bSigma_{1,2}\|_2+\|\hat\bSigma_{1,3}\|_2\,.
			\end{equation}
			
			Firstly, for $\|\hat\bSigma_{1,1}\|_2$, notice that 
			\[
			\begin{split}
				\|\hat\bSigma_{1,1}\|_2^2&\le\underbrace{\bigg[\max_{p\in[d_j]}\sum_{q=1}^{d_{\mminus j}}|\tilde\sigma_{p,q}^{(j,k)}-\tilde\sigma_{1,p,q}^{(j,k)}|I\{|\tilde\sigma_{p,q}^{(j,k)}|\ge\delta_1,|\tilde\sigma_{1,p,q}^{(j,k)}|\ge\delta_1\}\bigg]}_{\mathcal{I}_{1,1,1}} \\
				&\quad \times \underbrace{\bigg[\max_{q\in[d_{\mminus j}]}\sum_{p=1}^{d_j}|\tilde\sigma_{p,q}^{(j,k)}-\tilde\sigma_{1,p,q}^{(j,k)}|I\{|\tilde\sigma_{p,q}^{(j,k)}|\ge\delta_1,|\tilde\sigma_{1,p,q}^{(j,k)}|\ge\delta_1\}\bigg]}_{\mathcal{I}_{1,1,2}}\,.\\ 
			\end{split}
			\]
			By \eqref{sigma2} and \eqref{sigma1}, when $\delta_1>0$, we conclude that
			\begin{equation}\label{112}
				\begin{split}
					\mathcal{I}_{1,1,2}&\le \max_{p\in[d_j],q\in[d_{\mminus j}]}|\tilde\sigma_{p,q}^{(j,k)}-\tilde\sigma_{1,p,q}^{(j,k)}|\max_{q \in [d_{\mminus j}]}\sum^{d_j}_{p=1}I\{|\tilde\sigma_{1,p,q}^{(j,k)}|\ge\delta_1\}\\
					& \le\max_{p\in[d_j],q\in[d_{\mminus j}]}|\tilde\sigma_{p,q}^{(j,k)}-\tilde\sigma_{1,p,q}^{(j,k)}| \max_{q\in[d_{\mminus j}]}\sum^{d_j}_{p=1}I\{|\tilde\sigma_{1,p,q}^{(j,k)}| \neq 0\}  \\
     &=  O_{\rm p}\{ (n^{-1}\log D_n)^{1/2} s_{j} \}\,.
				\end{split}
			\end{equation}
		Similarly, one can verify that
			\begin{equation}\label{111}
				\mathcal{I}_{1,1,1}=  O_{\rm p} \bigg\{(n^{-1}\log D_n)^{1/2} \prod_{j^\prime \neq j} s_{j^\prime} \bigg\}\,.
			\end{equation}
			Combining  \eqref{112} and \eqref{111}, we have
			\[
			\|\hat\bSigma_{1,1}\|_2= O_{\rm p}
            \{ (S_n n^{-1}\log D_n)^{1/2}
            \}\,.
			\]
			
			Secondly, we consider $\|\hat\bSigma_{1,3}\|_2$. Write $\hat\bSigma_{1,3}=\hat\bSigma_{1,3,1}+\hat\bSigma_{1,3,2}$ with 
			\begin{equation}\label{hat Sigma 13 decompose}
			    		\begin{split}
				[\hat\bSigma_{1,3,1}]_{p,q}&=(\tilde\sigma_{1,p,q}^{(j,k)}-\tilde\sigma_{p,q}^{(j,k)})I\{|\tilde\sigma_{p,q}^{(j,k)}|<\delta_1,|\tilde\sigma_{1,p,q}^{(j,k)}|\ge\delta_1\}\,,\\
				[\hat\bSigma_{1,3,2}]_{p,q}&=\tilde\sigma_{p,q}^{(j,k)}I\{|\tilde\sigma_{p,q}^{(j,k)}|<\delta_1,|\tilde\sigma_{1,p,q}^{(j,k)}|\ge\delta_1\}\,,
			\end{split}
			\end{equation}
			respectively. 
			Similarly to \eqref{112} and \eqref{111}, we can verify that
			\[
			\|\hat\bSigma_{1,3,1}\|_2= O_{\rm p}
            \{(S_n n^{-1}\log D_n)^{1/2}
            \}\,.
			\]
			On the other hand, 
			\[
			\begin{split}
				\|\hat\bSigma_{1,3,2}\|^2_2 &\le \bigg[\max_{p\in[d_j]}\sum_{q=1}^{d_{\mminus j}}|\tilde\sigma_{p,q}^{(j,k)}|I\{|\tilde\sigma_{p,q}^{(j,k)}|<\delta_1,|\tilde\sigma_{1,p,q}^{(j,k)}|\ge\delta_1\}\bigg]\\
				&\quad\times\bigg[\max_{q \in [d_{\mminus j}]}\sum_{p=1}^{d_j}|\tilde\sigma_{p,q}^{(j,k)}|I\{|\tilde\sigma_{p,q}^{(j,k)}|<\delta_1,|\tilde\sigma_{1,p,q}^{(j,k)}|\ge\delta_1\}\bigg]\,.
			\end{split}
			\]
			Similarly to \eqref{112}, we have
			\[
			\max_{q \in [d_{\mminus j}]}\sum_{p=1}^{d_j}|\tilde\sigma_{p,q}^{(j,k)}|I\{|\tilde\sigma_{p,q}^{(j,k)}|<\delta_1,|\tilde\sigma_{1,p,q}^{(j,k)}|\ge\delta_1\}\le \delta_1 \max_{q \in [d_{\mminus j}]}\sum_{p=1}^{d_j} I\{|\tilde\sigma_{1,p,q}^{(j,k)}|\ge\delta_1\} \lesssim \delta_1 s_{j}  \,.
			\]
			A parallel bound holds for the second part of $\|\hat\bSigma_{1,3,2}\|_2$ by a similar argument with replacing $s_{j}$ by $\prod_{j^\prime\ne j}s_{j^\prime}$. Then, 
			\[
			\|\hat\bSigma_{1,3,2}\|_2\lesssim  \delta_1 S_n^{1/2} \,,
			\]
			which further implies that
			\begin{equation}\label{13}
				\|\hat\bSigma_{1,3}\|_2 =  O_{\rm p}\{\delta_1   S_n^{1/2} +  (S_n n^{-1}\log D_n)^{1/2} \}\,.
			\end{equation}
			
			Finally, for $\hat\bSigma_{1,2}$, write $\hat\bSigma_{1,2}=\hat\bSigma_{1,2,1}+\hat\bSigma_{1,2,2}$ with
	\begin{equation}\label{hat Sigma 12 decompose}
			    		\begin{split}
				[\hat\bSigma_{1,2,1}]_{p,q} &= \tilde\sigma_{2,p,q}^{(j,k)}I\{|\tilde\sigma_{p,q}^{(j,k)}|\ge\delta_1,|\tilde\sigma_{1,p,q}^{(j,k)}|<\delta_1\}\,,\\
				[\hat\bSigma_{1,2,2}]_{p,q}&=\tilde\sigma_{1,p,q}^{(j,k)}I\{|\tilde\sigma_{p,q}^{(j,k)}|\ge\delta_1,|\tilde\sigma_{1,p,q}^{(j,k)}|<\delta_1\}\,.
			\end{split}
		\end{equation}
			Similarly to \eqref{I2} and \eqref{sigma1}, we have
			\[
			\|\hat\bSigma_{1,2,2}\|_2  \lesssim   \delta_1 S_n^{1/2} \,.
			\]
			For $\hat\bSigma_{1,2,1}$, the bound relies on the inequality that
   \begin{align*}
       	\|\hat\bSigma_{1,2,1}\|_2^2 & \le  \bigg[\max_{q\in[d_{\mminus j}]}\sum_{p=1}^{d_j}|\tilde\sigma_{p,q}^{(j,k)}-\tilde\sigma_{1,p,q}^{(j,k)}|I\{|\tilde\sigma_{p,q}^{(j,k)}|\ge\delta_1,|\tilde\sigma_{1,p,q}^{(j,k)}|<\delta_1\}\bigg]\\
				&\quad\times\bigg[\max_{p\in[d_j]}\sum_{q=1}^{d_{\mminus j}}|\tilde\sigma_{p,q}^{(j,k)}-\tilde\sigma_{1,p,q}^{(j,k)}|I\{|\tilde\sigma_{p,q}^{(j,k)}|\ge\delta_1,|\tilde\sigma_{1,p,q}^{(j,k)}|<\delta_1\}\bigg]\,.
   \end{align*}
			 Notice that
\begin{align*}
    	 \max_{q\in[d_{\mminus j}]}\sum_{p=1}^{d_j}&|\tilde\sigma_{p,q}^{(j,k)}-\tilde\sigma_{1,p,q}^{(j,k)}|I\{|\tilde\sigma_{p,q}^{(j,k)}|\ge\delta_1,|\tilde\sigma_{1,p,q}^{(j,k)}|<\delta_1\}\\
				&\le \max_{q\in[d_{\mminus j}]}\sum_{p=1}^{d_j}|\tilde\sigma_{p,q}^{(j,k)}-\tilde\sigma_{1,p,q}^{(j,k)}|I\{|\tilde\sigma_{p,q}^{(j,k)}|\ge\delta_1,|\tilde\sigma_{1,p,q}^{(j,k)}| = 0\}\\
				&\qquad\qquad\qquad\qquad  +\max_{q\in[d_{\mminus j}]}\sum_{p=1}^{d_j}|\tilde\sigma_{p,q}^{(j,k)}-\tilde\sigma_{1,p,q}^{(j,k)}|I\{|\tilde\sigma_{p,q}^{(j,k)}|\ge\delta_1,|\tilde\sigma_{1,p,q}^{(j,k)}| \neq 0\}\\
    &\le \max_{p \in [d_j], q\in[d_{\mminus j}]} |\tilde\sigma_{p,q}^{(j,k)}-\tilde\sigma_{1,p,q}^{(j,k)}| \cdot \max_{q\in[d_{\mminus j}]}\sum_{p=1}^{d_j}I\{|\tilde\sigma_{p,q}^{(j,k)}-\tilde\sigma_{1,p,q}^{(j,k)}|\ge \delta_1\}\\
				&\qquad\qquad\qquad\qquad 
                +\max_{q\in[d_{\mminus j}]}\sum_{p=1}^{d_j}|\tilde\sigma_{p,q}^{(j,k)}-\tilde\sigma_{1,p,q}^{(j,k)}|I\{ |\tilde\sigma_{1,p,q}^{(j,k)}| \neq 0\}\\
				&  =   O_{\rm p}\{(n^{-1}\log D_n)^{1/2}\} \cdot \max_{q\in[d_{\mminus j}]}\sum_{p=1}^{d_j}I\{|\tilde\sigma_{p,q}^{(j,k)}-\tilde\sigma_{1,p,q}^{(j,k)}|\ge \delta_1\}\\
				&\qquad\qquad\qquad\qquad +  O_{\rm p}\{ (n^{-1}\log D_n)^{1/2} s_{j} \}\,,
\end{align*}			
			where the last two lines are by \eqref{sigma2} and \eqref{112}.   Set  $\delta_1=C_* (n^{-1}\log D_n)^{1/2}$ for some sufficiently large constant $C_*>0$, by \eqref{tail sigma2} and Markov inequality, we have
\begin{equation}\label{sum all error}
    \begin{split}
				&\mathbb{P}\bigg[\max_{q\in[d_{\mminus j}]}\sum_{p=1}^{d_j}I\{|\tilde\sigma_{p,q}^{(j,k)}-\tilde\sigma_{1,p,q}^{(j,k)}|\ge  \delta_1\}\ge \lambda\bigg]
				\le \frac{1}{\lambda}\sum^{d_{\mminus j}}_{q=1}\sum^{d_j}_{p=1}\mathbb{P}\{|\tilde\sigma_{2,p,q}^{(j,k)}|\ge  \delta_1\}\le \frac{C}{\lambda}
			\end{split}
\end{equation}			
			for any $\lambda>1$, provided that $\log D_n \ll  n^c$ for some constant $c\in(0,1)$ depending only on $c_1$ and $c_2$ specified in Assumptions {\rm\ref{tail}} and {\rm\ref{mixing}}.  It follows that
			\[
			\max_{q\in[d_{\mminus j}]}\sum_{p=1}^{d_j}I\{|\tilde\sigma_{p,q}^{(j,k)}-\tilde\sigma_{1,p,q}^{(j,k)}|\ge  \delta_1\}=O_{\rm p}(1)\,.
			\]
			A similar bound holds after exchanging the indices $p$ and $q$, which further implies that $
			\|\hat\bSigma_{1,2,1}\|_2  =  O_{\rm p}(\ubar\sigma_{\xi}^2\bar\sigma_{\xi}^{-1}\Pi_{n})$. 
			The same upper bound also holds for $\|\hat\bSigma_{1,2,2}\|_2$, $\|\hat\bSigma_{1,3}\|_2$, $\|\hat\bSigma_{1,1}\|_2$ and $\mathcal{I}_2$, and hence it holds for $\|T_{\delta_1}\{\tilde \bSigma_{\Yb_j,\xi}(k)\}-\tilde\bSigma_{\Cb_j,\xi}(k)\|_2$. Lemma \ref{lem: l2 truncated covariance} follows directly.  
$\hfill\Box$

\subsection{Proof of Lemma \ref{pro: singular K21j}}\label{sec:lem:gap K21j}

Let $\Ab_{j,-i}$ denote the $d_j\times (r-1)$ matrix obtained by deleting the $i$-th column of $\Ab_j$, and define $\Db_{j,-i}=(\Ib_{d_j}-\ab_{i,j}\ab_{i,j}^{\T})\Ab_{j,-i}$. Let $\Pb_{j,-i}$ be the matrix of left singular vectors of $\Db_{j,-i}$. We first show that there exists an $(r-1)\times (r-1)$ matrix $\Ub_{j,-i}$ such that $\Db_{j,-i}=\Pb_{j,-i}\Ub_{j,-i}$, and the singular values of $\Ub_{j,-i}$ are bounded away from $0$ and $\infty$. Without loss of generality, assume that the $i$-th column of $\Ab_j$ is placed as the first column, so that $\Ab_j=(\ab_{i,j},\Ab_{j,-i})$. Then
\[
\Ab_j^{\T}\Ab_j
=
\begin{pmatrix}
1 & \ab_{i,j}^{\T}\Ab_{j,-i}\\
\Ab_{j,-i}^{\T}\ab_{i,j} & \Ab_{j,-i}^{\T}\Ab_{j,-i}
\end{pmatrix}\,.
\]
Since $|\ab_{i,j}|_2=1$, the Schur complement of the upper-left block of $\Ab_j^{\T}\Ab_j$ is 
$$\Ab_{j,-i}^{\T}\Ab_{j,-i}-\Ab_{j,-i}^{\T}\ab_{i,j}\ab_{i,j}^{\T}\Ab_{j,-i}=\Db_{j,-i}^{\T}\Db_{j,-i}\,.$$
For any $\xb\in\mathbb{R}^{r-1}$ with $|\xb|_2=1$, since $\Ab_{j,-i}^{\T}\ab_{i,j}\ab_{i,j}^{\T}\Ab_{j,-i}$ is positive semidefinite, we have $\xb^{\T}\Db_{j,-i}^{\T}\Db_{j,-i}\xb\le \xb^{\T}\Ab_{j,-i}^{\T}\Ab_{j,-i}\xb$. Hence, $\sigma_{\max}(\Db_{j,-i}^{\T}\Db_{j,-i})\le \sigma_{\max}(\Ab_{j,-i}^{\T}\Ab_{j,-i})$. Since $\Ab_{j,-i}^{\T}\Ab_{j,-i}$ is a principal submatrix of $\Ab_j^{\T}\Ab_j$, it follows that $\sigma_{\max}(\Ab_{j,-i}^{\T}\Ab_{j,-i})\le \sigma_{\max}(\Ab_j^{\T}\Ab_j)$. Therefore, $\sigma_{\max}(\Db_{j,-i}^{\T}\Db_{j,-i})\le \sigma_{\max}(\Ab_j^{\T}\Ab_j)$. By Assumption \ref{sparsity}, we have
$$\sigma_{\max}(\Db_{j,-i}^{\T}\Db_{j,-i})\le \sigma_{\max}(\Ab_j^{\T}\Ab_j)\le C_6^2\,.$$
Moreover, by the inverse formula for block matrices, the inverse of the Schur complement
$\Db_{j,-i}^{\T}\Db_{j,-i}$ is the lower-right block of
$(\Ab_j^{\T}\Ab_j)^{-1}$. By Assumption \ref{sparsity}, we have
\[
\sigma_{\max}\{(\Db_{j,-i}^{\T}\Db_{j,-i})^{-1}\}
\le
\sigma_{\max}\{(\Ab_j^{\T}\Ab_j)^{-1}\}
\le
C_6^2\,,
\]
which implies $\sigma_{\min}(\Db_{j,-i}^{\T}\Db_{j,-i})
\ge
C_6^{-2}$. Consequently,
\[
C_6^{-1}
\le
\sigma_{r-1}(\Db_{j,-i})
\le
\sigma_1(\Db_{j,-i})
\le
C_6.
\]
Since $\Db_{j,-i}$ and $\Ub_{j,-i}$ have the same singular values, it holds that $C_6^{-1}\le \sigma_{r-1}(\Ub_{j,-i})\le \sigma_1(\Ub_{j,-i})\le C_6$.

 Left-multiplying both sides of $\Pb_{j,-i}=\Db_{j,-i}\Ub^{-1}_{j,-i}$ by $\ab_{i,j}^\T$, we have
			\[
			\ab_{i,j}^\T\Pb_{j,-i}=\ab_{i,j}^\T(\Ib_{d_j}-\ab_{i,j}\ab_{i,j}^\T)\Ab_{j,-i}\Ub_{j,-i}^{-1}={\bf 0}\,.
			\]
			In other words, $\ab_{i,j}$ is orthogonal to $\Pb_{j,-i}$. Let $\Gb_{k,\xi,-i}$ be the $(r-1)\times(r-1)$ diagonal matrix by deleting $g_{k,i,\xi}$ and the associated row and column in $\Gb_{k,\xi}$. Then, by definition,
            \begin{align*}
                &\Pb_{j,-i}^\T(\Kb_{1,2,j}-\bar\lambda_i \Ib_{d_j})\Pb_{j,-i}\\
    &~~~~~~=\Pb_{j,-i}^\T\Ab_j\Gb_{1,\xi}\Gb_{2,\xi}^{-1}(\Ab_j^\T\Ab_j)^{-1}\Ab_j^\T\Pb_{j,-i}-\bar\lambda_i\Ib_{r-1}\\
				&~~~~~~=\Pb_{j,-i}^\T\Ab_j\Gb_{1,\xi}\Gb_{2,\xi}^{-1}(\Ab_j^\T\Ab_j)^{-1}\Ab_j^\T(\Ib_{d_j}-\ab_{i,j}\ab_{i,j}^\T)\Ab_{j,-i}\Ub_{j,-i}^{-1}-\bar\lambda_i\Ib_{r-1}\\
				&~~~~~~=\Pb_{j,-i}^\T\Ab_{j,-i}\Gb_{1,\xi,-i}\Gb_{2,\xi,-i}^{-1}\Ub_{j,-i}^{-1}-\bar\lambda_i\Ib_{r-1}\\
				&~~~~~~=\Pb_{j,-i}^\T(\Ib_{d_j}-\ab_{i,j}\ab_{i,j}^\T)\Ab_{j,-i}\Gb_{1,\xi,-i}\Gb_{2,\xi,-i}^{-1}\Ub_{j,-i}^{-1}-\bar\lambda_i\Ib_{r-1}\\
				&~~~~~~=\Ub_{j,-i}(\Gb_{1,\xi,-i}\Gb_{2,\xi,-i}^{-1}-\bar\lambda_i\Ib_{r-1})\Ub_{j,-i}^{-1}\,,
            \end{align*}
			whose minimum singular value is always bounded away from $0$ and $\infty$ under Assumption \ref{gap new}. Furthermore, because $(\ab_{i,j},\Pb_{j,-i})^\T(\ab_{i,j},\Pb_{j,-i})=\Ib_r$, there exists a $d_j\times (d_j-r)$ matrix $\Ob_{j,-r}$ satisfying 
			\[
			(\ab_{i,j},\Pb_{j,-i})^\T\Ob_{j,-r}={\bf 0} ~~ \textup{and} ~~ \Ob_{j,-r}^\T\Ob_{j,-r}=\Ib_{d_j-r}\,.
			\]
			Let $\Ob_{j,\mminus i}=(\Pb_{j,-i},\Ob_{j,-r})$. By the fact that $\Ab_j^\T\Ob_{j,-r}={\bf 0}$, we have 
			\[
			\Ob_{j,\mminus i}^\T(\Kb_{1,2,j}-\bar\lambda_i\Ib_{d_j})\Ob_{j,\mminus i}=\left(\begin{matrix}
				&\Pb_{j,-i}^\T(\Kb_{1,2,j}-\bar\lambda_i\Ib_{d_j})\Pb_{j,-i}&{\bf 0}\\
				&{\bf 0}&-\bar\lambda_i\Ib_{d_j-r}
			\end{matrix}\right)\,,
			\]
			whose minimum singular value is also uniformly bounded away from $0$ and $\infty$.  Hence the first assertion of Lemma \ref{pro: singular K21j} is verified.
            
            On the other hand, we already know that $\max_{k\in[K],i\in[r]}|\tilde g_{k,i,\xi}-g_{k,i,\xi}| =  O_{\rm p}(w_1n^{-1/2})$ by Lemma \ref{lem: gkixi}. Then, under Assumptions \ref{eigenvalue} and \ref{gap new}, we can conclude that
			\[
			\hat\lambda_i= \Big\{\frac{g_{1,i,\xi}}{g_{2,i,\xi}} + o_{\rm p}(1)\Big\}\{1+o_{\rm p}(1)\}^{-1}   =\bar\lambda_i + o_{\rm p}(1)\,.
			\]
			Repeating the argument used to establish the first assertion of Lemma \ref{pro: singular K21j}, with $\Kb_{1,2,j}$ and $\bar\lambda_i$ replaced by $\hat\Kb_{1,2,j}$ and $\hat{\lambda}_i$, respectively, we can conclude the second assertion of Lemma \ref{pro: singular K21j}.  This completes the proof of Lemma \ref{pro: singular K21j}.
$\hfill\Box$

\subsection{Proof of Lemma \ref{lem: Omega n}}\label{sec:lem: Omega n}
For the events $\Xi_{h,n}(\tilde C)$ with $h \in \{1,2,5,6,8,9\}$, we only prove the results for the case $k_1=k_2=0$. The cases with general $k_1$ and $k_2$ can be handled similarly and are therefore omitted. We first prove that $\mathbb{P}\{\Xi_{1,n}(\tilde C)\}\to 1$ and 
$\mathbb{P}\{\Xi_{2,n}(\tilde C)\}\to 1$ as $n \to \infty$. 
The proofs for $\Xi_{3,n}(\tilde C)$ to $\Xi_{5,n}(\tilde C)$ are analogous and are therefore omitted. We then consider $\Xi_{7,n}(\tilde C)$ to $\Xi_{9,n}(\tilde C)$ using covering arguments, and finally handle $\Xi_{6,n}(\tilde C)$, which requires additional treat on $\bar f_i$.

        Under Assumptions \ref{tail}, \ref{mixing} and \ref{cross}, the proofs for the events $\Xi_{1,n}(\tilde C)$ to $\Xi_{5,n}(\tilde C)$ are  similar.
     For $\Xi_{1,n}(\tilde C)$,  notice that
\begin{align*} 
     &\frac{1}{n-k}\sum_{t=k+1}^n(f_{t,i}-\bar f_i)(f_{t-k,\ell}-\bar f_\ell)-\Upsilon_{k,i,\ell} \\
     &~~~~~~ =
\frac{1}{n-k}\sum_{t=k+1}^n
\{f_{t,i}f_{t-k,\ell}-\mathbb{E}(f_{t,i}f_{t-k,\ell})\}   -
\frac{1}{n-k}\sum_{t=k+1}^n
\{f_{t,i}\bar f_\ell-\mathbb{E}(f_{t,i}\bar f_\ell)\}  \\\nonumber
&~~~~~~ \quad -
\frac{1}{n-k}\sum_{t=k+1}^n
\{\bar f_i f_{t-k,\ell}-\mathbb{E}(\bar f_i f_{t-k,\ell})\}   +
\{\bar f_i\bar f_\ell-\mathbb{E}(\bar f_i\bar f_\ell)\}\,.
\end{align*}
Similarly to  \eqref{sigma2},  it holds that
\[
 \frac{1}{n-k}\sum_{t=k+1}^n
    \{f_{t,i}f_{t-k,\ell}-\mathbb{E}(f_{t,i}f_{t-k,\ell})\}
     =
    O_{\rm p}(n^{-1/2})\,.
\]
We further have
\begin{align}\label{eq: four terms for ftftk}
&\frac{1}{n-k}\sum_{t=k+1}^n
 \{f_{t,i}\bar f_\ell-\mathbb{E}(f_{t,i}\bar f_\ell)\} \\\nonumber
& ~~~~~~~~~~~~=
\bigg[\frac{1}{n-k}\sum_{t=k+1}^n
\{f_{t,i}-\mathbb{E}(f_{t,i})\}\bigg]
\{\bar f_\ell-\mathbb{E}(\bar f_\ell)\} \\ \nonumber
&~~~~~~~~~~~~ \quad +
\bigg\{\frac{1}{n-k}\sum_{t=k+1}^n
\mathbb{E}(f_{t,i})\bigg\}
\{\bar f_\ell-\mathbb{E}(\bar f_\ell)\} \\  \nonumber
& ~~~~~~~~~~~~ \quad + 
\bigg[\frac{1}{n-k}\sum_{t=k+1}^n
\{f_{t,i}-\mathbb{E}(f_{t,i})\}\bigg] \mathbb{E}(\bar f_\ell) \\
&~~~~~~~~~~~~ \quad + 
\bigg\{\frac{1}{n-k}\sum_{t=k+1}^n
\mathbb{E}(f_{t,i})\bigg\}\mathbb{E}(\bar f_\ell)
-
\frac{1}{n-k}\sum_{t=k+1}^n
\mathbb{E}(f_{t,i}\bar f_\ell)\,. \nonumber
\end{align}
By the same concentration
inequality used in proving \eqref{sigma2}, it holds that
\[
    \frac{1}{n-k}\sum_{t=k+1}^n
    \{f_{t,i}-\mathbb{E}(f_{t,i})\}
    =
    O_{\rm p}(n^{-1/2}) ~~ \textup{and}~~
    \bar f_\ell-\mathbb{E}(\bar f_\ell)
    =
    O_{\rm p}(n^{-1/2})\,.
\]
Moreover, Assumption \ref{tail} implies that
$\max_{t\in[n],i\in[r]}\mathbb{E}|f_{t,i}| = O(1)$. We can conclude that the first three terms on the right-hand side of \eqref{eq: four terms for ftftk} are $O_{\rm p}(n^{-1/2})$. It remains to bound the last deterministic term in \eqref{eq: four terms for ftftk}.
By the definition of $\bar f_\ell$, we have
\begin{align*}
    &\bigg|
\bigg\{\frac{1}{n-k}\sum_{t=k+1}^n
\mathbb{E}(f_{t,i})\bigg\}\mathbb{E}(\bar f_\ell)
-
\frac{1}{n-k}\sum_{t=k+1}^n
\mathbb{E}(f_{t,i}\bar f_\ell)
\bigg| \\
&~~~~~~ =
\bigg|
\frac{1}{n-k}\sum_{t=k+1}^n
\{
\mathbb{E}(f_{t,i})\mathbb{E}(\bar f_\ell)
-
\mathbb{E}(f_{t,i}\bar f_\ell)
\}
\bigg| \\
&~~~~~~ =
\bigg|
\frac{1}{n(n-k)}
\sum_{t=k+1}^n\sum_{s=1}^n
\operatorname{Cov}(f_{t,i},f_{s,\ell})
\bigg| \\
&~~~~~~ \le
\frac{1}{n(n-k)}
\sum_{t=k+1}^n\sum_{s=1}^n
|\operatorname{Cov}(f_{t,i},f_{s,\ell})|\,.
\end{align*}
By Theorem 16.2.3 in \citeS{athreya2006measure-app},  it holds that $\sup_{t\in[n]}
    \sum_{s=1}^n
    |\operatorname{Cov}(f_{t,i},f_{s,\ell})|
    \le C$ for some universal constant $C > 0$. Thus, the last deterministic term in \eqref{eq: four terms for ftftk} is $O(n^{-1})$. Similarly, we have
\[
    \frac{1}{n-k}\sum_{t=k+1}^n
    \{\bar f_i f_{t-k,\ell}-\mathbb{E}(\bar f_i f_{t-k,\ell})\}
    =
    O_{\rm p}(n^{-1/2}) ~~ \textup{and} ~~  \bar f_i\bar f_\ell-\mathbb{E}(\bar f_i\bar f_\ell)  =
    O_{\rm p}(n^{-1/2})\,.
\]
Combining the above bounds, we obtain
\[
    \bigg|
    \frac{1}{n-k}\sum_{t=k+1}^n(f_{t,i}-\bar f_i)(f_{t-k,\ell}-\bar f_\ell)
    -
    \Upsilon_{k,i,\ell}
    \bigg|
    =
    O_{\rm p}(n^{-1/2})\,.
\]
Therefore, for some sufficiently large
constant $\tilde C>0$, $\mathbb{P}\{\Xi_{1,n}(\tilde C)\}\rightarrow 1$ as $n\rightarrow\infty$.

 For $\Xi_{2,n}(\tilde C)$, recall that
$\bLambda_{\xi} = (\Lambda_{\xi,i,j})_{r\times r}$ is the diagonal matrix
with the $i$-th diagonal entry being
$n^{-1}\sum_{t=1}^n\mathbb{E}\{(\xi_{t,i}-\bar\xi_i)^2\}$. Write $u_{t,i}
    =
    (\ab_{i,m}^{\MP}\otimes\cdots\otimes\ab_{i,1}^{\MP})^\T
    \textup{vec}(\mathcal{E}_t)$ and $\bar u_i=n^{-1}\sum_{t=1}^n u_{t,i}$.
By definition,
\[
    \Lambda_{\xi,\ell,\ell}^{1/2}\xi_{t,\ell}^{\textup{s}}
    =
    w_\ell(f_{t,\ell}-\bar f_\ell)+(u_{t,\ell}-\bar u_\ell)\,.
\]
Thus, for $i,\ell\in[r]$ and $k\in\{0,1\}$,
\[
\begin{split}
& \Lambda_{\xi,\ell,\ell}^{1/2}
\bigg|
\frac{1}{n-k}\sum_{t=k+1}^n
(f_{t,i}-\bar f_i)\xi_{t-k,\ell}^{\textup{s}}
-
\frac{1}{n-k}\sum_{t=k+1}^n
\mathbb{E}\{(f_{t,i}-\bar f_i)\xi_{t-k,\ell}^{\textup{s}}\}
\bigg|                                                     \\
&~~~~~~ \le
w_\ell
\bigg|
\frac{1}{n-k}\sum_{t=k+1}^n
(f_{t,i}-\bar f_i)(f_{t-k,\ell}-\bar f_\ell)
-
\Upsilon_{k,i,\ell}
\bigg|                                                     \\
&~~~~~~\quad
+
\bigg|
\frac{1}{n-k}\sum_{t=k+1}^n
(f_{t,i}-\bar f_i)(u_{t-k,\ell}-\bar u_\ell)
-
\frac{1}{n-k}\sum_{t=k+1}^n
\mathbb{E}\{(f_{t,i}-\bar f_i)(u_{t-k,\ell}-\bar u_\ell)\}
\bigg|\, .
\end{split}
\]
Under Assumption \ref{error}, it follows that $\mathbb{E}\{(f_{t,i}-\bar f_i)(u_{t-k,\ell}-\bar u_\ell)\} = 0$.
Moreover, by Assumption \ref{cross} and the uniform boundedness of
$|\ab_{i,m}^{\MP}\otimes\cdots\otimes\ab_{i,1}^{\MP}|_2$, the process
$\{u_{t,\ell}\}_{t\ge1}$ has the same exponential-type tail bound as in
Assumption \ref{tail}. Therefore, 
\[
\begin{split}
& \Lambda_{\xi,\ell,\ell}^{1/2}
\bigg|
\frac{1}{n-k}\sum_{t=k+1}^n
(f_{t,i}-\bar f_i)\xi_{t-k,\ell}^{\textup{s}}
-
\frac{1}{n-k}\sum_{t=k+1}^n
\mathbb{E}\{(f_{t,i}-\bar f_i)\xi_{t-k,\ell}^{\textup{s}}\}
\bigg|  =
O_{\rm p}\bigg(\frac{w_\ell}{\sqrt n}\bigg)\,.
\end{split}
\]
Notice that 
  	\begin{equation*}\label{expectation xi si}
   \begin{split}
       \Lambda_{\xi,i,i} = \frac{1}{n}\sum_{s=1}^n\mathbb{E}\{(\xi_{s,i}-\bar\xi_i)^2\} =  	\frac{w_i^2}{n}\sum_{s=1}^n\mathbb{E}\{(f_{s,i} - \bar f_i)^2\}   +  	\frac{1}{n}\sum_{s=1}^n\mathbb{E}\{(u_{s,i}-\bar u_i)^2\}\,.
   \end{split}
			\end{equation*}
   By Assumption \ref{cross}, it follows that $	n^{-1}\sum_{s=1}^n\mathbb{E}\{(u_{s,i}-\bar u_i)^2\} \le C$ for some universal constant $C>0$. It then follows that 
\begin{equation}\label{eq: varxi-bar-xi}
      \Lambda_{\xi,i,i}=w_i^2\Upsilon_{0,i,i}+O(1)\asymp w_i^2 \,.
\end{equation}
Then $\mathbb{P}\{\Xi_{2,n}(\tilde C)\}\rightarrow1$ as $n\rightarrow\infty$. By the definitions of $\Fb_{\xi,\mminus i}^{\textup{s}}$ and $\bxi_i^{\textup{s}}$, and by similar calculations, the events $\Xi_{3,n}(\tilde C)$ to $\Xi_{5,n}(\tilde C)$ hold with probability approaching one for sufficiently large $\tilde C > 0$. The proofs are very similar so we omit the details. 


			In the following, we focus on $\mathbb{P}\{\Xi_{7,n}(\tilde C)\}\rightarrow 1$ as $n \to \infty$. 
			Recall that $\bbeta_j \in \mathbb{R}^{d_j}$ specified in $\Xi_{7,n}(\tilde C)$ satisfies $|\bbeta_j|_2 = 1$. Then,	according to (G.3) of Lemma G.1 in \citeS{han2024tensor-app}, there exist non-random vectors $\{\bgamma_{j,h_j} \in \mathbb{R}^{d_j},1\le h_j\le (17m)^{d_j}\}_{j\in[m]}$ and $\{\tilde\bgamma_{j,h_j^\prime}  \in \mathbb{R}^{d_j},1\le h_j^\prime \le (17m)^{d_j}\}_{j\in[m]}$ such that $|\bgamma_{j,h_j}|_2\le 1$, $|\tilde \bgamma_{j,h_j^\prime}|_2\le 1$, and
			\[
			\begin{split}
				&		\max_{\substack{|\tilde\bbeta_j|_2=1 = |\bbeta_j|_2,\\
                    \tilde\bbeta_j,\, \bbeta_j \in\mathbb{R}^{d_j},\forall j\in[m]}}\bigg|\frac{1}{n-1}\sum_{t=2}^n(\mathcal{E}_{t-1}\times_{j=1}^m\tilde\bbeta_j^{\T})(\mathcal{E}_t\times_{j=1}^m\bbeta_j^{\T})\bigg|\\
				& ~~~~~~ \le 2\max_{1\le h_j\le (17m)^{d_j},1\le  h_j^\prime \le (17m)^{d_j},\forall j\in[m]}\bigg|\frac{1}{n-1}\sum_{t=2}^n(\mathcal{E}_{t-1}\times_{j=1}^m\tilde\bgamma_{j,h^\prime_j}^{\T})(\mathcal{E}_t\times_{j=1}^m\bgamma_{j,h_j}^{\T})\bigg|\,.
			\end{split}
			\]
			Similarly to \eqref{tail error},  given such deterministic $\{ \bgamma_{j,h_j}\}_{j\in[m]}$ and $\{\tilde  \bgamma_{j,h_j^\prime}\}_{j\in[m]}$, we always have
		\begin{equation}\label{Xi 7n error bound}
		    \mathbb{P}\bigg(\bigg|\frac{1}{n-1}\sum_{t=2}^n(\mathcal{E}_{t-1}\times_{j=1}^m\tilde\bgamma_{j,h_j^\prime})(\mathcal{E}_t\times_{j=1}^m\bgamma_{j,h_j})\bigg|\ge  x\bigg)\lesssim \exp(- Cnx^2) +\exp(- Cn^{\tilde c}x^{\tilde c})
		\end{equation}
			for any $x > 0$ and some universal constant $C>0$, where $\tilde c=(1+2c_1^{-1}+c_2^{-1})^{-1}$.  Therefore,
			\[
			\begin{split}
				&\mathbb{P}\bigg(\max_{\substack{|\tilde\bbeta_j|_2=1 = |\bbeta_j|_2,\\
                    \tilde\bbeta_j,\, \bbeta_j \in\mathbb{R}^{d_j},\forall j\in[m]}} \bigg|\frac{1}{n-1}\sum_{t=2}^n(\mathcal{E}_{t-1}\times_{j=1}^m\tilde\bbeta_j^{\T})(\mathcal{E}_t\times_{j=1}^m\bbeta_j^{\T})\bigg|\ge x\bigg)\\
				& \qquad \qquad\qquad\qquad \qquad \qquad\lesssim  (17m)^{2\sum_{j=1}^m d_j}\{\exp(-Cnx^2)+\exp(-Cn^{\tilde c}x^{\tilde c})\}\,.
			\end{split}
			\]
			Taking $x=\tilde C\{(n^{-1}\sum_{j=1}^m d_j)^{1/2}+n^{-1}(\sum_{j=1}^m d_j)^{1/\tilde c}\}$ for sufficiently large constant $\tilde C>0$, we obtain
			\begin{equation}\label{Xi 4n rate}
            \begin{split}
            &\max_{\substack{|\tilde\bbeta_j|_2=1 = |\bbeta_j|_2,\\
                    \tilde\bbeta_j,\, \bbeta_j \in\mathbb{R}^{d_j},\forall j\in[m]}} \bigg|\frac{1}{n-1}\sum_{t=2}^n(\mathcal{E}_{t-1}\times_{j=1}^m\tilde\bbeta_j^{\T})(\mathcal{E}_t\times_{j=1}^m\bbeta_j^{\T})\bigg| \\
            & \qquad \qquad\qquad\qquad \qquad \qquad \le \tilde C \bigg\{\bigg(\frac{\sum_{j=1}^m d_j}{n}\bigg)^{1/2}+\frac{(\sum_{j=1}^m d_j)^{1/\tilde c}}{n} \bigg\}  
            \end{split}
			\end{equation}
		with probability at least $1-\tilde C_1\exp(-\tilde C_2{\textstyle{\sum}}_{j=1}^m d_j)$ for some universal constants $\tilde C_1,\tilde C_2>0$. Therefore, $\mathbb{P}\{\Xi_{7,n}(\tilde C)\} \rightarrow 1$ for sufficiently large $\tilde C>0$, as long as $\max_{j\in[m]}d_j\rightarrow \infty$ ($D_n\rightarrow\infty$) as $n\rightarrow \infty$. 
 
The event $\Xi_{8,n}(\tilde C)$ can be handled similarly. By Assumption \ref{cross} and the similar arguments in the proof for $\Xi_{7,n}(\tilde C)$, we have
\[
\begin{split}
   \max_{  |\bbeta_j|_2=1,\, \bbeta_j \in\mathbb{R}^{d_j}, \forall j\in[m] }
\bigg|
\frac{1}{n}\sum_{t=1}^n
\mathcal{E}_t\times_{j=1}^m\bbeta_j^{\T}
\bigg|
&\le
\tilde C \bigg\{\bigg(\frac{\sum_{j=1}^m d_j}{n}\bigg)^{1/2}+\frac{(\sum_{j=1}^m d_j)^{1/\check c}}{n} \bigg\} \\
&\le \tilde C\bigg\{\bigg(\frac{\sum_{j=1}^m d_j}{n}\bigg)^{1/2}+\frac{(\sum_{j=1}^m d_j)^{1/\tilde c}}{n} \bigg\} 
\end{split}
\]
with probability approaching one, where $\check c = (2+|c_1^{-1} -1|_{\MP} + c_2^{-1})^{-1}$. For $\Xi_{9,n}(\tilde C)$, given any deterministic
$\bgamma_{j,h_j},\tilde\bgamma_{j,h_j^\prime} \in \mathbb{R}^{d_j}$ for $j \in [m]$, define
\[
    Z_t
    =
    (\otimes_{m}^{j=1}\tilde\bgamma_{j,h_j^\prime})^\T
    [
    \textup{vec}(\mathcal{E}_t)\textup{vec}(\mathcal{E}_t)^\T
    -
    \mathbb{E}\{\textup{vec}(\mathcal{E}_t)\textup{vec}(\mathcal{E}_t)^\T\}
    ]
    (\otimes_{m}^{j=1}\bgamma_{j,h_j})\,.
\]
Then $\mathbb{E}(Z_t)=0$. By Assumption \ref{cross}  and the similar arguments in the proof for $\Xi_{7,n}(\tilde C)$, it holds that
$\mathbb{P}\{\Xi_{9,n}(\tilde C)\}\rightarrow1$ as $n \to \infty$.

For $\Xi_{6,n}(\tilde C)$, it is slightly different because
$f_{t-k,i}-\bar f_i$ depends on the factors across the time dimension. We only consider the case $k=1$, since the case $k=0$ can be handled analogously. Notice that
\[
\begin{split}
&\bigg|
\frac{1}{n-1}\sum_{t=2}^n
(f_{t-1,i}-\bar f_i)
(\mathcal{E}_t\times_{j=1}^m\bbeta_j^\T)
\bigg|                                                \\
&\quad\le
\bigg|
\frac{1}{n-1}\sum_{t=2}^n
\{f_{t-1,i}-\mathbb{E}(\bar f_i)\}
(\mathcal{E}_t\times_{j=1}^m\bbeta_j^\T)
\bigg| +
|\bar f_i-\mathbb{E}(\bar f_i)|
\bigg|
\frac{1}{n-1}\sum_{t=2}^n
(\mathcal{E}_t\times_{j=1}^m\bbeta_j^\T)
\bigg| \,.
\end{split}
\]
For the first term, given any deterministic $\bgamma_{j,h_j} \in \mathbb{R}^{d_j}$ for $j \in [m]$,
Assumption \ref{error} implies that $ \mathbb{E}
    [
    \{f_{t-1,i}-\mathbb{E}(\bar f_i)\}
    (\mathcal{E}_t\times_{j=1}^m \bgamma_{j,h_j}^\T)
    ]=0$. Moreover, Assumptions \ref{tail}, \ref{mixing}, and \ref{cross} imply that
\[
\mathbb{P}\bigg(
\bigg|
\frac{1}{n-1}\sum_{t=2}^n
\{f_{t-1,i}-\mathbb{E}(\bar f_i)\}
(\mathcal{E}_t\times_{j=1}^m \bgamma_{j,h_j}^\T)
\bigg|>x
\bigg)
\lesssim
\exp(-Cnx^2)+\exp(-Cn^{\tilde c}x^{\tilde c})
\]
for any $x > 0$ and some universal constant $C>0$, where $\tilde c=(1+2c_1^{-1}+c_2^{-1})^{-1}$. Following the similar arguments in the proof for $\Xi_{7,n}(\tilde C)$, it holds that
\begin{equation}\label{Xi 6 ft Et tail bound}
\begin{split}
    &\max_{|\bbeta_j|_2=1,\,\bbeta_j \in\mathbb{R}^{d_j}, \forall j\in[m]}
\bigg|
\frac{1}{n-1}\sum_{t=2}^n
\{f_{t-1,i}-\mathbb{E}(\bar f_i)\}
(\mathcal{E}_t\times_{j=1}^m\bbeta_j^\T)
\bigg| \\
& ~~~~~~~~~~~~~~~~~~~~~~~~\le
\tilde C \bigg\{\bigg(\frac{\sum_{j=1}^m d_j}{n}\bigg)^{1/2}+\frac{(\sum_{j=1}^m d_j)^{1/\tilde c}}{n} \bigg\}
\end{split}  
\end{equation}
with probability approaching one. For the second term, by Assumptions \ref{tail} and \ref{mixing}, we have $|\bar f_i-\mathbb{E}(\bar f_i)|=O_{\rm p}(n^{-1/2})=o_{\rm p}(1)$. 
Combining this with the bound established for $\Xi_{8,n}(\tilde C)$ gives
\[
\begin{split}
&|\bar f_i-\mathbb{E}(\bar f_i)|
\max_{|\bbeta_j|_2=1,\,\bbeta_j \in\mathbb{R}^{d_j},\forall j\in[m]}
\bigg|
\frac{1}{n-1}\sum_{t=2}^n
(\mathcal{E}_t\times_{j=1}^m\bbeta_j^\T)
\bigg|   \\
&~~~~~~~~~~~~~~~~~~~~~~~~ \le
\tilde C \bigg\{\bigg(\frac{\sum_{j=1}^m d_j}{n}\bigg)^{1/2}+\frac{(\sum_{j=1}^m d_j)^{1/\tilde c}}{n} \bigg\}
\end{split}
\]
with probability approaching one. Therefore, $\mathbb{P}\{\Xi_{6,n}(\tilde C)\}\rightarrow1$  as $n \to \infty$ for sufficiently large $\tilde C>0$.
$\hfill\Box$

\subsection{Proof of Lemma \ref{lemma: E varphi}}\label{sec:pro:varphi}
 Let $\tilde\bUpsilon_k = (\tilde\Upsilon_{k,i,\ell})_{r \times r}$, where 
	\[
	\tilde\Upsilon_{k,i,\ell}=\frac{1}{n-k}\sum_{t=k+1}^n\mathbb{E}(\xi_{t,i}^{\textup{s}}\xi_{t-k,\ell}^{\textup{s}})\,,\quad i,\ell\in[r]\,,\,k\in \{0,1\}\,.
	\]
Further let  $\tilde\bUpsilon_{k,i}$ be the  $i$-th column of $\tilde\bUpsilon_k$, $\tilde\bUpsilon_{k,i,-i}$ be the $(r-1)$-dimensional vector by deleting the $i$-th entry of $\tilde\bUpsilon_{k,i}$, and $\tilde\bUpsilon_{k,-i,-i}$ be the $(r-1)\times (r-1)$ matrix by deleting the $i$-th row and $i$-th column of $\tilde\bUpsilon_k$. Then, by definition, 
\begin{equation}\label{bound for F  xi}
    \frac{1}{n-1}\mathbb{E}\{(\Fb_{\xi,\mminus i}^{\textup{s}})^{\T}\bxi_i^{\textup{s}}\}=\tilde\bUpsilon_{1,i,-i}\,.
\end{equation}
We begin by deriving several bounds for $\tilde\bUpsilon_k$.

Recall that $\bLambda_{\xi} = (\Lambda_{\xi,i,j})_{r\times r}$ is the $r\times r$ diagonal matrix with the $i$-th diagonal entry being $n^{-1}\sum_{t = 1}^n\mathbb{E}\{(\xi_{t,i}-\bar\xi_i)^2\}$. Write $u_{t,i} = (\ab_{i,m}^{\MP} \otimes \cdots \otimes  \ab_{i,1}^{\MP})^{\T}\textup{vec}(\mathcal{E}_t)$  and  $\bar u_i = n^{-1}\sum_{t = 1}^n u_{t,i}$. Then,
\[
\begin{split}
    \Lambda_{\xi,i,i}^{1/2}\Lambda_{\xi,\ell,\ell}^{1/2}\mathbb{E}(\xi_{t,i}^{\textup{s}}\xi_{t-k,\ell}^{\textup{s}})&=\mathbb{E}\{(\xi_{t,i}-\bar\xi_{i})(\xi_{t-k,\ell}-\bar\xi_{\ell})\}\\
    &= w_iw_{\ell}\mathbb{E}\{(f_{t,i}-\bar f_{i})(f_{t-k,\ell}-\bar f_{\ell})\}   +\mathbb{E}\{(u_{t,i}-\bar u_{i})(u_{t-k,\ell}-\bar u_{\ell})\}\,.
\end{split}
\]
Therefore,
\begin{equation}\label{tilde Upsilon}
    \tilde\Upsilon_{k,i,\ell}=\frac{w_iw_{\ell}}{\Lambda_{\xi,i,i}^{1/2}\Lambda_{\xi,\ell,\ell}^{1/2}}\Upsilon_{k,i,\ell}+\frac{1}{\Lambda_{\xi,i,i}^{1/2}\Lambda_{\xi,\ell,\ell}^{1/2}}\frac{1}{n-k}\sum_{t=k+1}^n\mathbb{E}\{(u_{t,i}-\bar u_{i})(u_{t-k,\ell}-\bar u_{\ell})\}\,.
\end{equation}
Notice that $\Lambda_{\xi,i,i}=w_i^2\Upsilon_{0,i,i}+O(1)\asymp w_i^2$ by \eqref{eq: varxi-bar-xi}.  Then, taking $k=0$ in \eqref{tilde Upsilon}, since $w_r\ge C_3$ and $C_9^{-1}\le \sigma_r(\bUpsilon_0)\le \sigma_1(\bUpsilon_0)\le C_9$, we can conclude that 
   \[
   \sigma_1(\tilde\bUpsilon_0)\asymp \sigma_1(\bUpsilon_0)~~\textup{and}~~\sigma_r(\tilde\bUpsilon_0)\asymp \sigma_r(\bUpsilon_0)\,.
   \]
   Next, take $k=1$ in \eqref{tilde Upsilon} and note that $\{u_{t,i}\}_{t \ge 1}$ is serially uncorrelated. Hence, 
   \[
   \frac{1}{n-1}\sum_{t=2}^n\mathbb{E}\{(u_{t,i}-\bar u_{i})(u_{t-1,\ell}-\bar u_{\ell})\} = O(n^{-1})\,.
   \]
   It follows that
   \[
   \max_{i\ne \ell}|\tilde\Upsilon_{1,i,\ell}| = O\bigg(\gamma_{\max}+\frac{1}{nw_r^2}\bigg)\,,
   \]
   thus $|\tilde\bUpsilon_{1,i,-i}|_2=O(\gamma_{\max}+n^{-1}w_r^{-2})$.

   Further note that
   \[
   \frac{n}{n-1}\tilde\Upsilon_{0,i,\ell}-\frac{1}{n-1}\sum_{t=2}^n\mathbb{E}(\xi_{t,i}^{\textup{s}}\xi_{t,\ell}^{\textup{s}})=\frac{1}{n-1}\mathbb{E}(\xi_{1,i}^{\textup{s}}\xi_{1,\ell}^{\textup{s}})=O(n^{-1})\,.
   \]
   Combining the preceding display with the definition of $\tilde\bUpsilon_0$, we obtain
\begin{equation}\label{lower bound for FF}
       \sigma_{r-1} [\mathbb{E}  \{(n-1)^{-1}(\Fb_{\xi,\mminus i}^{\textup{s}})^{\T}\Fb_{\xi,\mminus i}^{\textup{s}}  \}  ]
       \asymp \sigma_{r-1}(\tilde\bUpsilon_{0,-i,-i})\,.
\end{equation}
   By \eqref{bound for F xi}, \eqref{lower bound for FF}, and the fact that $|\tilde\bUpsilon_{1,i,-i}|_2=O(\gamma_{\max}+n^{-1}w_r^{-2})$, we obtain
   \[
   \begin{split}
          \max_{i\ne \ell}|\bar\varphi_{i,\ell}| &\le \max_{i\in[r]}\bigg\|\bigg[\frac{1}{n-1}\mathbb{E}\{(\Fb_{\xi,\mminus i}^{\textup{s}})^{\T}\Fb_{\xi,\mminus i}^{\textup{s}}\}\bigg]^{-1}\bigg\|_2 \cdot \bigg|\frac{1}{n-1}\mathbb{E}\{(\Fb_{\xi,\mminus i}^{\textup{s}})^{\T}\bxi_i^{\textup{s}}\}\bigg|_2 \\
          &=O\bigg(\gamma_{\max}+\frac{1}{nw_r^2}\bigg)\,.
   \end{split}
   \]
 This completes the proof of Lemma \ref{lemma: E varphi}.
 $\hfill\Box$

\subsection{Proof of Lemma \ref{bound for Lambda and phi}}

In the proof, $\tilde C_0 > 0$ is a universal constant and may vary in different lines, but is independent of $\tilde C$ and $(i,j,\textit{v})$.
  To simplify the notation, given the $\textit{v}$-th round and the $j$-th mode, for $i\in[r]$,  define 
\begin{equation}\label{define check a}
  \check\ab_{i,j^\prime}^{(\textit{v},j)}
=
\left\{
\begin{array}{ll}
\tilde\ab_{i,j^\prime}^{(\textit{v})}\,, & j^\prime< j\,,\\
\tilde\ab_{i,j^\prime}^{(\textit{v}-1)}\,, & j\le j^\prime\le m\,,
\end{array}
\right. ~~\textup{and}~~ (\check\ab_{i,j^\prime}^{(\textit{v},j)})^{\MP}
=
\left\{
\begin{array}{ll}
(\tilde\ab_{i,j^\prime}^{(\textit{v})})^{\MP}\,, & j^\prime< j\,,\\
(\tilde\ab_{i,j^\prime}^{(\textit{v}-1)})^{\MP}\,, & j\le j^\prime\le m\,.
\end{array}
\right.
\end{equation}
Write $  \check\Ab_{j^\prime}^{(\textit{v},j)}
    =
    (
    \check\ab_{1,j^\prime}^{(\textit{v},j)},
    \ldots,
    \check\ab_{r,j^\prime}^{(\textit{v},j)}
    )$. Under the condition
$w_r^{-1}w_1\bar\theta_j^{(\textit{v})}\le \tilde C^{-2}$, we have
\[
    \max_{j^\prime\in[m]}
    \|\check\Ab_{j^\prime}^{(\textit{v},j)}-\Ab_{j^\prime}\|_2
    \le
    \tilde C_0\bar\theta_j^{(\textit{v})}
    \le \tilde C_0 \tilde C^{-2}\,.
\]
Together with Assumption \ref{sparsity}, Weyl's theorem implies that
$\sigma_r(\check\Ab_{j^\prime}^{(\textit{v},j)})$ is uniformly bounded away
from $0$ for all $j^\prime\in[m]$ when $\tilde C$ is sufficiently large. 
Hence,
\begin{equation}\label{eq:amp_hat_amp}
      \max_{i\in[r]}
    |
    (\check\ab_{i,j^\prime}^{(\textit{v},j)})^{\MP}
    -
    \ab_{i,j^\prime}^{\MP}
    |_2
    \le
    \tilde C_0\bar\theta_j^{(\textit{v})}\,,
    ~~~~
    j^\prime\in[m]\,.
\end{equation} 
Recall that
$$\check f_{t,i}^{(\textit{v},j)}= \{\otimes_{m}^{j^\prime=1}(\check\ab_{i,j^\prime}^{(\textit{v},j)})^{\MP}\}^\T \textup{vec}(\mathcal{Y}_t) \,,$$ 
 where $\otimes_{m}^{j^\prime=1} \bbeta_{j'}$  is shorthand for $\bbeta_m \otimes \cdots \otimes \bbeta_1$, and $\tilde f_{t,i}^{(\textit{v},j)} = (\check f_{t,i}^{(\textit{v},j)} - \bar{\check f}_{i}^{(\textit{v},j)})/\tilde{\sigma}_{\check f, i}^{(\textit{v},j)}$ with $\bar{\check f}_{i}^{(\textit{v},j)} = n^{-1}\sum_{t=1}^n \check f_{t,i}^{(\textit{v},j)}$ and
\begin{equation}\label{notation tilde sigma fi vj}
     (\tilde{\sigma}_{\check f, i}^{(\textit{v},j)})^2 = \frac{1}{n-1}\sum_{t=1}^n(\check f_{t,i}^{(\textit{v},j)} - \bar{\check f}_{i}^{(\textit{v},j)})^2 \,.
\end{equation}
  To prove the first assertion of Lemma \ref{bound for Lambda and phi},  we begin by showing that
\begin{equation}\label{check f ti minus xi ti}
       \frac{1}{n-1}\sum_{t=1}^n\{\check f_{t,i}^{(\textit{v},j)}-\xi_{t,i} - (\bar{\check f}_{i}^{(\textit{v},j)}-\bar\xi_i)\}^2\le \tilde C_0\tilde C^{-4} w_i^2\,.
\end{equation}
   Because $\textup{vec}(\mathcal{Y}_t)=\sum_{i=1}^r w_if_{t,i}(\otimes^{j=1}_m\ab_{i,j})+\textup{vec}(\mathcal{E}_t)$ while $\xi_{t,i}=w_if_{t,i}+(\otimes^{j=1}_m \ab^{\MP}_{i,j})^\T \textup{vec}(\mathcal{E}_t)$, we can write
\begin{equation}\label{check f ti minus xi ti decomp}
       \begin{split}
&   \check f_{t,i}^{(\textit{v},j)}-\xi_{t,i} - ( \bar{\check f}_{i}^{(\textit{v},j)}-\bar\xi_i) \\
   &~~~~~~~~= \bigg[\prod_{j^\prime=1}^m\{\ab_{i,j^\prime}^\T(\check\ab_{i,j^\prime}^{(\textit{v},j)})^{\MP} \}-1 \bigg]w_i(f_{t,i} - \bar f_i) \\
   &~~~~~~~~\quad + \sum_{\ell\ne i} \bigg[\prod_{j^\prime = 1}^m\{\ab_{\ell,j^\prime}^\T(\check\ab_{i,j^\prime}^{(\textit{v},j)})^{\MP} \} \bigg] w_\ell (f_{t,\ell} - \bar f_\ell)\\
   &~~~~~~~~\quad +[\otimes_{m}^{j^\prime=1}(\check\ab_{i,j^\prime}^{(\textit{v},j)})^{\MP}-\otimes_{m}^{j^\prime=1}\ab^{\MP}_{i,j^\prime}]^\T \{\textup{vec}(\mathcal{E}_t) - \textup{vec}(\bar{\mathcal{E}}) \}\,,
   \end{split}
\end{equation}
where $ \bar{\mathcal{E}}  = n^{-1}\sum_{t = 1}^n  \mathcal{E}_t $. We bound the three terms on the right-hand side of \eqref{check f ti minus xi ti decomp} separately.
 Notice that under the event $\Xi_{1,n}(\tilde C)$, as long as $\tilde C$ is large, we  have $\max_{i \in [r]} n^{-1}\sum_{t=1}^n(f_{t,i} - \bar f_i)^2 \le \tilde C_0$ for some constant $\tilde C_0 > 0$ that is independent of $\tilde C$ and $(i,j,\textit{v})$. Furthermore, by \eqref{eq:amp_hat_amp}, it holds that
\[
   \bigg| \prod_{j^\prime=1}^m\{\ab_{i,j^\prime}^\T(\check\ab_{i,j^\prime}^{(\textit{v},j)})^{\MP} \}-1 \bigg| \le \tilde C_0\bar\theta_j^{(\textit{v})}\le \tilde C_0\tilde C^{-2}w_1^{-1}w_r\,,
   \]
   and the same bound also holds for $|{\textstyle\prod\nolimits}_{j^\prime = 1}^m\{\ab_{\ell,j^\prime}^\T(\check\ab_{i,j^\prime}^{(\textit{v},j)})^{\MP} \}|$ with $\ell \neq i$.  
Let $ \check{\db}_{i}^{(\textit{v},j)}
    =
     \otimes_{m}^{j^\prime=1}
    (\check\ab_{i,j^\prime}^{(\textit{v},j)})^{\MP} 
    -
    \otimes_{m}^{j^\prime=1}\ab_{i,j^\prime}^{\MP}$. By the decomposition of Kronecker products and \eqref{eq:amp_hat_amp}, it holds that $ |\check{\db}_i^{(\textit{v},j)}|_2
    \le
    \tilde C_0\bar\theta_j^{(\textit{v})}$ for some constant $\tilde C_0 > 0$ that is independent of $\tilde C$ and $(i,j,\textit{v})$. 
Moreover, $\check{\db}_i^{(\textit{v},j)}$ can be written as the sum of  $m$
Kronecker product vectors, each with spectral norm bounded by
$\tilde C_0\bar\theta_j^{(\textit{v})}$. Hence,
\begin{equation}\label{check a minus a product E}
\begin{split}
     \frac{1}{n}\sum_{t=1}^n
    \{
    (\check{\db}_i^{(\textit{v},j)})^\T
    \textup{vec}(\mathcal{E}_t)
    \}^2 
    &=
    (\check{\db}_i^{(\textit{v},j)})^\T
    \frac{1}{n}\sum_{t=1}^n
    [
    \textup{vec}(\mathcal{E}_t)\textup{vec}(\mathcal{E}_t)^\T
    -
    \mathbb{E}\{
    \textup{vec}(\mathcal{E}_t)\textup{vec}(\mathcal{E}_t)^\T
    \}
    ]
    \check{\db}_i^{(\textit{v},j)}                                      \\
    &\quad+
    (\check{\db}_i^{(\textit{v},j)})^\T
    \frac{1}{n}\sum_{t=1}^n
    \mathbb{E}\{
    \textup{vec}(\mathcal{E}_t)\textup{vec}(\mathcal{E}_t)^\T
    \}
    \check{\db}_i^{(\textit{v},j)}                                      \\
    &\le
    \tilde C_0
    (\bar\theta_j^{(\textit{v})})^2
    (
    \tilde C L_n
    +
    1
    )        \le
    \tilde C_0\tilde C^{-4}w_r^2\,,
\end{split}
\end{equation}
where the last line is by the event $\Xi_{9,n}(\tilde C)$ and condition \eqref{strong factor condition}.  Moreover, by Jensen inequality,
\[
\begin{split}
    \{
    (\check{\db}_i^{(\textit{v},j)})^\T
    \textup{vec}(\bar{\mathcal{E}})
    \}^2
     \le
    \frac{1}{n}\sum_{t=1}^n
    \{
    (\check{\db}_i^{(\textit{v},j)})^\T
    \textup{vec}(\mathcal{E}_t)
    \}^2                                      
     \le
    \tilde C_0\tilde C^{-4}w_r^2\,.
\end{split}
\]
Combining the preceding bounds for the three terms on the right-hand side of \eqref{check f ti minus xi ti decomp}, and applying  Triangle inequality and Cauchy--Schwarz inequality, we can conclude \eqref{check f ti minus xi ti}.

   Moreover, under the event $\Xi_{5,n}(\tilde C)$ and by \eqref{eq: varxi-bar-xi}, we have
   \[
   \bigg|\frac{1}{n-1}\sum_{t=1}^n(\xi_{t,i}-\bar\xi_i)^2-\frac{1}{n}\sum_{t=1}^n\mathbb{E}\{(\xi_{t,i}-\bar\xi_i)^2\}\bigg|\le \tilde C_0\tilde C^{-1}w_i^2\,.
   \]
   Therefore, we can conclude that
\begin{equation}\label{tilde sigma fi vj}
       \bigg|(\tilde{\sigma}_{\check f, i}^{(\textit{v},j)})^2-\frac{1}{n}\sum_{t=1}^n\mathbb{E}\{(\xi_{t,i}-\bar\xi_i)^2\}\bigg|\le \tilde C_0\tilde C^{-1}w_i^2\,,
\end{equation}
and the first assertion of Lemma \ref{bound for Lambda and phi} holds because  $\Lambda_{\xi,i,i}=n^{-1}\sum_{t=1}^n\mathbb{E}\{(\xi_{t,i}-\bar\xi_i)^2\} \asymp w_i^2$  by \eqref{eq: varxi-bar-xi}.

We now turn to show $|\tilde{\bvarphi}^{(\textit{v},j)}_i-\bar\bvarphi_i|_2\le \tilde C_0\tilde C^{-1}$. To do this, define $\bvarphi_i =(\varphi_{i,1},\ldots,\varphi_{i,r})^{\T}$ with the $i$-th entry being 1 and the remaining entries are given by $-\{(\Fb_{\xi,\mminus i}^{\textup{s}})^{\T}\Fb_{\xi,\mminus i}^{\textup{s}}\}^{-1}(\Fb_{\xi,\mminus i}^{\textup{s}})^{\T}\bxi_i^{\textup{s}}$.   We first show $|\bvarphi_i-\bar\bvarphi_i|_2\le \tilde C_0\tilde C^{-1}$. By Triangle inequality, it holds that
	\[
	\begin{split}
		|\bvarphi_i-\bar\bvarphi_i|_2 &\le |\{(\Fb_{\xi,\mminus i}^{\textup{s}})^{\T}\Fb_{\xi,\mminus i}^{\textup{s}}\}^{-1}(\Fb_{\xi,\mminus i}^{\textup{s}})^{\T}\bxi_i^{\textup{s}}-[\mathbb{E}\{(\Fb_{\xi,\mminus i}^{\textup{s}})^{\T}\Fb_{\xi,\mminus i}^{\textup{s}}\}]^{-1}(\Fb_{\xi,\mminus i}^{\textup{s}})^{\T}\bxi_i^{\textup{s}}|_2\\
		&\quad+|[\mathbb{E}\{(\Fb_{\xi,\mminus i}^{\textup{s}})^{\T}\Fb_{\xi,\mminus i}^{\textup{s}}\}]^{-1}[(\Fb_{\xi,\mminus i}^{\textup{s}})^{\T}\bxi_i^{\textup{s}}-\mathbb{E}\{(\Fb_{\xi,\mminus i}^{\textup{s}})^{\T}\bxi_i^{\textup{s}}\}]|_2\,.
	\end{split}
	\]
    We only show how to bound the first term on the right-hand side, while the second term can be handled similarly. By \eqref{matrix inverse}, the first line can be bounded by
\[
\begin{split}
    &| \{(\Fb_{\xi,\mminus i}^{\textup{s}})^{\T}\Fb_{\xi,\mminus i}^{\textup{s}}\}^{-1}[(\Fb_{\xi,\mminus i}^{\textup{s}})^{\T}\Fb_{\xi,\mminus i}^{\textup{s}}-\mathbb{E}\{(\Fb_{\xi,\mminus i}^{\textup{s}})^{\T}\Fb_{\xi,\mminus i}^{\textup{s}}\}][\mathbb{E}\{(\Fb_{\xi,\mminus i}^{\textup{s}})^{\T}\Fb_{\xi,\mminus i}^{\textup{s}}\}]^{-1}(\Fb_{\xi,\mminus i}^{\textup{s}})^{\T}\bxi_i^{\textup{s}}|_2\,.
\end{split}
\]
On the event $\Xi_{3,n}(\tilde C) \cap \Xi_{4,n}(\tilde C)$, this term admits the upper bound
$\tilde C_0\tilde C^{-1}$ due to \eqref{bound for F  xi} and \eqref{lower bound for FF}. Handling the second line similarly, we conclude that 
    \begin{equation}\label{eq:bvarphi_i-Ebvarphi_i}
        |\bvarphi_i-\bar\bvarphi_i|_2 \le \tilde C_0\tilde{C}^{-1}.
    \end{equation}
Now,  Lemma \ref{bound for Lambda and phi} follows from Triangle inequality once we show
    \begin{equation}\label{zeta 1 decomp 3}
        |\tilde{\bvarphi}^{(\textit{v},j)}_i- \bvarphi_i |_2 \le \tilde C_0 \tilde{C}^{-1}.
    \end{equation}
     To this end, we should first show that
\begin{equation}\label{zeta 1 decomp 31}
	n^{-1}\|(\tilde\Fb_{\mminus i}^{(\textit{v},j)})^{\T}\tilde\Fb_{\mminus i}^{(\textit{v},j)}-(\Fb_{\xi,\mminus i}^{\textup{s}})^{\T}\Fb_{\xi,\mminus i}^{\textup{s}}\|_2 \le \tilde C_0\tilde C^{-1}\,,
\end{equation}
where $\tilde\Fb_{\mminus i}^{(\textit{v},j)}$ is defined similarly to $\tilde\Fb_{\mminus i}$ in \eqref{eq: tilde xi} by replacing $\check f_{t,i}$ with $\check f_{t,i}^{(\textit{v},j)}$. 
	By Cauchy--Schwarz inequality and Triangle inequality, it suffices to show that $n^{-1}\|\tilde\Fb_{\mminus i}^{(\textit{v},j)}-\Fb_{\xi,\mminus i}^{\textup{s}}\|^2_2 \le \tilde C_0 \tilde C^{-2}$, or sufficiently, 
\begin{equation}\label{zeta 1 decomp 32}
	\max_{i \in [r]}\frac{1}{n}\sum_{t=2}^n(\tilde f^{(\textit{v},j)}_{t,i}-\xi_{t,i}^{\textup{s}})^2\le \tilde C_0\tilde C^{-2}\,.
\end{equation}
	By the definition of $\tilde f^{(\textit{v},j)}_{t,i}$, write
	\[
	\begin{split}
\tilde f^{(\textit{v},j)}_{t,i}-\xi_{t,i}^{\textup{s}}&=\frac{1}{\tilde{\sigma}^{(\textit{v},j)}_{\check f, i}} ( \check f^{(\textit{v},j)}_{t,i}- \bar{\check f}^{(\textit{v},j)}_{i}-\xi_{t,i}+\bar\xi_i ) +\bigg( \frac{\Lambda_{\xi,i,i}^{1/2}}{\tilde{\sigma}^{(\textit{v},j)}_{\check f, i}}-1\bigg) \xi_{t,i}^{\textup{s}}\,.
	\end{split}
	\]
By \eqref{check f ti minus xi ti}, \eqref{tilde sigma fi vj}, and the first assertion of  Lemma \ref{bound for Lambda and phi} proved above, together with the fact that
$n^{-1}\sum_{t=1}^n(\xi_{t,i}^{\textup{s}})^2 \le \tilde{C}_0$ under
$\Xi_{5,n}(\tilde C)$, we obtain \eqref{zeta 1 decomp 32} and hence \eqref{zeta 1 decomp 31}.  Similarly, we also conclude that
	\[
		n^{-1}\|(\tilde\Fb^{(\textit{v},j)}_{\mminus i})^{\T}\tilde\fb^{(\textit{v},j)}_i-(\Fb_{\xi,\mminus i}^{\textup{s}})^{\T}\bxi_i^{\textup{s}}\|_2 \le \tilde C_0\tilde C^{-1}\,.
	\] 
	Then, by a procedure parallel to the proof of $|\bvarphi_i-\bar\bvarphi_i|_2\le \tilde C_0\tilde C^{-1}$, we can conclude (\ref{zeta 1 decomp 3}).   Then Lemma \ref{bound for Lambda and phi} holds.
$\hfill\Box$

\subsection{Proof of Lemma \ref{lem: zeta_1 lower bound}}

       In the proof, $\tilde C_0 > 0$ is a universal constant and may vary in different lines, but is independent of $\tilde C$ and $(i,j,\textit{v})$. Notice that
\begin{equation}\label{ab ij zeta 1}
    			\ab_{i,j}^\T\bzeta_1^{(i,j,\textit{v})} =\frac{w_i}{n-1}\sum_{t=2}^n(f_{t,i}-\bar f_i)\tilde\xi_{t-1,i}^{(\textit{v},j)} \{ \bbb_{i,j}^{\T}(\tilde\bbb_{i,j}^{(\textit{v})})^{\MP}\}\,.
\end{equation}
By the definition of $\tilde\xi_{t-1,i}^{(\textit{v},j)}$, write
			\begin{equation}\label{zeta 1 decomp}
			\begin{split}
				&\frac{w_i}{n-1}\sum_{t=2}^n(f_{t,i}-\bar f_i)\tilde\xi_{t-1,i}^{(\textit{v},j)}\\
				&~~~~~~ =\frac{w_i}{n-1}\sum_{t=2}^n(f_{t,i}-\bar f_i)\tilde f^{(\textit{v},j)}_{t-1,i}+\sum_{\ell\ne i}\frac{w_i\tilde\varphi_{i,\ell}^{(\textit{v},j)}}{n-1}\sum_{t=2}^n(f_{t,i}-\bar f_i)\tilde f^{(\textit{v},j)}_{t,\ell}\,.
			\end{split}
			\end{equation}
	As we will show in Section \ref{sec:lemma8-step-2}, the following three auxiliary bounds hold:
		\begin{align}
            & | (\tilde\bbb_{i,j}^{(\textit{v})})^{\MP}   - \bbb_{i,j}^{\MP}|_2 \le \tilde C_0\tilde C \bar\theta^{(\textit{v})}_j \le \tilde C_0\tilde C^{-1} \, \label{bij(v)-bij}\,,\\
	&	\bigg|	\frac{1}{n-1}\sum_{t=2}^n(f_{t,i}-\bar f_i)\tilde f^{(\textit{v},j)}_{t-1,i}-	\frac{1}{n-1}\sum_{t=2}^n(f_{t,i}-\bar f_i)\xi_{t-1,i}^{\textup{s}}\bigg|\le \tilde C_0\tilde C^{-1/2}\,,\label{zeta 1 decomp 1}\\
	&	\max_{\ell\ne i}\bigg|	\frac{1}{n-1}\sum_{t=2}^n(f_{t,i}-\bar f_i)\tilde f^{(\textit{v},j)}_{t,\ell}-	\frac{1}{n-1}\sum_{t=2}^n(f_{t,i}-\bar f_i)\xi_{t,\ell}^{\textup{s}}\bigg|\le \tilde C_0\tilde C^{-1/2}\,\label{zeta 1 decomp 2}.
		\end{align}
Notice that by \eqref{bij(v)-bij}, we have
			\[
			|\bbb_{i,j}^{\T}(\tilde\bbb_{i,j}^{(\textit{v})})^{\MP}-1|\le | \{(\tilde\bbb_{i,j}^{(\textit{v})})^{\MP}  - \bbb_{i,j}^{\MP}\}^{\T}\bbb_{i,j}|\le \tilde C_0\tilde C^{-1}\,.
			\]
			Then, by \eqref{ab ij zeta 1}, to complete the proof, it remains to bound \eqref{zeta 1 decomp}. By \eqref{zeta 1 decomp 1}, the event $\Xi_{2,n}(\tilde C)$, and Triangle inequality, we conclude that
            \begin{align*}
                 &\bigg|	\frac{1}{n-1}\sum_{t=2}^n(f_{t,i}-\bar f_i)\tilde f^{(\textit{v},j)}_{t-1,i}-\mathbb{E}\bigg\{\frac{1}{n-1}\sum_{t=2}^n(f_{t,i}-\bar f_i)\xi_{t-1,i}^{\textup{s}}\bigg\}\bigg|\\
    & ~~~~~~\le \bigg|	\frac{1}{n-1}\sum_{t=2}^n(f_{t,i}-\bar f_i)\tilde f^{(\textit{v},j)}_{t-1,i}-\frac{1}{n-1}\sum_{t=2}^n(f_{t,i}-\bar f_i)\xi_{t-1,i}^{\textup{s}}\bigg|\\
    &~~~~~~ \quad +\bigg|	\frac{1}{n-1}\sum_{t=2}^n(f_{t,i}-\bar f_i)\xi_{t-1,i}^{\textup{s}}-\mathbb{E}\bigg\{\frac{1}{n-1}\sum_{t=2}^n(f_{t,i}-\bar f_i)\xi_{t-1,i}^{\textup{s}}\bigg\}\bigg|\\
   & ~~~~~~  \le \tilde C_0\tilde C^{-1/2}\,.
            \end{align*}
Similarly,   we can also conclude that
\[
\max_{\ell\ne i}\bigg|\frac{1}{n-1}\sum_{t=2}^n(f_{t,i}-\bar f_i)\tilde f^{(\textit{v},j)}_{t,\ell}-\mathbb{E}\bigg\{\frac{1}{n-1}\sum_{t=2}^n(f_{t,i}-\bar f_i)\xi_{t,\ell}^{\textup{s}}\bigg\}\bigg|\le \tilde C_0\tilde C^{-1/2}.
\]
Further by Triangle inequality and Lemmas \ref{lemma: E varphi} and \ref{bound for Lambda and phi}, we have
         \[
         \begin{split}
             & \max_{\ell\ne i}\bigg|\frac{\tilde\varphi_{i,\ell}^{(\textit{v},j)}}{n-1}\sum_{t=2}^n(f_{t,i}-\bar f_i)\tilde f^{(\textit{v},j)}_{t,\ell}-\mathbb{E}\bigg\{\frac{\bar\varphi_{i,\ell}}{n-1}\sum_{t=2}^n(f_{t,i}-\bar f_i)\xi_{t,\ell}^{\textup{s}}\bigg\}\bigg|\\
           &~~~~~~\le \max_{\ell\ne i}\bigg|\frac{\tilde\varphi_{i,\ell}^{(\textit{v},j)}-\bar\varphi_{i,\ell}}{n-1}\sum_{t=2}^n(f_{t,i}-\bar f_i)\tilde f^{(\textit{v},j)}_{t,\ell}\bigg|\\
            &~~~~~~\quad +\max_{\ell\ne i}|\bar\varphi_{i,\ell}| \cdot \bigg|\frac{1}{n-1}\sum_{t=2}^n(f_{t,i}-\bar f_i)\tilde f^{(\textit{v},j)}_{t,\ell}-\mathbb{E}\bigg\{\frac{1}{n-1}\sum_{t=2}^n(f_{t,i}-\bar f_i)\xi_{t,\ell}^{\textup{s}}\bigg\}\bigg|\\
           & ~~~~~~\le  \tilde C_0\tilde C^{-1/2} 
         \end{split}   
         \]
         for sufficiently large $n$. Recall
         \[
         \sigma_{f_i,\xi_i} = \mathbb{E}\bigg\{\frac{1}{n-1}\sum_{t=2}^n(f_{t,i}-\bar f_i)\xi_{t-1,i}^{\textup{s}}\bigg\} + \mathbb{E}\bigg\{\sum_{\ell\ne i}\frac{ \bar\varphi_{i,\ell}}{n-1}\sum_{t=2}^n(f_{t,i}-\bar f_i)\xi_{t,\ell}^{\textup{s}}\bigg\}\,.
         \]
 Together with \eqref{zeta 1 decomp}, we can conclude that
\begin{equation}\label{zeta 1 final}
		\begin{split}
		&\bigg|\frac{w_i}{n-1}\sum_{t=2}^n(f_{t,i}-\bar f_i)\tilde\xi_{t-1,i}^{(\textit{v},j)} - w_i\sigma_{f_i,\xi_i} \bigg| \le \tilde C_0 \tilde C^{-1/2}w_i\,.
	\end{split}
\end{equation}
Combining \eqref{ab ij zeta 1}, \eqref{bij(v)-bij}, and
\eqref{zeta 1 final}, we have
\[
\begin{split}
    |\ab_{i,j}^{\T}\bzeta_1^{(i,j,\textit{v})} -  w_i\sigma_{f_i,\xi_i}|
    \le \tilde C_0 \tilde C^{-1/2}w_i \,.
\end{split}
\] 
Since $|\sigma_{f_i,\xi_i}|$ is uniformly bounded away from $0$ by
condition \eqref{strong factor condition}, it follows that
\[
    \bigg|
    \frac{\ab_{i,j}^{\T}\bzeta_1^{(i,j,\textit{v})}}
    {w_i\sigma_{f_i,\xi_i}}
    -
    1
    \bigg|
    \le
    \tilde C_0\tilde C^{-1/2}\,.
\]
Moreover, since $\bzeta_1^{(i,j,\textit{v})}$ is proportional to $\ab_{i,j}$, we have
\[
    \bigg|
    \frac{1}{w_i\sigma_{f_i,\xi_i}}
    \bzeta_1^{(i,j,\textit{v})}
    -
    \ab_{i,j}
    \bigg|_2
    \le
    \tilde C_0\tilde C^{-1/2}\,,
\]
which implies Lemma \ref{lem: zeta_1 lower bound}.  
$\hfill\Box$

\subsubsection{Proofs of \eqref{bij(v)-bij}--\eqref{zeta 1 decomp 2}}\label{sec:lemma8-step-2}
\underline{{\it Proof of \eqref{bij(v)-bij}.}} Following the notation in \eqref{define check a}, recall 
$$\{(\tilde\bbb_{1,j}^{(\textit{v})})^{\MP},\ldots,(\tilde\bbb_{r,j}^{(\textit{v})})^{\MP}\}^{\T}= \{(\tilde\Bb_j^{(\textit{v})})^\T \tilde\Bb_j^{(\textit{v})}\}^{-1}(\tilde\Bb_j^{(\textit{v})})^{\T}
\,,$$
where $\tilde \Bb_j^{(\textit{v})} =  (\tilde{\bbb}_{1,j}^{(\textit{v})},\ldots,\tilde{\bbb}_{r,j}^{(\textit{v})}
)$ with $\tilde\bbb_{i,j}^{(\textit{v})}= \otimes_{m}^{j^\prime \neq j} \check\ab_{i,j^\prime}^{(\textit{v},j)}$. Here $\otimes_{m}^{j^\prime \neq j} \bbeta_{j'} = \bbeta_m \otimes \cdots\otimes \bbeta_{j+1} \otimes \bbeta_{j-1} \otimes \cdots \otimes \bbeta_1$ for short. 
Write $\{(\tilde\Bb_j^{(\textit{v})})^\T \tilde\Bb_j^{(\textit{v})}\}^{-1} = (\tilde\varpi_{p,q}^{(\textit{v},j)})_{r\times r}$ and $(\Bb_j^\T \Bb_j)^{-1} = (\varpi_{p,q})_{r\times r}$. We can express $(\tilde\bbb_{i,j}^{(\textit{v})})^{\MP}$ and $\bbb_{i,j}^{\MP}$  as
\begin{equation}\label{eq: bij(b) expression}
    (\tilde\bbb_{i,j}^{(\textit{v})})^{\MP} = \sum_{\ell = 1}^r \tilde\varpi^{(\textit{v},j)}_{\ell,i}(\otimes_{m}^{j^\prime \neq j} \check\ab_{\ell,j^\prime}^{(\textit{v},j)}) ~~\textup{and}~~ \bbb_{i,j}^{\MP} = \sum_{\ell = 1}^r \varpi_{\ell,i}(\otimes_{m}^{j^\prime \neq j} \ab_{\ell,j^\prime})\,.
\end{equation}
Notice that
\[
\|\tilde\Bb_j^{(\textit{v})} - \Bb_j\|_2 \le  \tilde C_0 \max_{i \in [r]}|\tilde\bbb_{i,j}^{(\textit{v})} - \bbb_{i,j}|_2\\
\le  \tilde C_0 \bar\theta^{(\textit{v})}_j \le \tilde C_0 \tilde C^{-1}\,. 
\]
By Triangle inequality and \eqref{matrix inverse}, it holds that
\[
\begin{split}
     \max_{i,\ell \in [r]}|\tilde\varpi^{(\textit{v},j)}_{\ell,i} - \varpi_{\ell,i}| &\le \|\{(\tilde\Bb_j^{(\textit{v})})^\T \tilde\Bb_j^{(\textit{v})}\}^{-1} - (\Bb_j^\T \Bb_j)^{-1}\|_2 \\
     &\le \|\{(\tilde\Bb_j^{(\textit{v})})^\T \tilde\Bb_j^{(\textit{v})}\}^{-1}\|_2 \cdot \|(\tilde\Bb_j^{(\textit{v})})^\T \tilde\Bb_j^{(\textit{v})} - \Bb_j^\T \Bb_j\|_2 \cdot \| (\Bb_j^\T \Bb_j)^{-1}\|_2 \\
     & \le \tilde C_0\tilde C \bar\theta^{(\textit{v})}_j \le \tilde C_0\tilde C^{-1} \,.
\end{split}
\]
Hence, it follows that
\[
\begin{split}
    |(\tilde\bbb_{i,j}^{(\textit{v})})^{\MP} - \bbb_{i,j}^{\MP}|_2 &\le \bigg|\sum_{\ell = 1}^r \tilde\varpi^{(\textit{v},j)}_{\ell,i}(\otimes_{m}^{j^\prime \neq j} \check\ab^{(\textit{v},j)}_{\ell,j^\prime}-\otimes_{m}^{j^\prime \neq j} \ab_{\ell,j^\prime}) \bigg|_2 \\
 &\quad + \bigg|\sum_{\ell = 1}^r (\tilde\varpi^{(\textit{v},j)}_{\ell,i} - \varpi_{\ell,i})(\otimes_{m}^{j^\prime \neq j} \ab_{\ell,j^\prime}) \bigg|_2 \\
 & \le \tilde C_0\tilde C \bar\theta^{(\textit{v})}_j \le \tilde C_0 \tilde C^{-1}\,,
\end{split}
\]
which implies \eqref{bij(v)-bij}.

\underline{{\it Proof of \eqref{zeta 1 decomp 1}.}}  Following the notation from \eqref{define check a} to \eqref{notation tilde sigma fi vj},  it holds that
			\begin{equation}\label{zeta 1 decomp 11}
				\begin{split}
				\text{left-hand side of \eqref{zeta 1 decomp 1}} &\le \bigg| \frac{(\tilde{\sigma}^{(\textit{v},j)}_{\check f, i})^{-1}}{n-1}\sum_{t=2}^n(f_{t,i}-\bar f_i)(\check f^{(\textit{v},j)}_{t-1,i}-\bar{\check f}^{(\textit{v},j)}_{i}-\xi_{t-1,i}+\bar\xi_i)\bigg|\\
				&\quad +\bigg|\bigg(\frac{\Lambda_{\xi,i,i}^{1/2}}{\tilde{\sigma}^{(\textit{v},j)}_{\check f, i}}-1\bigg) \frac{1}{n-1}\sum_{t=2}^n(f_{t,i}-\bar f_i)\xi_{t-1,i}^{\textup{s}}\bigg|\,.
			\end{split}
			\end{equation}
            By Lemma \ref{bound for Lambda and phi}, we already have
			\begin{equation*} 
			 \bigg| \frac{\Lambda_{\xi,i,i}^{1/2}}{\tilde{\sigma}^{(\textit{v},j)}_{\check f, i}}-1 \bigg|	\le \tilde C_0  \tilde{C}^{-1}\,.
			\end{equation*}
            Meanwhile, $|(n-1)^{-1}\sum_{t=2}^n(f_{t,i}-\bar f_i)\xi_{t-1,i}^{\textup{s}}|$ is upper bounded by some constant under the event $\Xi_{2,n}(\tilde C)$. Therefore, the second line of \eqref{zeta 1 decomp 11} is upper bounded by $\tilde C_0\tilde C^{-1}$. Moreover, by  \eqref{eq: varxi-bar-xi}, \eqref{check f ti minus xi ti}, and Cauchy--Schwarz inequality, we have
           \[
\begin{split}
&\bigg|
\frac{\Lambda_{\xi,i,i}^{-1/2}}{n-1}
\sum_{t=2}^n
(f_{t,i}-\bar f_i)
(
\check f^{(\textit{v},j)}_{t-1,i} -
\bar{\check f}^{(\textit{v},j)}_{i} - \xi_{t-1,i} + \bar\xi_i )
\bigg|  \\
&~~~~~~\le
|\Lambda_{\xi,i,i}^{-1/2}|
\bigg\{
\frac{1}{n-1}\sum_{t=2}^n
(f_{t,i}-\bar f_i)^2
\bigg\}^{1/2}
\bigg\{
\frac{1}{n-1}\sum_{t=2}^n
(
\check f^{(\textit{v},j)}_{t-1,i}
-
\bar{\check f}^{(\textit{v},j)}_{i}
-
\xi_{t-1,i}
+
\bar\xi_i
)^2
\bigg\}^{1/2}  \\
&~~~~~~\le
\tilde C_0\tilde C^{-1/2}\,.
\end{split}
\]
            Therefore, the first line of \eqref{zeta 1 decomp 11} is upper bounded by $\tilde C_0\tilde C^{-1/2}$. Then, we conclude \eqref{zeta 1 decomp 1}.


	\underline{{\it Proof of \eqref{zeta 1 decomp 2}.}}  It is similar to the proof of (\ref{zeta 1 decomp 1}), so we omit the details. 
$\hfill\Box$

\subsection{Proof of Lemma \ref{lem: zeta_2 and zeta_4}}\label{sec:lem: zeta_2 and zeta_4}

In the proof, $\tilde C_0 > 0$ is a universal constant and may vary in different lines, but is independent of $\tilde C$ and $(i,j,\textit{v})$. Steps 1 and 2 establish  the upper bounds for $|\bzeta^{(i,j,v)}_2|_2$ and $|\bzeta^{(i,j,v)}_4|_2$, respectively.

\subsubsection{Step 1: Upper bound of \texorpdfstring{$|\bzeta^{(i,j,v)}_2|_2$}{zeta_2}}
 
By the definition of $\bzeta^{(i,j,v)}_2$ above (\ref{iterative sigma decompose}) and Cauchy--Schwarz inequality, we have
				\[
				\begin{split}
					|\bzeta_2^{(i,j,\textit{v})}|_2 &\le \underbrace{\bigg|\sum_{\ell\ne i}\tilde{\sigma}^{(\textit{v},j)}_{\check f,\ell} \bigg\{ \frac{1}{n-1}\sum_{t=2}^n \tilde f^{(\textit{v},j)}_{t,\ell} \tilde\xi^{(\textit{v},j)}_{t-1,i}\bigg\} \{\bbb_{\ell,j}^{\T}(\tilde\bbb_{i,j}^{(\textit{v})})^{\MP} \}\bigg|}_{|\zeta_{2,1}^{(i,j,\textit{v})}|} \\
					&\quad +\underbrace{\bigg|\sum_{\ell\ne i}\frac{1}{n-1}\sum_{t=2}^n\{\tilde{\sigma}^{(\textit{v},j)}_{\check f, \ell}\tilde f^{(\textit{v},j)}_{t,\ell}-w_{\ell}(f_{t,\ell} - \bar f_{\ell})\}\tilde\xi^{(\textit{v},j)}_{t-1,i}\{\bbb_{\ell,j}^{\T}(\tilde\bbb_{i,j}^{(\textit{v})})^{\MP} \}\bigg|}_{|\zeta_{2,2}^{(i,j,\textit{v})}|}\,.
				\end{split}
				\]
			By the construction of $(\tilde\xi^{(\textit{v},j)}_{1,i},\ldots,\tilde\xi^{(\textit{v},j)}_{n-1,i})^{\T}$, we always have 
   \[
   \frac{1}{n-1}\sum_{t=2}^n \tilde f^{(\textit{v},j)}_{t,\ell} \tilde\xi^{(\textit{v},j)}_{t-1,i} = 0
   \]
   for all $\ell\ne i$, which implies $|\zeta_{2,1}^{(i,j,\textit{v})}|=0$. For $|\zeta_{2,2}^{(i,j,\textit{v})}|$, under the event $\Xi_n(\tilde C)$, it follows from \eqref{bij(v)-bij} that
			\[
			\max_{\ell \ne i}|\bbb_{\ell,j}^{\T}(\tilde\bbb_{i,j}^{(\textit{v})})^{\MP} | \le \tilde C_0  \tilde C\bar\theta^{(\textit{v})}_j\,.
			\]
			For any $\ell \neq i$, we have
\begin{align}\label{check f ti minus}
    	\tilde{\sigma}^{(\textit{v},j)}_{\check f, \ell}\tilde f^{(\textit{v},j)}_{t,\ell}-w_{\ell}(f_{t,\ell} - \bar f_{\ell})&= \check f^{(\textit{v},j)}_{t,\ell}-w_{\ell}f_{t,\ell}-\frac{1}{n}\sum_{s=1}^n(\check f^{(\textit{v},j)}_{s,\ell}-w_{\ell} f_{s,\ell}) \\ \nonumber
                    &= w_{\ell}( f_{t,\ell} -\bar f_{\ell}) \bigg[\prod_{j^\prime=1}^m\{\ab_{\ell,j^\prime}^{\T}(\check\ab^{(\textit{v},j)}_{\ell,j^\prime})^{\MP}\}-1 \bigg] \\\nonumber
                    & \quad +\sum_{\ell^\prime \ne \ell}w_{\ell^\prime}(f_{t,\ell^\prime} - \bar f_{\ell^\prime}) \bigg[\prod_{j^\prime=1}^m\{\ab_{\ell^\prime,j^\prime}^{\T}(\check\ab^{(\textit{v},j)}_{\ell,j^\prime})^{\MP}\} \bigg] \\\nonumber
     & \quad + \{\otimes^{j^\prime =1}_{m} (\check\ab^{(\textit{v},j)}_{\ell,j^\prime})^{\MP}\}^{\T}\textup{vec}(\mathcal{E}_t - \bar{\mathcal{E}})\,,
\end{align}
			where $(\check\ab^{(\textit{v},j)}_{\ell,1})^{\MP},\ldots,(\check\ab^{(\textit{v},j)}_{\ell,m})^{\MP}$ are defined in \eqref{define check a}. Then, under the event $\Xi_n(\tilde C)$, by Triangle inequality and Cauchy--Schwarz inequality, it holds that
\begin{align}\label{check f ti minus step 2}
 	&\bigg|\sum_{\ell\ne i}\frac{1}{n-1}\sum_{t=2}^n w_{\ell} (f_{t,\ell} - \bar f_{\ell}) \bigg[\prod_{j^\prime=1}^m\{\ab_{\ell,j^\prime}^{\T}(\check\ab^{(\textit{v},j)}_{\ell,j^\prime})^{\MP}\}-1 \bigg] \tilde\xi^{(\textit{v},j)}_{t-1,i}\{\bbb_{\ell,j}^{\T}(\tilde\bbb_{i,j}^{(\textit{v})})^{\MP} \}\bigg| \nonumber \\ \nonumber
					 & ~~~~  \le \sum_{\ell \neq i}\bigg| \frac{1}{n-1}\sum_{t=2}^n w_{\ell} (f_{t,\ell} - \bar f_{\ell})\tilde\xi^{(\textit{v},j)}_{t-1,i} \bigg| \cdot \bigg|\prod_{j^\prime=1}^m\{\ab_{\ell,j^\prime}^{\T}(\check\ab^{(\textit{v},j)}_{\ell,j^\prime})^{\MP}\}-1 \bigg| \cdot |\bbb_{\ell,j}^{\T}(\tilde\bbb_{i,j}^{(\textit{v})})^{\MP} | \\\nonumber
       &  ~~~~ \le \tilde C_0\tilde C (\bar\theta^{(\textit{v})}_j)^2  \sum_{\ell \neq i} w_{\ell}   \bigg\{\frac{1}{n-1}\sum_{t=2}^n (f_{t,\ell} - \bar f_{\ell})^2\bigg\}^{1/2 }\cdot  \bigg\{\frac{1}{n-1}\sum_{t=2}^n(\tilde\xi^{(\textit{v},j)}_{t-1,i})^2 \bigg\}^{1/2} \\  
       & ~~~~ \le \tilde C_0\tilde Cw_1(\bar\theta^{(\textit{v})}_j)^2 \,, 
\end{align}
		where we use the fact that $(n-1)^{-1}\sum_{t=2}^n (\tilde\xi^{(\textit{v},j)}_{t-1,i})^2 \le (n-1)^{-1}\sum_{t=2}^n(\tilde f^{(\textit{v},j)}_{t-1,i})^2 \le \tilde{C}_0$. Similarly, 
			\[
			\bigg|\sum_{\ell\ne i}\frac{1}{n-1}\sum_{t=2}^n\sum_{\ell^{\prime}\ne \ell}w_{\ell^{\prime}}(f_{t,\ell^{\prime}} - \bar f_{\ell^{\prime}}) \bigg[\prod_{j^\prime=1}^m\{\ab_{\ell^{\prime}, j^\prime}^{\T}(\check\ab^{(\textit{v},j)}_{\ell,j^\prime})^{\MP}\} \bigg] \tilde\xi^{(\textit{v},j)}_{t-1,i}\{\bbb_{\ell,j}^{\T}(\tilde\bbb_{i,j}^{(\textit{v})})^{\MP} \}\bigg|\le \tilde C_0\tilde Cw_1(\bar\theta^{(\textit{v})}_j)^{m+1}\,.
			\]
Moreover, for each $\ell \neq i$, by the definition of $\tilde\bvarphi^{(\textit{v},j)}_{i}$ given above \eqref{iterative sigma decompose}, we can write
	\begin{equation}\label{zeta 220}
				\begin{split}
			&\frac{1}{n-1}\sum_{t=2}^n[\{\otimes^{j^\prime=1}_{m} (\check\ab^{(\textit{v},j)}_{\ell,j^\prime})^{\MP}\}^{\T}\textup{vec}(\mathcal{E}_t - \bar{\mathcal{E}})]\tilde\xi^{(\textit{v},j)}_{t-1,i}\\
        &    ~~~~~~~~~~~~ = \frac{1}{n-1}\sum_{t=2}^n[\{\otimes^{j^\prime=1}_{m} (\check\ab^{(\textit{v},j)}_{\ell,j^\prime})^{\MP}\}^{\T}\textup{vec}(\mathcal{E}_t - \bar{\mathcal{E}})]\tilde f^{(\textit{v},j)}_{t-1,i}\\
 & ~~~~~~~~~~~~\quad +\sum_{i^{\prime}\ne i}\frac{\tilde\varphi^{(\textit{v},j)}_{i,i^{\prime}}}{n-1}\sum_{t=2}^n[\{\otimes^{j^\prime=1}_{m} (\check\ab^{(\textit{v},j)}_{\ell,j^\prime})^{\MP}\}^{\T}\textup{vec}(\mathcal{E}_t - \bar{\mathcal{E}})]\tilde f^{(\textit{v},j)}_{t,i^{\prime}}\,.
		\end{split}
	\end{equation}
On the one hand, under the event $\cap_{h=6}^8\Xi_{h,n}(\tilde C)$ for sufficiently large $\tilde C > 0$, by \eqref{check f ti minus}, Triangle inequality, and the fact $\tilde{\sigma}^{(\textit{v},j)}_{\check f, i}  \gtrsim w_i$, it follows that
 \begin{align}\label{zeta 221}
   \nonumber   &\bigg|\frac{1}{n-1}\sum_{t=2}^n[\{\otimes^{j^\prime=1}_{m} (\check\ab^{(\textit{v},j)}_{\ell,j^\prime})^{\MP}\}^{\T}\textup{vec}(\mathcal{E}_t - \bar{\mathcal{E}})]\tilde f^{(\textit{v},j)}_{t-1,i}\bigg|\\
      \nonumber    &\quad \le \bigg|\frac{w_{i}}{\tilde{\sigma}^{(\textit{v},j)}_{\check f, i} }\frac{1}{n-1}\sum_{t=2}^n[\{\otimes^{j^\prime=1}_{m} (\check\ab^{(\textit{v},j)}_{\ell,j^\prime})^{\MP}\}^{\T}\textup{vec}(\mathcal{E}_t - \bar{\mathcal{E}})](f_{t-1,i}-\bar f_{i})\bigg|\\
   \nonumber        &\qquad + \bigg|\frac{w_{i}}{\tilde{\sigma}^{(\textit{v},j)}_{\check f, i} }\frac{1}{n-1}\sum_{t=2}^n[\{\otimes^{j^\prime=1}_{m} (\check\ab^{(\textit{v},j)}_{\ell,j^\prime})^{\MP}\}^{\T}\textup{vec}(\mathcal{E}_t - \bar{\mathcal{E}})](f_{t-1,i}-\bar f_{i})\bigg[\prod_{j^\prime=1}^m\{\ab_{i,j^\prime}^{\T}(\check\ab^{(\textit{v},j)}_{i,j^\prime})^{\MP}\} - 1 \bigg]\bigg|\\
 \nonumber 		&\qquad + \bigg|\sum_{i^{\prime}\neq i}\frac{w_{i^{\prime}}}{\tilde{\sigma}^{(\textit{v},j)}_{\check f, i} }\frac{1}{n-1}\sum_{t=2}^n[\{\otimes^{j^\prime=1}_{m} (\check\ab^{(\textit{v},j)}_{\ell,j^\prime})^{\MP}\}^{\T}\textup{vec}(\mathcal{E}_t - \bar{\mathcal{E}})](f_{t-1,i^{\prime}}-\bar f_{i^{\prime}}) \bigg[\prod_{j^\prime=1}^m\{\ab_{i^{\prime},j^\prime}^{\T}(\check\ab^{(\textit{v},j)}_{i,j^\prime})^{\MP}\}\bigg] \bigg|\\
 \nonumber 		&\qquad+\bigg|\frac{1}{\tilde{\sigma}^{(\textit{v},j)}_{\check f, i} }\frac{1}{n-1}\sum_{t=2}^n [\{\otimes^{j^\prime=1}_{m} (\check\ab^{(\textit{v},j)}_{\ell,j^\prime})^{\MP}\}^{\T}\textup{vec}(\mathcal{E}_t - \bar{\mathcal{E}}) ] [\{\otimes^{j^\prime=1}_{m} (\check\ab^{(\textit{v},j)}_{i,j^\prime})^{\MP}\}^{\T}\textup{vec}(\mathcal{E}_{t-1}- \bar{\mathcal{E}}) ]\bigg| \\
         &\quad \le  \tilde C_0\tilde C L_n  + \tilde C_0\tilde C \frac{ w_1}{w_r}L_n(\bar\theta^{(\textit{v})}_j)^m+  \frac{\tilde C_0\tilde C}{w_r}L_n^2 \,. 
 \end{align}

On the other hand, similarly, for $i^\prime \neq i$, it follows that
\begin{align}\label{zeta 221-part2}
    &\bigg|\frac{1}{n-1}\sum_{t=2}^n [\{\otimes^{j^\prime=1}_{m} (\check\ab^{(\textit{v},j)}_{\ell,j^\prime})^{\MP}\}^{\T}\textup{vec}(\mathcal{E}_t - \bar{\mathcal{E}})]\tilde f^{(\textit{v},j)}_{t,i^{\prime}}\bigg|\\ \nonumber
	 &~~~~~~\le  \tilde C_0  \tilde C L_n
     + \tilde C_0  \tilde C \frac{  w_1}{w_r}L_n(\bar\theta^{(\textit{v})}_j)^m +  \frac{\tilde C_0  \tilde C}{w_r}L_n^2 \\ \nonumber
	&~~~~~~\quad +\frac{\tilde C_0 \tilde C}{w_{r}}\bigg|\frac{1}{n-1}\sum_{t=2}^n\{\otimes^{j^\prime=1}_m(\check\ab^{(\textit{v},j)}_{\ell,j^\prime})^{\MP}\}^{\T}\mathbb{E}\{\text{vec}(\mathcal{E}_t)\text{vec}(\mathcal{E}_t)^{\T}\} \{\otimes^{j^\prime=1}_m(\check\ab^{(\textit{v},j)}_{i^{\prime},j^\prime})^{\MP}\}\bigg|\\ \nonumber
	 &~~~~~~ \le \tilde C_0  \tilde C L_n 
     + \tilde C_0  \tilde C  \frac{ w_1}{w_r}L_n(\bar\theta^{(\textit{v})}_j)^m+  \frac{\tilde C_0  \tilde C}{w_r}L_n^2 +\frac{\tilde C_0  \tilde C}{w_r}\,,
\end{align}
			where the major difference from \eqref{zeta 221} is that the contemporary covariance of the idiosyncratic errors is not zero, which is from the estimation error of the factors in the projection step. Then, to show the upper bound of $|\zeta_{2,2}^{(i,j,\textit{v})}|$, it remains to calculate $\max_{\ell\ne i}|\tilde\varphi_{i,\ell}^{(\textit{v},j)}|$. Different from Lemma \ref{bound for Lambda and phi}, here we aim to find a more accurate bound for $\max_{\ell\ne i}|\tilde\varphi_{i,\ell}^{(\textit{v},j)}|$.
			
			By  \eqref{zeta 1 decomp 31}  and $\Xi_{3,n}(\tilde C)$, we can  conclude that $\|\{n^{-1}(\tilde\Fb^{(\textit{v},j)}_{\mminus i})^{\T}\tilde\Fb^{(\textit{v},j)}_{\mminus i}\}^{-1}\|_2 \le \tilde C_0$. In the following, we aim to show that
	\begin{equation}\label{zeta 222}
    \begin{split}
     |n^{-1}(\tilde\Fb^{(\textit{v},j)}_{\mminus i})^{\T}\tilde\fb^{(\textit{v},j)}_{i}|_2 &\le   \tilde C_0\tilde C\bigg\{\gamma_{\max}+ \Delta_{0,n}+\frac{1}{n w_r^2}\\
    & \qquad \qquad +\frac{w_1}{w_r} (\bar\theta^{(\textit{v})}_j)^m  + \frac{w_1}{w_r^2}L_n\bar\theta^{(\textit{v})}_j  + \frac{1}{w_r^2}L_n^2\bar\theta^{(\textit{v})}_j \bigg\} 
    \end{split}
	\end{equation}
 which is then also the upper bound of $\max_{\ell\ne i}|\tilde\varphi_{i,\ell}^{(\textit{v},j)}|$ since $\tilde C_0$ is a universal positive constant. Because $r$ is fixed, it suffices to bound each entry. By definition, given $\ell\ne i$,
\begin{equation}\label{zeta 223}
	\begin{split}
	\frac{1}{n-1}\sum_{t=2}^n\tilde f^{(\textit{v},j)}_{t,\ell}\tilde f^{(\textit{v},j)}_{t-1,i}=
		\frac{1}{\tilde{\sigma}^{(\textit{v},j)}_{\check f, i}\tilde{\sigma}^{(\textit{v},j)}_{\check f, \ell}}\frac{1}{n-1}\sum_{t=2}^n(\check f^{(\textit{v},j)}_{t,\ell}- \bar{\check f}^{(\textit{v},j)}_{\ell})(\check f^{(\textit{v},j)}_{t-1,i}-\bar{\check f}^{(\textit{v},j)}_{i})\,.
	\end{split}
\end{equation}
We aim to replace all the $\check f^{(\textit{v},j)}_{t,i}$ with $\xi_{t,i}$ and bound the error. By the definition of $\xi_{t,i}$, we have
\[
\mathbb{E}\bigg\{\frac{1}{n-1}\sum_{t=2}^n(\xi_{t,\ell}-\bar\xi_{\ell})(\xi_{t-1,i}-\bar\xi_i)\bigg\}=w_iw_{\ell}\Upsilon_{1,\ell,i}+O(n^{-1})\,.
\]
Meanwhile, by $\Xi_{4,n}(\tilde C)$ and \eqref{eq: varxi-bar-xi}, we also have
\[
\begin{split}
	&\bigg|\frac{1}{n-1}\sum_{t=2}^n(\xi_{t,\ell}-\bar\xi_{\ell})(\xi_{t-1,i}-\bar\xi_i)-\mathbb{E}\bigg\{\frac{1}{n-1}\sum_{t=2}^n(\xi_{t,\ell}-\bar\xi_{\ell})(\xi_{t-1,i}-\bar\xi_i)\bigg\}\bigg|\\
	&~~~~~~ \le \tilde C_0 \tilde C\Lambda_{\xi,i,i}^{1/2}\Lambda_{\xi,\ell,\ell}^{1/2} \Delta_{0,n} \le \tilde C_0\tilde Cw_iw_{\ell} \Delta_{0,n} \,.
\end{split}
\]
To bound \eqref{zeta 223}, it remains to control 
\[
\bigg| \frac{1}{n-1}\sum_{t=2}^n(\check f^{(\textit{v},j)}_{t,\ell}- \bar{\check f}^{(\textit{v},j)}_{\ell})(\check f^{(\textit{v},j)}_{t-1,i}-\bar{\check f}^{(\textit{v},j)}_{i}) - \frac{1}{n-1}\sum_{t=2}^n(\xi_{t,\ell}-\bar\xi_{\ell})(\xi_{t-1,i}-\bar\xi_i)\bigg|\,.
\]
By the definition of $\Delta_{0,n}$, we have
\[
\max_{i\ne \ell}\bigg|\frac{1}{n-1}\sum_{t=2}^n(f_{t,\ell}-\bar f_\ell)(f_{t-1,i}-\bar f_i)\bigg| \le \gamma_{\max} + \Delta_{0,n}\,.
\]
Furthermore, by \eqref{check f ti minus xi ti decomp}, together with the arguments similar to those used from \eqref{check f ti minus step 2} to \eqref{zeta 221}, we can conclude that
\begin{align*}
    &\bigg|\frac{1}{n-1}\sum_{t=2}^n(\check f^{(\textit{v},j)}_{t,\ell}- \bar{\check f}^{(\textit{v},j)}_{\ell} - \xi_{t,\ell} + \bar{\xi}_\ell ) (\check f^{(\textit{v},j)}_{t-1,i} - \bar{\check f}^{(\textit{v},j)}_{i} -\xi_{t-1,i} + \bar{\xi}_i)\bigg| \\
& ~~~~~~ \le w_iw_{\ell} \tilde C_0\tilde C \bigg\{(\gamma_{\max} + \Delta_{0,n})(\bar\theta^{(\textit{v})}_j)^2 + \frac{w_1}{w_r} (\bar\theta^{(\textit{v})}_j)^m + \frac{w_1}{w_r^2} L_n(\bar\theta^{(\textit{v})}_j)^2 + \frac{1}{w_r^2} L_n^2(\bar\theta^{(\textit{v})}_j)^2 \bigg\}\,, \\
&\bigg|\frac{1}{n-1}\sum_{t=2}^n(\check f^{(\textit{v},j)}_{t,\ell}- \bar{\check f}^{(\textit{v},j)}_{\ell} - \xi_{t,\ell} + \bar{\xi}_\ell)(\xi_{t-1,i} - \bar\xi_i)\bigg|\\
& ~~~~~~\le w_iw_{\ell}\tilde C_0 \tilde C \bigg\{(\gamma_{\max} + \Delta_{0,n})\bar\theta^{(\textit{v})}_j + \frac{w_1}{w_r} (\bar\theta^{(\textit{v})}_j)^m + \frac{w_1}{w_r^2} L_n\bar\theta^{(\textit{v})}_j + \frac{1}{w_r^2} L_n^2\bar\theta^{(\textit{v})}_j \bigg\}\,,\\
&\bigg|\frac{1}{n-1}\sum_{t=2}^n(\xi_{t,\ell}-\bar\xi_\ell)(\check f^{(\textit{v},j)}_{t-1,i} - \bar{\check f}^{(\textit{v},j)}_{i} -\xi_{t-1,i} + \bar{\xi}_i)\bigg| \\
& ~~~~~~\le w_iw_{\ell} \tilde C_0\tilde C \bigg\{(\gamma_{\max} + \Delta_{0,n})\bar\theta^{(\textit{v})}_j + \frac{w_1}{w_r} (\bar\theta^{(\textit{v})}_j)^m + \frac{w_1}{w_r^2} L_n\bar\theta^{(\textit{v})}_j + \frac{1}{w_r^2} L_n^2\bar\theta^{(\textit{v})}_j \bigg\}\,.
\end{align*}
Returning to (\ref{zeta 223}), now we can conclude that
\[
\begin{split}
	\bigg|\frac{1}{n-1}\sum_{t=2}^n\tilde f^{(\textit{v},j)}_{t,\ell}\tilde f^{(\textit{v},j)}_{t-1,i}\bigg|  & \le \tilde C_0\tilde C\bigg\{\gamma_{\max}+\Delta_{0,n}+\frac{1}{nw_r^2} + \frac{w_1}{w_r} (\bar\theta^{(\textit{v})}_j)^m  + \frac{w_1}{w_r^2}L_n\bar\theta^{(\textit{v})}_j + \frac{1}{w_r^2}L_n^2\bar\theta^{(\textit{v})}_j\bigg\}\,.
\end{split}
\]
Hence, \eqref{zeta 222} holds, and the same bound holds for $\max_{\ell\ne i}|\tilde\varphi_{i,\ell}^{(\textit{v},j)}|$. Now return to \eqref{zeta 220}, by \eqref{zeta 221}--\eqref{zeta 222}, it follows that
\begin{equation}\label{zeta 224}
	\begin{split}
		&\bigg|\frac{1}{n-1}\sum_{t=2}^n[\{\otimes^{j^\prime=1}_{m} (\check\ab^{(\textit{v},j)}_{\ell,j^\prime})^{\MP}\}^{\T}\textup{vec}(\mathcal{E}_t - \bar{\mathcal{E}})]\tilde\xi^{(\textit{v},j)}_{t-1,i}\bigg|\\
		  & ~~~~~~\le  \tilde C_0\tilde C \bigg\{ \frac{ w_1}{w_r}L_n\bar\theta^{(\textit{v})}_j+ L_n + \frac{1}{w_r} L_n^2   + \frac{w_1}{w_r^2} (\bar\theta^{(\textit{v})}_j)^m + \frac{\gamma_{\max} + \Delta_{0,n}}{w_r} +  \frac{1}{n w_r^3}  \bigg\}\,.
	\end{split}
\end{equation}
Therefore, we can eventually conclude that
		\begin{equation*}
            \begin{split}                				|\zeta_{2,2}^{(i,j,\textit{v})}| &\le \tilde C_0 \tilde C\bigg(w_1\bar\theta^{(\textit{v})}_j+\frac{w_1}{w_r}L_n\bar\theta^{(\textit{v})}_j+L_n  + \frac{1}{w_r} L_n^2+ 
                \frac{\gamma_{\max} + \Delta_{0,n}}{w_r} +  \frac{1}{n w_r^3}
                \bigg)\bar\theta^{(\textit{v})}_j\,,
            \end{split}
		\end{equation*}			
which implies that
		\begin{equation}\label{zeta2 result}
            \begin{split}                				|\bzeta_2^{(i,j,\textit{v})}|_2 &\le \tilde C_0 \tilde C\bigg(w_1\bar\theta^{(\textit{v})}_j+\frac{w_1}{w_r}L_n\bar\theta^{(\textit{v})}_j+L_n  + \frac{1}{w_r} L_n^2+ 
                \frac{\gamma_{\max} + \Delta_{0,n}}{w_r} +  \frac{1}{n w_r^3}
                \bigg)\bar\theta^{(\textit{v})}_j\,.
            \end{split}
		\end{equation}

\subsubsection{Step 2: Upper bound of \texorpdfstring{$|\bzeta^{(i,j,v)}_4|_2$}{zeta_4}}

  By definition above (\ref{iterative sigma decompose}), write
  \begin{align*}
      \bzeta_4^{(i,j,\textit{v})}&=  \underbrace{\frac{1}{n-1}\sum_{t=2}^n\tilde\xi^{(\textit{v},j)}_{t-1,i}(\Eb_{t,j} - \bar \Eb_j)\{(\tilde\bbb_{i,j}^{(\textit{v})})^{\MP} -\bbb_{i,j}^{\MP}\}}_{\bzeta_{4,1}^{(i,j,\textit{v})}}
		\\
		&\quad+\underbrace{\frac{1}{\tilde{\sigma}^{(\textit{v},j)}_{\check f, i} }\frac{1}{n-1}\sum_{t=2}^n(\check f^{(\textit{v},j)}_{t-1,i}-\bar{\check f}^{(\textit{v},j)}_{i}-\xi_{t-1,i}+\bar\xi_i)(\eb_{t,i,j} - \bar{\eb}_{i,j})}_{\bzeta_{4,2}^{(i,j,\textit{v})}}\\
		&\quad+\underbrace{\sum_{\ell\ne i}\frac{\tilde\varphi^{(\textit{v},j)}_{i,\ell}}{\tilde{\sigma}^{(\textit{v},j)}_{\check f, \ell}}\frac{1}{n-1}\sum_{t=2}^n(\check f^{(\textit{v},j)}_{t,\ell}-\bar{\check f}^{(\textit{v},j)}_{\ell}-\xi_{t,\ell}+\bar\xi_\ell)(\eb_{t,i,j} - \bar{\eb}_{i,j})}_{\bzeta_{4,3}^{(i,j,\textit{v})}}\\
		&\quad+\underbrace{\sum_{\ell\ne i}\frac{\tilde\varphi^{(\textit{v},j)}_{i,\ell}}{\tilde{\sigma}^{(\textit{v},j)}_{\check f, \ell}}\frac{1}{n-1}\sum_{t=2}^n\mathbb{E}\{(\xi_{t,\ell}-\bar\xi_\ell)(\eb_{t,i,j} - \bar{\eb}_{i,j})\}}_{\bzeta_{4,4}^{(i,j,\textit{v})}}\,.
  \end{align*}
   			
For $	|\bzeta_{4,1}^{(i,j,\textit{v})}|_2$, by \eqref{eq: bij(b) expression}, we have
\begin{equation*}
\begin{split}
\bzeta_{4,1}^{(i,j,\textit{v})} &=  \underbrace{\frac{1}{n-1}\sum_{t=2}^n\tilde\xi^{(\textit{v},j)}_{t-1,i}(\Eb_{t,j} - \bar \Eb_j)\bigg\{ \sum_{\ell = 1}^r (\tilde\varpi^{(\textit{v},j)}_{\ell,i} - \varpi_{\ell,i})(\otimes_{m}^{j^\prime \neq j} \ab_{\ell,j^\prime}) \bigg\}}_{\bzeta_{4,1,1}^{(i,j,\textit{v})}} \\
&\quad+ \underbrace{\frac{1}{n-1}\sum_{t=2}^n\tilde\xi^{(\textit{v},j)}_{t-1,i}(\Eb_{t,j} - \bar \Eb_j)\bigg\{ \sum_{\ell = 1}^r  \tilde\varpi^{(\textit{v},j)}_{\ell,i}(\otimes_{m}^{j^\prime \neq j} \check\ab^{(\textit{v},j)}_{\ell,j^\prime}- \otimes_{m}^{j^\prime \neq j} \ab_{\ell,j^\prime}) \bigg\}}_{\bzeta_{4,1,2}^{(i,j,\textit{v})}} \,.
\end{split}
\end{equation*}
Notice   that 
\[
\begin{split}
    |\bzeta_{4,1,1}^{(i,j,\textit{v})}|_2 &\le  \tilde C_0\tilde C \sum_{\ell = 1}^r |\tilde\varpi^{(\textit{v},j)}_{\ell,i} - \varpi_{\ell,i}| \\
    & \quad  \times \max_{|\bbeta_{j'}|_2 = 1,\,\bbeta_{j'}\in\mathbb{R}^{d_{j'}},\forall j' \in [m]} \bigg| \frac{1}{n-1}\sum_{t=2}^n \{(\otimes^{j'\ne j}_{m} \bbeta_{j'}) \otimes \bbeta_j\}^{\T} \textup{vec}(\Eb_{t,j} -\bar{\Eb}_j) \tilde\xi^{(\textit{v},j)}_{t-1,i} \bigg| \,.
\end{split}
\]
Using calculations similar to those leading to \eqref{zeta 224}, together with $ \max_{i,\ell \in [r]}|\tilde\varpi^{(\textit{v},j)}_{\ell,i} - \varpi_{\ell,i}|  \le \tilde{C}_0 \bar\theta^{(\textit{v})}_j$, we can conclude that
	\[
\begin{split}
		|\bzeta_{4,1,1}^{(i,j,\textit{v})}|_2 &\le\tilde C_0 \tilde C \bigg\{ \frac{ w_1}{w_r}L_n\bar\theta^{(\textit{v})}_j+ L_n  +  \frac{1}{w_r} L_n^2 + \frac{w_1}{w_r^2} (\bar\theta^{(\textit{v})}_j)^m + \frac{\gamma_{\max} + \Delta_{0,n}}{w_r} +  \frac{1}{n w_r^3} \bigg\} \bar\theta^{(\textit{v})}_j \,.
\end{split}
	\]	
Analogously, replacing each $\check\ab^{(\textit{v},j)}_{\ell,j^\prime}$ in $\otimes_{m}^{j^\prime \neq j} \check\ab^{(\textit{v},j)}_{\ell,j^\prime}$  with $\ab_{\ell,j^\prime}$ one by one in $\bzeta_{4,1,2}^{(i,j,\textit{v})}$, we have
\[
\begin{split}
    |\bzeta_{4,1,2}^{(i,j,\textit{v})}|_2 &\le \tilde C_0\tilde C \sum_{\ell = 1}^r |\tilde\varpi^{(\textit{v},j)}_{\ell,i} | \sum_{j^\prime \neq j} |\check\ab^{(\textit{v},j)}_{\ell,j^\prime}-  \ab_{\ell,j^\prime} |_2 \\
    & \quad  \times  \max_{|\bbeta_{j'}|_2 = 1,\,\bbeta_{j'}\in\mathbb{R}^{d_{j'}},\forall j' \in [m]} \bigg| \frac{1}{n-1}\sum_{t=2}^n \{(\otimes^{j'\ne j}_{m} \bbeta_{j'})\otimes \bbeta_j\}^{\T} \textup{vec}(\Eb_{t,j} -\bar{\Eb}_j) \tilde\xi^{(\textit{v},j)}_{t-1,i} \bigg| \\
    &\le \tilde C_0\tilde C \bigg\{ \frac{ w_1}{w_r}L_n\bar\theta^{(\textit{v})}_j+ L_n +  \frac{1}{w_r} L_n^2 + \frac{w_1}{w_r^2} (\bar\theta^{(\textit{v})}_j)^m + \frac{\gamma_{\max} + \Delta_{0,n}}{w_r} +  \frac{1}{n w_r^3}  \bigg\} \bar\theta^{(\textit{v})}_j \,,
\end{split}
\]
which implies that
	\[
\begin{split}
			|\bzeta_{4,1}^{(i,j,\textit{v})}|_2&\le 
            \tilde C_0\tilde C \bigg\{ \frac{ w_1}{w_r}L_n\bar\theta^{(\textit{v})}_j+ L_n +  \frac{1}{w_r} L_n^2 + \frac{w_1}{w_r^2} (\bar\theta^{(\textit{v})}_j)^m + \frac{\gamma_{\max} + \Delta_{0,n}}{w_r} +  \frac{1}{n w_r^3}  \bigg\} \bar\theta^{(\textit{v})}_j\,.
\end{split}
	\]	
    Moreover, by \eqref{check f ti minus xi ti decomp} and the arguments similar to those used from \eqref{check f ti minus step 2}--\eqref{zeta 221}, under the event $\Xi_{n}(\tilde C)$, we can conclude that
	\[
\begin{split}
			|\bzeta_{4,2}^{(i,j,\textit{v})}|_2 &\le \tilde C_0 \tilde CL_n\bigg\{1+ \frac{w_1}{w_r} (\bar\theta^{(\textit{v})}_j)^{m-1} + \frac{1}{w_r}L_n\bigg\}\bar\theta^{(\textit{v})}_j\,,\\
					|\bzeta_{4,3}^{(i,j,\textit{v})}|_2 &\le \tilde C_0 \tilde C \bigg\{\frac{\gamma_{\max}+\Delta_{0,n}}{w_r} + \frac{1}{n w_r^3}+\frac{w_1}{w_r^2} (\bar\theta^{(\textit{v})}_j)^m  + \frac{w_1}{w_r} L_n + \frac{1}{w_r} L_n^2\bigg\} \bar\theta^{(\textit{v})}_j  \,.
\end{split}
	\]
   By   Lemma \ref{bound for Lambda and phi}, the explicit upper bound of  $\max_{\ell\ne i}|\tilde\varphi_{i,\ell}^{(\textit{v},j)}|$ shown in (\ref{zeta 222}), and Assumption \ref{cross}, we conclude that
	\[
    \begin{split}
        	|\bzeta_{4,4}^{(i,j,\textit{v})}|_2 &\le  \tilde C_0\tilde C \bigg\{\frac{\gamma_{\max}+\Delta_{0,n}}{w_r} +\frac{1}{n w_r^3}+\frac{w_1}{w_r^2} (\bar\theta^{(\textit{v})}_j)^m  + \frac{w_1}{w_r^3}L_n\bar\theta^{(\textit{v})}_j + \frac{1}{w_r^3}L_n^2\bar\theta^{(\textit{v})}_j\bigg\}\,.
    \end{split}
	\]
	By Triangle inequality and the bounds for
$\bzeta_{4,1}^{(i,j,\textit{v})}$, $\bzeta_{4,2}^{(i,j,\textit{v})}$, $\bzeta_{4,3}^{(i,j,\textit{v})}$ and $\bzeta_{4,4}^{(i,j,\textit{v})}$, we have
\[
\begin{split}
    |\bzeta_4^{(i,j,\textit{v})}|_2
    &\le
    \tilde C_0\tilde C L_n\bigg( \frac{w_1}{w_r} +\frac{1}{w_r}L_n\bigg)\bar\theta_j^{(\textit{v})}  +
    \tilde C_0\tilde C
    \bigg\{
    \frac{\gamma_{\max}+\Delta_{0,n}}{w_r}
    +
    \frac{1}{n w_r^3}
    +
    \frac{w_1}{w_r^2}
    (\bar\theta_j^{(\textit{v})})^m
    \bigg\}\,.
\end{split}
\]
Here we use
$w_r^{-1}w_1\bar\theta_j^{(\textit{v})}\le \tilde C^{-2}$ and condition
\eqref{strong factor condition} to absorb the higher-order terms.
Overall, by condition \eqref{strong factor condition}, it follows that
\begin{equation*} 
   \begin{split}
				& |\bzeta_2^{(i,j,\textit{v})}|_2 + |\bzeta_4^{(i,j,\textit{v})}|_2 \\
                 &~~~~~~\le  \tilde C_0\tilde C\bigg( \frac{\gamma_{\max}+\Delta_{0,n}}{w_r} +\frac{1}{n w_r^3}\bigg) +\tilde C_0\tilde C \bigg\{w_1\bar\theta^{(\textit{v})}_j+L_n\bigg(\frac{w_1}{w_r}+\frac{1}{w_r}L_n\bigg)\bigg\}\bar\theta^{(\textit{v})}_j \\
                &~~~~~~\le  \tilde C_0\tilde C\bigg( \frac{\gamma_{\max}+\Delta_{0,n}}{w_r} +\frac{1}{n w_r^3}\bigg) +\tilde C_0\tilde C \bigg(w_1\bar\theta^{(\textit{v})}_j+\frac{w_1}{w_r}L_n\bigg)\bar\theta^{(\textit{v})}_j\,.
\end{split}
\end{equation*}
We complete the proof of Lemma \ref{lem: zeta_2 and zeta_4}.	$\hfill\Box$

\subsection{Proof of Lemma \ref{lem: truncated projected covariance}}\label{sec:lem: truncated projected covariance}

In the proof, $\tilde C_0 > 0$ is a universal constant and may vary in different lines, but is independent of $\tilde C$ and $(i,j,\textit{v})$.
According to the decomposition in (\ref{iterative sigma decompose}), to ease the notation,  for $p\in[d_j]$, let $\tilde\sigma_p^{(i,j,\textit{v})}$, $\tilde\sigma_{1,p}^{(i,j,\textit{v})}$, $\tilde\sigma_{2,p}^{(i,j,\textit{v})}$, $\tilde\sigma_{3,p}^{(i,j,\textit{v})}$, and $\tilde\sigma_{4,p}^{(i,j,\textit{v})}$ be the $p$-th entries of $\tilde\bSigma_{\tilde\yb_{i,j},\tilde\xi_{i}}^{(\textit{v},j)}(1)$, $\bzeta_1^{(i,j,\textit{v})}$, $\bzeta_2^{(i,j,\textit{v})}$, $\bzeta_3^{(i,j,\textit{v})}$, and $\bzeta_4^{(i,j,\textit{v})}$, respectively. Similarly to (\ref{decomposition truncation}), it holds that
			\[
			\begin{split}
				|T_{\delta_{2,j}}\{\tilde\bSigma^{(\textit{v},j)}_{\tilde\yb_{i,j},\tilde\xi_{i}}(1)\}-\bzeta_1^{(i,j,\textit{v})}|_2& \le \underbrace{|T_{\delta_{2,j}}\{\tilde\bSigma^{(\textit{v},j)}_{\tilde\yb_{i,j},\tilde\xi_{i}}(1)\}-T_{\delta_{2,j}}(\bzeta_1^{(i,j,\textit{v})})|_2}_{\mathcal{\uppercase\expandafter{\romannumeral3}}_{1}}+\underbrace{|T_{\delta_{2,j}}(\bzeta_1^{(i,j,\textit{v})})-\bzeta_1^{(i,j,\textit{v})}|_2}_{\mathcal{\uppercase\expandafter{\romannumeral3}}_{2}}\,.
			\end{split}
			\]
			 By Lemma \ref{lem: zeta_1 lower bound}, we have $|\bzeta_1^{(i,j,\textit{v})}|_2 \ge  w_i C_9^{-1}  (1 - \tilde{C}_0 \tilde{C}^{-1/2}) > C$ for some universal constant $C > 0$.
           It follows that
			\begin{equation}\label{11l}
				\begin{split}
					\mathcal{\uppercase\expandafter{\romannumeral3}}_{2}^2 &\le \sum_{p=1}^{d_j}|\tilde\sigma_{1,p}^{(i,j,\textit{v})}|^2I\{|\tilde\sigma_{1,p}^{(i,j,\textit{v})}|<\delta_{2,j}\}\le \delta_{2,j}^2\sum_{p=1}^{d_j}I\{|\tilde\sigma_{1,p}^{(i,j,\textit{v})}|\ne 0\}	\le \tilde C_0\delta_{2,j}^2 s_{j}\,.
				\end{split}
			\end{equation}
			For $\mathcal{\uppercase\expandafter{\romannumeral3}}_{1}$, similarly to (\ref{I1}), we define three vectors $\tilde\bSigma_{1,1}^{(i,j,\textit{v})}$, $\tilde \bSigma_{1,2}^{(i,j,\textit{v})}$ and $\tilde\bSigma_{1,3}^{(i,j,\textit{v})}$, whose $p$-th entries are defined respectively as
			\[
			\begin{split}
				\tilde\Sigma_{1,1,p}^{(i,j,\textit{v})}&= (\tilde\sigma_{p}^{(i,j,\textit{v})}-\tilde\sigma_{1,p}^{(i,j,\textit{v})})I\{|\tilde\sigma_{p}^{(i,j,\textit{v})}|\ge\delta_{2,j},|\tilde\sigma_{1,p}^{(i,j,\textit{v})}|\ge\delta_{2,j}\}\,,\\
				\tilde\Sigma_{1,2,p}^{(i,j,\textit{v})}&= \tilde\sigma_{p}^{(i,j,\textit{v})}I\{|\tilde\sigma_{p}^{(i,j,\textit{v})}|\ge\delta_{2,j},|\tilde\sigma_{1,p}^{(i,j,\textit{v})}|<\delta_{2,j}\}\,,\\
				\tilde\Sigma_{1,3,p}^{(i,j,\textit{v})}&= \tilde\sigma_{1,p}^{(i,j,\textit{v})}I\{|\tilde\sigma_{p}^{(i,j,\textit{v})}|<\delta_{2,j},|\tilde\sigma_{1,p}^{(i,j,\textit{v})}|\ge\delta_{2,j}\}\,.
			\end{split}
			\] 
			It follows that 
			\[
			\mathcal{\uppercase\expandafter{\romannumeral3}}_{1}\le |\tilde\bSigma_{1,1}^{(i,j,\textit{v})}|_2+|\tilde\bSigma_{1,2}^{(i,j,\textit{v})}|_2+|\tilde\bSigma_{1,3}^{(i,j,\textit{v})}|_2 \,.
			\]
			We will handle the three terms one by one similarly to the proof of Lemma \ref{lem: l2 truncated covariance}, while the major difference is on the additional errors $\bzeta_2^{(i,j,\textit{v})}$ and $\bzeta_4^{(i,j,\textit{v})}$. 
			
			For  $\tilde\bSigma_{1,1}^{(i,j,\textit{v})}$, we have
			\begin{equation}\label{tilde Sigma 11l}
				|\tilde\bSigma_{1,1}^{(i,j,\textit{v})}|_2^2\le \tilde{C}_0 \bigg\{ \sum_{p=1}^{d_j}|\tilde\sigma_{3,p}^{(i,j,\textit{v})}|^2I\{|\tilde\sigma_{1,p}^{(i,j,\textit{v})}|\ge\delta_{2,j}\}+\sum_{p=1}^{d_j}|\tilde\sigma_{2,p}^{(i,j,\textit{v})}|^2+\sum_{p=1}^{d_j}|\tilde\sigma_{4,p}^{(i,j,\textit{v})}|^2\bigg\}\,.
			\end{equation}
			 Lemma \ref{bound for Lambda and phi} shows that
			\[
			1 - \tilde C_0\tilde C^{-1}\le \frac{\Lambda_{\xi,i,i}^{1/2}}{\tilde{\sigma}^{(\textit{v},j)}_{\check f, i} }\le 1 + \tilde C_0\tilde C^{-1}\,.
			\]		
			Then, similarly to \eqref{112} and \eqref{111}, we can conclude that
			\begin{equation}\label{11l1}
				\sum_{p=1}^{d_j}|\tilde\sigma_{3,p}^{(i,j,\textit{v})}|^2I\{|\tilde\sigma_{1,p}^{(i,j,\textit{v})}|\ge\delta_{2,j}\}	\le \tilde C_0\tilde C \Delta_{1,n,j}^2 s_{j}\,.
			\end{equation}
 Combining \eqref{11l1} and the bound \eqref{eq: zeta2+zeta4} in Lemma \ref{lem: zeta_2 and zeta_4}, we have
\[
\begin{split}
    |\tilde\bSigma_{1,1}^{(i,j,\textit{v})}|_2
    &\le
    \tilde C_0\tilde C
    \Delta_{1,n,j}s_j^{1/2}
    +
    \tilde C_0\tilde C
    \bigg(
    \frac{\gamma_{\max}+\Delta_{0,n}}{w_r}
    +
    \frac{1}{n w_r^3}
    \bigg)  +
    \tilde C_0\tilde C
    \bigg(w_1\bar\theta_j^{(\textit{v})}+\frac{w_1}{w_r}L_n\bigg)\bar\theta_j^{(\textit{v})}\,.
\end{split}
\]

			Next, for $\tilde\bSigma_{1,3}^{(i,j,\textit{v})}$, note that
			\[
			|\tilde\Sigma_{1,3,p}^{(i,j,\textit{v})}|\le |\tilde\sigma_{p}^{(i,j,\textit{v})}-\tilde\sigma_{1,p}^{(i,j,\textit{v})}|I\{|\tilde\sigma_{1,p}^{(i,j,\textit{v})}|\ge\delta_{2,j}\}+|\delta_{2,j}|I\{|\tilde\sigma_{1,p}^{(i,j,\textit{v})}|\ge\delta_{2,j}\}\,.
			\]
			Therefore, by \eqref{11l} and the bound \eqref{eq: zeta2+zeta4} in Lemma \ref{lem: zeta_2 and zeta_4}, similarly to (\ref{13}), we can conclude that
			\[
			\begin{split}
				|\tilde\bSigma_{1,3}^{(i,j,\textit{v})}|_2
				&\le \tilde C_0\tilde C(\Delta_{1,n,j} s_{j}^{1/2}+\delta_{2,j}s_{j}^{1/2}) +\tilde C_0\tilde C\bigg( \frac{\gamma_{\max}+\Delta_{0,n}}{w_r} +\frac{1}{n w_r^3}\bigg)\\
				&\quad +\tilde C_0\tilde C \bigg(w_1\bar\theta_j^{(\textit{v})}+\frac{w_1}{w_r}L_n\bigg)\bar\theta^{(\textit{v})}_j\,.
			\end{split}
			\]
			
			Finally, for $\tilde\bSigma_{1,2}^{(i,j,\textit{v})}$, by Triangle inequality, we have
			\[
			\begin{split}
|\tilde\Sigma_{1,2,p}^{(i,j,\textit{v})}|&\le|\tilde\sigma_{1,p}^{(i,j,\textit{v})}|I\{|\tilde\sigma_{1,p}^{(i,j,\textit{v})}|<\delta_{2,j}\}  +|\tilde\sigma_{2,p}^{(i,j,\textit{v})}|+|\tilde\sigma_{4,p}^{(i,j,\textit{v})}| \\
& \quad + |\tilde\sigma_{3,p}^{(i,j,\textit{v})}|I\{|\tilde\sigma_{p}^{(i,j,\textit{v})}|\ge\delta_{2,j},|\tilde\sigma_{1,p}^{(i,j,\textit{v})}|<\delta_{2,j}\}\,.
			\end{split}
			\]
			On the one hand, similarly to the upper bounds for $|\tilde\bSigma_{1,1}^{(i,j,\textit{v})}|_2$ and $|\tilde\bSigma_{1,3}^{(i,j,\textit{v})}|_2$,  we have
			\[
			\begin{split}
				& \bigg(\sum_{p=1}^{d_j}	|\tilde\sigma_{1,p}^{(i,j,\textit{v})}|^2 I\{|\tilde\sigma_{1,p}^{(i,j,\textit{v})}|<\delta_{2,j}\}+\sum_{p=1}^{d_j}|\tilde\sigma_{2,p}^{(i,j,\textit{v})}|^2+\sum_{p=1}^{d_j}|\tilde\sigma_{4,p}^{(i,j,\textit{v})}|^2\bigg)^{1/2}\\
				& ~~~~~~\le \tilde C_0\tilde C\delta_{2,j}s_{j}^{1/2} +\tilde C_0\tilde C\bigg( \frac{\gamma_{\max}+\Delta_{0,n}}{w_r} +\frac{1}{n w_r^3}\bigg)+\tilde C_0\tilde C \bigg(w_1\bar\theta_j^{(\textit{v})}+\frac{w_1}{w_r}L_n\bigg)\bar\theta^{(\textit{v})}_j\,.
			\end{split}
			\]
			Meanwhile, similarly to the decomposition of $\tilde\bSigma_{1,2,1}$ in the proof of Lemma \ref{lem: l2 truncated covariance}, we have
			\[
			\begin{split}
				&|\tilde\sigma_{3,p}^{(i,j,\textit{v})}|I\{|\tilde\sigma_{p}^{(i,j,\textit{v})}|\ge\delta_{2,j},|\tilde\sigma_{1,p}^{(i,j,\textit{v})}|<\delta_{2,j}\} \\
                &~~~~~~~~~~~~\le |\tilde\sigma_{3,p}^{(i,j,\textit{v})}|I\{|\tilde\sigma_{1,p}^{(i,j,\textit{v})}|\ne 0\}\\
				&~~~~~~~~~~~~\quad +|\tilde\sigma_{3,p}^{(i,j,\textit{v})}|I\{|\tilde\sigma_{2,p}^{(i,j,\textit{v})}+\tilde\sigma_{3,p}^{(i,j,\textit{v})}+\tilde\sigma_{4,p}^{(i,j,\textit{v})}|\ge \delta_{2,j}\}\,.
			\end{split}
			\]
			Similarly to the bound for (\ref{11l1}), we have
			\[
			\bigg(\sum_{p=1}^{d_j}|\tilde\sigma_{3,p}^{(i,j,\textit{v})}|^2I\{|\tilde\sigma_{1,p}^{(i,j,\textit{v})}|\ne 0\}\bigg)^{1/2}\le \tilde C_0\tilde C \Delta_{1,n,j}s_{j}^{1/2}\,.
			\]
			Moreover, 
			\[
			\begin{split}
				&|\tilde\sigma_{3,p}^{(i,j,\textit{v})}|I\{|\tilde\sigma_{2,p}^{(i,j,\textit{v})}+\tilde\sigma_{3,p}^{(i,j,\textit{v})}+\tilde\sigma_{4,p}^{(i,j,\textit{v})}|\ge\delta_{2,j}\}\\
				&~~~~~~ \le |\tilde\sigma_{3,p}^{(i,j,\textit{v})}|I\{|\tilde\sigma_{3,p}^{(i,j,\textit{v})}|< |\tilde\sigma_{2,p}^{(i,j,\textit{v})}+\tilde\sigma_{4,p}^{(i,j,\textit{v})}|\}+|\tilde\sigma_{3,p}^{(i,j,\textit{v})}|I\{|\tilde\sigma_{3,p}^{(i,j,\textit{v})}|\ge 0.5\delta_{2,j}\}\\
				&~~~~~~ \le  |\tilde\sigma_{2,p}^{(i,j,\textit{v})}|+|\tilde\sigma_{4,p}^{(i,j,\textit{v})}|+|\tilde\sigma_{3,p}^{(i,j,\textit{v})}|I\{|\tilde\sigma_{3,p}^{(i,j,\textit{v})}|\ge 0.5\delta_{2,j}\}\,.
			\end{split}
			\]
By Lemma \ref{bound for Lambda and phi},  for a sufficiently
large $\tilde C_* > 0$, the event $\{|\tilde\sigma_{3,p}^{(i,j,\textit{v})}|\ge0.5\delta_{2,j}\}$ implies
\[
    |[\tilde\bSigma_{\eb_{i,j},\xi_i}(1)]_p|
    +
    \max_{\ell\ne i}|[\tilde\bSigma_{\eb_{i,j},\xi_\ell}(0)]_p|
    >
    0.4\tilde C_*
    \bigg(\frac{\log d_j}{n}\bigg)^{1/2}\,.
\]
Hence, we can conclude that
			\[
			\begin{split}	       
      \bigg(\sum_{p=1}^{d_j}|\tilde\sigma_{3,p}^{(i,j,\textit{v})}|^2I\{|\tilde\sigma_{3,p}^{(i,j,\textit{v})}|\ge 0.5\delta_{2,j}\}\bigg)^{1/2}
				&\le \tilde C_0\tilde C\Delta_{1,n,j}\Delta_{2,n,j}(0.4\tilde C_*)      
                \end{split}
			\]
			for sufficiently large $\tilde C_* > 0$. It follows that
			\[
			\begin{split}
				|\tilde\bSigma_{1,2}^{(i,j,\textit{v})}|_2
				&\le \tilde C_0\tilde C\{\Delta_{1,n,j}\Delta_{2,n,j}(0.4\tilde C_*)+\Delta_{1,n,j} s_{j}^{1/2}+\delta_{2,j}s_{j}^{1/2}\}\\
				&\quad +\tilde C_0\tilde C\bigg( \frac{\gamma_{\max}+\Delta_{0,n}}{w_r} +\frac{1}{n w_r^3}\bigg) +\tilde C_0\tilde C \bigg(w_1\bar\theta_j^{(\textit{v})}+\frac{w_1}{w_r}L_n\bigg)\bar\theta^{(\textit{v})}_j\,.
			\end{split}
			\]
Combining with the bounds for $|\tilde\bSigma_{1,1}^{(i,j,\textit{v})}|_2$, $|\tilde\bSigma_{1,3}^{(i,j,\textit{v})}|_2$, and $|\tilde\bSigma_{1,2}^{(i,j,\textit{v})}|_2$, we have
			\[
			\begin{split}
				&|T_{\delta_{2,j}}\{\tilde\bSigma_{\tilde\yb_{i,j},\tilde\xi_i}^{(\textit{v},j)}(1)\}-\bzeta_1^{(i,j,\textit{v})}|_2	 \\		
                &~~~~~~\le \tilde C_0\tilde C\{\Delta_{1,n,j}\Delta_{2,n,j}(0.4\tilde C_*)+\Delta_{1,n,j} s_{j}^{1/2}+\delta_{2,j}s_{j}^{1/2}\}\\
				&~~~~~~\quad +\tilde C_0\tilde C\bigg( \frac{\gamma_{\max}+\Delta_{0,n}}{w_r} +\frac{1}{n w_r^3}\bigg)+\tilde C_0\tilde C \bigg(w_1\bar\theta_j^{(\textit{v})}+\frac{w_1}{w_r}L_n\bigg)\bar\theta^{(\textit{v})}_j\,,
			\end{split}
			\]
			which concludes Lemma \ref{lem: truncated projected covariance}.
$\hfill\Box$

\subsection{Proof of Lemma \ref{lem: Delta}}\label{sec:lem: Delta}			
		Following the arguments used in the proof of Lemma \ref{lem: Omega n} for $\Xi_{1,n}(\tilde C)$ and $\Xi_{4,n}(\tilde C)$, we can derive the convergence rate of $\Delta_{0,n}$. 
			 Specifically, recall that $\bLambda_{\xi} = (\Lambda_{\xi,i,j})_{r\times r}$ is the diagonal matrix with the $i$-th diagonal entry being $n^{-1}\sum_{t = 1}^n\mathbb{E}\{(\xi_{t,i}-\bar\xi_i)^2\}$, and $u_{t,i} = (\ab_{i,m}^{\MP} \otimes \cdots \otimes  \ab_{i,1}^{\MP})^{\T}\textup{vec}(\mathcal{E}_t)$ with $(\ab_{1,j}^{\MP},\ldots,\ab_{ r,j}^{\MP})^{\T} =(\Ab_j^\T\Ab_j)^{-1}{\Ab}_j^{\T}$. By definition, we can write
			\[
		\begin{split}			
        &\frac{\Lambda_{\xi,i,i}^{1/2}\Lambda_{\xi,\ell,\ell}^{1/2}}{n}\sum_{t=2}^n\{\xi_{t,\ell}^{\textup{s}}\xi_{t-1,i}^{\textup{s}}-\mathbb{E}(\xi_{t,\ell}^{\textup{s}}\xi_{t-1,i}^{\textup{s}})\}\\
			&~~~~~~=\frac{1}{n}\sum_{t=2}^n\{w_{\ell}(f_{t,\ell}-\bar f_{\ell})+(u_{t,\ell}-\bar u_{\ell})\}\{w_{i}(f_{t-1,i}-\bar f_i)+(u_{t-1,i}-\bar u_{i})\}\\
			&~~~~~~\quad -\frac{1}{n}\sum_{t=2}^n\mathbb{E}[\{w_{\ell}(f_{t,\ell}-\bar f_{\ell})+(u_{t,\ell}-\bar u_{\ell})\}\{w_{i}(f_{t-1,i}-\bar f_i)+(u_{t-1,i}-\bar u_{i})\}]\,,
		\end{split}
			\]
			 where $\bar u_i = n^{-1}\sum_{t = 1}^n u_{t,i}$. It is already shown in the proof of Lemma \ref{lem: Omega n} that
			\[
			\begin{split}
				&\frac{1}{n}\sum_{t=2}^n[(f_{t,\ell}-\bar f_{\ell})(f_{t-1,i}-\bar f_i)-\mathbb{E}\{(f_{t,\ell}-\bar f_{\ell})(f_{t-1,i}-\bar f_i)\}]\\
				&~~~~~~=\frac{1}{n}\sum_{t=2}^n\{f_{t,\ell}f_{t-1,i}-\mathbb{E}(f_{t,\ell}f_{t-1,i})\}
				-\frac{1}{n}\sum_{t=2}^n\{f_{t,\ell}\bar f_{i}-\mathbb{E}(f_{t,\ell}\bar f_{i})\}\\
				&~~~~~~\quad-\frac{1}{n}\sum_{t=2}^n\{f_{t-1,i}\bar f_{\ell}-\mathbb{E}(f_{t-1,i}\bar f_{\ell})\}+\frac{1}{n}\sum_{t=1}^n \{\bar f_{\ell} \bar f_{i}-\mathbb{E}(\bar f_{\ell} \bar f_{i})\}\\
				&~~~~~~=  O_{\rm p}\bigg(\frac{1}{\sqrt{n}}\bigg)\,,
			\end{split}
			\]
			which implies the second part in $\Delta_{0,n}$ is $O_{\rm p}(n^{-1/2})$. Handling the other interaction terms similarly, we claim that
			\[
			\frac{\Lambda_{\xi,i,i}^{1/2}\Lambda_{\xi,\ell,\ell}^{1/2}}{n}\sum_{t=2}^n\{\xi_{t,\ell}^{\textup{s}}\xi_{t-1,i}^{\textup{s}}-\mathbb{E}(\xi_{t,\ell}^{\textup{s}}\xi_{t-1,i}^{\textup{s}})\} = O_{\rm p}\bigg(\frac{w_iw_{\ell}}{\sqrt{n}}\bigg)\,.
			\] 
The rate for $\Delta_{0,n}$ then holds because $\Lambda_{\xi,i,i}^{1/2} \asymp w_i$ by  \eqref{eq: varxi-bar-xi}.


The proof for $\Delta_{1,n,j}$ is similar to (\ref{tail sigma2}) and (\ref{sigma2}) under Assumptions \ref{tail}, \ref{mixing} and \ref{cross}. Write $\eb_{t,\ell,j} = (e_{t,\ell,j,1},\ldots,e_{t,\ell,j,d_j})^\T$ and $\bar\eb_{\ell,j} = (\bar e_{\ell,j,1},\ldots,\bar e_{\ell,j,d_j})^\T$. 
By definition of $\xi_{t,i}^{\textup{s}}$, it holds that
\[
\begin{split}
  \frac{\Lambda_{\xi,i,i}^{1/2}}{w_i}[\tilde\bSigma_{\eb_{\ell,j},\xi_i}(1)]_p &= \frac{1}{n-1}\sum_{t=2}^n e_{t,\ell,j,p}w_i^{-1}\xi_{t-1,i} - \frac{w_i^{-1}\bar \xi_i}{n-1}\sum_{t=2}^ne_{t,\ell,j,p} \\
& \quad - \frac{\bar e_{\ell,j,p}}{n-1}\sum_{t=2}^n w_i^{-1}\xi_{t-1,i} + \bar e_{\ell,j,p}   w_i^{-1}\bar \xi_i\,.  
\end{split}
\]
Therefore, similarly to \eqref{tail sigma2} and together with \eqref{eq: varxi-bar-xi}, we have
\[
    \mathbb{P}(
    |[\tilde\bSigma_{\eb_{\ell,j},\xi_i}(1)]_p|\ge x
    )
    \lesssim
    \exp(-Cnx^2)
    +
    \exp(-Cn^{\tilde c}x^{\tilde c})
    +
    \exp(-Cn^{\check c}x^{\check c})
\]
		 for any $x \in (0,1)$ and some universal constant $C>0$, where $\tilde c=(1+2c_1^{-1}+c_2^{-1})^{-1}$ and $\check c = (2+|c_1^{-1} -1|_{\MP} + c_2^{-1})^{-1}$.  Analogous to $\tilde\bSigma_{\eb_{\ell,j},\xi_i}(1)$,  it holds that
\[
			\mathbb{P}(|[\tilde\bSigma_{\eb_{\ell,j},\xi_i}(0)]_p|\ge x)\lesssim \exp(-Cnx^2) +\exp(-Cn^{\tilde c}x^{\tilde c}) +\exp(-Cn^{\check c}x^{\check c})
			\]
			for any $x \in (0,1)$ and some universal constant $C>0$, where $\tilde c=(1+2c_1^{-1}+c_2^{-1})^{-1}$ and $\check c = (2+|c_1^{-1} -1|_{\MP} + c_2^{-1})^{-1}$. We then can conclude the result of $\Delta_{1,n,j}$ provided that $\max_{j \in [m]}\log d_j \ll n^c$ for some constant $c\in(0,1)$ depending only on $c_1$ and $c_2$ specified in Assumptions {\rm\ref{tail}} and {\rm\ref{mixing}}.
            
            For $\Delta_{2,n,j}(0.4\tilde C_*)$, by Markov's inequality, we have
\[
\begin{split}
    &\mathbb{P}\{\Delta_{2,n,j}(0.4\tilde C_*)\ge \lambda\}  \\
    &\quad\le
    \lambda^{-1}
    \sum_{i,\ell\in[r]}\sum_{p=1}^{d_j}
    \mathbb{P}\bigg\{
    |[\tilde\bSigma_{\eb_{\ell,j},\xi_i}(1)]_p|
    +
    |[\tilde\bSigma_{\eb_{\ell,j},\xi_i}(0)]_p|
    >
    0.4\tilde C_*
    \bigg(\frac{\log d_j}{n} \bigg)^{1/2}
    \bigg\} 
\end{split}
\]
for any $\lambda > 1$. By above arguments, it holds that 
			\[
			\mathbb{P}(|[\tilde\bSigma_{\eb_{\ell,j},\xi_i}(1)]_p|+|[\tilde\bSigma_{\eb_{\ell,j},\xi_i}(0)]_p|\ge x)\lesssim \exp(-Cnx^2) +\exp(-Cn^{\tilde c}x^{\tilde c}) +\exp(-Cn^{\check c}x^{\check c})
			\]
			for any $x \in (0,1)$ and some universal constant $C>0$, where $\tilde c=(1+2c_1^{-1}+c_2^{-1})^{-1}$ and $\check c = (2+|c_1^{-1} -1|_{\MP} + c_2^{-1})^{-1}$. Therefore, for sufficiently large $\tilde C_* > 0$, we have
			\[
			\mathbb{P}\{\Delta_{2,n,j}(0.4 \tilde C_*)\ge \lambda\} \le \lambda^{-1}C 
			\]
			for any $\lambda  > 1$, provided that $\max_{j \in [m]}\log d_j \ll n^c$ for some constant $c\in(0,1)$ depending only on $c_1$ and $c_2$ specified in Assumptions {\rm\ref{tail}} and {\rm\ref{mixing}}, which implies that $\Delta_{2,n,j}(0.4 \tilde C_*) = O_{\rm p}(1)$. 
$\hfill\Box$

\subsection{Proof of Lemma \ref{lemma: 4th order cross moment}}\label{sec:lem: 4th order}			
By Triangle inequality, to prove Lemma \ref{lemma: 4th order cross moment}, it suffices to show that
\begin{equation}\label{lem: 4th eq 1}
    \begin{split}
          &\frac{1}{n-1}\sum_{t=2}^n\{(\otimes_{m}^{j^\prime=1}\check\ab_{i,j^\prime}^{\MP}-\otimes_{m}^{j^\prime=1}\ab^{\MP}_{i,j^\prime})^\T \textup{vec}(\mathcal{E}_{t-1}) \}^2\{\bbeta_{i,j}(\hb)^{\T} \textup{vec}(\Eb_{t,j})\}^2\,\\
           &~~~~~~ =O_{\rm p} \Big\{ (\tilde L_n+1 ) \max_{j' \in [m]}| \check\ab_{i,j^\prime}^{\MP}- \ab^{\MP}_{i,j^\prime}|_2^2 \Big\}\,,
    \end{split}
\end{equation}
and 
\begin{equation}\label{lem: 4th eq 2}
    \begin{split}
            &\{(\otimes_{m}^{j^\prime=1}\check\ab_{i,j^\prime}^{\MP}-\otimes_{m}^{j^\prime=1}\ab^{\MP}_{i,j^\prime})^\T \textup{vec}(\bar{\mathcal{E}}) \}^2 \cdot \frac{1}{n-1}\sum_{t=2}^n\{\bbeta_{i,j}(\hb)^{\T} \textup{vec}(\Eb_{t,j})\}^2\,\\
            &~~~~~~=O_{\rm p} \Big(\tilde L_n^2 \max_{j' \in [m]}| \check\ab_{i,j^\prime}^{\MP}- \ab^{\MP}_{i,j^\prime}|_2^2 \Big)\,.
    \end{split}
\end{equation}
Notice that \eqref{lem: 4th eq 2} follows directly from Lemma \ref{lem: Omega n} on the event $\Xi_{8,n}(\tilde C)$, together with the facts that $L_n\lesssim \tilde L_n$ and $(n-1)^{-1}\sum_{t=2}^n\{\bbeta_{i,j}(\hb)^{\T} \textup{vec}(\Eb_{t,j})\}^2=O_{\rm p}(1)$. Therefore, we focus on the proof of \eqref{lem: 4th eq 1}.

To start, define  $\mathcal{Z}_{t}=   \textup{vec}(\mathcal{E}_{t-1}) \otimes \textup{vec}(\mathcal{E}_{t})$ for $t \in \{2,\ldots,n\}$. Then, by the similar arguments in the proof of Lemma \ref{lem: Omega n} for the event $\Xi_{9,n}(\tilde C)$, we have
\begin{align}
    &\max_{\substack{
    |\tilde\bbeta_{1,j}|_2= 1 = |\tilde\bbeta_{2,j}|_2,\,
    |\bbeta_{1,j}|_2= 1 = |\bbeta_{2,j}|_2,\\
    \tilde\bbeta_{1,j},\,\tilde\bbeta_{2,j},\,
    \bbeta_{1,j},\,\bbeta_{2,j}\in\mathbb{R}^{d_j},
    \ \forall j\in[m]}}
    \bigg|
    \frac{1}{n-1}\sum_{t=2}^{n}
    \{(\otimes^{j=1}_m \tilde\bbeta_{1,j})
    \otimes
    (\otimes^{j=1}_m \tilde\bbeta_{2,j})\}^{\T}
    \notag\\
    &\qquad\qquad \times
    \{
    \mathcal{Z}_{t}\mathcal{Z}_{t}^{\T}
    -
    \mathbb{E}(\mathcal{Z}_{t}\mathcal{Z}_{t}^{\T})
    \}
    \{
    (\otimes^{j=1}_m\bbeta_{1,j})
    \otimes
    (\otimes^{j=1}_m\bbeta_{2,j})
    \}
    \bigg|
    =
    O_{\rm p}(\tilde L_n)\,.
    \label{Omega 9n-revised}
\end{align}
Similarly, define $\mathcal{Z}_{t,j}= \textup{vec}(\mathcal{E}_{t-1}) \otimes \textup{vec}(\Eb_{t,j})$ for $t \in \{2,\ldots,n\}$. By \eqref{eq: kronecker product equality}, \eqref{Omega 9n-revised} and the one-by-one replacement argument for the Kronecker products,  it follows that
\[
\begin{split}
    &\bigg|
    \frac{1}{n-1}\sum_{t=2}^n
    \{(\otimes_{m}^{j^\prime=1}\check\ab_{i,j^\prime}^{\MP}
    -\otimes_{m}^{j^\prime=1}\ab^{\MP}_{i,j^\prime})
    \otimes \bbeta_{i,j}(\hb)\}^{\T}
    \{
    \mathcal{Z}_{t,j}\mathcal{Z}_{t,j}^{\T}
    -
    \mathbb{E}(\mathcal{Z}_{t,j}\mathcal{Z}_{t,j}^{\T})
    \} \\
    &~~~~~~\qquad\qquad
    \times
    \{(\otimes_{m}^{j^\prime=1}\check\ab_{i,j^\prime}^{\MP}
    -\otimes_{m}^{j^\prime=1}\ab^{\MP}_{i,j^\prime})
    \otimes \bbeta_{i,j}(\hb)\}
    \bigg| \\
    &~~~~~~ =
    O_{\rm p} \Big(
    \tilde L_n
    \max_{j' \in [m]}|
    \check\ab_{i,j^\prime}^{\MP}
    - \ab^{\MP}_{i,j^\prime}|_2^2
    \Big)\,.
\end{split}
\]
On the other hand, by Assumptions \ref{mixing} and \ref{cross}, we can conclude that $|(\bbeta\otimes\tilde\bbeta)^{\T}\mathbb{E}(\mathcal{Z}_{t,j}\mathcal{Z}_{t,j}^{\T})(\bbeta\otimes\tilde\bbeta)| \le C$ for some constant $C>0$ and any unit vectors $\bbeta,\tilde\bbeta\in\mathbb{R}^{D_n}$. Therefore, 
\[
\begin{split}
      &\bigg|
    \frac{1}{n-1}\sum_{t=2}^n
    \{(\otimes_{m}^{j^\prime=1}\check\ab_{i,j^\prime}^{\MP}
    -\otimes_{m}^{j^\prime=1}\ab^{\MP}_{i,j^\prime})
    \otimes \bbeta_{i,j}(\hb)\}^{\T}
    \{
    \mathbb{E}(\mathcal{Z}_{t,j}\mathcal{Z}_{t,j}^{\T})
    \} \\
    &~~~~~~\qquad\qquad
    \times
    \{(\otimes_{m}^{j^\prime=1}\check\ab_{i,j^\prime}^{\MP}
    -\otimes_{m}^{j^\prime=1}\ab^{\MP}_{i,j^\prime})
    \otimes \bbeta_{i,j}(\hb)\}
    \bigg| \\
    &~~~~~~ =
    O_{\rm p} \Big(
    \max_{j' \in [m]}|
    \check\ab_{i,j^\prime}^{\MP}
    - \ab^{\MP}_{i,j^\prime}|_2^2
    \Big)\,.
\end{split}
\]
Then, \eqref{lem: 4th eq 1} holds, which further implies Lemma \ref{lemma: 4th order cross moment}. $\hfill\Box$

\section{Relaxation of technical assumptions}\label{sec:relaxation assumptions}

\subsection{Serial dependence of the idiosyncratic error tensor}\label{sec: relax error serial dependence}
 
We assume that the idiosyncratic errors are serially uncorrelated in Assumption \ref{error}, which enables a direct separation of
the signal part and the noise part through the auto-covariances of the observed data.  In fact, our proposed procedures can still work if the idiosyncratic errors are serially correlated. For $\Ab_j$ and $\Bb_j$ specified in Section \ref{sec: model}, and $\Gb_{k,\xi}$ specified in Section \ref{sec: initial},    notice that $\bSigma_{\Yb_j,\xi}(k) = \Ab_j \Gb_{k,\xi} \Bb_j^{\T}$ and the corresponding representation of $\Kb_{1,2,j}$ in \eqref{Kbj} no longer hold when the idiosyncratic errors are serially correlated. Instead,
\[
\bSigma_{\Yb_j,\xi}(k)=\Ab_j\Gb_{k,\xi}\Bb_j^\T+\bSigma_{\Eb_j,\xi}(k) \,,
\]
where the additional term $\bSigma_{\Eb_j,\xi}(k)$ arises from the serial dependence in the idiosyncratic errors and is defined as
\[
\bSigma_{\Eb_j,\xi}(k)
=
\frac{1}{n-k}\sum_{t=k+1}^n
\mathbb{E} [\{\Eb_{t,j}-\mathbb{E}(\bar{\Eb}_j)\}\{\xi_{t-k}-\mathbb{E}(\bar{\xi})\} ] 
\]
with $\bar{\Eb}_j=n^{-1}\sum_{t=1}^n \Eb_{t,j}$. 
Write $\bSigma_{\Cb_j,\xi}(k)=\Ab_j\Gb_{k,\xi}\Bb_j^\T$. Then, $\ab_{i,j}$ and $\Qb_j$ can be identified in the same manner as in Section~\ref{sec: initial} based on \eqref{Kbj} and \eqref{tilde Mj}, with $\bSigma_{\Yb_j,\xi}(k)$ replaced by $\bSigma_{\Cb_j,\xi}(k)$.  Therefore, in order to ensure our procedures  still work in such case, the key step is to establish the relationship between the estimator $\tilde{\bSigma}_{k,j}$ specified in \eqref{hat Sigma kj} and $\bSigma_{\Cb_j,\xi}(k)$. 
For $s_j$ specified in Assumption \ref{sparsity}, if we further assume  
\[
\begin{split}
  \max_{k\in[K]}\max_{j\in[m]}\|&\bSigma_{\Eb_j,\xi}(k)\|_2  = O(1)\,, ~~  \max_{k\in[K]} \max_{q\in [d_{\mminus j}]}\sum_{p=1}^{d_j} I\{[\bSigma_{\Eb_j,\xi}(k)]_{p,q} \neq 0\}  \le s_j\,, \\
  &\textup{and}~~\max_{k\in[K]} \max_{p \in [d_j]}\sum_{q=1}^{d_{\mminus j}} I\{[\bSigma_{\Eb_j,\xi}(k)]_{p,q} \neq 0\}  \le \prod_{j' \neq j} s_{j'}  
\end{split}
\]
for each $j\in[m]$, then following the same strategy as in the proof of Lemma \ref{lem: l2 truncated covariance}, we can conclude
\[
\| T_{\delta_1}\{\tilde\bSigma_{\Yb_j,\xi}(k)\} - \tilde\bSigma_{\Cb_j,\xi}(k) - \bSigma_{\Eb_j,\xi}(k) \|_2 
= \ubar\sigma_{\xi}^2 \bar\sigma_{\xi}^{-1} \cdot O_{\rm p}(\Pi_n)\,,
\]
where $\tilde\bSigma_{\Cb_j,\xi}(k)$ defined in \eqref{sigma C}  is the sample estimate of $\bSigma_{\Cb_j,\xi}(k)$.
By Triangle inequality, it follows that
\[
\| T_{\delta_1}\{\tilde\bSigma_{\Yb_j,\xi}(k)\} - \tilde\bSigma_{\Cb_j,\xi}(k) \|_2 
\le \ubar\sigma_{\xi}^2 \bar\sigma_{\xi}^{-1} \cdot O_{\rm p}(\Pi_n) + \|\bSigma_{\Eb_j,\xi}(k)\|_2\,.
\]
 Further, following the proof of Theorem \ref{thm: aij}, if $\Pi_n + \ubar\sigma_{\xi}^{-2} \bar\sigma_{\xi} \max_{k\in[K]}\max_{j\in[m]}  \|\bSigma_{\Eb_j,\xi}(k)\|_2 \ll 1$, we can conclude that the one-pass estimator satisfies
\begin{equation}\label{rate serial correlation}
|\tilde\ab_{z_i,j} - \kappa_{i,j} \ab_{i,j}|_2 
= O_{\rm p}\bigg( \frac{1}{\ubar\sigma_{\xi}} \sqrt{\frac{S_n \log D_n}{n}} + \frac{1}{\ubar\sigma_{\xi}} \max_{k\in[K]}  \|\bSigma_{\Eb_j,\xi}(k)\|_2 \bigg)
\end{equation}
for $z_i$ and $\kappa_{i,j}$ specified in Theorem \ref{thm: aij}.  

  In comparison to Theorem \ref{thm: aij}, the additional error term $\ubar\sigma_{\xi}^{-1}\max_{k\in[K]}  \|\bSigma_{\Eb_j,\xi}(k)\|_2$ in \eqref{rate serial correlation} originates from the serial correlation of the noise. To ensure consistency, it is required that $\ubar\sigma_{\xi}\rightarrow\infty$ as $n\rightarrow \infty$, which is the cost of relaxing Assumption \ref{error}. Similar requirements also appear in \citeS{bai2003inferential-app} and \citeS{chen2026estimation-app}, which are used to  guarantee that the factor signal is strong enough relative to the idiosyncratic errors for consistent estimation. 

We further evaluate the robustness of the proposed methods through a simulation study with serially correlated idiosyncratic errors. Specifically, we modify the setting in Section~\ref{sec:numerical} by generating each entry of the error tensor $\mathcal{E}_t$ as an AR(1) process, where the autoregressive coefficient is independently drawn from a uniform distribution on $[-0.3,0.3]$. All other aspects of the data-generating process remain unchanged. Table~\ref{table:ARerror-rf-all} reports the finite-sample performance of the seven methods discussed in the paper (Pro.iter, HOPE, CC-ISO, Pro.init, cPCA, RP-PCA, and RCP) in estimating the factor loading vectors.  The results are similar to those for the uncorrelated error case in Table~\ref{table:rf-all}, indicating that the proposed methods remain effective and robust even in the presence of serially correlated errors.

\begin{table}[htbp]
\centering
\scriptsize
\renewcommand{\arraystretch}{1.15}
\setlength{\tabcolsep}{3pt}
\caption{
The averages and standard deviations (in parentheses) of the estimation errors \eqref{eq:estimation error} for different methods based on 2000 repetitions. The elements of the error tensor sequence $\{\mathcal{E}_t\}_{t=1}^n$ are generated independently as AR(1) processes, with the autoregressive coefficients independently drawn from the uniform distribution on $[-0.3,0.3]$. Bold numbers indicate the smallest average estimation error among all competing methods. All numbers reported below are multiplied by 100.}
\label{table:ARerror-rf-all}
\begin{tabular}{c|c|c|c|ccc|cccc}
\hline\hline
\multirow{2}{*}{\textbf{$\rho$}} & \multirow{2}{*}{\textbf{$\phi$}} & \multirow{2}{*}{\textbf{$s$}} & \multirow{2}{*}{\textbf{$n$}} & \multicolumn{3}{c|}{\textbf{Iterative estimates}} & \multicolumn{4}{c}{\textbf{One-pass estimates}} \\ \cline{5-11}
& & & & \textbf{Pro.iter} & \textbf{HOPE} & \textbf{CC-ISO} & \textbf{Pro.init} & \textbf{cPCA} & \textbf{RP-PCA} & \textbf{RCP} \\
\hline
\multirow{12}{*}{0}
& \multirow{6}{*}{0.25}
& \multirow{2}{*}{0}
& 400 & \textbf{0.24} (4.25) & 0.71 (7.81) & 0.46 (6.02) & 4.91 (9.50) & 17.14 (17.43) & 19.26 (18.06) & 33.54 (39.48) \\
& & & 800 & \textbf{0.12} (3.05) & 0.36 (5.35) & 0.20 (3.98) & 2.39 (5.98) & 14.17 (15.31) & 16.88 (16.35) & 29.11 (38.25) \\
\cline{3-11}
& & \multirow{2}{*}{0.3}
& 400 & 0.54 (6.81) & 0.93 (9.03) & \textbf{0.36} (5.51) & 4.27 (8.89) & 15.17 (17.58) & 16.31 (17.54) & 29.83 (38.74) \\
& & & 800 & \textbf{0.30} (5.09) & 0.40 (5.84) & 0.51 (6.58) & 2.33 (7.28) & 11.89 (15.13) & 13.62 (15.94) & 27.64 (37.86) \\
\cline{3-11}
& & \multirow{2}{*}{0.6}
& 400 & \textbf{0.45} (6.28) & 1.20 (10.29) & 0.61 (7.33) & 4.08 (8.56) & 13.73 (18.02) & 14.84 (17.88) & 29.84 (38.43) \\
& & & 800 & \textbf{0.25} (4.76) & 0.51 (6.74) & 0.39 (6.08) & 1.99 (6.44) & 9.90 (14.82) & 10.93 (15.06) & 26.51 (37.39) \\
\cline{2-11}
& \multirow{6}{*}{0.75}
& \multirow{2}{*}{0}
& 400 & 1.54 (7.50) & 2.04 (8.26) & \textbf{0.58} (4.07) & 14.12 (17.57) & 32.20 (10.53) & 33.60 (10.34) & 53.83 (30.78) \\
& & & 800 & \textbf{0.67} (5.05) & 1.39 (6.44) & 0.88 (5.24) & 5.56 (10.79) & 33.51 (10.36) & 34.45 (10.25) & 55.60 (31.43) \\
\cline{3-11}
& & \multirow{2}{*}{0.3}
& 400 & \textbf{0.30} (4.29) & 0.70 (6.80) & 0.40 (4.36) & 6.82 (11.94) & 25.68 (15.38) & 28.49 (14.90) & 41.30 (38.42) \\
& & & 800 & \textbf{0.49} (5.92) & 0.72 (6.78) & 0.55 (5.34) & 3.41 (9.56) & 26.37 (14.54) & 30.17 (14.14) & 43.44 (39.56) \\
\cline{3-11}
& & \multirow{2}{*}{0.6}
& 400 & \textbf{0.21} (3.84) & 0.76 (7.85) & 0.61 (6.66) & 4.82 (9.20) & 18.21 (18.17) & 19.64 (17.91) & 33.35 (39.74) \\
& & & 800 & 0.33 (5.31) & 0.56 (6.72) & \textbf{0.26} (4.49) & 2.14 (6.75) & 15.32 (16.33) & 17.37 (16.68) & 30.02 (38.78) \\
\hline
\multirow{12}{*}{0.75}
& \multirow{6}{*}{0.25}
& \multirow{2}{*}{0}
& 400 & \textbf{0.39} (2.77) & 25.26 (38.12) & 27.07 (38.60) & 9.32 (13.56) & 48.41 (15.04) & 49.86 (14.44) & 21.57 (24.77) \\
& & & 800 & \textbf{0.12} (0.05) & 23.76 (37.46) & 24.35 (37.37) & 4.83 (9.47) & 48.64 (14.24) & 48.93 (13.67) & 20.69 (24.64) \\
\cline{3-11}
& & \multirow{2}{*}{0.3}
& 400 & \textbf{0.38} (3.34) & 28.19 (40.08) & 29.48 (40.35) & 9.29 (14.25) & 49.79 (16.16) & 50.53 (15.04) & 22.35 (25.97) \\
& & & 800 & \textbf{0.10} (0.04) & 29.67 (40.58) & 29.58 (40.52) & 3.60 (6.87) & 50.24 (15.64) & 51.27 (15.05) & 19.05 (24.09) \\
\cline{3-11}
& & \multirow{2}{*}{0.6}
& 400 & \textbf{0.25} (2.20) & 29.79 (41.28) & 31.08 (41.52) & 8.14 (13.49) & 51.09 (17.95) & 52.24 (16.77) & 20.70 (24.84) \\
& & & 800 & \textbf{0.08} (0.03) & 32.03 (41.99) & 32.39 (41.93) & 3.17 (6.39) & 51.08 (17.27) & 52.14 (16.45) & 19.54 (24.63) \\
\cline{2-11}
& \multirow{6}{*}{0.75}
& \multirow{2}{*}{0}
& 400 & \textbf{4.64} (10.17) & 7.19 (15.62) & 31.17 (21.03) & 23.26 (18.73) & 38.28 (8.62) & 42.33 (9.13) & 27.89 (19.83) \\
& & & 800 & \textbf{0.38} (1.60) & 3.86 (14.17) & 19.08 (21.78) & 10.35 (13.64) & 38.58 (9.08) & 40.09 (8.32) & 23.66 (20.62) \\
\cline{3-11}
& & \multirow{2}{*}{0.3}
& 400 & \textbf{0.64} (4.01) & 13.82 (28.31) & 18.43 (30.24) & 13.64 (17.32) & 44.22 (10.80) & 45.89 (10.50) & 24.32 (24.55) \\
& & & 800 & \textbf{0.13} (0.05) & 13.83 (28.46) & 14.43 (28.68) & 5.54 (10.09) & 44.44 (10.86) & 45.25 (10.58) & 21.80 (24.51) \\
\cline{3-11}
& & \multirow{2}{*}{0.6}
& 400 & \textbf{0.26} (1.71) & 23.66 (37.58) & 24.71 (38.06) & 9.40 (14.06) & 49.49 (14.92) & 50.51 (14.37) & 22.87 (25.88) \\
& & & 800 & \textbf{0.11} (0.73) & 24.40 (37.77) & 25.90 (38.40) & 3.65 (7.72) & 50.45 (14.20) & 50.97 (13.91) & 20.61 (24.74) \\
\hline\hline
\end{tabular}
\end{table}

\subsection{Tail probability and mixing conditions}\label{sec: relax tail condition}

The exponentially decaying tail probabilities and $\alpha$-mixing coefficients assumed in Assumptions \ref{tail} and \ref{mixing} ensure exponential-type upper bounds for the tail probabilities of the statistics involved in the proofs, such as in \eqref{tail error}. In fact, the proposed procedure would still be valid if these conditions were relaxed to allow polynomially decaying tail probabilities and $\alpha$-mixing coefficients, by applying Fuk--Nagaev-type inequalities to construct appropriate upper bounds.

Specifically, assume that all tail probabilities in Assumption \ref{tail} are upper bounded by $O\{x^{-2(C_{1,*}+C_{2,*})}\}$ as $x\rightarrow\infty$, and that the $\alpha$-mixing coefficients in Assumption \ref{mixing} satisfy $\alpha(k) = O\{k^{-(C_{1,*}-1)(C_{1,*}+C_{2,*})/C_{2,*}}\}$ as $k\rightarrow\infty$  
for some constants $C_{1,*} > 2$ and $C_{2,*} > 0$. 
 Then, by Fuk--Nagaev-type inequalities (Lemma 4 in Appendix E of \citeS{chang2018principal-app}), the uniform bound in \eqref{sigma2} is modified to $O_{\rm p}(\max\{D_n^{1/C_{1,*}} n^{-(C_{1,*}-1)/C_{1,*}}, (n^{-1}\log D_n)^{1/2}\})$. Following the proof of Lemma \ref{lem: l2 truncated covariance}, we can similarly conclude that
\[
\| T_{\delta_1}\{\tilde\bSigma_{\Yb_j,\xi}(k)\} - \tilde\bSigma_{\Cb_j,\xi}(k) \|_2 = O_{\rm p}( n^{-1/2} S_n^{1/2} D_n^{1/C_{1,*}} )
\]
with thresholding level $\delta_1 = \tilde C_* n^{-1/2} D_n^{1/C_{1,*}}$ for some sufficiently large constant $\tilde C_* > 0$. Then, following the proof of Theorem \ref{thm: aij}, if $\bar\sigma_{\xi} \ubar{\sigma}_{\xi}^{-2} n^{-1/2} S_n^{1/2} D_n^{1/C_{1,*}} \ll 1$ and $\tilde r=r$, we can conclude that the one-pass estimator remains consistent with the convergence rate  
$$O_{\rm p}(\ubar\sigma_{\xi}^{-1}n^{-1/2}S_n^{1/2}D_n^{1/C_{1,*}})\,,$$
although this rate is slower than that reported in Theorem \ref{thm: aij}. 

We further provide a robustness check via simulation under relaxed tail conditions. Specifically, we replace the Gaussian error distribution in Section~\ref{sec:numerical} with a heavy-tailed $t(5)$ distribution, keeping all other aspects of the data-generating process unchanged. Table~\ref{table:heavytail-rf-all} reports the finite-sample performance of the seven methods discussed in the paper (Pro.iter, HOPE, CC-ISO, Pro.init, cPCA, RP-PCA, and RCP) in estimating the factor loading vectors. 
We can find that: (i) when the factors are uncorrelated $(\rho = 0)$, the proposed iterative estimator (Pro.iter) performs worse than CC-ISO and outperforms HOPE  when the sample size is small ($n = 400$), but Pro.iter works comparably with CC-ISO when the sample size is large ($n = 800$); (ii) when the factors are correlated $(\rho = 0.75)$, Pro.iter significantly outperforms HOPE and CC-ISO.
These findings are similar to that for the Gaussian error distribution case, which indicate that the proposed iterative estimator remains effective and robust even in the presence of heavy-tailed errors.

\begin{table}
\centering
\scriptsize
\renewcommand{\arraystretch}{1.15}
\setlength{\tabcolsep}{3pt}
\caption{
The averages and standard deviations (in parentheses) of the estimation errors \eqref{eq:estimation error} for different methods based on 2000 repetitions. The elements of the error term sequence $\{\mathcal{E}_t\}_{t = 1}^n$ are independently drawn from $t(5)$. Bold numbers indicate the smallest average estimation error among all competing methods. All numbers reported below are multiplied by 100.}
\label{table:heavytail-rf-all}
\begin{tabular}{c|c|c|c|ccc|cccc}
\hline\hline
\multirow{2}{*}{\textbf{$\rho$}} & \multirow{2}{*}{\textbf{$\phi$}} & \multirow{2}{*}{\textbf{$s$}} & \multirow{2}{*}{\textbf{$n$}} & \multicolumn{3}{c|}{\textbf{Iterative estimates}} & \multicolumn{4}{c}{\textbf{One-pass estimates}} \\
\cline{5-11}
& & & & \textbf{Pro.iter} & \textbf{HOPE} & \textbf{CC-ISO} & \textbf{Pro.init} & \textbf{cPCA} & \textbf{RP-PCA} & \textbf{RCP} \\
\hline
\multirow{12}{*}{0}
& \multirow{6}{*}{0.25}
& \multirow{2}{*}{0}
& 400 & 1.31 (10.62) & 1.85 (12.70) & \textbf{0.49} (6.09) & 9.04 (16.67) & 18.09 (18.90) & 19.19 (17.94) & 39.96 (40.61) \\
& & & 800 & 0.68 (7.76) & 0.77 (8.18) & \textbf{0.30} (4.49) & 3.66 (9.68) & 14.51 (16.12) & 16.98 (16.45) & 35.05 (40.24) \\
\cline{3-11}
& & \multirow{2}{*}{0.3}
& 400 & 1.45 (11.30) & 1.73 (12.38) & \textbf{0.44} (5.77) & 8.01 (15.85) & 15.93 (18.82) & 16.50 (17.81) & 37.92 (40.71) \\
& & & 800 & 0.48 (6.43) & 0.59 (7.11) & \textbf{0.41} (5.76) & 3.04 (8.26) & 11.98 (15.52) & 13.46 (15.60) & 35.50 (40.60) \\
\cline{3-11}
& & \multirow{2}{*}{0.6}
& 400 & 0.81 (8.48) & 1.52 (11.49) & \textbf{0.71} (7.85) & 6.67 (13.15) & 13.93 (18.40) & 14.86 (17.99) & 37.38 (40.58) \\
& & & 800 & \textbf{0.40} (5.97) & 0.58 (7.21) & 0.42 (6.18) & 2.69 (8.24) & 10.01 (15.09) & 10.97 (15.19) & 33.88 (40.28) \\
\cline{2-11}
& \multirow{6}{*}{0.75}
& \multirow{2}{*}{0}
& 400 & 3.92 (11.98) & 3.97 (11.70) & \textbf{0.53} (3.64) & 21.64 (20.78) & 32.53 (10.70) & 33.56 (10.31) & 60.03 (27.54) \\
& & & 800 & 1.25 (6.39) & 2.10 (7.95) & \textbf{0.82} (4.87) & 8.63 (13.68) & 33.53 (10.38) & 34.45 (10.27) & 61.59 (27.73) \\
\cline{3-11}
& & \multirow{2}{*}{0.3}
& 400 & 1.74 (11.36) & 2.12 (12.24) & \textbf{0.44} (4.78) & 13.03 (19.49) & 26.55 (16.58) & 28.50 (14.89) & 49.15 (38.03) \\
& & & 800 & 0.56 (6.20) & 0.81 (7.04) & \textbf{0.46} (4.84) & 5.03 (11.73) & 26.31 (14.62) & 30.15 (14.10) & 50.51 (38.99) \\
\cline{3-11}
& & \multirow{2}{*}{0.6}
& 400 & 1.14 (9.73) & 1.58 (11.49) & \textbf{0.57} (6.38) & 8.94 (16.70) & 18.81 (19.07) & 19.74 (18.03) & 41.22 (40.95) \\
& & & 800 & 0.56 (6.97) & 0.94 (8.82) & \textbf{0.30} (4.72) & 3.00 (8.48) & 15.54 (16.88) & 17.31 (16.52) & 35.73 (39.79) \\
\hline
\multirow{12}{*}{0.75}
& \multirow{6}{*}{0.25}
& \multirow{2}{*}{0}
& 400 & \textbf{0.78} (5.13) & 24.79 (37.73) & 30.66 (39.62) & 16.07 (20.65) & 48.49 (15.15) & 51.61 (15.61) & 25.77 (26.27) \\
& & & 800 & \textbf{0.19} (0.07) & 23.31 (37.21) & 24.50 (37.41) & 6.45 (10.28) & 48.39 (14.23) & 48.94 (13.61) & 25.27 (26.40) \\
\cline{3-11}
& & \multirow{2}{*}{0.3}
& 400 & \textbf{0.79} (6.50) & 28.86 (40.29) & 30.64 (40.46) & 14.06 (18.89) & 50.00 (16.37) & 51.32 (15.64) & 26.30 (26.21) \\
& & & 800 & \textbf{0.15} (0.06) & 29.55 (40.64) & 29.73 (40.68) & 5.24 (8.48) & 50.36 (15.70) & 51.59 (15.27) & 23.04 (25.91) \\
\cline{3-11}
& & \multirow{2}{*}{0.6}
& 400 & \textbf{1.11} (8.74) & 31.79 (41.81) & 32.35 (41.86) & 12.81 (18.59) & 51.48 (18.26) & 52.99 (17.29) & 26.14 (26.78) \\
& & & 800 & \textbf{0.13} (0.05) & 32.92 (42.27) & 32.73 (42.10) & 4.37 (7.59) & 50.91 (17.40) & 52.03 (16.47) & 23.69 (26.03) \\
\cline{2-11}
& \multirow{6}{*}{0.75}
& \multirow{2}{*}{0}
& 400 & \textbf{9.02} (14.10) & 10.88 (17.35) & 43.91 (14.20) & 30.35 (20.31) & 38.69 (8.66) & 45.38 (10.68) & 32.89 (19.19) \\
& & & 800 & \textbf{0.61} (2.44) & 4.14 (14.57) & 37.18 (18.15) & 15.50 (16.43) & 38.59 (9.03) & 43.44 (9.17) & 27.48 (20.54) \\
\cline{3-11}
& & \multirow{2}{*}{0.3}
& 400 & \textbf{1.44} (6.34) & 14.50 (28.43) & 34.84 (31.02) & 21.61 (22.41) & 44.37 (10.82) & 50.81 (11.36) & 28.76 (25.34) \\
& & & 800 & \textbf{0.22} (0.68) & 13.65 (28.19) & 19.57 (30.20) & 7.91 (12.17) & 44.48 (10.82) & 46.48 (10.57) & 25.45 (24.94) \\
\cline{3-11}
& & \multirow{2}{*}{0.6}
& 400 & \textbf{0.60} (4.48) & 22.85 (37.12) & 27.69 (38.59) & 15.21 (19.89) & 49.61 (15.12) & 52.08 (14.86) & 26.54 (26.65) \\
& & & 800 & \textbf{0.15} (0.07) & 24.68 (37.93) & 26.93 (38.62) & 5.35 (9.46) & 50.43 (14.22) & 51.11 (13.86) & 24.35 (25.67) \\
\hline\hline
\end{tabular}
\end{table}


\bibliographystyleS{jasa}
 
 \begingroup
\setlength{\bibsep}{2pt}      
\linespread{0.9}\selectfont 
\bibliographyS{Ref-app-abbreviation}
\endgroup

\end{document}